\documentclass[a4paper,10pt]{article}
\usepackage{graphicx} 
\usepackage{amsmath}
\usepackage[latin1]{inputenc}
\usepackage{amssymb}
\usepackage{amsthm}
\usepackage{mathrsfs}
\usepackage{bbm}
\usepackage{enumerate}
\usepackage{xcolor}
\usepackage{mathtools,leftidx}
\usepackage{upgreek}
\usepackage{bigints}
\usepackage{enumitem}

\usepackage{enumitem}
\usepackage{comment}
\usepackage[backend=biber,maxbibnames=99]{biblatex}
\usepackage{appendix}
\usepackage[affil-it]{authblk}
\usepackage{slashed}
\usepackage[hidelinks]{hyperref}
\usepackage[margin=3cm]{geometry}
\usepackage{verbatim}
\usepackage{fullpage}

\newtheorem{thm}{Theorem}
\newtheorem{theorem*}{Theorem}
\newtheorem{lemma}[thm]{Lemma}
\newtheorem{proposition}[thm]{Proposition}
\newtheorem{corollary}{Corollary}
\newtheorem{defin}{Definition}
\theoremstyle{definition}

\newtheorem{remark}{Remark}

\numberwithin{corollary}{section}
\numberwithin{remark}{section}
\numberwithin{defin}{section}
\newtheorem*{notation*}{Notation}
\numberwithin{equation}{subsection}
\numberwithin{thm}{section}

\DeclareMathOperator{\arcsec}{arcsec}

\DeclareMathOperator{\arcsech}{arcsech}

\newcommand{\dd}{\mathrm{d}}
\newcommand{\ddu}{\dd\bar{u}}
\newcommand{\ddv}{\dd\bar{v}}

\newcommand{\Scrim}{\mathscr{I}^-}
\newcommand{\Scrip}{\mathscr{I}^+}

\newcommand{\hp}{\mathscr{H}^+}
\newcommand{\hm}{\mathscr{H}^-}

\newcommand{\isp}{\raisebox{\depth}{\scalebox{1}[-1]{$\psi$}}}

\newcommand{\kspacehp}{{}^{(\ell)}\mathcal{K}^{N,n}_{\hp}}
\newcommand{\kspacehpbar}{{}^{(\ell)}\mathcal{K}^{N,n}_{\overline{\hp}}}

\newcommand{\kspacehpzero}{{}^{(\ell)}\mathcal{K}^{N}_{\hp}}
\newcommand{\kspacehpbarzero}{{}^{(\ell)}\mathcal{K}^{N}_{\overline{\hp}}}

\newcommand{\cspacehp}{{}^{(\ell)}\mathcal{C}^{N,n}_{\hp}}
\newcommand{\cspacehpbar}{{}^{(\ell)}\mathcal{C}^{N,n}_{\overline{\hp}}}

\newcommand{\sinespace}{{}^{(\ell)}\mathcal{E}_{\Sigma}^{N,\sin}}
\newcommand{\cosspace}{{}^{(\ell)}\mathcal{E}_{\Sigma}^{N,\cos}}

\newcommand{\sinespacen}{{}^{(\ell)}\mathcal{E}_{\Sigma}^{N,\sin,n}}
\newcommand{\cosspacen}{{}^{(\ell)}\mathcal{E}_{\Sigma}^{N,\cos,n}}

\newcommand{\cosspacebar}{{}^{(\ell)}\mathcal{E}_{\overline{\Sigma}}^{N,\cos}}

\newcommand{\sinespacebarn}{{}^{(\ell)}\mathcal{E}_{\overline{\Sigma}}^{N,\sin,n}}
\newcommand{\cosspacebarn}{{}^{(\ell)}\mathcal{E}_{\overline{\Sigma}}^{N,\cos,n}}

\newcommand{\cosmapforfutn}{{}^{(\ell)}\mathscr{F}_{+,\hp}^{N,\cos,n}}
\newcommand{\sinmapforfutn}{{}^{(\ell)}\mathscr{F}_{+,\hp}^{N,\sin,n}}

\newcommand{\cosmapforfutbarn}{{}^{(\ell)}\mathscr{F}_{+,\overline{\hp}}^{N,\cos,n}}

\newcommand{\cosmapforbacn}{{}^{(\ell)}\mathscr{B}_{+,\hp}^{N,\cos,n}}
\newcommand{\sinmapforbacn}{{}^{(\ell)}\mathscr{B}_{+,\hp}^{N,\sin,n}}

\newcommand{\cosmapforbacscripn}{{}^{(\ell)}\mathscr{B}_{+,\Scrip}^{N,\cos,n}}

\newcommand{\cosspacehp}{{}^{(\ell)}\mathcal{E}_{\hp}^{N,\cos}}
\newcommand{\sinespacehp}{{}^{(\ell)}\mathcal{E}_{\hp}^{N,\sin}}

\newcommand{\cosspacehpbar}{{}^{(\ell)}\mathcal{E}_{\overline{\hp}}^{N,\cos}}

\newcommand{\cosspacehpn}{{}^{(\ell)}\mathcal{E}_{\hp}^{N,\cos,n}}
\newcommand{\sinespacehpn}{{}^{(\ell)}\mathcal{E}_{\hp}^{N,\sin,n}}

\newcommand{\cosspacehpbarn}{{}^{(\ell)}\mathcal{E}_{\overline{\hp}}^{N,\cos,n}}

\newcommand{\sinespacehpbarn}{{}^{(\ell)}\mathcal{E}_{\overline{\hp}}^{N,\sin,n}}

\newcommand{\spacehp}{{}^{(\ell)}\mathcal{E}_{\hp}^{N}}

\newcommand{\spacehpbar}{{}^{(\ell)}\mathcal{E}_{\overline{\hp}}^{N}}

\newcommand{\spacehpn}{{}^{(\ell)}\mathcal{E}_{\hp}^{N,n}}

\newcommand{\spacehpbarn}{{}^{(\ell)}\mathcal{E}_{\overline{\hp}}^{N,n}}

\newcommand{\cosspacescrip}{{}^{(\ell)}\mathcal{E}_{\Scrip}^{N,\cos}}
\newcommand{\sinespacescrip}{{}^{(\ell)}\mathcal{E}_{\Scrip}^{N,\sin}}

\newcommand{\cosspacescripo}{{}^{(\ell)}\mathring{\mathcal{E}}_{\Scrip}^{N,\cos}}

\newcommand{\cosspacescripn}{{}^{(\ell)}\mathcal{E}_{\Scrip}^{N,\cos,n}}
\newcommand{\sinespacescripn}{{}^{(\ell)}\mathcal{E}_{\Scrip}^{N,\sin,n}}

\newcommand{\cosspacescripon}{{}^{(\ell)}\mathring{\mathcal{E}}_{\Scrip}^{N,\cos,n}}

\newcommand{\sinespacescripon}{{}^{(\ell)}\mathring{\mathcal{E}}_{\Scrip}^{N,\sin,n}}

\newcommand{\fancyapcos}{\mathfrak{a}_{\hp}^{\cos}}

\newcommand{\fancyapsin}{\mathfrak{a}_{\hp}^{\sin}}

\newcommand{\epsnh}{\epsilon_{{\mathrm{NH}}}}
\newcommand{\deltnh}{\delta_{{\mathrm{NH}}}}

\newcommand{\Kinfty}{\mathrm{K}_{|\omega_R|\gg1}}
\newcommand{\Kinftyi}{\mathcal{K}_{|\omega_R|\gg1}}

\newcommand{\spacescriexp}[1]{{}^{(#1)}\mathcal{E}_{\Scrip}^{\mathrm{exp}}}

\newcommand{\spacescriexpn}[1]{{}^{(#1)}\mathcal{E}_{\Scrip}^{\mathrm{exp},n}}

\newcommand{\spacescriexpon}[1]{{}^{(#1)}\mathring{\mathcal{E}}_{\Scrip}^{\mathrm{exp},n}}

\newcommand{\pu}{\partial_u}
\newcommand{\pv}{\partial_v}

\newcommand{\lapo}{\mathring{\slashed{\Delta}}}

\newcommand{\hor}{U_{\mathrm{hor}}}
\newcommand{\out}{U_{\mathrm{inf}}}
\newcommand{\horb}{\overline{U}_{\mathrm{hor}}}
\newcommand{\outb}{\overline{U}_{\mathrm{inf}}}

\newcommand{\ghorb}{\overline{G}_{\mathrm{hor}}}

\newcommand{\trr}{\tilde{R}}
\newcommand{\ttt}{\tilde{T}}

\newtheoremstyle{named}{}{}{\itshape}{}{\bfseries}{}{.5em}{\thmnote{#3}}
\theoremstyle{named}
\newtheorem*{namedtheorem}{Theorem}

\makeatletter
\patchcmd{\@maketitle}{\LARGE \@title}{\fontsize{16}{19.2}\selectfont\@title}{}{}
\makeatother

\title{On the image of the scattering map for horizon-regular solutions of the linear scalar wave equation\\ on the Schwarzschild black hole exterior}

\author{Hamed Masaood\thanks{hamed.masaood@princeton.edu}}
\affil{Princeton Gravity Initiative, Princeton University\\\mbox{Washington Road, Princeton NJ 08544, United States of America}}
\begin{document}

\maketitle
\begin{abstract}
    We construct Hilbert-space isomorphisms identifying the spaces of radiation fields on the event horizon and null infinity that are induced by the forward scattering map on the Schwarzschild exterior for fixed spherical harmonic mode solutions to the linear scalar wave equation arising from Cauchy data sets at $t=0$ which are regular at the event horizon, in the sense that the local energy at the event horizon is finite. We show that fixed spherical harmonic mode solutions with no radiation on the event horizon must decay exponentially in time, at a rate that is determined by the surface gravity. On the other hand, we construct examples of fixed spherical harmonic mode solutions with no radiation on null infinity and which have polynomial decay along the event horizon. Finally, we construct examples of polynomially decaying solutions to the linear scalar wave equation which are regular at the event horizon, have unbounded support in spherical harmonic modes, and induce no radiation on the event horizon.
\end{abstract}
\tableofcontents

\section{Introduction and detailed overview}
\subsection{Statement of the problem and preliminary results}\label{Section 1.1}
The cornerstone of scattering theory for the linear scalar wave equation 
\begin{align}\label{full wave equation}
    \Box_{g}\phi=0,
\end{align}
on the exterior of the Schwarzschild family of black hole solutions to the Einstein equations,
is the existence of a global time translation isometry, which leads to energy spaces of scattering data that define via energy conservation Hilbert-space isomorphisms, giving rise to a scattering map (See for instance \cite{BachelotWaveOperator}, \cite{DRSR14}, \cite{DimockKayI}, \cite{Nicolas}). The conservation law associated with time-translation invariance is however not compatible with the energy dissipation near the event horizon arising due to the redshift effect. Indeed, redshift becomes a blueshift effect when solutions are evolved backwards from scattering data, and the conflict between this blueshift and time-translation invariance is resolved by the conserved energy losing information near the event horizon (see \cite{DR13}, \cite{DRSR14}, \cite{DSR17} for a detailed discussion). Thus, a scattering solution arising via $\partial_t$-generated energy conservation is generically singular at the event horizon, in the sense that the transverse null derivative is not locally square integrable at the event horizon, leading to a \emph{degenerate} scattering theory. Concretely, in the Schwarzschild spacetime, and using the Eddington--Finkelstein double null coordinates $(u,v,\theta^A)$, solutions that arise via the scattering theory associated with the time translation isometry generically fail to satisfy
\begin{align}\label{regularity assumption}
    \int_{u_0}^\infty\int_{S^2} du\,d\mathrm{Vol}_{S^2}\,\frac{1}{\Omega^2}|\partial_u\phi({u},v)|^2<\infty,
\end{align}
for any finite $u$, $v$, where $\Omega^2=1-\frac{2M}{r}$.
Explicit examples were constructed in \cite{DRSR14} and \cite{DSR17}, where the following statements was proved:
\begin{theorem*}[Dafermos, Rodnianski, Shlapentokh-Rothman 2018]\label{Theorem A DRSR}
    For any \mbox{$p>2$}, the solution $\phi$ to \eqref{full wave equation} such that $\phi$, $\partial_v\phi$ extends continuously to $(v+1)^{-p}$, $-p(v+1)^{-p-1}$ respectively at $\hp\cap\{v\geq0\}$, and with $\lim_{v\longrightarrow\infty}\partial_ur\phi(u,v)=0$, violates \eqref{regularity assumption}.
\end{theorem*}
\begin{theorem*}[Dafermos, Shlapentokh-Rothman 2017]\label{Theorem B DSR}
    For $p$ sufficiently large, there exist $\psi_{\Scrip}$ with  $\partial_u\psi_{\Scrip}\in L^2(\Scrip)$ satisfying $|\psi_{\Scrip}|\sim C(1+|u|)^{-p}$ as $u\longrightarrow\infty$, such that the solution $\phi$ to \eqref{full wave equation} satisfying $\lim_{v\longrightarrow\infty}r\partial_u\phi(u,v)=\partial_u\psi_{\Scrip}(u)$, $\phi|_{\hp}=0$, violates \eqref{regularity assumption}.
\end{theorem*}
The theorems above say that the space of scattering states given by the $\partial_t$-energy conservation law is much too large to describe the image of a non-degenerate forward scattering map. We recall that in Schwarzschild, we take as the space of future scattering data the space $\mathcal{E}^T_+:=\mathcal{E}^T_{\hp}\oplus \mathcal{E}^T_{\Scrip}$, where each of $\mathcal{E}^T_{\hp}$, $\mathcal{E}^T_{\Scrip}$ is the completion of the space of smooth, compactly supported functions on $\hp$, $\Scrip$, respectively, under the norms
\begin{align}
    \|f\|_{\mathcal{E}^T_{\hp}}^2&=\int_{\hp}dv\,d\mathrm{Vol}_{S^2}\,|\partial_v f|^2,\\
    \|f\|_{\mathcal{E}^T_{\Scrip}}^2&=\int_{\Scrip}du\,d\mathrm{Vol}_{S^2}\,|\partial_u f|^2.
\end{align}
Then the conservation of $\partial_t$-energy implies the existence of the unitary Hilbert space isomorphisms
\begin{align}
    \mathscr{F}^T_+:\mathcal{E}^T_{\Sigma}\longrightarrow \mathcal{E}^T_+,\qquad \mathscr{B}^T_+:\mathcal{E}^T_+\longrightarrow\mathcal{E}^T_{\Sigma},
\end{align}
where for $\Sigma:=\{t=0\}$, the space $\mathcal{E}^T_{\Sigma}$ is the completion of the space of smooth, compactly supported Cauchy data for \eqref{full wave equation} under the norm
\begin{align}
\|(\uppsi,\uppsi')\|_{\mathcal{E}_{\Sigma}^T}^2=\int_{\Sigma}d\mathrm{Vol}_{\Sigma}\, T_{\mu\nu}[\phi]\partial_t^{\mu}n_{\Sigma^*}^\nu.
\end{align}
and $\mathscr{F}^T_+$ and $\mathscr{B}^T_+$ are inverse to one another. By the staticity of the Schwarzschild spacetime, analogous statements hold for evolution towards the past of $\Sigma$, leading to the maps 
\begin{align}
    \mathscr{F}^T_-:\mathcal{E}^T_{\Sigma}\longrightarrow \mathcal{E}^T_-,\qquad \mathscr{B}^T_-:\mathcal{E}^T_-\longrightarrow\mathcal{E}^T_{\Sigma},
\end{align}
where $\mathcal{E}_-^T:=\mathcal{E}^T_{\hm}\oplus \mathcal{E}^T_{\Scrim}$, and $\mathcal{E}^T_{\hm}$, $\mathcal{E}^T_{\Scrim}$ are defined by the completion of smooth, compactly supported functions under the norms 
\begin{align}
    \|f\|_{\mathcal{E}^T_{\hm}}^2&=\int_{\hm}du\,d\mathrm{Vol}_{S^2}\,|\partial_u f|^2,\\
    \|f\|_{\mathcal{E}^T_{\Scrim}}^2&=\int_{\Scrim}dv\,d\mathrm{Vol}_{S^2}\,|\partial_v f|^2.
\end{align}
Analogous statements to the above hold when $\Sigma$ is replaced by $\overline{\Sigma}$ and each of $\mathscr{H}^{\pm}$ is replaced by $\overline{\mathscr{H}^{\pm}}$, where $\overline{\Sigma}$, $\overline{\mathscr{H}^{\pm}}$ are the extensions of $\Sigma$, $\mathscr{H}^{\pm}$ to include the bifurcation sphere.

On the other hand, neither the image, nor the inverse of the non-degenerate forwards scattering map $\mathscr{F}^N_+$, defined to be the restriction of $\mathscr{F}^T_+$ to $\mathcal{E}_{\Sigma}^N\subset \mathcal{E}^T_{\Sigma}$ or $\mathcal{E}_{\overline{\Sigma}}^N\subset \mathcal{E}^T_{\overline{\Sigma}}$, with $\mathcal{E}^N_{\Sigma}$, $\mathcal{E}^N_{\overline{\Sigma}}$ being the closure of smooth, compactly supported Cauchy data for \eqref{full wave equation} on $\Sigma$, $\overline{\Sigma}$ respectively, under the norms
\begin{align}
\begin{split}
    \|(\uppsi,\uppsi')\|_{\mathcal{E}^N_{\Sigma}}^2:&=\int_{\Sigma}d\mathrm{Vol}_{\Sigma}\, T_{\mu\nu}[\phi](\partial_t+e^u\partial_u+e^{-v}\partial_v)^{\mu}n_{\Sigma}^\nu\\&\simeq \|\uppsi\|_{H^1(\Sigma,d\mathrm{Vol}_{\Sigma})}^2+\|\uppsi'\|_{L^2(\Sigma,d\mathrm{Vol}_{\Sigma})}^2,
\end{split}
\end{align}
\begin{align}
\begin{split}
    \|(\uppsi,\uppsi')\|_{\mathcal{E}^N_{\overline{\Sigma}}}^2:&=\int_{\overline{\Sigma}}d\mathrm{Vol}_{\overline{\Sigma}}\, T_{\mu\nu}[\phi](\partial_t+e^u\partial_u+e^{-v}\partial_v)^{\mu}n_{\Sigma}^\nu\\&\simeq \|\uppsi\|_{H^1(\overline{\Sigma},d\mathrm{Vol}_{\overline{\Sigma}})}^2+\|\uppsi'\|_{L^2(\overline{\Sigma},d\mathrm{Vol}_{\overline{\Sigma}})}^2
\end{split}
\end{align}
is known. In fact, the backwards scattering problem for horizon-regular solutions has only been solved in the case of scattering data that decay exponentially in the sense that
\begin{align}\label{state of the art}
    e^{\frac{v}{4M}}\partial_v\psi_{\hp}\in L^2(\hp),\qquad e^{\frac{u}{4M}}\partial_u\psi_{\Scrip}\in  L^2(\Scrip),
\end{align}
while the the backwards scattering problem for horizon-regular solutions arising from polynomially decaying scattering data, or even when scattering data decay at a rate slower than in \eqref{state of the art}, remains an open problem. This is a substantial shortcoming, given that solutions constructed in the forward scattering problem can have decay in time that is arbitrarily slow, and that even in the case of smooth, compactly supported Cauchy data, solutions generically attain a sharp polynomial decay rate in time. See for example \cite{AAG18b}, \cite{AAG18}, \cite{AAG19}, \cite{DDS11}, \cite{Hintz}.

Motivated by Theorems \ref{Theorem A DRSR} and \ref{Theorem B DSR}, we consider in this paper the following problems on the Schwarzschild exterior:

\begin{namedtheorem}[Existence:] Given $\psi_{\Scrip}\in \mathcal{E}^T_{\Scrip}$, does there exist $\psi_{\hp}$ in $\mathcal{E}^T_{\hp}$ or $\mathcal{E}^T_{\overline{\hp}}$ such that the solution $\phi$ to \eqref{full wave equation} arising from $(\psi_{\hp},\psi_{\Scrip})$ satisfies \eqref{regularity assumption}? Likewise, given $\psi_{\hp}$ in $\mathcal{E}^T_{\hp}$ or $\mathcal{E}^T_{\overline{\hp}}$, does there exist $\psi_{\Scrip}\in \mathcal{E}^T_{\Scrip}$ such that the solution $\phi$ satisfies \eqref{regularity assumption}?
\end{namedtheorem}
\begin{namedtheorem}[Rigidity of asymptotics:]
    Assume that $\phi^{(1)}$, $\phi^{(2)}$ arising from scattering data $(\psi_{\hp},\psi_{\Scrip}^{(1)})$ and $(\psi_{\hp},\psi_{\Scrip}^{(2)})$, respectively, have finite local energy at the horizon. What can we say about the asymptotic behaviour of radiation field at $\Scrip$, $\psi_{\Scrip}=\psi_{\Scrip}^{(1)}-\psi_{\Scrip}^{(2)}$ of the difference $\phi=\phi^{(1)}-\phi^{(2)}$? Similarly, if the solutions arise from data $(\psi_{\hp}^{(1)},\psi_{\Scrip})$, $(\psi_{\hp}^{(2)},\psi_{\Scrip})$, then what constraints apply to $\psi_{\hp}=\psi_{\hp}^{(1)}-\psi_{\hp}^{(2)}$? 
\end{namedtheorem}


To start with, we can formulate a preliminary conjecture on the question of rigidity of asymptotics. Using \eqref{full wave equation} and taking $\psi=r\phi$, we have
\begin{align}\label{wave equation psi intro}
    \partial_u\partial_v\psi-\frac{\Omega^2}{r^2}\lapo\psi+2M\frac{\Omega^2}{r^3}\psi=0,
\end{align}
where $\lapo$ is the Laplacian on the standard unit sphere $S^2$. We find that the transverse null derivative near $\hp$, $\Omega^{-2}\partial_u\psi$, satisfies
\begin{align}\label{sketchy}
    \partial_v\left(\frac{1}{\Omega^2}\partial_u\psi\right)+\frac{1}{2M}\left(\frac{1}{\Omega^2}\partial_u\psi\right)\sim\frac{1}{(2M)^2}(\lapo-1)\psi.
\end{align}
If $\psi$ and its derivatives decay like $e^{-\frac{t}{4M}}$ in $L^2_t$, we are then guaranteed that \eqref{regularity assumption} is satisfied. Combined with the examples of Theorems \ref{Theorem A DRSR} and \ref{Theorem B DSR}, we may be led to make the following conjectures:
\begin{enumerate}[label=\Roman*]
    \item \emph{Rigidity of asymptotics at $\Scrip$}: If $\psi_{\hp}=0$, then \eqref{regularity assumption} implies $e^{\frac{u}{4M}}\partial_u\psi_{\Scrip}\in L^2(\Scrip)$.\label{Statement I}
    \item \emph{Rigidity of asymptotics at $\hp$: }If $\psi_{\Scrip}=0$, then \eqref{regularity assumption} implies $e^{\frac{v}{4M}}\partial_v\psi_{\hp}\in L^2(\hp)$.\label{Statement II}
\end{enumerate}

The results of this paper however suggest that the picture is more nuanced than this heuristic picture given in statements \ref{Statement I} and \ref{Statement II} above. Consider the decomposition of a solution $\psi$ of \eqref{full wave equation} into spherical harmonic modes, 
\begin{align}
\psi=\sum_{\ell=0}^{\infty}\sum_{m\in\mathbb{N},m\in[-\ell,\ell]}\psi_{\ell,m}(u,v)Y^m_{\ell}(\theta,\phi),    
\end{align}
where each $\psi_{\ell,m}$ satisfies
\begin{align}\label{fixed ell mode wave equation}
    \pu\pv\psi_{\ell}+\ell(\ell+1)\frac{\Omega^2}{r^2}\psi_{\ell}+2M\frac{\Omega^2}{r^3}\psi_{\ell}=0.
\end{align}
In the remainder of this paper, we suppress dependence on the azimuthal number $m$. We show in this paper that
\begin{itemize}
    \item When restricting attention to solutions of fixed spherical harmonic mode $\ell$, rigidity of asymptotics at $\Scrip$ holds, i.e.~statement \ref{Statement I} above is \textbf{true} for solutions supported on a fixed (or finitely many) $\ell$-modes.
    \item On the other hand, there exist solutions to \eqref{full wave equation}, with support in $\ell$ unbounded, satisfying \eqref{regularity assumption} which have vanishing radiation at $\hp$ and which decay only \emph{polynomially} at $\Scrip$. Thus rigidity of asymptotics at $\Scrip$ does not hold, and statement \ref{Statement I} above is \textbf{false} in general.
    \item For any $\ell$, there exist solutions to \eqref{fixed ell mode wave equation} satisfying \eqref{regularity assumption} which have vanishing radiation at $\Scrip$ and polynomially decaying radiation at $\hp$. Thus statement \ref{Statement II} is \textbf{false} in general, even when restricted to fixed $\ell$.
\end{itemize}

We then utilise the fact that statement \ref{Statement I} holds for fixed $\ell$ to construct non-degenerate scattering theories which demonstrate that for fixed $\ell$, the radiation field at the horizon determines the radiation field at null infinity, up to a correction which decays towards the future like $e^{-\frac{u}{4M}}$, and which we can {freely prescribe} as a component of the scattering data to find a unique solution attaining the given scattering data at the horizon and null infinity and is regular at the horizon, leading to a well-posed scattering theory. We then obtain Hilbert space isomorphisms that give a complete description of the space of radiation fields at the event horizon and null infinity for solutions arising from data in $\mathcal{E}^N_{\Sigma}$ and data in $\mathcal{E}^N_{\overline{\Sigma}}$. In particular, the scattering theories we construct in this paper provide a complete answer to the question of existence in the case of fixed spherical harmonic mode solutions to \eqref{wave equation}. 

\subsection{High frequency obstructions to rigidity of asymptotics}
The relation \eqref{sketchy} can be used to show that a solution is regular at $\hp$, in the sense that \eqref{regularity assumption} is satisfied, if and only if
\begin{align}\label{regularity at H- of H+}
    \int_{u}^{\infty}\int_{S^2}d\bar{u}\,d\mathrm{Vol}_{S^2}\, e^{\frac{u}{2M}}|\partial_u\psi_{\hm}|^2<\infty.
\end{align}
In frequency space, the constraint \eqref{regularity at H- of H+} is equivalent by Paley--Wiener theory to the requirement that $a_{\hm}$, the Fourier transform of $\psi_{\hm}$, is the non-tangential limit of a holomorphic function in a strip below the real axis:
\begin{align}\label{fundamental strip}
    \left\{\omega: \Im\omega\in\left(-\frac{1}{4M},0\right)\right\},
\end{align}
which is square integrable on horizontal lines in the strip \eqref{fundamental strip}. A precise description of the contributions to $\psi_{\hm}$ from radiation at $\Scrip$ and at $\hp$ can then be obtained from the analytic properties of the scattering matrix elements, i.e.~the reflection and transmission coefficients, and their asymptotic behaviour at high frequency. To this end, the bulk of this work is devoted to the identification of the precise leading order behaviour of reflection and transmission coefficients at high angular momenta, and high, complex frequencies. Refer to Sections \ref{section: Jost solutions} and \ref{section: definition of reflection and transmission coefficients} for the definition of the reflection and transmission coefficients.

In this language, we can already describe the obstruction to rigidity of asymptotics at either of $\hp$, $\Scrip$. In the case of $\hp$, the obstruction is the exponential decay of the reflection coefficient associated with the horizon, which we denote by $\tilde{R}$, as $|\Re\omega|\longrightarrow\infty$ for fixed $\Im\omega$ in the strip $\{\Im\omega\in[0,\frac{1}{4M}]\}$, and we find that for any fixed $\ell$,
\begin{align}\label{high frequency behaviour of reflection coefficient intro}
    |{\trr}(\omega,\ell)|\lesssim e^{-4M(1-\delta)\pi |\Re\omega|}.
\end{align}
for any $\delta>0$. For a solution to \eqref{fixed ell mode wave equation} with vanishing radiation at $\Scrip$, the radiation fields at $\hp$ and $\hm$ are related in frequency space by
\begin{align}
    a_{\hm}(\omega,\ell)=\tilde{R}(-\omega,\ell)a_{\hp}(\omega,\ell),
\end{align}
where $a_{\hm}$ and $a_{\hp}$ are the Fourier transforms of $\psi_{\hm}$ and $\psi_{\hp}$ respectively. We find a counterexample to rigidity of asymptotics at $\hp$ by devising a highly oscillatory radiation field at $\psi_{\hp}$ which is of exponential order less than or equal to $1$ on the fundamental strip, i.e.~$|a_{\hp}|\sim e^{2\pi|\Im\omega\Re\omega|}$ as $\Re\omega$ grows. The exponential decay of ${\trr}$ then overcomes the growth of $a_{\hp}$, and \eqref{regularity assumption} is satisfied, yet the radiation field at $\hp$ is not exponentially decaying. In fact, it can be chosen to decay at a finite polynomial rate when measured in $L^2(\hp)$. Thus, due to \eqref{high frequency behaviour of reflection coefficient intro}, rigidity of asymptotics at $\hp$ fails even at fixed $\ell$. See already Theorem \ref{construction of counterexample I}

The situation for rigidity of asymptotics at $\Scrip$ is more varied: when $\ell$ is fixed, the transmission coefficient approaches unity as $\Re\omega$ grows with $\Im\omega$ fixed. Therefore, the only potential obstruction to exponential decay comes from the low frequency behaviour of the transmission coefficient. We show below that when a solution with finite local energy at the horizon also satisfies $\psi_{\hp}=0$, the low frequency obstruction disappears, and for fixed $\ell$, $\psi_{\hm}$ decays exponentially if and only if $\psi_{\Scrip}$ decays exponentially at the same rate. Regularity at $\hp$ then dictates that the rate of exponential decay of $\psi_{\hm}$ is determined by the surface gravity. Thus, rigidity of asymptotics at $\Scrip$ is satisfied for any fixed spherical harmonic mode. See already Theorem \ref{exponential decay theorem'}.

Rigidity of asymptotics at $\Scrip$ still fails nonetheless when support in $\ell$ is allowed to be unbounded. Indeed, we show in Proposition \ref{holy grail of transmission} that there exists a constant $C$ such that, for any $\omega$ with $\Im\omega=\frac{1}{4M}$, and $\ell(\ell+1)\geq C|\omega|^9$, we have
\begin{align}
    |T(\omega,\ell)|^2\lesssim \left(\frac{\ell}{(2M)^3|\omega|^2}\right)^{-\ell}.
\end{align}
The exponential decay of $T$ allows us to provide a counterexample to  rigidity of asymptotics at $\Scrip$, by constructing a polynomially decaying radiation field $\psi_{\Scrip}$ (as measured in $L^2(\Scrip)$), which is highly damped in $\omega$ but highly oscillatory in $\ell$. See Theorem \ref{construction of counter example Scrip}

\subsection{Non-degenerate scattering theory for $\ell=0$}

The fact that statement \ref{Statement I} above holds for fixed $\ell$ tells us that the radiation field at the event horizon carries all the information needed to deduce the part of the solution that decays slowly compared to $e^{-\frac{t}{4M}}$. We use this insight to answer the existence question and construct non-degenerate scattering theories for fixed $\ell$ solutions to \eqref{full wave equation}. We denote by $-i\omega a_{\hp}$, $-i\omega a_{\hm}$ $-i\omega a_{\Scrip}$, $-i\omega a_{\Scrim}$ the Fourier transforms of $\partial_v\psi_{\hp}$, $\partial_u\psi_{\hm}$, $\partial_u\psi_{\Scrip}$, $\partial_v\psi_{\Scrim}$ respectively. By a slight abuse of notation, we will use $\mathscr{F}^{N}_+$ to refer to the restriction of $\mathscr{F}^{T}_+$ to either $\mathcal{E}^{N}_{\Sigma}$ or $\mathcal{E}^{N}_{\overline{\Sigma}}$.

Consider $\psi_{\hp}\in \mathcal{E}_{\hp}^T$ a given function on $\hp$ which is spherically symmetric. If $\psi_{\hp}$ were to arise as the radiation field at $\hp$ of a solution with no incoming radiation from $\hm$, then we must have
\begin{align}\label{some relation in spherical symmetry}
    \omega a_{\Scrip}=\omega \frac{\tilde{R}(-\omega,\ell=0)}{T(-\omega,\ell=0)}a_{\hp}.
\end{align}
For spherically symmetric solutions to \eqref{full wave equation}, i.e.~solutions to \eqref{fixed ell mode wave equation} with $\ell=0$, we have that \mbox{$|\omega T(\omega,\ell=0)^{-1}|$} is bounded near $\omega=0$. Therefore, the solution induces a finite energy radiation field at $\Scrip$ if and only if
\begin{align}\label{spherical symmetry H1}
    \psi_{\hp}\in H^1(\mathbb{R}).
\end{align}
Note now that all radiation fields in $\mathcal{E}^T_{\hp}$ which arise out of Cauchy data sets on $\{t=0\}$ that have finite local energy at $r=2M$ will satisfy \eqref{spherical symmetry H1}. Therefore, we may always decompose any solution to \eqref{fixed ell mode wave equation} with $\ell=0$ arising via $\mathscr{F}^N_{\Sigma}$ into a sum of two solutions, 
\begin{align}
    \psi=\psi_{\nwarrow}+\psi_{\nearrow},
\end{align}
where $\psi_{\nwarrow}$ has no incoming radiation from $\hm$, and $\psi_{\nearrow}$ has no outgoing radiation at $\hp$. Moreover, rigidity of asymptotics at $\Scrip$ for $\ell=0$ dictates that $\psi_{\nearrow}$ induces a radiation field at $\Scrip$ that decays like $e^{-\frac{u}{4M}}$ in $L^2(\mathscr{I}^+)$. This allows to construct a scattering map that gives a complete description of the image of $\mathscr{F}^N_{\Sigma}$ in $\mathcal{E}^T_{\hp}\oplus \mathcal{E}^T_{\Scrip}$. We show the existence of the Hilbert space isomorphism
\begin{align}
    {}^{(\ell=0)}\mathcal{E}_{\Sigma}^{N}\simeq {}^{(\ell=0)}\mathcal{E}^N_{\hp}\oplus {}^{(\ell=0)}\mathcal{E}_{\Scrip}^{\mathrm{exp}}, 
\end{align}
where ${}^{(\ell=0)}\mathcal{E}_{\hp}^N$, ${}^{(\ell=0)}\mathcal{E}_{\Scrip}^{\mathrm{exp}}$ are the closures of $C^{\infty}_c(\mathbb{R})$ under the respective norms
\begin{align}
    \|f\|^2_{{}^{(\ell=0)}\mathcal{E}_{\hp}^N}:=\int_{-\infty}^{\infty} dv\, |f|^2+(1+e^{-v})|f'|^2,\qquad \|f\|_{{}^{(\ell=0)}\mathcal{E}_{\Scrip}^{\mathrm{exp}}}^2:=\int_{-\infty}^{\infty}du\, (1+e^{u})|f'|^2.
\end{align}
We can readily deduce the Hilbert space isomorphism
\begin{align}
    {}^{(\ell=0)}\mathcal{E}_{\overline{\Sigma}}^{N}\simeq {}^{(\ell=0)}\mathcal{E}^N_{\overline{\hp}}\oplus {}^{(\ell=0)}\mathcal{E}_{\Scrip}^{\mathrm{exp}},
\end{align}
where ${}^{(\ell=0)}\mathcal{E}_{\overline{\Sigma}}^{N}$ and ${}^{(\ell=0)}\mathcal{E}^N_{\overline{\hp}}$ are modifications of ${}^{(\ell=0)}\mathcal{E}_{{\Sigma}}^{N}$ and ${}^{(\ell=0)}\mathcal{E}^N_{\overline{\hp}}$ that allow solutions of \eqref{wave equation} to have support on $\mathcal{B}$.

An immediate corollary of Theorem \ref{non degenerate backwards scattering hp spherical symmetry} is a classification of all radiation fields at $\Scrip$ arising from finite energy solutions which are spherically symmetric and which are regular at $\hp$. For a solution satisfying \eqref{regularity assumption} which has no incoming radiation from $\hm$, we have by \eqref{some relation in spherical symmetry} that 
\begin{align}\label{some other relation in spherical symmetry}
    \omega^2 (1+|\omega|^2)^{-\frac{1}{2}}e^{4M\pi|\omega|}a_{\Scrip}\in L^2(\mathbb{R}).
\end{align}
It then follows from Theorem \ref{non degenerate backwards scattering hp spherical symmetry} that $\psi_{\Scrip}$ is the radiation field at $\Scrip$ of a solution satisfying \eqref{regularity assumption} if and only if 
\begin{align}
    \psi_{\Scrip}=\psi_{\Scrip}^{(1)}+\psi_{\Scrip}^{(2)},
\end{align}
where $\psi_{\Scrip}^{(2)}\in \spacescriexp{\ell=0}$ and the Fourier transform of $\psi_{\Scrip}^{(1)}$ satisfies \eqref{some other relation in spherical symmetry}. See already Corollary \ref{corollary non degenerate backwards scattering hp spherical symmetry}.
\subsection{Non-degenerate scattering theory for fixed spherical harmonic modes with $\ell\geq1$}

When $\ell\geq1$, the fact that $|T(\omega)|^{-1}$ diverges faster than $O(|\omega|^{-1})$ as $\omega\longrightarrow0$ presents an obstruction to the scheme described in the previous section. It is clear that for $\psi_{\hp}\in L^2(\mathbb{R})\cap \dot{H}^{-\ell}(\mathbb{R})$, we can repeat the above construction for backwards scattering, since we will then have $\omega T^{-1}a_{\hp}\in L^2(\mathbb{R})$ by assumption. However, this will not be sufficient to account for the full image of $\mathscr{F}^N_{+}$. We then devise a different method that makes use of the time reversal isometry of the Schwarzschild spacetime.

Let us now describe how a non-degenerate scattering theory for \eqref{fixed ell mode wave equation} may be obtained by utilising the time reversal isometry. Let us fix $\ell\in\mathbb{N}_{\geq0}$, and let ${}^{(\ell)}\mathcal{E}^N_{\Sigma}$ be the projection of $\mathcal{E}^N_{\Sigma}$ to the span of $\{Y^{\ell}_{m}\}_{m=-\ell}^{\ell}$. Then ${}^{(\ell)}\mathcal{E}^N_{\Sigma}$ is simply the space
\begin{align}
    {}^{(\ell)}\mathcal{E}^N_{\Sigma}=\cosspace\oplus \sinespace,
\end{align}
where
\begin{align}
    \cosspace:=H^1((2M,\infty),\Omega^{-1}r^2dr),\qquad \sinespace:=L^2((2M,\infty),\Omega^{-1}r^2dr).
\end{align}
Note that the measure $\Omega^{-1}r^2 dr$ is obtained by considering the hypersurface $\{t=0\}$ in coordinates induced by the Kruskal foliation, and represents the true volume form in the radial direction on $\{t=0\}$ when expressed in the Boyer--Lindquist coordinate $r$, utilised here as an effective radial coordinate.

For data $(\uppsi,0)$ with $\uppsi\in \cosspace$, the time-reversal isometry of \eqref{Schwarzschild metric in Boyer Lindquist} implies that the solution is symmetric in $t$, and therefore, the radiation field at $\hp$ determines the full solution. In particular, we have
\begin{align}\label{a scrip cosine intro}
    a_{\Scrip}(\omega)=\frac{1}{T(-\omega)}\left[a_{\hp}(-\omega)+\tilde{R}(-\omega)a_{\hp}(\omega)\right].
\end{align}
At low frequency, $\omega\sim0$, the scattering matrix elements in frequency space have the behaviour
\begin{align}
    |T(\omega)|\sim |\omega|^{\ell+1}
\end{align}
(see for instance \cite{DDS11} or \cite{StarobinskiiChurilov}). However, while $\tilde{R}$ in absolute value is close to unity up to an $|\omega|^{2(\ell+1)}$, the phase of $\tilde{R}$ can have terms that decay more slowly than $\omega^{\ell}$. We then make the following assumption: for any given $\ell\geq0$, there exists a function $\widetilde{W}_{\ell}$ such that 
\begin{align}\label{condition ell neq 0}
    \tilde{R}(\omega,\ell)= \widetilde{W}_{\ell}^{(\ell)}(\omega)+\widetilde{H}^{(\ell)}_{\ell}(\omega),
\end{align}
where $\omega^{-ell}\widetilde{H}^{(\ell)}_{\ell}$ is locally square integrable at $\omega=0$. With this assumption, the fact that $\uppsi\in \cosspace$ implies
\begin{align}
    \omega^{-\ell}(a_{\hp}(-\omega)+\widetilde{W}_{\ell}^{(\ell)}(-\omega)a_{\hp}(\omega))
\end{align}
is locally square integrable at $\omega=0$. Conversely, if we assume $a_{\hp}\in \cosspacehp$, where
\begin{align}\label{def of E cos hp}
    \begin{split}
        \cosspacehp:=\Big\{&f:\mathbb{R}\longrightarrow \mathbb{R}; f\in L^2(\mathbb{R}), (e^{-\frac{v}{4M}}+1)f'\in L^2(\mathbb{R}), \\&\omega^{-\ell}[\mathcal{F}[f](\omega)+\widetilde{W}_{\ell}^{(\ell)}(\omega)\mathcal{F}[f](-\omega)]\in L^2_{loc}(\mathbb{R})\Big\},
    \end{split}
\end{align}
equipped with the norm
\begin{align}\label{def of E cos hp norm}
        \begin{split}
            \|f\|_{\cosspacehp}^2=\|f\|_{L^2(\mathbb{R})}^2+\|(1+e^{-\frac{1}{4M}v})f'\|_{L^2(\mathbb{R})}^2+\int_{-1}^{1}d\omega\; \omega^{-2\ell}\left|\mathcal{F}[f](\omega)+\widetilde{W}_{\ell}^{(\ell)}(\omega)\mathcal{F}[f](-\omega)\right|^2.
        \end{split}
    \end{align}
We define $a_{\Scrip}$ via \eqref{a scrip cosine intro}, and show that $a_{\hp}$, $a_{\Scrip}$ give rise to $\psi_{\hp}$, $\psi_{\Scrip}$ are in $\mathcal{E}^T_{\hp}$, $\mathcal{E}^T_{\Scrip}$ respectively, and moreover, that the resulting $\psi_{\hm}$, $\psi_{\Scrim}$ satisfy $\psi_{\hm}(u)=\psi_{\hp}(-v)$, $\psi_{\Scrim}(v)=\psi_{\Scrip}(-v)$. Therefore, $\partial_t\psi|_{t=0}=0$, and a local estimate near the bifurcation sphere using the time symmetry of the solution ensures $\uppsi\in H^1_{loc}(\Sigma,d\mathrm{Vol}_{\Sigma})$, including at $r=2M$. Thus we can show $\cosspace$ is isomorphic to $\cosspacehp$, with the isomorphism being mediated by time-symmetric evolution.

We may repeat the exact same argument for solutions with vanishing initial value. Since these solutions are \mbox{necessarily} time-antisymmetric, we can apply the same argument as above with $\psi_{\hm}$ in the space
\begin{align}\label{def of E sin hp}
    \begin{split}
        \sinespacehp:=\Big\{&f:\mathbb{R}\longrightarrow \mathbb{R}: f\in L^2(\mathbb{R}), (e^{-\frac{v}{4M}}+1)f'\in L^2(\mathbb{R}), \\&\omega^{-\ell}[\mathcal{F}[f](\omega)-\widetilde{W}_{\ell}^{(\ell)}(\omega)\mathcal{F}[f](-\omega)]\in L^2_{loc}(\mathbb{R})\Big\},
    \end{split}
    \end{align}
    equipped with the norm
    \begin{align}\label{def of E sin hp norm}
        \begin{split}
            \|f\|_{\sinespacehp}^2=\|f\|_{L^2(\mathbb{R})}^2+\|(1+e^{-\frac{1}{4M}v})f'\|_{L^2(\mathbb{R})}^2+\int_{-1}^{1}d\omega\; \omega^{-2\ell}|\mathcal{F}[f](\omega)-\widetilde{W}_{\ell}^{(\ell)}(\omega)\mathcal{F}[f](-\omega)|^2,
        \end{split}
    \end{align}
and we can show that $\sinespace$ is isomorphic to $\sinespacehp$, with the isomorphism being mediated by time-antisymmetric evolution. Thus we have the canonical isomorphism
\begin{align}\label{canonical isomorphism intro}
    {}^{(\ell)}\mathcal{E}^N_{\Sigma}\simeq \cosspacehp\oplus \sinespacehp.
\end{align}
Therefore, $\psi_{\hp}$ is in the image of $\mathscr{F}^N_{+}$ if and only $\psi_{\hp}$ is in the algebraic sum
\begin{align}
     \cosspacehp+\sinespacehp.
\end{align}
The construction described above is the content of Theorem \ref{non degenerate backwards scattering hp} below, and it provides an answer to the question of existence as stated above when it comes to radiation fields at $\hp$, provided that the condition \eqref{condition ell neq 0} is satisfied by $\tilde{R}(\omega,\ell)$. 

Having identified the space of radiation fields at $\hp$ arising via $\mathscr{F}_{+}^N$, we can now adapt a strategy similar to how Theorem \ref{non degenerate backwards scattering hp spherical symmetry} provides a scattering theory for $\ell=0$. Since the space 
\begin{align}
    \kspacehpzero:=\{(v,-v); v\in {\cosspacehp\cap\sinespacehp}\}
\end{align}
is a closed subspace of $\cosspacehp\oplus\sinespacehp$, there exists a unique $v^*$ in ${\cosspacehp\cap\sinespacehp}$ such that $(\psi_{\hp}^{\cos}+v^*,\psi_{\hp}^{\sin}-v^*)\perp\kspacehpzero$. We then define
\begin{align}
    \spacehp:=\{f\in\mathcal{E}^T_{\hp}; f=f^{\cos}+f^{\sin},\;\; (f^{\cos},f^{\sin})\in \big(\kspacehpzero\big)^{\perp}\},
\end{align}
equipped with the norm
\begin{align}
    \|\psi_{\hp}\|_{\spacehp}^2:=\|\psi_{\hp}^{\cos}+v^*\|_{\cosspacehp}^2+\|\psi_{\hp}^{\sin}-v^*\|_{\sinespacehp}^2.
\end{align}
Thus, for any $\psi$ solving \eqref{fixed ell mode wave equation}, we define $\psi^*$ to be the solution arising via the projection map and the isomorphism \eqref{canonical isomorphism intro}, and we will then have that $\psi-\psi^*$ has no radiation at $\hp$. Rigidity of asymptotics at $\Scrip$ then implies that the image of $\mathcal{E}_{\Sigma}^N$ under $\mathscr{F}^{N}_+$ is the space
\begin{align}
    \spacehp\oplus \spacescriexp{\ell},
\end{align}
where $\spacescriexpn{\ell}$ is defined as the closure of smooth, compactly supported functions on $\mathbb{R}$ under the norm
\begin{align}
    \|f\|_{{}^{(\ell)}\mathcal{E}_{\Scrip}^{\mathrm{exp}}}^2:=\int_{-\infty}^{\infty}du\, (1+e^{\frac{u}{2M}})|f'|^2.
\end{align}

The case of solutions that are supported on the bifurcation sphere $\mathcal{B}$ is similarly handled, but in addition to  \eqref{formal assumption on low frequency expansion of tilde R intro}, we work with the assumption
\begin{align}
    \tilde{R}(\omega,\ell)=\widetilde{W}_{\ell}^{(\ell+1)}(\omega)+\widetilde{H}^{(\ell+1)}_{\ell}(\omega), 
\end{align}
where $\omega^{-\ell-1}\widetilde{H}^{(\ell+1)}_{\ell}$ is locally square integrable at $\omega=0$,
and we may apply the argument used to find the image of $\cosspace$ under $\mathscr{F}^{N,\cos}_{+,\hp}$ to show
\begin{align}
    \cosspacebar\simeq \cosspacehpbar,
\end{align}
where the space $\cosspacehp$ is given by
\begin{align}\label{def of E cos hp bar }
    \begin{split}
        \cosspacehpbar:=\Big\{&f:\mathbb{R}\longrightarrow \mathbb{R}: (e^{-\frac{v}{4M}}+1)f'\in L^2(\mathbb{R}), \\&\omega^{-\ell-1}[\mathcal{F}[f'](\omega)-\widetilde{W}_{\ell}^{(\ell+1)}(\omega)\mathcal{F}[f'](-\omega)]\in L^2_{loc}(\mathbb{R})\Big\},
    \end{split}
\end{align}
and is equipped with the norm
\begin{align}\label{def of E cos hp bar norm}
        \begin{split}
            \|f\|_{\cosspacebar}^2=\|(1+e^{-\frac{1}{4M}v})f'\|_{L^2(\mathbb{R})}^2+\int_{-1}^{1}d\omega\; \omega^{-2\ell-2}\left|\mathcal{F}[f'](\omega)-\widetilde{W}_{\ell}^{(\ell+1)}(\omega)\mathcal{F}[f'](-\omega)\right|^2.
        \end{split}
    \end{align}
This shows that, taking 
\begin{align}
    \kspacehpbarzero:=\{(v,-v); v\in {\cosspacehpbar\cap\sinespacehp}\},
\end{align}
\begin{align}
    \spacehpbar:=\{f\in\mathcal{E}^T_{\hp}; f=f^{\cos}+f^{\sin},\;\; (f^{\cos},f^{\sin})\in \big(\kspacehpbarzero\big)^{\perp}\},
\end{align}
we have the Hilbert space isomorphism
\begin{align}
    {}^{(\ell)}\mathcal{E}_{\overline{\Sigma}}^{N}\simeq \spacehpbar\oplus \spacescriexp{\ell}.
\end{align}

We may use the ideas above to identify the projection of the image of $\mathscr{F}_+^N$ to $\mathcal{E}_{\Scrip}^T$, and show that
\begin{align}\label{space of radiation at Scrip for general ell Sigma bar}
    {}^{(\ell)}\mathcal{E}^N_{\overline{\Sigma}}\simeq \cosspacescrip\oplus \sinespacescrip,
\end{align}
where 
\begin{align}\label{def of E cos scrip}
    \begin{split}
        \cosspacescrip=\Bigg\{&f:\mathbb{R}\longrightarrow\mathbb{R}; f'\in L^2(\mathbb{R}),\,\frac{1}{T(-\omega)}\Big[\mathcal{F}[f'](-\omega)-R(-\omega)\mathcal{F}[f'](\omega)\Big]\in PW\left(0,\frac{1}{4M}\right)\Bigg\},
    \end{split}
\end{align}

\begin{align}\label{def of E sin scrip}
    \begin{split}
        \sinespacescrip=\Bigg\{&f:\mathbb{R}\longrightarrow\mathbb{R}; f'\in L^2(\mathbb{R}),\,\frac{1}{T(-\omega)}\Big[\mathcal{F}[f'](-\omega)+R(-\omega)\mathcal{F}[f'](\omega)\Big]\in PW\left(0,\frac{1}{4M}\right),\\& \frac{1}{\omega T(\omega)}\Big[\mathcal{F}[f'](\omega)+R(\omega)\mathcal{F}[f'](-\omega)\Big]\in L^2(\mathbb{R})\Bigg\},
    \end{split}
\end{align}
equipped with the norms
\begin{align}
    \begin{split}
        \|f\|_{\cosspacescrip}^2:=\|f'\|_{L^2(\mathbb{R})}+\left\|\frac{1}{T(-\omega)}\Big[\mathcal{F}[f'](-\omega)-R(-\omega)\mathcal{F}[f'](\omega)\Big]\right\|_{PW(0,\frac{1}{4M})}^2,
    \end{split}
\end{align}
\begin{align}
    \begin{split}
        \|f\|_{\sinespacescrip}^2:=&\|f'\|_{L^2(\mathbb{R})}+\left\|\frac{1}{T(-\omega)}\Big[\mathcal{F}[f'](-\omega)+R(-\omega)\mathcal{F}[f'](\omega)\Big]\right\|_{PW(0,\frac{1}{4M})}^2\\&\,+\int_{-\infty}^{\infty}d\omega \frac{1}{\omega^2|T(-\omega)|^2}\Big|\mathcal{F}[f'](\omega)+R(\omega)\mathcal{F}[f'](-\omega)\Big|^2,
    \end{split}
\end{align}
where $PW(0,b)$ refers to the space of Paley--Wiener functions over the strip $\Im\omega\in (0,b)$ (see Definition \ref{def of paley wiener space}). Thus $\psi_{\Scrip}$ arises in the image of ${}^{(\ell)}\mathcal{E}^N_{\overline{\Sigma}}$ under ${}^{(\ell)}\mathscr{F}_+^N$ at $\Scrip$ if and only if
\begin{align}
    \psi_{\Scrip}\in \cosspacescrip+\sinespacescrip.
\end{align}
The construction described here is the content of Theorem \ref{non degenerate backwards scattering Scrip} below, and it provides the answer to the question of existence as stated above when it comes to radiation fields at $\Scrip$. 

It is straightforward to specialise to the case of the image of  ${}^{(\ell)}\mathcal{E}^N_{{\Sigma}}$ under ${}^{(\ell)}\mathscr{F}^N_+$, where we have
\begin{align}
    {}^{(\ell)}\mathcal{E}^N_{{\Sigma}}\simeq\cosspacescripo\oplus \sinespacescrip,
\end{align}
with $\cosspacescripo$ defined by
\begin{align}\label{def of E cos scrip o}
    \begin{split}
        \cosspacescripo=\Bigg\{&f:\mathbb{R}\longrightarrow\mathbb{R}; f'\in L^2(\mathbb{R}),\,\frac{1}{T(-\omega)}\Big[\mathcal{F}[f'](-\omega)-R(-\omega)\mathcal{F}[f'](\omega)\Big]\in PW\left(0,\frac{1}{4M}\right),\\&\frac{1}{\omega T(\omega)}\Big[\mathcal{F}[f'](\omega)-R(\omega)\mathcal{F}[f'](-\omega)\Big]\in L^2(\mathbb{R})\Bigg\}. 
    \end{split}
\end{align}
We give the precise statements in Theorem \ref{non degenerate backwards scattering Scrip}, 

Finally, note that analogous statements hold for radiation at $\hm$, $\Scrim$.
\subsection{Rigidity of asymptotics and nondegenerate scattering theory under higher order regularity conditions at the horizon}
We also address rigidity of asymptotics and non-degenerate scattering for solutions satisfying a higher regularity analogue of \eqref{regularity assumption}. Recall that time-translation invariance also implies the Hilbert-space isomorphisms
\begin{align}
    \mathscr{F}^{T,n}_{+}:\mathcal{E}^{T,n}_{\Sigma}\longrightarrow \mathcal{E}^{T,n}_{\hp}\oplus \mathcal{E}^{T,n}_{\Scrip},\qquad \mathscr{B}^{T,n}_{+}:\mathcal{E}^{T,n}_{\hp}\oplus \mathcal{E}^{T,n}_{\Scrip}\longrightarrow\mathcal{E}^{T,n}_{\Sigma},
\end{align}
where $\mathcal{E}^{T,n}_{\Sigma}$ is the closure of smooth, compactly supported data under the norm
\begin{align}
    \|(\uppsi,\uppsi')\|_{\mathcal{E}_{\Sigma}^{T,n}}^2=\sum_{k=0}^n\int_{\Sigma}d\mathrm{Vol}_{\Sigma}\, T_{\mu\nu}[\partial_t^k\phi]\partial_t^{\mu}n_{\Sigma^*}^\nu,
\end{align}
and each of $\mathcal{E}^{T,n}_{\hp}$, $\mathcal{E}^{T,n}_{\Scrip}$ is the closure of smooth, compactly supported functions on $\mathbb{R}\times S^2$ under the respective norms
\begin{align}
    \|\psi_{\hp}\|^2_{\mathcal{E}^{T,n}_{\hp}}:=\sum_{k=0}^{n}\int_{\hp}dvd\mathrm{Vol}_{S^2} |\partial_v^{k+1}\psi_{\hp}|^2,\qquad \|\psi_{\Scrip}\|^2_{\mathcal{E}^{T,n}_{\Scrip}}:=\sum_{k=0}^{n}\int_{\Scrip}dud\mathrm{Vol}_{S^2} |\partial_u^{k+1}\psi_{\Scrip}|^2.
\end{align}
On the other hand, commuting \eqref{wave equation psi intro} with $\partial_U$ (or $\Omega^{-2}\partial_u$) leads to an enhanced redshift effect and higher order correlations between radiation at the horizon and at null infinity. We show that
\begin{itemize}
    \item A solution $\psi$ to \eqref{fixed ell mode wave equation} that has finite $\partial_t$-energy to $n^{\mathrm{th}}$ order, has no radiation on $\hp$, and satisfying 
        \begin{align}\label{regularity assumption n intro}
             \sum_{k=0}^{n}\int_{u_0}^{\infty}\int_{S^2}dud\mathrm{Vol}_{S^2}\frac{1}{\Omega^2}\left|\partial_u\left(\frac{1}{\Omega^2}\partial_u\right)^k\phi(u,v)\right|^2<\infty
        \end{align}
     will also have $e^{\frac{u}{4M}}\partial_u\left(e^{\frac{u}{2M}}\partial_u\right)^k\psi_{\Scrip}\in L^2(\mathbb{\Scrip})$ for $k=0,\dots,n$.
    \item On the other hand, for any $n\in\mathbb{N}_{\geq0}$, there exist examples of finite $\partial_t$-energy solutions to \eqref{fixed ell mode wave equation} with no radiation on $\Scrip$ which satisfy \eqref{regularity assumption n intro} and which decay polynomially along $\hp$.
    \item Similarly, for any $n\in\mathbb{N}_{\geq0}$, there exist solutions to \eqref{wave equation psi intro} with unbounded support in $\ell$ which satisfy \eqref{regularity assumption n intro} and and which decay polynomially along $\hp$.
\end{itemize}
The statements above are shown in an analogous fashion to the case where \eqref{regularity assumption n intro} holds with $n=0$. We give the precise statements in Section \ref{Section precise statement of results} below. 

We also address the question of the existence of horizon-regular solutions attaining a given radiation field at either $\hp$ or $\Scrip$ when the regularity assumption at $\hp$ is strengthened to \eqref{regularity assumption n intro}. Let $\mathscr{F}^{N,n}_{+}$ be the restriction of $\mathscr{F}^{T}_{+}$ to $\mathcal{E}^{N,n}_{\Sigma}=H^n(\Sigma,d\mathrm{Vol}_{\Sigma})\oplus H^{n-1}(\Sigma,d\mathrm{Vol}_{\Sigma})$, and let ${}^{(\ell)}\mathscr{F}^{N,n}_{+}$ be the restriction of $\mathscr{F}^{N,n}_{+}$ to Cauchy data supported in the span of $\{Y^{\ell}_m\}^{\ell}_{m=-\ell}$. In the case $\ell=0$, we identify the image of $\mathcal{E}^{N,n}_{\Sigma}$ under $\mathscr{F}^{N,n}_{\Sigma}$ to be
\begin{align}
     {}^{(\ell=0)}\mathcal{E}^{N,n}_{{\hp}}\oplus \spacescriexpon{\ell=0},
\end{align}
where ${}^{(\ell=0)}\mathcal{E}_{\hp}^{N,n}$, $\spacescriexpon{\ell=0}$ are the closures of $C^{\infty}_c(\mathbb{R})$ under the norms
\begin{align}
    \|f\|^2_{{}^{(\ell=0)}\mathcal{E}_{\hp}^{N,n}}:=\sum_{k=0}^{n+1}\int_{-\infty}^{\infty} dv\, \left|\left(1+e^{-\frac{(2n+1)v}{4M}}\right)\partial_v^{k}f\right|^2,
\end{align}
\begin{align}
     \|f\|_{\spacescriexpon{\ell=0}}^2:=\sum_{k=1}^{n+1}\int_{-\infty}^{\infty}du\, \left|\left(1+e^{\frac{(2n+1)u}{4M}}\right)\partial_u^{k}f\right|^2.
\end{align}
See already Theorem \ref{non degenerate backwards scattering hp spherical symmetry} below.

The image under $\mathscr{F}^{N,n}_+$ of ${}^{(\ell=0)}\mathcal{E}_{\overline{\Sigma}}^{N,n}$, the subspace of $H^n(\overline{\Sigma},d\mathrm{Vol}_{\overline{\Sigma}})\oplus H^{n-1}(\overline{\Sigma},d\mathrm{Vol}_{\overline{\Sigma}})$ consisting of spherically symmetric data, is given by
\begin{align}
    {}^{(\ell=0)}\mathcal{E}^{N,n}_{\overline{\hp}}\oplus \spacescriexpn{\ell=0},
\end{align}
where ${}^{(\ell=0)}\mathcal{E}^{N,n}_{\overline{\hp}}$, $\spacescriexpn{\ell=0}$ are the closure of $C^{\infty}_c(\mathbb{R})$ under the respective norms
\begin{align}
    \|f\|_{{}^{(\ell=0)}\mathcal{E}^{N,n}_{\overline{\hp}}}^2:=\sum_{k=1}^{n+1}\|f^{(k)}\|_{L^2(\mathbb{R})}^2+\sum_{k=0}^{n}\|e^{-\frac{v}{4M}}\partial_v(e^{-\frac{v}{2M}}\partial_v)^kf\|_{L^2(\mathbb{R})}^2,
\end{align}
\begin{align}
    \|f\|_{{}^{(\ell=0)}\mathcal{E}_{\Scrip}^{\mathrm{exp},n}}^2:=\sum_{k=0}^{n}\int_{-\infty}^{\infty}du\, |\partial_u^{k+1}f|^2+e^{\frac{u}{2M}}|\partial_u(e^{\frac{u}{2M}}\partial_u)^kf|^2.
\end{align}

For any fixed $\ell\geq1$, we identify the image of $\mathcal{E}^{N,n}_{\Sigma}$ under $\mathscr{F}^{N,n}_{\Sigma}$ to be 
\begin{align}\label{def of kspace intro}
    \spacehpn\oplus\spacescriexpon{\ell},
\end{align}
where $\spacehpn$ is defined out of $\cosspacehpn$ and $\sinespacehpn$ analogously to $\spacehp$,
\begin{align}
    \kspacehp\{(v,-v); v\in {\cosspacehpn\cap\sinespacehpn}\},
\end{align}
\begin{align}\label{def of spacehpn intro}
    \spacehpn:=\{f\in\mathcal{E}^{T,n}_{\hp}; f=f^{\cos}+f^{\sin},\;\; (f^{\cos},f^{\sin})\in \big(\kspacehp\big)^{\perp}\},
\end{align}
where
\begin{align}\label{def of E cos hp n}
    \begin{split}
        \cosspacehpn:=\Big\{&f:\mathbb{R}\longrightarrow \mathbb{R}; \;\left(1+e^{-\frac{(2n+1)v}{4M}}\right)f^{(k)}\in L^2(\mathbb{R})\;\;\forall k=0\dots n+1, \\&
        \omega^{-\ell}[\mathcal{F}[f](\omega)+\widetilde{W}_{\ell}^{(\ell)}(\omega)\mathcal{F}[f](-\omega)]\in L^2_{loc}(\mathbb{R})\Big\},
    \end{split}
\end{align}
\begin{align}\label{def of E sin hp n}
    \begin{split}
        \sinespacehpn:=\Big\{&f:\mathbb{R}\longrightarrow \mathbb{R}; \;\left(1+e^{-\frac{(2n+1)v}{4M}}\right)f^{(k)}\in L^2(\mathbb{R})\;\;\forall k=0,\dots,n+1,\\&
        \omega^{-\ell}[\mathcal{F}[f](\omega)-\widetilde{W}_{\ell}^{(\ell)}(\omega)\mathcal{F}[f](-\omega)]\in L^2_{loc}(\mathbb{R})\Big\},
    \end{split}
\end{align}
equipped with the norms
\begin{align}
        \begin{split}
            \|f\|_{\cosspacehpn}^2=&\sum_{k=0}^{n+1}\left\|\left(1+e^{-\frac{(2n+1)v}{4M}}\right)f^{(k)}\right\|_{L^2(\mathbb{R})}^2+\int_{-1}^{1}d\omega\; \omega^{-2\ell}\left|\mathcal{F}[f](\omega)+\widetilde{W}_{\ell}^{(\ell)}(\omega)\mathcal{F}[f](-\omega)\right|^2,
        \end{split}
    \end{align}
    \begin{align}
        \begin{split}
            \|f\|_{\sinespacehpn}^2=&\sum_{k=0}^{n+1}\left\|\left(1+e^{-\frac{(2n+1)v}{4M}}\right)f^{(k)}\right\|_{L^2(\mathbb{R})}^2+\int_{-1}^{1}d\omega\; \omega^{-2\ell}\left|\mathcal{F}[f](\omega)-\widetilde{W}_{\ell}^{(\ell)}(\omega)\mathcal{F}[f](-\omega)\right|^2,
        \end{split}
    \end{align}
    and $\spacescriexpon{\ell}$ is defined to be the closure of smooth, compactly supported functions on $\mathbb{R}$ under the norm
    \begin{align}
    \|f\|_{\spacescriexpon{\ell}}^2:=\sum_{k=1}^{n+1}\int_{-\infty}^{\infty}du\, \left|\left(1+e^{\frac{(2n+1)u}{4M}}\right)\partial_u^{k}f\right|^2.
\end{align}
We identify the image of $\mathcal{E}^{N,n}_{\overline{\Sigma}}$ under $\mathscr{F}^{N,n}_{\Sigma}$ to be
\begin{align}
    \spacehpbarn\oplus\spacescriexpn{\ell},
\end{align}
where $\spacehpbarn$ is defined out of $\cosspacehpbarn$ and $\sinespacehpbarn$ analogously to $\spacehp$,
\begin{align}
    \kspacehpbar\{(v,-v); v\in {\cosspacehpbarn\cap\sinespacehpbarn}\},
\end{align}
\begin{align}\label{def of spacehpbarn intro}
    \spacehpbarn:=\{f\in\mathcal{E}^{T,n}_{\hp}; f=f^{\cos}+f^{\sin},\;\; (f^{\cos},f^{\sin})\in \big(\kspacehpbar\big)^{\perp}\},
\end{align}
where 
\begin{align}\label{def of E cos hp bar n}
    \begin{split}
        \cosspacehpbarn:=\Big\{&f:\mathbb{R}\longrightarrow \mathbb{R}; \;f^{(k)}\in L^2(\mathbb{R})\;\;\forall k=1,\dots,n+1,\\&\;
        e^{-\frac{v}{4M}}\partial_v(e^{-\frac{v}{2M}}\partial_v)^kf\in L^2(\mathbb{R}) \;\;\forall k=1,\dots n, \\&
        \omega^{-\ell-1}[\mathcal{F}[f'](\omega)-\widetilde{W}_{\ell}^{(\ell+1)}(\omega)\mathcal{F}[f'](-\omega)]\in L^2_{loc}(\mathbb{R})\Big\},
    \end{split}
\end{align}
\begin{align}\label{def of E sin hp bar n}
    \begin{split}
        \sinespacehpbarn:=\Big\{&f:\mathbb{R}\longrightarrow \mathbb{R}; \;f^{(k)}\in L^2(\mathbb{R})\;\;\forall k=0\dots,n+1,\\&\;
        e^{-\frac{v}{4M}}\partial_v(e^{-\frac{v}{2M}}\partial_v)^kf\in L^2(\mathbb{R}) \;\;\forall k=1,\dots n, \\&
        \omega^{-\ell}[\mathcal{F}[f](\omega)-\widetilde{W}_{\ell}^{(\ell)}(\omega)\mathcal{F}[f](-\omega)]\in L^2_{loc}(\mathbb{R})\Big\},
    \end{split}
\end{align}
equipped with the norms
\begin{align}\label{def of E cos hp bar norm n}
        \begin{split}
            \|f\|_{\cosspacehpbarn}^2=&\sum_{k=1}^{n+1}\|f^{(k)}\|_{L^2(\mathbb{R})}^2+\sum_{k=0}^{n}\|e^{-\frac{v}{4M}}\partial_v(e^{-\frac{v}{2M}}\partial_v)^kf\|_{L^2(\mathbb{R})}^2\\&+\int_{-1}^{1}d\omega\; \omega^{-2\ell-2}\left|\mathcal{F}[f](\omega)+\widetilde{W}_{\ell}^{(\ell+1)}(\omega)\mathcal{F}[f](-\omega)\right|^2.
        \end{split}
    \end{align}
    \begin{align}\label{def of E sin hp bar norm n}
        \begin{split}
            \|f\|_{\sinespacehpbarn}^2=&\sum_{k=0}^{n+1}\|f^{(k)}\|_{L^2(\mathbb{R})}^2+\sum_{k=0}^{n}\|e^{-\frac{v}{4M}}\partial_v(e^{-\frac{v}{2M}}\partial_v)^kf\|_{L^2(\mathbb{R})}^2\\&\int_{-1}^{1}d\omega\; \omega^{-2\ell}\left|\mathcal{F}[f](\omega)-\widetilde{W}_{\ell}^{(\ell)}(\omega)\mathcal{F}[f](-\omega)\right|^2.
        \end{split}
    \end{align}
    
We identify the projection to $\Scrip$ of the image of $\mathcal{E}^{N,n}_{\overline{\Sigma}}$ under $\mathscr{F}^{N,n}_{+}$ to be the algebraic sum
\begin{align}
    \cosspacescripn+ \sinespacescripn,
\end{align}
where 
\begin{align}\label{def of E cos scrip n}
    \begin{split}
        \cosspacescripn=\Bigg\{&f:\mathbb{R}\longrightarrow\mathbb{R}; f^{(k)}\in L^2(\mathbb{R})\;\;\forall k=1,\dots,n+1,\\& \,\frac{\omega^k}{T(\omega)}\Big[\mathcal{F}[f'](\omega)-R(\omega)\mathcal{F}[f'](-\omega)\Big]\in L^2(\mathbb{R})\;\;\forall k=0,\dots,n,\\&\left[\prod_{q=0}^{k}\omega-iq\right] \frac{1}{\omega T(-\omega)}\Big[\mathcal{F}[f'](-\omega)-R(-\omega)\mathcal{F}[f'](\omega)\Big]\in PW\left(0,\frac{2k+1}{4M}\right)\;\;\forall k=0,\dots,n\Bigg\},
    \end{split}
\end{align}

\begin{align}\label{def of E sin scrip n}
    \begin{split}
        \sinespacescripn=\Bigg\{&f:\mathbb{R}\longrightarrow\mathbb{R}; f^{(k)}\in L^2(\mathbb{R})\;\;\forall k=1,\dots,n+1,\\&\,\frac{\omega^{k}}{T(\omega)}\Big[\mathcal{F}[f'](\omega)+R(\omega)\mathcal{F}[f'](-\omega)\Big]\in L^2(\mathbb{R}) \;\;\forall k=-1,\dots,n,\\&\left[\prod_{q=0}^{k}\omega-iq\right] \frac{1}{\omega T(-\omega)}\Big[\mathcal{F}[f'](-\omega)+R(-\omega)\mathcal{F}[f'](\omega)\Big]\in PW\left(0,\frac{2k+1}{4M}\right)\;\;\forall k=0,\dots,n\Bigg\},
    \end{split}
\end{align}
equipped with the norms
\begin{align}
        \|f\|_{\cosspacescripn}:=\sum_{k=0}^{n}\|f^{(k+1)}\|_{L^2(\mathbb{R})}^2+\sum_{k=0}^{n}\left\|\left[\prod_{q=0}^{k}(\omega-iq)\right]\frac{1}{\omega T(-\omega)}\Big[\mathcal{F}[f'](-\omega)-R(-\omega)\mathcal{F}[f'](\omega)\Big]\right\|^2_{PW\left(0,\frac{2k+1}{4M}\right)}.
\end{align}
\begin{align}
\begin{split}
    \|f\|_{\sinespacescripn}:=&\sum_{k=0}^{n}\|f^{(k+1)}\|_{L^2(\mathbb{R})}^2+\sum_{k=0}^{n}\left\|\left[\prod_{q=0}^{k}(\omega-iq)\right]\frac{1}{\omega T(-\omega)}\Big[\mathcal{F}[f'](-\omega)+R(-\omega)\mathcal{F}[f'](\omega)\Big]\right\|^2_{PW\left(0,\frac{2k+1}{4M}\right)}\\&\,+\int_{-\infty}^{\infty}d\omega\, \frac{1}{\omega^2 |T(\omega)|^2}\Big|\mathcal{F}[f'](\omega)+R(\omega)\mathcal{F}[f'](-\omega)\Big|^2.
\end{split}
\end{align}
We identify the projection to $\Scrip$ of the image of $\mathcal{E}^{N,n}_{{\Sigma}}$ under $\mathscr{F}^{N,n}_{\Sigma}$ to be the algebraic sum
\begin{align}
    \cosspacescripon+ \sinespacescripon,
\end{align}
where
\begin{align}\label{def of E cos scrip o n}
    \begin{split}
        \cosspacescripon=\Bigg\{&f:\mathbb{R}\longrightarrow\mathbb{R}; f^{(k)}\in L^2(\mathbb{R})\;\;\forall k=1,\dots,n+1,\\& \frac{\omega^k}{T(-\omega)}\Big[\mathcal{F}[f'](-\omega)-R(-\omega)\mathcal{F}[f'](\omega)\Big]\in PW\left(0,\frac{2n+1}{4M}\right)\;\;\forall k=-1,\dots,n\Bigg\},
    \end{split}
\end{align}
\begin{align}\label{def of E sin scrip o n}
    \begin{split}
        \sinespacescripon=\Bigg\{&f:\mathbb{R}\longrightarrow\mathbb{R}; f^{(k)}\in L^2(\mathbb{R})\;\;\forall k=1,\dots,n+1,\\& \frac{\omega^k}{T(-\omega)}\Big[\mathcal{F}[f'](-\omega)+R(-\omega)\mathcal{F}[f'](\omega)\Big]\in PW\left(0,\frac{2n+1}{4M}\right)\;\;\forall k=-1,\dots,n\Bigg\},
    \end{split}
\end{align}
equipped with the norms
\begin{align}
    \begin{split}
        \|f\|_{\cosspacescripon}^2=\sum_{k=0}^{n}\|f^{(k+1)}\|_{L^2(\mathbb{R}}^2+\sum_{k=-1}^{n}\left\|\frac{\omega^k}{T(-\omega)}\Big[\mathcal{F}[f'](-\omega)-R(-\omega)\mathcal{F}[f'](\omega)\Big]\right\|^2_{PW\left(0,\frac{2n+1}{4M}\right)},
    \end{split}
\end{align}
\begin{align}
    \begin{split}
        \|f\|_{\sinespacescripon}^2=\sum_{k=0}^{n}\|f^{(k+1)}\|_{L^2(\mathbb{R}}^2+\sum_{k=-1}^{n}\left\|\frac{\omega^k}{T(-\omega)}\Big[\mathcal{F}[f'](-\omega)+R(-\omega)\mathcal{F}[f'](\omega)\Big]\right\|^2_{PW\left(0,\frac{2n+1}{4M}\right)}.
    \end{split}
\end{align}
\subsection{Precise statement of results}\label{Section precise statement of results}
\subsubsection{Theorems on the question of rigidity of asymptotics}
We first state the theorems addressing rigidity of asymptotics:
\begin{thm}\label{exponential decay theorem'}
    For a given, fixed $\ell\in \mathbb{N}_{\geq0}$, let $\psi$ be a solution to \eqref{fixed ell mode wave equation}. If $\psi$ is such that 
    \begin{align}
        &\psi|_{\overline{\hp}}=0,\\
        &\|(\psi|_{t=0},\partial_t\psi|_{t=0})\|_{{}^{(\ell)}\mathcal{E}^{T,n}_{\overline{\Sigma}}}<\infty,
    \end{align}
    and $\psi$ satisfies \eqref{regularity assumption n intro} for some $n\in\mathbb{N}_{\geq0}$, then the radiation field of $\psi$ at $\Scrip$, $\psi_{\Scrip}$, satisfies
    \begin{align}\label{exponential decay result}
        e^{\frac{u}{4M}}\partial_u\left(e^{\frac{u}{2M}}\partial_u\right)^k\psi_{\Scrip}\in L^2(\mathbb{R})
    \end{align}
    for all $k=1,\dots,n$, and we have
    \begin{align}
        \sum_{k=0}^{n}\int_{-\infty}^{\infty} du\, e^{\frac{u}{2M}}(\partial_u(e^{\frac{u}{2M}}\partial_u)^k\psi_{\Scrip})^2\lesssim \sum_{k=0}^{n}\int_{-\infty}^{\infty} du\, |\partial_u^{k+1}\psi_{\Scrip}|^2+\sum_{k=0}^{n}\int_{-\infty}^{\infty} du\, e^{\frac{u}{2M}}(\partial_u(e^{\frac{u}{2M}}\partial_u)^k\psi_{\hm})^2.
    \end{align}
\end{thm}
\begin{thm}\label{construction of counterexample I}
    For any $n\in\mathbb{N}_{\geq0}$, there exists $\psi_{\hp}$ with $\partial_v^k \psi_{\hp}\in L^2(\hp)$ for $k=1,\dots,n+1$, satisfying
    \begin{align}
        \int_{v_0}^{\infty}\int_{S^2}dvd\mathrm{Vol}_{S^2}\,v^{2m}|\partial_v\psi_{\hp}|^2=\infty
    \end{align}
    for any $v_0$ and for $m>n$, such that there exists a solution $\psi$ to \eqref{wave equation psi intro} with
    \begin{align}
        \lim_{v\longrightarrow\infty}{\partial_u\psi}(u,v)=0\qquad \psi|_{\hp}=\psi_{\hp},
    \end{align}
    and which satisfies \eqref{regularity assumption n intro}.
\end{thm}
\begin{thm}\label{construction of counter example Scrip}
    For any $n\in\mathbb{N}_{\geq0}$, there exists $\psi_{\Scrip}$ with $\partial_u \psi_{\Scrip}\in L^2(\mathbb{R})$ and $m>2n$ with 
    \begin{align}
        \int_{u_0}^{\infty}\int_{S^2}du d\mathrm{Vol}_{S^2}\,u^{2m}|\psi_{\Scrip}|^2=\infty
    \end{align}
    for any $u_0$, such that there exists a solution $\psi$ to \eqref{wave equation psi intro} with
    \begin{align}
        \lim_{v\longrightarrow\infty}{\partial_u\psi}(u,v)=\partial_u\psi_{\Scrip}\qquad \psi|_{\hp}=0,
    \end{align}
    and which satisfies \eqref{regularity assumption n intro}.
\end{thm}
\subsubsection{Theorems on non-degenerate scattering theory}
We now give the precise statements of the results of this paper on the construction of non-degenerate scattering theories and on the question of existence.
\begin{thm}\label{non degenerate backwards scattering hp spherical symmetry}
     For $n\in\mathbb{N}_{\geq0}$, let ${}^{(\ell=0)}\mathscr{F}_+^{N,n}$ be the restriction of $\mathscr{F}^N_+$ to ${}^{(\ell=0)}\mathcal{E}^{N,n}_{\Sigma}$, the subspace of $\mathcal{E}^{N}_{\Sigma}$ whose elements are spherically symmetric. Then the image of ${}^{(\ell=0)}\mathscr{F}_+^{N,n}$ is ${{}^{(\ell=0)}\mathcal{E}_{\hp}^{N,n}}\oplus \spacescriexpon{\ell=0}$, and evolution under \eqref{fixed ell mode wave equation} with $\ell=0$ extends to the Hilbert space isomorphisms
        \begin{align*}
            &{}^{(\ell=0)}\mathscr{F}^N_+:{}^{(\ell=0)}\mathcal{E}_{\Sigma}^{N,n}\longrightarrow {{}^{(\ell=0)}\mathcal{E}_{\hp}^{N,n}}\oplus \spacescriexpon{\ell=0},\\& {}^{(\ell=0)}\mathscr{B}^N_+:{{}^{(\ell=0)}\mathcal{E}_{\hp}^{N,n}}\oplus \spacescriexpon{\ell=0}\longrightarrow {}^{(\ell=0)}\mathcal{E}_{\Sigma}^{N,n},
        \end{align*}
        such that 
        \begin{align*}
            &{}^{(\ell=0)}\mathscr{F}^{N,n}_+\circ {}^{(\ell=0)}\mathscr{B}^{N,n}_+=\mathrm{Id}_{{{}^{(\ell=0)}\mathcal{E}_{\hp}^{N,n}}\oplus \spacescriexpon{\ell=0}},\\&{}^{(\ell=0)}\mathscr{B}^{N,n}_+\circ {}^{(\ell=0)}\mathscr{F}^{N,n}_+=\mathrm{Id}_{{}^{(\ell=0)}\mathcal{E}_{\Sigma}^{N,n}}.
        \end{align*}
        An analogous statement holds with $\Sigma$ replaced by $\overline{\Sigma}$, $\hp$ replaced by $\overline{\hp}$, and we have the Hilbert space isomorphisms
        \begin{align*}
            &{}^{(\ell=0)}\mathscr{F}^N_+:{}^{(\ell=0)}\mathcal{E}_{\overline{\Sigma}}^{N,n}\longrightarrow {{}^{(\ell=0)}\mathcal{E}_{\overline{\hp}}^{N,n}}\oplus \spacescriexpn{\ell=0},\\& {}^{(\ell=0)}\mathscr{B}^N_+:{{}^{(\ell=0)}\mathcal{E}_{\overline{\hp}}^{N,n}}\oplus \spacescriexpn{\ell=0}\longrightarrow {}^{(\ell=0)}\mathcal{E}_{\overline{\Sigma}}^{N,n},
        \end{align*}
        satisfying
        \begin{align*}
            &{}^{(\ell=0)}\mathscr{F}^{N,n}_+\circ {}^{(\ell=0)}\mathscr{B}^{N,n}_+=\mathrm{Id}_{{{}^{(\ell=0)}\mathcal{E}_{\overline{\hp}}^{N,n}}\oplus \spacescriexpn{\ell=0}},\\&{}^{(\ell=0)}\mathscr{B}^{N,n}_+\circ {}^{(\ell=0)}\mathscr{F}^{N,n}_+=\mathrm{Id}_{{}^{(\ell=0)}\mathcal{E}_{\overline{\Sigma}}^{N,n}}.
        \end{align*}
\end{thm}
\begin{corollary}\label{corollary non degenerate backwards scattering hp spherical symmetry}
    A spherically symmetric radiation field $\psi_{\Scrip}$ with $\partial_u^k\psi_{\Scrip}\in\mathcal{E}^T_{\Scrip}$ for $k=1,\dots,n+1$ arises as the radiation field at $\Scrip$ of a solution to \eqref{fixed ell mode wave equation} with $\ell=0$ satisfying \eqref{regularity assumption n intro} if and only if
    \begin{align}
        \psi_{\Scrip}=\psi_{\Scrip}^{(1)}+\psi_{\Scrip}^{(2)},
    \end{align}
    where $\psi_{\Scrip}^{(2)}\in\spacescriexpn{\ell=0}$, and for $k=0,\dots,n$
    \begin{align}
        \omega^{k+1}(1+|\omega|^2)^{-\frac{1}{2}} e^{4M\pi|\omega|}\mathcal{F}[\partial_u\psi_{\Scrip}^{(1)}](\omega)\in L^2(\mathbb{R}).
    \end{align}
\end{corollary}
\begin{thm}\label{non degenerate backwards scattering hp}
     For $\ell\in\mathbb{N}_{>0}$ and $n\in\mathbb{N}_{\geq0}$, let ${}^{(\ell)}\mathscr{F}_+^{N,n}$ be the restriction of $\mathscr{F}^N_+$ to ${}^{(\ell)}\mathcal{E}^{N,n}_{\Sigma}$, the subspace of $\mathcal{E}^{N}_{\Sigma}$ whose elements are in the span of $\{Y^{\ell}_{m}\}_{m=-\ell}^{\ell}$. Assume that 
    \begin{align}\label{formal assumption on low frequency expansion of tilde R intro}
        \tilde{R}(\omega,\ell)=\widetilde{W}_{\ell}^{(\ell)}(\omega)+\widetilde{H}_{\ell}^{(\ell)}(\omega),
    \end{align}
     where $\widetilde{W}_{\ell}^{(\ell)}$ and $\omega^{-\ell}\widetilde{H}_{\ell}^{(\ell)}$ are square integrable on $\omega\in[-1,1]$. Then the image of ${}^{(\ell)}\mathcal{E}_{\Sigma}^{N,n}$ under ${}^{(\ell)}\mathscr{F}_+^{N,n}$ is ${{}^{(\ell)}\mathcal{E}_{\hp}^{N,n}}\oplus \spacescriexpon{\ell}$, and evolution under \eqref{fixed ell mode wave equation} extends to the Hilbert space isomorphisms
        \begin{align*}
            &{}^{(\ell)}\mathscr{F}^{N,n}_+:{}^{(\ell)}\mathcal{E}_{\Sigma}^{N,n}\longrightarrow {{}^{(\ell)}\mathcal{E}_{\hp}^{N,n}}\oplus \spacescriexpon{\ell},\\& {}^{(\ell)}\mathscr{B}^{N,n}_+:{{}^{(\ell)}\mathcal{E}_{\hp}^{N,n}}\oplus \spacescriexpon{\ell}\longrightarrow {}^{(\ell)}\mathcal{E}_{\Sigma}^{N,n},
        \end{align*}
        such that 
        \begin{align*}
            &{}^{(\ell)}\mathscr{F}^{N,n}_+\circ {}^{(\ell)}\mathscr{B}^{N,n}_+=\mathrm{Id}_{{{}^{(\ell)}\mathcal{E}_{\hp}^{N,n}}\oplus \spacescriexpon{\ell}},\\&{}^{(\ell)}\mathscr{B}^{N,n}_+\circ {}^{(\ell)}\mathscr{F}^{N,n}_+=\mathrm{Id}_{{}^{(\ell)}\mathcal{E}_{\Sigma}^{N,n}}.
        \end{align*}
        An analogous statement holds with $\Sigma$ replaced by $\overline{\Sigma}$, $\hp$ replaced by $\overline{\hp}$, under the assumption that there exist $\widetilde{W}_{\ell}^{(\ell+1)}$, $\widetilde{H}_{\ell}^{(\ell+1)}$, such that 
    \begin{align}\label{formal assumption on low frequency expansion of tilde R ell+1 intro}
        \tilde{R}(\omega,\ell)=\widetilde{W}_{\ell}^{(\ell+1)}(\omega)+\widetilde{H}_{\ell}^{(\ell+1)}(\omega),
    \end{align}
    where $\widetilde{W}_{\ell}^{(\ell)}$ and $\omega^{-\ell}\widetilde{H}_{\ell}^{(\ell)}$ are square integrable on $\omega\in[-1,1]$. We then have the Hilbert space isomorphisms
    \begin{align*}
            &{}^{(\ell)}\mathscr{F}^{N,n}_+:{}^{(\ell)}\mathcal{E}_{\overline{\Sigma}}^{N,n}\longrightarrow {{}^{(\ell)}\mathcal{E}_{\overline{\hp}}^{N,n}}\oplus \spacescriexpn{\ell},\\& {}^{(\ell)}\mathscr{B}^{N,n}_+:{{}^{(\ell)}\mathcal{E}_{\hp}^{N,n}}\oplus \spacescriexpn{\ell}\longrightarrow {}^{(\ell)}\mathcal{E}_{\overline{\Sigma}}^{N,n},
        \end{align*}
        satisfying
        \begin{align*}
            &{}^{(\ell)}\mathscr{F}^{N,n}_+\circ {}^{(\ell)}\mathscr{B}^{N,n}_+=\mathrm{Id}_{{{}^{(\ell)}\mathcal{E}_{\overline{\hp}}^{N,n}}\oplus \spacescriexpn{\ell}},\\&{}^{(\ell)}\mathscr{B}^{N,n}_+\circ {}^{(\ell)}\mathscr{F}^{N,n}_+=\mathrm{Id}_{{}^{(\ell)}\mathcal{E}_{\overline{\Sigma}}^{N,n}}.
        \end{align*}
     Analogous statements hold with $\hp$ replaced by $\hm$, or when $\hp$ is replaced by $\overline{\hm}$ and $\Sigma$ is replaced by $\overline{\Sigma}$.
\end{thm}
\begin{thm}\label{non degenerate backwards scattering Scrip}
    A function $\psi_{\Scrip}\in\mathcal{E}^T_{\Scrip}$ which is in the span of $\{Y^{\ell}_{m}\}_{m=-\ell}^{\ell}$ is the radiation field at $\Scrip$ of a solution to \eqref{fixed ell mode wave equation} arising from data in $\mathcal{E}^N_{{\Sigma}}$ if and only if 
    \begin{align}
       \psi_{\Scrip}\in \cosspacescripo+\sinespacescrip.
    \end{align}
    $\psi_{\Scrip}$ arises as the radiation field at $\Scrip$ in evolution from data in $\mathcal{E}^N_{{\overline{\Sigma}}}$ if and only if
    \begin{align}
       \psi_{\Scrip}\in \cosspacescripo+\sinespacescrip.
    \end{align}
    Analogous statements hold with $\Scrip$ replaced by $\Scrim$
\end{thm}
\begin{remark}
    Assumption \eqref{formal assumption on low frequency expansion of tilde R intro} is trivially satisfied by choosing $\widetilde{W}_{\ell}^{(\ell)}=(1-\omega^{\ell})\tilde{R}$ and $\widetilde{H}_{\ell}^{(\ell)}=\omega^{\ell}\tilde{R}$. However, we expect that $\tilde{R}$ possesses a polyhomogeneous asymptotic expansion as $\omega\longrightarrow0$. Such a polyhomogeneous expansion enables us via \eqref{formal assumption on low frequency expansion of tilde R intro} to obtain an explicit definition of $\spacehpn$ in terms of elementary functions in the weight $\widetilde{W}_{\ell}^{(\ell)}$. The same comment applies to $\widetilde{W}_{\ell}^{(\ell+1)}$, $\widetilde{H}_{\ell}^{(\ell+1)}$ in \eqref{formal assumption on low frequency expansion of tilde R ell+1 intro} 
\end{remark}
\begin{remark}
    We emphasise that all the maps defined in Theorems \ref{non degenerate backwards scattering hp spherical symmetry}, \ref{non degenerate backwards scattering hp}, \ref{non degenerate backwards scattering Scrip}, and in Corollary \ref{corollary non degenerate backwards scattering hp spherical symmetry} are \underline{not} unitary. We also emphasise that the continuity of said maps depends on constants that are $\ell$-dependent.
\end{remark}
\paragraph*{Acknowledgments} The author is grateful to Yakov Shlapentokh-Rothman for helpful discussions of the results of this paper.

\section{Preliminaries}

\subsection{The Schwarzschild spacetime}
For the coordinate system $(t,r,\theta,\phi)\in\mathbb{R}\times(2M,\infty)\times S^2$, the Schwarzschild metric reads
\begin{align}\label{Schwarzschild metric in Boyer Lindquist}
    g_{Schw}=-\left(1-\frac{2M}{r}\right)dt^2+\left(1-\frac{2M}{r}\right)^{-1}dr^2+r^2(d\theta^2+\sin^2\theta d\phi).
\end{align}
Define the tortoise coordinate $r_*$ by
\begin{align}\label{tortoise}
    r_*=r+2M\log\left(\frac{r}{2M}-1\right).
\end{align}
Let 
\begin{align}\label{EF DNG}
    u=\frac{1}{2}(t-r_*),\qquad v=\frac{1}{2}(t+r_*).
\end{align}
Then the wave equation reads
\begin{align}\label{wave equation}
    \pu\pv\psi-\frac{\Omega^2}{r^2}\mathring{\slashed{\Delta}}\psi+2M\frac{\Omega^2}{r^3}\psi=0.
\end{align}
where
\begin{align}
    \Omega^2:=1-\frac{2M}{r},
\end{align}
and $\mathring{\slashed{\Delta}}$ is the Laplacian operator associated with the standard metric of a round sphere of radius $r$.
We note that 
\begin{align}
    \Omega^2=\frac{1}{r}e^{\frac{r_*-r}{2M}}.
\end{align}
Thus the potential $2M\frac{\Omega^2}{r^3}$ decays exponentially in $u$ as $u\longrightarrow\infty$ for any fixed finite $v$, and only polynomially towards $\Scrip$ for any fixed $u$. Similarly, $2M\frac{\Omega^2}{r^3}$ decays exponentially in $v$ as $v\longrightarrow-\infty$ for any fixed finite $v$, and only polynomially towards $\Scrim$. 

If we separate a solution $\psi$ to \eqref{wave equation} into its spherical harmonic modes, 
\begin{align}
    \psi=\sum_{\ell=0}^{\infty}\sum_{m=-\ell}^{\ell}\psi_{\ell,m}(u,v)Y^m_{\ell}(\theta,\phi), 
\end{align}
then $\psi_{\ell}$ satisfies
\begin{align}\label{fixed ell mode wave equation 2}
    \pu\pv\psi_{\ell}+\ell(\ell+1)\frac{\Omega^2}{r^2}\psi_{\ell}+2M\frac{\Omega^2}{r^3}\psi_{\ell}=0.
\end{align}

The coordinate system $(t,r,\theta,\phi)$ is singular at $r=2M$. We can extend the Schwarzschild exterior to $r=2M$ using the Kruskal coordinate system $(U,V,\theta,\phi)$, where
\begin{align}
    U=-e^{-\frac{1}{2M}u},\qquad V=e^{\frac{1}{2M}v}.
\end{align}
Note that 
\begin{align}\label{relation of UV to r}
    -UV=\frac{r}{2M}e^{\frac{r}{2M}}\Omega^2.
\end{align}
In the coordinate system $(U,V,\theta,\phi)$, \eqref{Schwarzschild metric in Boyer Lindquist} reads
\begin{align}\label{Schwarzschild metric in Kruskal coordinates}
    -\frac{32M^3}{r}e^{-\frac{r}{2M}}dUdV+r(U,V)^2(d\theta^2+\sin^2\theta d\phi^2),
\end{align}
where $r(U,V)$ is implicitly defined via \eqref{relation of UV to r}. We also define
\begin{align}
    T:=V+U,\qquad R=V-U,
\end{align}
and we note that the submanifold $\{t=0\}$ in the coordinates $(R,\theta^A)$ extends analytically to $\overline{\Sigma}$, defined to be 
\begin{align}
    \overline{\Sigma}=\Sigma\cup\mathcal{B},
\end{align}
where $\mathcal{B}$ is the sphere $U=0,V=0$.

In the Kruskal coordinate system, \eqref{wave equation} reads
\begin{align}\label{wave equation in Kruskal coordinates}
    \partial_U\partial_V\psi+\frac{8M^3}{r^3}e^{-\frac{r}{2M}}\lapo\psi+\frac{16M^4}{r^4}e^{-\frac{r}{2M}}\psi=0.
\end{align}
We denote by $\mathscr{C}_{u^*}$ the outgoing null hypersurface $\{(u,v,\theta^A); u=u^*\}$. We denote by $\underline{\mathscr{C}}_{v^*}$ the ingoing null hypersurface $\{(u,v,\theta^A); v=v^*\}$.
\subsection{The Fourier transform}
Our sign convention for the Fourier transform is
\begin{align}
    \mathcal{F}[f](k)=\hat{f}(k)=\int_{-\infty}^\infty e^{ikx}f(x)dx,\qquad {f}(x)=\frac{1}{2\pi}\int_{-\infty}^\infty e^{-ikx}\hat{f}(k)dk=\mathcal{F}^{-1}[\hat{f}](x).
\end{align}
Under our sign convention, exponential decay as $x\longrightarrow\infty$ is equivalent by Paley--Wiener theory to the Fourier transform being holomorphic in an open horizontal strip below the real axis, whose width is determined by the rate of exponential decay, such that $\hat{f}$ is square integrable along horizontal lines in that strip, with the $L^2$ norms along said horizontal lines being uniformly bounded over the width of the strip. We formally state this fact below:
\begin{thm}\label{Paley Wiener L2}
    A function $f\in L^2(\mathbb{R})$ satisfies
    \begin{align}
        \| e^{\alpha x}f(x)\|_{L^2_x(\mathbb{R})}<\infty
    \end{align}
    if and only if the Fourier transform of $f$ is the non-tangential limit almost every where at $\mathbb{R}$ of a function $\hat{f}$ which is holomorphic in the strip $\Im\omega\in(-\alpha,0)$ and which satisfies
    \begin{align}
        \hat{f}(\cdot-i\epsilon)\in L^2(\mathbb{R})
    \end{align}
    for every $\epsilon\in(0,\alpha)$, with
    \begin{align}
        \sup_{\epsilon\in(0,\alpha)}\|\hat{f}(\cdot-i\epsilon)\|_{L^2(\mathbb{R})}<\infty.
    \end{align}
    Moreover, we have
    \begin{align}
        c\| (1+e^{\alpha x})f(x)\|_{L^2_x(\mathbb{R})}\leq \sup_{\epsilon\in(0,\alpha)}\|\hat{f}(\cdot-i\epsilon)\|_{L^2(\mathbb{R})}\leq C \| (1+e^{\alpha x})f(x)\|_{L^2_x(\mathbb{R})}
    \end{align}
    for $c,C>0$.
\end{thm}
On the basis of Theorem \ref{Paley Wiener L2}, we define a Hilbert space of functions modeled after Theorem \ref{Paley Wiener L2} as follows:
\begin{defin}\label{def of paley wiener space}
    For $a\in\mathbb{R}_{\geq0}$, a function $\hat{f}:\{\Im\omega\in[-a,0]\}\longrightarrow \mathbb{C}$ is said to belong to the Paley--Wiener space $PW(0,-a)$ if $\hat{f}$ is holomorphic for $\Im\omega\in(-a,0)$, attains its values at $\Im\omega=0$, $\Im\omega=-a$ as non-tangential boundary values almost everywhere, and satisfies $\hat{f}(\cdot-i\epsilon)\in L^2(\mathbb{R})$ for $\epsilon\in(0,a)$ and
    \begin{align}
        \sup_{\epsilon\in(0,a)}\|\hat{f}(\cdot-i\epsilon)\|_{L^2(\mathbb{R})}<\infty.
    \end{align}
    We equip $PW(0,-a)$ with the norm
    \begin{align}
        \|\hat{f}\|_{PW(0,-a)}^2= \sup_{\epsilon\in(0,a)}\|\hat{f}(\cdot-i\epsilon)\|_{L^2(\mathbb{R})}^2.
    \end{align}
    The space $PW(0,a)$ is analogously defined.
\end{defin}

\section{Stationary scattering theory}

\subsection{The Jost solutions}\label{section: Jost solutions}
If we formally decompose $\psi$ solving \eqref{wave equation} into its spherical harmonic components, and take the Fourier transform of the resulting equation, i.e.~equation \eqref{fixed ell mode wave equation 2}, in the variable $t$, we get the radial $1$-dimensional Schr\"odinger equation
\begin{align}\label{ode}
    u''(\omega,\ell,r_*)+(\omega^2-V_{\ell}(r_*))u(\omega,\ell,r_*)=0,
\end{align}
where
\begin{align}
    V_{\ell}=\frac{\ell(\ell+1)\Omega^2}{r^2}+\frac{2M\Omega^2}{r^3}.
\end{align}
\begin{proposition}
    For any $\omega\neq0$ with $\Im\omega\geq0$, there exists a unique solution $\out(\omega,\ell,r_*)$ to \eqref{ode} satisfying
    \begin{align}\label{defining condition of Jost out}
        \lim_{r_*\longrightarrow\infty}e^{-i\omega r_*}\out=1,
    \end{align}
     and a unique solution $\hor(\omega,\ell,r_*)$ to \eqref{ode} satisfying 
     \begin{align}\label{defining condition of Jost hor}
        \lim_{r_*\longrightarrow-\infty}e^{i\omega r_*}\hor=1.
    \end{align}
\end{proposition}
\begin{proof}
    This is a standard result in scattering theory on the line for Schr\"odinger's equation, and was proven in the Schwarzschild context in \cite{BMB}.
\end{proof}
\begin{remark}
The Jost functions $\out,\hor$ can be constructed via the Volterra integral equation
\begin{align}\label{def of U hor}
    e^{i\omega r_*}\hor(\omega,\ell,r_*)=1-\int_{-\infty}^{r_*} e^{i\omega (r_*-y)}\frac{\sin \omega(r_*-y)}{\omega}V_{\ell}(y) e^{i\omega y}\hor(\omega,\ell,y)dy.
\end{align}
\begin{align}\label{def of U inf}
    e^{-i\omega r_*}\out(\omega,\ell,r_*)=1+\int_{r_*}^{\infty} e^{-i\omega (r_*-y)}\frac{\sin \omega(r_*-y)}{\omega}V_{\ell}(y) e^{-i\omega y}\out(\omega,\ell,y)dy.
\end{align}
\end{remark}

Both $\hor$ and $\out$ can be extended meromorphically to the complex plane according to the following theorem of Bachelot--Motet-Bachelot (\cite{BMB}, Theorem IV.1):
\begin{theorem*}[Bachelot--Motet-Bachelot '93]\label{Bachelot thm}
    For any fixed $\ell$, $r_*\in\mathbb{R}$, $\hor(\omega,\ell,r_*)$ is analytic in $\{\omega\in\mathbb{C},\omega\neq \lambda i,\lambda\in(-\infty,-\frac{1}{4M}]$ and $\out(\omega,\ell,r_*)$ is analytic in $\{\omega\in\mathbb{C},\omega\neq\lambda i,\lambda\in(-\infty,0]\}$.
\end{theorem*}
We now give an argument showing that for any fixed $\omega$ with $\{\omega\in\mathbb{C}, \Im\omega<0,\Re\omega\neq0 \}$,  the analytic extension of $\out(\omega,\ell,r_*)$, viewed as a function of $r_*\in \mathbb{R}$ is a solution to \eqref{ode} satisfying \eqref{defining condition of Jost out}.
\begin{proposition}\label{analytic extension of r}
    For $r_*\in\mathbb{R}$, let $r(y)$ denote $r$ such that $y=r+2M\log(\frac{r}{2M}-1)$. Then $r$ is a continuous function of $y$ and admits a holomorphic extension to $\{y\in\mathbb{C};|\Re y|>A\}$ for any $A>0$ and satisfies
    \begin{align}
        r(y)=2M+e^{\frac{y}{2M}-1}+o(e^{\frac{y}{2M}}),\qquad y\longrightarrow-\infty,
    \end{align}
    \begin{align}
        r(y)=y+o(y),\qquad y\longrightarrow-\infty.
    \end{align}
\end{proposition}
\begin{proof}
    This is Proposition IV.2 of \cite{BMB}.
\end{proof}
\begin{proposition}
    Fix $\omega\in\{\omega\in\mathbb{C}, \Im\omega<0,\Re\omega\neq0 \}$, and for each $r_*$, let $f=f(\omega,\ell,r_*)$ be the value of the analytic extension of $\out(\omega,\ell,r_*)$ at $r_*$ to $\{\omega\in\mathbb{C}, \Im\omega<0,\Re\omega\neq0 \}$ at $\omega$. Then $f$ satisfies \eqref{ode} and \eqref{defining condition of Jost out}. An analogous statement applies to $\hor$.
 \end{proposition}

\begin{proof}
    This is a simple corollary of the argument leading to Proposition IV.4 of \cite{BMB}, and we provide a sketch of the argument here in the case of $\out$. Let $\theta\in(-\frac{\pi}{2},\frac{\pi}{2})$, then for $x\in\mathbb{C}$ we define the contour $\Gamma_{\theta}(x)$ to be the path $x\rightarrow(\Re x,0)\rightarrow (B,0)\rightarrow B+\infty e^{i\theta}$. For $\omega$ as in the hypothesis with $\arg\omega\in[-\theta,\pi-\theta]$, take $W(x,y,\omega)=(2i\omega)^{-1}(1-e^{2i\omega(y-x)})V_{\ell}(y)$, and note that $\int_{\Gamma_{\theta}(x)}|V_{\ell}(y)|dy<\infty$. Define $m_0(x,\omega;\theta):=1$ and 
    \begin{align}\label{contour thing}
        m_{n+1}(x,\omega;\theta):=-\int_{\Gamma_{\theta}(x)}W(x,y,\omega)m_n(y,\omega;\theta).
    \end{align}
    Then taking $m_+(x,\omega;\theta):=\sum_{n=0}^{\infty}m_n(x,\omega;\theta)$, we note
    \begin{align}
        |\partial_x^jm_n(x,\omega,\theta)|\leq \frac{1}{n!}\left(2^j|\omega|^{j-1}\int_{\Gamma_{\theta}(x)}|V_{\ell}|(y)dy\right)^n.
    \end{align}
    Then $m_+$ satisfies the integral equation
    \begin{align}
        m_+(x,\omega;\theta)=1-\int_{\Gamma_{\theta}(x)}dy\, W(x,y,\omega)m_+(y,\omega;\theta).
    \end{align}
    By Lebesgue's dominated convergence theorem, $m_+$ also satisfies
    \begin{align}
        \partial_x^2m_++2i\omega \partial_xm_+-V_{\ell}m_+=0.
    \end{align}
    Thus $f(\omega,\ell,x):=m_+(x,\omega;\theta)e^{i\omega x}$ satisfies \eqref{ode} for $\omega$ with $\arg\omega\in[-\theta,\pi-\theta]$. To show that $m_+\longrightarrow1$ as $x\longrightarrow\infty$ with $\Im x=0$, we apply Cauchy's theorem to deform the contour $\Gamma_{\theta}(x)$ to the contour $\Xi_{\theta}(x)$ given by $(x,0)\rightarrow (x,(x-B)\tan\theta)\rightarrow B+\infty e^{i\theta}$. A simple induction argument shows that $|m_n|\leq \frac{1}{n!}\left(|\omega|^{-1}\int_{\Xi_{\theta}(x)}|V|(y)dy\right)^n$. We estimate using Proposition \ref{analytic extension of r}
    \begin{align}
    \begin{split}
        \int_{0}^{(x-B)\tan\theta}dy \,|V_{\ell}(y)|&\leq \sup_{y\in[0,(x-B)\tan\theta]}\sqrt{x^2+y^2}|V_{\ell}(x+iy)|\lesssim \frac{\ell(\ell+1)}{|x|}
    \end{split}
    \end{align}
   Therefore,
    \begin{align}
        |m_+-1|\lesssim \exp\left\{\frac{\ell(\ell+1)+1}{|x|}+\int_{(x-B)\tan\theta}^{\infty}dy\,|V_{\ell}(B+ye^{i\theta})|\right\}-1,
    \end{align}
    and we obtain by 
    \begin{align}
        \lim_{x\longrightarrow\infty} m_+(x,\omega,\theta)=1.
    \end{align}
\end{proof}

\begin{defin}
    Define 
    \begin{align}
        \horb(\omega,\ell,r_*):=\hor(-\omega,\ell,r_*),\qquad \outb(\omega,\ell,r_*)=\out(-\omega,\ell,r_*).
    \end{align}
\end{defin}
\begin{remark}
    Note that for $\omega\in\mathbb{R}$, $\horb$ and $\outb$ are the complex conjugates of $\hor$, $\out$. This is \emph{not} the case however when $\Im\omega\neq0$.
\end{remark}

\begin{defin}
    The Wronskian of two function $f,g$ of $r_*$ is defined by
    \begin{align}
        W(f,g):=fg'-f'g.
    \end{align}
\end{defin}
\begin{remark}
    We can find the Wronskian of $\hor, \horb$ by considering the limit $r_*\longrightarrow-\infty$ to get 
    \begin{align}
        W(\hor,\horb)=2i\omega.
    \end{align}
    Similarly, the Wronskian of $\out, \outb$ is obtained by considering the limit \mbox{$r_*\longrightarrow\infty$}
    \begin{align}
        W(\out,\outb)=-2i\omega.
    \end{align}
    Thus $\hor,\horb$ are linearly independent for $\omega\neq0$, and the same applies to $\out,\outb$.
\end{remark}

\subsection{The scattering matrix}\label{section: definition of reflection and transmission coefficients}
We state some known facts about $\hor,\horb,\out,\outb$
\begin{defin}\label{definition of reflection and transmission}
    Define the reflection coefficients $R(\omega,\ell),\tilde{R}(\omega,\ell)$, $T(\omega,\ell),\tilde{T}(\omega,\ell)$ by
    \begin{align}\label{def:TR}
        \outb(\omega,\ell,r_*)=T(\omega,\ell)\hor(\omega,\ell,r_*)+R(\omega,\ell)\out(\omega,\ell,r_*),
    \end{align}
    \begin{align}\label{def:TR'}
        \horb(\omega,\ell,r_*)=\tilde{T}(\omega,\ell)\out(\omega,\ell,r_*)+\tilde{R}(\omega,\ell)\hor(\omega,\ell,r_*).
    \end{align}
\end{defin}
We may express $T,R,\ttt,\trr$ in terms of the Wronskians of $\out,\hor,\outb,\horb$ as follows: the definition \eqref{def:TR} implies
\begin{align}
    \begin{pmatrix}
        \out&\hor\\\out'&\hor'
    \end{pmatrix}
    \begin{pmatrix}
        \ttt\\\trr
    \end{pmatrix}
    =
    \begin{pmatrix}
        \horb\\\horb'
    \end{pmatrix},
\end{align}
which can be inverted to give
\begin{align}\label{relation Tt Rt to Wronskians}
\begin{split}
    \begin{pmatrix}
        \ttt\\\trr
    \end{pmatrix}
    &=\frac{1}{\mathfrak{W}(\out,\hor)}\begin{pmatrix}
        \mathfrak{W}(\horb,\hor)\\\mathfrak{W}(\out,\horb)
    \end{pmatrix}\\&
    =\frac{1}{\mathfrak{W}(\out,\hor)}
    \begin{pmatrix}
        -2i\omega\\\mathfrak{W}(\out,\horb)
    \end{pmatrix}.
\end{split}
\end{align}
Similarly, we have
\begin{align}\label{relation T R to Wronskians}
\begin{split}
    \begin{pmatrix}
        T\\R
    \end{pmatrix}
    &=\frac{1}{\mathfrak{W}(\hor,\out)}
    \begin{pmatrix}
        \mathfrak{W}(\outb,\out)\\\mathfrak{W}(\hor,\outb)
    \end{pmatrix}
    \\&=\frac{1}{\mathfrak{W}(\hor,\out)}
    \begin{pmatrix}
        2i\omega\\\mathfrak{W}(\hor,\outb)
    \end{pmatrix}.
\end{split}
\end{align}
We also record
    \begin{align}\label{def:TR2}
        \out(\omega,\ell,r_*)=T(-\omega,\ell)\horb(\omega,\ell,r_*)+R(-\omega,\ell)\outb(\omega,\ell,r_*),
    \end{align}
    \begin{align}\label{def:TR'2}
        \hor(\omega,\ell,r_*)=\tilde{T}(-\omega,\ell)\outb(\omega,\ell,r_*)+\tilde{R}(-\omega,\ell)\horb(\omega,\ell,r_*).
    \end{align}
\begin{align}\label{relation Tt Rt minus to Wronskians}
    \begin{pmatrix}
        \ttt(-\omega)\\\trr(-\omega)
    \end{pmatrix}
    =\frac{1}{\mathfrak{W}(\outb,\horb)}
    \begin{pmatrix}
        2i\omega\\\mathfrak{W}(\outb,\hor)
    \end{pmatrix},
\end{align}
\begin{align}\label{relation T R minus to Wronskians}
    \begin{pmatrix}
        T(-\omega)\\R(-\omega)
    \end{pmatrix}
    =\frac{1}{\mathfrak{W}(\horb,\outb)}
    \begin{pmatrix}
        -2i\omega\\\mathfrak{W}(\horb,\out)
    \end{pmatrix}.
\end{align}
\begin{remark}\label{equality of transmission coefficients}
    Note in particular that $T(\omega)=\tilde{T}(\omega)$ for all $\omega\in \mathbb{C}$. 
\end{remark}
\begin{remark}
In what follows, we may refer to $T(-\omega)$ by $\overline{T}(\omega)$, and similarly for the remaining coefficients of the scattering matrix.
\end{remark}
\begin{lemma}\label{lem: matrix element identities at R}
    For $\omega\in\mathbb{R}$, we have the following:
    \begin{align}\label{unitarity for real frequency}
        |T|^2+|R|^2=|{T}|^2+|\tilde{R}|^2=T^2-\tilde{R}R=1,
    \end{align}
    \begin{align}
        \tilde{R}(-\omega,\ell)T(\omega,\ell)+T(-\omega,\ell)R(\omega,\ell)=0.
    \end{align}
\end{lemma}
\begin{proof}
    For the second identity we can combine \eqref{def:TR}, \eqref{def:TR'} with \eqref{def:TR2}, \eqref{def:TR'2}. For the first identity, we use
    \begin{align}
        \mathfrak{W}(a,b)\mathfrak{W}(c,d)=\mathfrak{W}(a,c)\mathfrak{W}(b,d)-\mathfrak{W}(a,d)\mathfrak{W}(b,c).
    \end{align}
\end{proof}
Finally, note that by Theorem \ref{Bachelot thm} of \cite{BMB} we have 
\begin{corollary}\label{analyticity of the wronskians}
    $\mathfrak{W}(\hor,\outb)$ is holomorphic in the strip $\{\omega:\Im\omega\in(-\frac{1}{4M},0)\}$. Similarly, $\mathfrak{W}(\outb,\horb)$ is holomorphic for $\Im\omega<0$. In particular, $T(-\omega)$ is holomorphic for $\Im\omega<0$, and $\tilde{R}(-\omega)$ is holomorphic for $\{\omega:\Im\omega\in(-\frac{1}{4M},0)\}$.
\end{corollary}

\begin{proposition}\label{nonvanishing of T}
    For $\omega\neq0$ with $\Im\omega\geq0$, $\mathfrak{W}(\hor,\out)(\omega)$ is non-zero.
\end{proposition}

\begin{proof}
    See for instance \cite{Whiting}. See also \cite{SR15}.
\end{proof}

\begin{corollary}\label{analyticity of the matrix elements}
    The coefficient $T(\omega)$ is holomorphic for $\Im\omega>0$ and $ \Re\omega\neq0, \Im\omega=0$. The reflection coefficient $R$ is not holomorphic for $\Re\omega=0$, while the reflection coefficient $\tilde{R}$ is holomorphic for \mbox{$\{\omega\in \mathbb{C}; \Im\omega>0, \omega\notin i[\frac{1}{4M},\infty) \}$}, and for $\omega\in \mathbb{R}, \omega\neq0$.
\end{corollary}
\subsection{Formulae for the reflection and transmission coefficients via variation of parameters}

We may express the scattering matrix in terms of integral operators acting on $Vu$, where $u$ solves \eqref{ode}. If we take $u=T\hor=\outb-R\out$, then $u$ satisfies
\begin{alignat}{2}
    u&\sim  e^{-i\omega r_*}+Re^{i\omega r_*},\qquad &&r_*\longrightarrow\infty,\\
    u&\sim T e^{-i\omega r_*},\qquad &&r_*\longrightarrow-\infty.
\end{alignat}
Let $\omega\in\mathbb{R}, \omega\neq0$. Variation of parameters on \eqref{ode} then implies
\begin{align}
    u=e^{-i\omega r_*}+\frac{1}{2i\omega}e^{i\omega r_*}\int_{-\infty}^{r_*}dy e^{-i\omega y}V_{\ell}(y)u(y)+\frac{1}{2i\omega}e^{-i\omega r_*}\int_{r_*}^\infty dy e^{i\omega y}V_{\ell}(y)u(y),
\end{align}
We see that
\begin{align}
    T(\omega,\ell)=1+\frac{1}{2i\omega}\int_{-\infty}^\infty dy e^{i\omega y}V_{\ell}u=1+\frac{1}{2i\omega}\int_{-\infty}^\infty dy e^{i\omega y}V_{\ell}(T\hor),
\end{align}
and conclude
\begin{align}\label{Volterra definition of T}
    T(\omega,\ell)=\frac{1}{1-\frac{1}{2i\omega}\int_{-\infty}^\infty dy e^{i\omega y}V_{\ell}(y)\hor(\omega,\ell,y)}.
\end{align}
\begin{align}\label{Volterra definition of R}
\begin{split}
    R(\omega,\ell)&=\frac{1}{2i\omega}\int_{-\infty}^\infty dy e^{-i\omega y}V_{\ell}u\\&=\frac{\frac{1}{2i\omega}\int_{-\infty}^\infty dy e^{-i\omega y}V_{\ell}(y)\hor(\omega,\ell,y)}{{1-\frac{1}{2i\omega}\int_{-\infty}^\infty dy e^{i\omega y}V_{\ell}(y)\hor(\omega,\ell,y)}}.
\end{split}
\end{align}
Similarly, we find
\begin{align}\label{Volterra definition of tilde T}
    T(\omega,\ell)=\frac{1}{1-\frac{1}{2i\omega}\int_{-\infty}^\infty dy e^{-i\omega y}V_{\ell}(y)\out(\omega,\ell,y)},
\end{align}
\begin{align}\label{Volterra definition of tilde R}
    \trr(\omega,\ell)=\frac{\frac{1}{2i\omega}\int_{-\infty}^\infty dy e^{i\omega y}V_{\ell}(y)\out(\omega,\ell,y)}{{1-\frac{1}{2i\omega}\int_{-\infty}^\infty dy e^{-i\omega y}V_{\ell}(y)\out(\omega,\ell,y)}}.
\end{align}
Flipping the sign of $t$, we get
\begin{align}
    \overline{T}(\omega,\ell)=\frac{1}{1+\frac{1}{2i\omega}\int_{-\infty}^\infty dy e^{-i\omega y}V_{\ell}(y)\horb(\omega,\ell,y)},
\end{align}
\begin{align}
    \overline{R}(\omega,\ell)=-\frac{\frac{1}{2i\omega}\int_{-\infty}^\infty dy e^{i\omega y}V_{\ell}(y)\horb(\omega,\ell,y)}{{1+\frac{1}{2i\omega}\int_{-\infty}^\infty dy e^{-i\omega y}V_{\ell}(y)\horb(\omega,\ell,y)}},
\end{align}
\begin{align}\label{variation of parameters ttt}
    \overline{T}(\omega,\ell)=\frac{1}{1+\frac{1}{2i\omega}\int_{-\infty}^\infty dy e^{i\omega y}V_{\ell}(y)\outb(\omega,\ell,y)},
\end{align}
\begin{align}\label{variation of parameters trr}
    \overline{\trr}(\omega,\ell)=-\frac{\frac{1}{2i\omega}\int_{-\infty}^\infty dy e^{-i\omega y}V_{\ell}(y)\outb(\omega,\ell,y)}{{1+\frac{1}{2i\omega}\int_{-\infty}^\infty dy e^{i\omega y}V_{\ell}(y)\outb(\omega,\ell,y)}}.
\end{align}
We can extend the above relations beyond $\omega\in \mathbb{R}$. First, we record the following crude estimates:

\begin{lemma}\label{crude boundedness of horb}
For $\Im\omega\in [0,\frac{1}{2})$ and $\theta\neq \pi$ we have for $\mathrm{ph}\, r=\theta$,
    \begin{align}
        \begin{split}
            \sup_{|r|\in[0,\infty)}|e^{i\omega r_*}\hor-1|\lesssim e^{\frac{\ell(\ell+1)+1}{2|\omega|}}\left(\frac{\ell(\ell+1)+1}{2|\omega|}\right).
        \end{split}
    \end{align}
    For $\Im\omega\geq0$, we have 
    \begin{align}
        \begin{split}
            \sup_{|r|\in[0,\infty)}|e^{-i\omega r_*}\out-1|\lesssim e^{\frac{\ell(\ell+1)+1}{2|\omega|}}\left(\frac{\ell(\ell+1)+1}{2|\omega|}\right).
        \end{split}
    \end{align}
\end{lemma}
\begin{proof}
    We apply Gr\"onwall's inequality to the formulae for $e^{i\omega r_*}\hor-1$, $e^{-i\omega r_*}\out-1$ given by variation of parameters \eqref{def of U hor}, \eqref{def of U inf}.
\end{proof}
\begin{corollary}
    Formulae \eqref{Volterra definition of T}, \eqref{Volterra definition of tilde T} hold for $\Im\omega\geq0$, $\omega\neq0$. Formula \eqref{Volterra definition of tilde R} holds for all $\omega\neq0$ with $\Im\omega\in[0,\frac{1}{4M})$.
\end{corollary}
\begin{proof}
    Using Lemma \ref{crude boundedness of horb}, we can apply 
    Lebesgue's dominated convergence theorem to conclude that each of $\int_{-\infty}^\infty dy e^{-i\omega y}V_{\ell}(y)\out(\omega,\ell,y)$, $\int_{-\infty}^\infty dy e^{-i\omega y}V_{\ell}(y)\out(\omega,\ell,y)$ has a holomorphic extension to $\Im\omega>0$. Applying the same argument to $\int_{-\infty}^\infty dy e^{i\omega y}V_{\ell}(y)\out(\omega,\ell,y)$, we have that $\tilde{R}$ is meromorphic for as long as $e^{2i\omega r_*}V_{\ell}$ is integrable, which is the case whenever $\Im\omega\in[0,\frac{1}{4M})$.
\end{proof}

\subsection{The radiation fields in frequency space}
\begin{defin}\label{def: microlocal radiation fields}
Let $\psi$ be a finite energy solution to \eqref{fixed ell mode wave equation 2} with radiation fields $\psi_{\hm}$, $\psi_{\hp}$, $\psi_{\Scrim}$, $\psi_{\Scrip}$ on $\hm$, $\hp$, $\Scrim$, $\Scrip$. Define the \emph{fixed frequency radiation fields} by
\begin{align}
    \omega a_{\hp}(\omega,\ell,m):=i\int_{-\infty}^\infty\ddv e^{i\omega\bar{v}} \partial_v\psi_{\hp}^{(\ell,m)}(\bar{v}),\qquad \omega a_{\hm}(\omega,\ell,m):=i\int_{-\infty}^\infty\ddu e^{i\omega\bar{u}} \partial_u\psi_{\hm}^{(\ell,m)}(\bar{u}),
\end{align}
\begin{align}
    \omega a_{\Scrip}(\omega,\ell,m):=i\int_{-\infty}^\infty\ddu\, e^{i\omega\bar{u}} \partial_u\psi_{\Scrip}^{(\ell,m)}(\bar{u}),\qquad \omega a_{\Scrim}(\omega,\ell,m):=i\int_{-\infty}^\infty\ddv\, e^{i\omega\bar{v}} \partial_v\psi_{\Scrim}^{(\ell,m)}(\bar{v}),
\end{align}
where $f^{(\ell,m)}$ refers to the $(\ell,m)$-spherical harmonic component of $f$.
\end{defin}
In what follows, we omit the azimuthal number $m$. The following was proven in \cite{DRSR14}:
\begin{proposition}\label{scattering matrix relations among fixed frequency radiation fields}
    For a solution $\psi$ to \eqref{wave equation} arising via $\mathscr{B}^T_{-}\circ \mathscr{F}^T_{+}$ from scattering data $(\psi_{\hm},\psi_{\Scrim})\in \mathcal{E}^T_{\hm}\oplus \mathcal{E}^T_{\Scrim}$, the radiation fields $\psi_{\hp}$, $\psi_{\Scrip}$ satisfy
    \begin{align}\label{relation microlocal hp to hm scrim}
         a_{\hp}=-\trr a_{\hm}+T a_{\Scrim},
    \end{align}
    \begin{align}\label{relation microlocal scrip to hm scrim}
        a_{\Scrip}=Ta_{\hm}-Ra_{\Scrim}.
    \end{align}
    Similarly, for a solution $\psi$ to \eqref{wave equation} arising via $\mathscr{B}^T_{+}\circ \mathscr{F}^T_{-}$ from scattering data $(\psi_{\hp},\psi_{\Scrip})\in \mathcal{E}^T_{\hp}\oplus \mathcal{E}^T_{\Scrip}$, the radiation fields $\psi_{\hm}$, $\psi_{\Scrim}$ satisfy
    \begin{align}\label{relation microlocal hm to hp scrip}
      a_{\hm}=-\tilde{R}(-\omega)a_{\hp}+T(-\omega)a_{\Scrip},
    \end{align}
    \begin{align}\label{relation microlocal scrim to hp scrip}
      a_{\Scrim}={T}(-\omega)a_{\hp}-R(-\omega)a_{\Scrip}.
    \end{align}
\end{proposition}

\section{Analytic and asymptotic properties of reflection and transmission coefficients}
\begin{remark}\label{heads up 2M=1}
    To reduce clutter, we set 
\begin{align*}
    2M=1
\end{align*}
in the remainder of the paper.
\end{remark}

\subsection{High frequency asymptotics of the reflection coefficient ${\trr}$ for fixed $\ell$}

In this section we establish the following for fixed $\ell$:
\begin{enumerate}
    \item $T$ asymptotes to $1$ as $|\Re\omega|\longrightarrow\infty$ for any $\omega$ with $\Im\omega\geq0$.
    \item ${\trr}$ decays exponentially like $e^{-2\pi|\omega|}$ as $|\Re\omega|\longrightarrow\infty$ for $\Im\omega=0$, and for $\omega$ with fixed imaginary part $\Im\omega\in[0,\frac{1}{2})$ it decays faster than $e^{-2\pi(1-\delta)|\omega|}$ as $|\Re\omega|\longrightarrow\infty$ for any $\delta>0$.
\end{enumerate}
We begin with item $2$ in what follows. Before embarking on proving Proposition \ref{holy grail of transmission}, we quote the following theorem of Olver, which contains the statements of Theorems 10.1 and 10.2, Chapter 6 of \cite{Olver}:
\begin{thm}\label{Olver's error approx theorem}
    For the integral equation
    \begin{align}\label{Olvers error integral equation general form}
        \varepsilon(\xi)=\int_{a}^{\xi}\mathrm{K}(\xi,v)[\upphi(v)J(v)+\uppsi_0(v)\varepsilon(v)+\uppsi_1(v)\varepsilon'(v)] dv,
    \end{align}
    with $a, \xi \in \mathbb{R}\cup\{-\infty,\infty\}$ and $v\in[a,\xi]$. Assume that $J$, $\upphi$, $\uppsi_0$, $\uppsi_1$ are continuous in $(a,\xi)$, except for a finite number of discontinuities and infinities. Assume that $\mathrm{K}$ and its first and second derivatives with respect to either $\xi$ or $v$ are continuous whenever $\xi, v\in(a,\xi)$, and satisfies $\mathrm{K}(\xi,\xi)=0$, and 
    \begin{align}
        |K(\xi,v)|&\leq P_0(\xi)Q(v),\\
        \left|{\partial_{\xi}\mathrm{K}(\xi,v)}\right|&\leq P_1(\xi)Q(v),\\
        \left|{\partial^2_{\xi}\mathrm{K}(\xi,v)}\right|&\leq P_2(\xi)Q(v),
    \end{align}
    with $P_0$, $P_1$, $P_2$, $Q$ continuous real-valued functions and $P_0$, $P_1$, $P_2$ all positive. If for $\xi\in(a,b)$ the following integrals converge
    \begin{align}
        \Phi(\xi)=\int_{a}^{\xi}|\upphi(v)|dv,\qquad \Psi_0(\xi)=\int_{a}^{\xi}|\uppsi_0(v)|dv,\qquad \Psi_1(\xi)=\int_{a}^{\xi}|\uppsi_1(v)|dv,
    \end{align}
    and the following suprema are finite
    \begin{align}
        \kappa=\sup_{\xi\in(a,b)} Q(\xi)|J(\xi)|,\qquad \kappa_0=\sup_{\xi\in(a,b)} P_0(\xi)Q(\xi),\qquad \kappa_1=\sup_{\xi\in(a,b)} P_1(\xi)Q(\xi),
    \end{align}
    then there exists a unique solution $\varepsilon$ to \eqref{Olvers error integral equation general form} which is continuously differentiable in $(a,b)$ and satisfies
    \begin{align}
        \varepsilon(\xi)/P_0(\xi),\;\varepsilon'(\xi)/P_1(\xi)\longrightarrow0
    \end{align}
    as $\xi\longrightarrow a$. Moreover,
    \begin{align}
        \frac{|\varepsilon(\xi)|}{P_0(\xi)},  \frac{|\varepsilon'(\xi)|}{P_1(\xi)}\leq \kappa \Phi(\xi)\exp\{\kappa_0\Psi_0+\kappa_1\Psi_1\},
    \end{align}
    and $\varepsilon''$ is continuous except at the discontinuities of $\upphi$, $J$, $\uppsi_0$, $\uppsi_1$.
    
    If we have, in addition to the above conditions, that
    \begin{align}
        \upphi=\uppsi_0,\qquad \uppsi_1=0,
    \end{align}
    then $\varepsilon$ satisfies
    \begin{align}
        \frac{|\varepsilon(\xi)|}{P_0(\xi)},  \frac{|\varepsilon'(\xi)|}{P_1(\xi)}\leq \frac{\kappa}{\kappa_0} \left[\exp\{\kappa_0\Phi(\xi)\}-1\right].
    \end{align}
\end{thm}

\begin{lemma}\label{hint at a pole of R tilde}
    For $m\in [1,\infty)$, the function $(m+2i\omega)\tilde{R}(\omega,\ell)$ extends holomorphically to a neighbourhood of $\omega=\frac{im}{2}$.
\end{lemma}
\begin{proof}
    For any $m\in \mathbb{R}_{>0}$, we have by Corollary \ref{analyticity of the matrix elements} that $\tilde{R}$ is holomorphic at $x\in\mathcal{D}_{\delta}(m):=\mathcal{N}_{\delta}(\frac{im}{2})\setminus i[\frac{1}{2},\infty)$ for any $\delta>0$. Pick $x\in\mathcal{D}_{\delta}(m)$. For $\overline{G}_{\mathrm{hor}}=e^{-i\omega r_*}\horb$. Then for $\omega\notin i[\frac{1}{2},\infty)$ we have 
\begin{align}
    \left(1-\frac{1}{r}\right)\partial_r^2 \overline{G}_{\mathrm{hor}}+\left(\frac{1}{r^2}+2i\omega\right)\partial_r \overline{G}_{\mathrm{hor}}+\left(\frac{\ell(\ell+1)}{r^2}+\frac{1}{r^2}\right)\overline{G}_{\mathrm{hor}}=0.
\end{align}
Near $r=1$ we may write $\ghorb$ in terms of a Taylor series centered at $r=1$. Let
\begin{align}
    a_n=(m+2i\omega)(-1)^n\frac{1}{n!}\partial_r^n \ghorb|_{r=1}.
\end{align}
We have the recurrence relation for $n\geq1$
\begin{align}\label{recurrent relation ghorb at horizon}
\begin{split}
    a_{n+1}=\frac{1}{(n+1)(n+1+2i\omega)}\Bigg[&\sum_{k=0}^{n}[L(n-k+1)+(n-k+2)(n-k+1)]a_{k}\\&-\sum_{k=1}^{n}k(n-k+2)a_k+\sum_{k=2}^{n}k(k-1)a_k\Bigg].
\end{split}
\end{align}
It is clear by \eqref{recurrent relation ghorb at horizon} that $a_{i}$ are all functions of $\omega$ which are analytic in $\mathcal{N}_{\delta}(\frac{im}{2})$. Assume for an induction hypothesis that $|a_{j}|<(C(L+2))^j$ for $i\leq n$ and $n\geq m$, then \eqref{recurrent relation ghorb at horizon} implies
\begin{align}
\begin{split}
    |a_{n+1}|\leq&\, L\frac{[C(L+2)]^{n+1}-1}{C(L+2)-1} +\frac{5}{2}\left[\frac{n}{n-m+1+m+2i\omega}\right]L\frac{[C(L+2)]^{n+1}-1}{C(L+2)-1}\\\leq&\,\frac{L+3(m+1)}{2C-1}[C(L+2)]^{n+1}.
\end{split}
\end{align}
Choosing $C$ large enough closes the induction step. It is easy to show that for $n<m$ there exist $C$ such that $|a_{j}|\leq (C(L+2))^{j}$. We have shown that for $r-1$ small enough, the functions $(m+2i\omega)\ghorb|_{r}$, $(m+2i\omega)\partial_r\ghorb|_{r}$ are holomorphic in $\mathcal{N}_{\delta}(\frac{im}{2})$, which implies that $(m+2i\omega)\tilde{R}$ is also holomorphic in $\mathcal{N}_{\delta}(\frac{im}{2})$. The argument above also shows that $\tilde{R}$ also extends holomorphically to points on the positive imaginary line bounded away from the points $\frac{im}{2}$, $m\in\mathbb{N}, m\geq1$. 
\end{proof}

\begin{proposition}\label{high frequency behaviour of reflection coefficient proposition 2}
    For fixed $\ell$, the reflection coefficient ${\tilde{R}}$ satisfies
    \begin{align}\label{exact high frequency behaviour of reflection coefficient}
        {\trr}(\omega,\ell)\sim e^{-2\pi |\omega|}(1+o(1))
    \end{align}
    as $|\omega|\longrightarrow\infty$ with $\Im\omega=0$.
\end{proposition}
\begin{proof}
    We write \eqref{ode} in the coordinate $z=-2i\omega(r-1)$ and treat it as a perturbation to the confluent hypergeometric equation as in Section 4.3 of \cite{Correia}. We then apply Theorem \ref{Olver's error approx theorem} to show that the approximate bases of solutions normalised at infinity and the horizon are both valid in an intersection region, which we then use to compute the scattering matrix elements.

    For $u$ a solution to \eqref{ode}, define 
    \begin{align}
        v:=r^{-\frac{1}{2}}(r-1)^{i\omega}e^{-i\omega(r-1)}u.
    \end{align}
    Let $z:=-2i\omega(r-1)$. Then we have
    \begin{align}\label{approx hypergeometric equation}
        \frac{d^2}{dz^2}v+\frac{1}{z}(1-2i\omega-z)\frac{d}{dz}v+\frac{1}{z}\left(2i\omega-\frac{1}{2}\right)v=\mathcal{R}v,
    \end{align}
    where
    \begin{align}
        \mathcal{R}:=-\frac{2\ell(\ell+1)+1}{z(z-2i\omega)}+\frac{1}{4(z-2i\omega)^2}.
    \end{align}
    Were it the case that $\mathcal{R}=0$, then equation \eqref{approx hypergeometric equation} would be exactly the confluent hypergeometric equation, thus we propose as an approximate basis near infinity
    \begin{align}
        \out^{\mathrm{CHG}}(\omega,\ell,r_*):=\sqrt{r}(r-1)^{-2i\omega}(-2i\omega)^{\frac{1}{2}-2i\omega}e^{i\omega r_*}U\left(\frac{1}{2}-2i\omega,1-2i\omega,-2i\omega(r-1)\right),
    \end{align}
    \begin{align}
        \outb^{\mathrm{CHG}}(\omega,\ell,r_*):=\sqrt{r}(2i\omega)^{\frac{1}{2}}e^{-i\omega r_*}U\left(\frac{1}{2},1-2i\omega,2i\omega(r-1)\right)
    \end{align}
    Near the horizon, we propose as an approximate basis
    \begin{align}
        \hor^{\mathrm{CHG}}(\omega,\ell,r_*):=\sqrt{r}(r-1)^{-i\omega}e^{i\omega (r-1)}M\left(\frac{1}{2}-2i\omega,1-2i\omega,-2i\omega(r-1)\right),
    \end{align}
    \begin{align}
        \horb^{\mathrm{CHG}}(\omega,\ell,r_*):=\sqrt{r}(r-1)^{i\omega}e^{i\omega(r-1)}M\left(\frac{1}{2},1+2i\omega,-2i\omega(r-1)\right)
    \end{align}
    where $M(a,b,z)$ and $U(a,b,z)$ solve
    \begin{align}
        z\frac{d^2}{dz^2}w+(b-z)\frac{d}{dz}w-aw=0.
    \end{align}
    Taking
    \begin{align}
        v_1=U\left(\frac{1}{2}-2i\omega,1-2i\omega,-2i\omega(r-1)\right),\qquad v_2=e^{-2i\omega(r-1)}U\left(\frac{1}{2},1-2i\omega,2i\omega(r-1)\right),
    \end{align}
    \begin{align}
        \begin{split}
            \out=\sqrt{r}(r-1)^{-2i\omega}(-2i\omega)^{\frac{1}{2}-2i\omega}e^{i\omega r_*}\left[v_1+\epsilon_{\infty}^{(1)}\right],
        \end{split}
    \end{align}
    \begin{align}
        \begin{split}
            \outb=\sqrt{r}(2i\omega)^{\frac{1}{2}}e^{-i\omega r_*}e^{2i\omega(r-1)}\left[v_2+\epsilon_{\infty}^{(2)}\right],
        \end{split}
    \end{align}
    we now use Theorem \ref{Olver's error approx theorem} to estimate $\epsilon_{\infty}$, $\epsilon_H$. We start with $\epsilon_{\infty}$, which satisfies
    \begin{align}
        \epsilon_{\infty}^{(1)}(z)=\int_{z}^{\infty}dy \left[\frac{v_1(y)v_2(z)-v_1(z)v_2(y)}{\mathfrak{W}(v_1,v_2)(y)}\mathcal{R}(y)(v_1(y)+\epsilon_{\infty}^{(1)}(y))\right]. 
    \end{align}
    The Wronskian is
    \begin{align}
        \mathfrak{W}(v_1,v_2)(z)=i(-2i\omega(r-1))^{-1+2i\omega}e^{-2i\omega(r-1)}=iz^{-1+2i\omega}e^z.
    \end{align}
    To estimate $v_1(y)v_2(z)-v_1(z)v_2(y)$, we approximate $U$ via
    \begin{align}
        U\left(\frac{1}{2}-2i\omega,1-2i\omega,z\right)=z^{-\left(\frac{1}{2}-2i\omega\right)}\left(1+R_1\right),\qquad |R_1|\lesssim \left|\frac{(1-4i\omega)}{8\omega(r-1)}\right|,
    \end{align}
    \begin{align}
         U\left(\frac{1}{2},1-2i\omega,2i\omega(r-1)\right)=z^{-\frac{1}{2}}(1+R_2),\qquad |R_2|\lesssim \left|\frac{(1+4i\omega)}{8\omega(r-1)}\right|.
    \end{align}
    We fix $|\omega_0|>0$ $r_0>1$ so that $|R_1|+|R_2|<\frac{3}{4}$ for all $r>r_0$, $|\omega|\geq |\omega_0|$. Then we find
    \begin{align}
    \begin{split}
        \left|\frac{v_1(y)v_2(z)-v_1(z)v_2(y)}{\mathfrak{W}(v_1,v_2)(y)}\right|&\leq y^{\frac{1}{2}}z^{-\frac{1}{2}}e^{z-y}+z^{-\frac{1}{2}}y^{\frac{1}{2}}
    \end{split}
    \end{align}
    Take $Q(y)=y^{\frac{1}{2}}$, $P_0(z)=z^{-\frac{1}{2}}$, $P_1(z)=z^{-\frac{3}{2}}$, $P_2(z)=z^{-\frac{5}{2}}$, $\upphi=\uppsi_0=\mathcal{R}$, $J=v_1$. We find
    \begin{align}
        \sup_{r\in (R_0,\infty)} P_0Q=1,\qquad \sup_{r\in (R_0,\infty)} QJ=1,\qquad \int_{z}^{\infty}|\upphi|dy=O\left(\frac{\ell(\ell+1)+1}{|\omega|}\right).
    \end{align}
    Therefore, we have by Theorem \ref{Olver's error approx theorem}
    \begin{align}
        \epsilon_{\infty}^{(1)}/|v_1|,\qquad \epsilon_{\infty}^{(1)}{}'/|v_1'|=O\left(\frac{\ell(\ell+1)+1}{|\omega|}\right).
    \end{align}
    An identical argument can be applied to estimate $\epsilon_{\infty}^{(2)}$, leading to
    \begin{align}
        \epsilon_{\infty}^{(2)}/|v_2|,\qquad \epsilon_{\infty}^{(2)}{}'/|v_2'|=O\left(\frac{\ell(\ell+1)+1}{|\omega|}\right).
    \end{align}
    We now repeat the above argument for $\hor^{\mathrm{CHG}}$, $\horb^{\mathrm{CHG}}$, taking
    \begin{align}
        w_1=M\left(\frac{1}{2}-2i\omega,1-2i\omega,-2i\omega(r-1)\right),\qquad w_2=(-2i\omega(r-1))^{2i\omega}M\left(\frac{1}{2},1+2i\omega,-2i\omega(r-1)\right),
    \end{align}
    \begin{align}
        \hor:=\sqrt{r}(r-1)^{-i\omega}e^{i\omega (r-1)}(w_1+\epsilon_H^{(1)}),
    \end{align}
    \begin{align}
         \horb=\sqrt{r}(r-1)^{-i\omega}e^{i\omega(r-1)}(-2i\omega)^{-2i\omega}(w_2+\epsilon_H^{(2)}).
    \end{align}
    We write the integral equation satisfied by $\epsilon_{H}^{(1)}$ as
    \begin{align}
        \epsilon_{H}^{(1)}(z)=\int_{0}^{z}dy \left[\frac{w_1(y)w_2(z)-w_1(z)w_2(y)}{y\cdot\mathfrak{W}(w_1,w_2)(y)}(y\mathcal{R}(y))(w_1(y)+\epsilon_{H}^{(1)}(y))\right]. 
    \end{align}
    We compute
    \begin{align}
        \mathfrak{W}(w_1,w_2)(z)=2i\omega(-2i\omega(r-1))^{-1+2i\omega}e^{-2i\omega(r-1)}=2i\omega z^{-1+2i\omega}e^z.
    \end{align}
    Using the integral representation formula for $M$,
    \begin{align}
        M(a,b,z)=\frac{\Gamma(b)}{\Gamma(a)\Gamma(b-a)}\int_{0}^{1}dt\, e^{zt}t^{a-1}(1-t)^{b-a-1},
    \end{align}
    Taking $a=\frac{1}{2}-2i\omega$, $b=1-2i\omega$, we immediately find
    \begin{align}
        \left|M\left(\frac{1}{2}-2i\omega,1-2i\omega,-2i\omega(r-1)\right)\right|\lesssim \left|\frac{\Gamma(1-2i\omega)}{\Gamma\left(\frac{1}{2}-2i\omega\right)}\right|\lesssim \sqrt{2|\omega|}.
    \end{align}
    Taking $a=\frac{1}{2}$, $b=1+2i\omega$, we find
    \begin{align}
        \left|M\left(\frac{1}{2},1+2i\omega,-2i\omega(r-1)\right)\right|\lesssim \left|\frac{\Gamma(1+2i\omega)}{\Gamma\left(\frac{1}{2}+2i\omega\right)}\right|\lesssim \sqrt{2|\omega|}.
    \end{align}
    Then we have
    \begin{align}
        \left|\frac{w_1(y)w_2(z)-w_1(z)w_2(y)}{y\cdot\mathfrak{W}(v_1,v_2)(y)}\right|\lesssim {|\omega|}^{-\frac{1}{2}}.
    \end{align}
    We then take $P_0=Q=1$, $P_1=|\omega|, P_2=|\omega|^2$, $\upphi=\uppsi_0=y\mathcal{R}(y)$, $J=w_1$. We take $r\in[1,r_0+1]$. Then we have
    \begin{align}
        \int_{0}^{z}|\upphi|dy\lesssim (\ell(\ell+1)+1)\log(R_0).
    \end{align}
    Therefore, by Theorem \ref{Olver's error approx theorem} we have for $r\in[1,R_0+1]$,
    \begin{align}
        \epsilon_{H}^{(1)}/|w_1|,\qquad \epsilon_{H}^{(1)}{}'/|w_1'|=O\left(\frac{\ell(\ell+1)+1}{|\omega|^{\frac{1}{2}}}\right).
    \end{align}
    An identical argument can be applied to estimate $\epsilon_{\infty}^{(2)}$, leading to
    \begin{align}
        \epsilon_{H}^{(2)}/|w_2|,\qquad \epsilon_{H}^{(2)}{}'/|w_2'|=O\left(\frac{\ell(\ell+1)+1}{|\omega|^{\frac{1}{2}}}\right).
    \end{align}
    We can match the approximate bases at $\tilde{r}\in[r_0,r_0+1]$ using the connection formula for confluent hypergeometric functions
    \begin{align}
        U(a,b,z)=\frac{\Gamma(1-b)}{\Gamma(a-b+1)}M(a,b,z)+\frac{\Gamma(b-1)}{\Gamma(a)}z^{1-b}M(a-b+1,2-b,z),
    \end{align}
    Starting from
    \begin{align}
        \out=\frac{1}{T}\left(\horb-\tilde{R}\hor\right),
    \end{align}
    and taking
    \begin{align}
        \eta^{(i)}_{\infty}=\epsilon^{(i)}/v_i,\qquad \eta^{(i)}_{H}=\epsilon^{(i)}/w_i,
    \end{align}
    we have
    \begin{align}
        \begin{split}
            \out=&(1+\eta_{\infty}^{(1)})\sqrt{\tilde{r}}(\tilde{r}-1)^{-i\omega}(-2i\omega)^{\frac{1}{2}-2i\omega}e^{i\omega \tilde{r}}v_1\\=&(1+\eta_{\infty}^{(1)})\sqrt{\tilde{r}}(\tilde{r}-1)^{-i\omega}(-2i\omega)^{\frac{1}{2}-2i\omega}e^{i\omega \tilde{r}}\left[\frac{\Gamma(2i\omega)}{\Gamma\left(\frac{1}{2}\right)}w_1+\frac{\Gamma(-2i\omega)}{\Gamma\left(\frac{1}{2}-2i\omega\right)}w_2\right]\\=&\frac{1+\eta_{\infty}^{(1)}}{1+\eta^{(1)}_H}\frac{\Gamma(2i\omega)}{\Gamma\left(\frac{1}{2}\right)}(-2i\omega)^{\frac{1}{2}-2i\omega}e^{i\omega}\hor\\&+\frac{1+\eta_{\infty}^{(1)}}{1+\eta^{(2)}_H}\frac{\Gamma(-2i\omega)}{\Gamma\left(\frac{1}{2}-2i\omega\right)}(-2i\omega)^{\frac{1}{2}-2i\omega}(-2i\omega)^{2i\omega}e^{i\omega}\horb.
        \end{split}
    \end{align}
    Therefore, 
    \begin{align}
        |\tilde{R}|=e^{-2\pi|\omega|}(1+\eta_{H}^{(2)})(1+\eta_{H}^{(1)})^{-1}=e^{-2\pi|\omega|}(1+o(1)).
    \end{align}
\end{proof}
\begin{proposition}\label{high frequency behaviour of reflection coefficient proposition}
    For fixed $\ell\in\mathbb{N}_{\geq0}$, $\lambda\in(0,\frac{1}{2})$ and any $\delta>0$, the reflection coefficient ${\tilde{R}}$ satisfies
    \begin{align}\label{high frequency behaviour of reflection coefficient}
        {\trr}(\omega,\ell)\lesssim e^{-2\pi(1-\delta) |\omega|}(1+o(1))
    \end{align}
    as $|\omega|\longrightarrow\infty$ with $\Im\omega=\lambda$.
\end{proposition}
\begin{proof}
Without loss of generality, assume $\Re\omega>0$.
Using relations \eqref{variation of parameters ttt}, \eqref{variation of parameters trr}, we can write
\begin{align}
    \begin{split}
        {\trr}&=-\frac{{{\ttt}}}{2i\omega}\int_{r=1}^{\infty} dr \frac{e^{2i\omega r}}{(r-1)^{-2i\omega}}\left(\frac{\ell(\ell+1)}{r^2}+\frac{1}{r^3}\right)(e^{-i\omega r_*}\out(r_*,\omega)). 
    \end{split}
\end{align}
We use \eqref{def:TR'},
\begin{align}
    \tilde{T}(\omega,\ell)\out=\horb(\omega,\ell,r_*)-\tilde{R}(\omega,\ell) \hor(\omega,\ell,r_*),
\end{align}
and the fact that $(e^{2i\omega r_*}V_{\ell})e^{-i\omega r_*}\horb$, $V_{\ell}e^{i\omega r_*}\hor$ are integrable to split the integral in two
\begin{align}
    2i\omega{\trr}(\omega,\ell)=I_1(\omega,\ell)+{\trr}(\omega,\ell)I_2(\omega,\ell),
\end{align}
where
\begin{align}\label{I1}
    \begin{split}
        I_1(\omega,\ell)=\int_{r=1}^{\infty} dr \frac{e^{2i\omega r}}{(r-1)^{-2i\omega}}\left(\frac{\ell(\ell+1)}{r^2}+\frac{1}{r^3}\right)(e^{-i\omega r_*}\horb(r_*,\omega,\ell)),
    \end{split}
\end{align}
\begin{align}
    \begin{split}
        I_2(\omega,\ell)=\int_{r=1}^{\infty}  dr\left(\frac{\ell(\ell+1)}{r^2}+\frac{1}{r^3}\right)(e^{i\omega r_*}\hor(r_*,\omega,\ell)).
    \end{split}
\end{align}

Let us consider $I_1'(\omega)$ given by
\begin{align}\label{def of I 1 ' for some reason}
    \begin{split}
        I_1'(\omega,\ell)=\int_{r=1}^{\infty} dr \frac{e^{2i\omega r}}{(r-1)^{-2i\omega}r^2}(e^{-i\omega r_*}\horb(r_*,\omega,\ell)),
    \end{split}
\end{align}
Let $s=r-1$, then \eqref{ode} in the variable $s$ reads
\begin{align}
    \frac{s^2}{(s+1)^2}\frac{d^2}{ds^2}u+\frac{s}{(s+1)^3}\frac{d}{ds}u+\left(\omega^2-\frac{s}{(s+1)^3}\left[\ell(\ell+1)+\frac{1}{s+1}\right]\right)u=0.
\end{align}
Note that in the variable $s$, the function $e^{-i\omega r_*(s)}\horb(r_*(s),\omega,\ell)-1$, at any $s\in\mathbb{C}\setminus \{0\}$, is holomorphic with analyticity radius $|s+1|$. We rewrite $I_1'$ as
\begin{align}
    \begin{split}
        I_1'=\int_{s=0}^{\infty} ds \frac{e^{2i\omega (s+1)}}{s^{-2i\omega}(s+1)^2}(e^{-i\omega r_*(s)}\horb(r_*(s),\omega,\ell)).
    \end{split}
\end{align}
We use the holomorphy and boundedness of $e^{i\omega r_*(s)}\hor(r_*(s),\omega,\ell)-1$ in the sector $\arg s\in[0,\pi)$ to rotate the integration path by $\pi(1-\delta)$, and get
\begin{align}
    \begin{split}
        I_1'=e^{2i\omega \cos(\pi\delta)}e^{2i\omega}e^{-2\pi\Re\omega (1-\delta)}\int_{s=0}^{\infty} ds \frac{e^{-2\omega \sin(\pi\delta) s}}{s^{2i\omega}(se^{i\pi(1-\delta)}+1)^2}(e^{-i\omega r_*(e^{-i\pi(1-\delta)}s)}\horb(r_*(e^{i\pi(1-\delta)}s),\omega,\ell)).
    \end{split}
\end{align}
We now use Lemma \ref{crude boundedness of horb} to estimate the integral and find
\begin{align}
    |I_1'|\lesssim e^{-2\pi|\Re\omega| (1-\delta)}.
\end{align}
Finally, we simply flatten $I_2$ using Lemma \ref{crude boundedness of horb}
to get $|I_2|\lesssim |\omega|^{-1}$, and thus
\begin{align}
    \begin{split}
        \trr\lesssim  e^{-2\pi|\Re\omega|(1-\delta)}+C'|\omega|^{-2}\trr\sim e^{-2\pi|\Re\omega|(1-\delta)}\implies \trr\lesssim e^{-2\pi|\Re\omega|(1-\delta)}.
    \end{split}
\end{align}
\end{proof}

\begin{proposition}\label{high frequency behaviour of reflection coefficient proposition n}
    For fixed $\ell\in\mathbb{N}_{\geq0}$, $\lambda\geq0$ and any $\delta>0$, the reflection coefficient ${\tilde{R}}$ satisfies
    \begin{align}\label{high frequency behaviour of reflection coefficient everywhere in the UHP}
        {\trr}(\omega,\ell)\lesssim e^{-2\pi(1-\delta) |\omega|}(1+o(1))
    \end{align}
    as $|\omega|\longrightarrow\infty$ with $\Im\omega=\lambda$.
\end{proposition}
\begin{proof}
    For any $m\in \mathbb{N}_{>0}$, and $\omega$ with $\Re\omega>0$, there exists $\epsilon_0>0$ such that for $|r-1|<\epsilon_0$, Lemma \ref{hint at a pole of R tilde} gives
    \begin{align}
        \horb =e^{i\omega r_*}\sum_{k=0}^\infty a_k(\omega,\ell)(r-1)^k,
    \end{align}
    where $a_{k}$ is holomorphic in $\omega$ away from $\frac{ip}{2}$, $p\in\mathbb{N}_{>0}$. For $\Im\omega\in[0,\frac{1}{2})$, we may rewrite $\tilde{R}$ as
    \begin{align}\label{recipe for analytic extension of tilde R}
    \begin{split}
        2i\omega \tilde{R}(\omega,\ell)=
        &\int_{r=1}^{\infty} dr \frac{e^{2i\omega r}}{(r-1)^{-2i\omega}}\left(\frac{\ell(\ell+1)}{r^2}+\frac{1}{r^3}\right)\left(\sum_{k=0}^m a_k(r-1)^k\right)\\&+\int_{r=1}^{\infty} dr \frac{e^{2i\omega r}}{(r-1)^{-2i\omega}}\left(\frac{\ell(\ell+1)}{r^2}+\frac{1}{r^3}\right)\left(e^{-i\omega r_*}\horb-\sum_{k=0}^m a_k(r-1)^k\right)\\&+\tilde{R}(\omega,\ell)\int_{r=1}^{\infty}dr \left(\frac{\ell(\ell+1)}{r^2}+\frac{1}{r^3}\right)e^{i\omega r_*}\hor.
    \end{split}
    \end{align}
    Note that for $k=0,\dots,m$, 
    \begin{align}
    \begin{split}
        &\int_{r=1}^{\infty} dr \frac{e^{2i\omega r}}{(r-1)^{-2i\omega-k}}\left(\frac{\ell(\ell+1)}{r^2}+\frac{1}{r^3}\right)\\&=e^{2i\omega} \Gamma(k+1+2i\omega)\left[\ell(\ell+1)U(k+1+2i\omega,k+2i\omega,-2i\omega)+U(k+1+2i\omega,k-1+2i\omega,-2i\omega)\right],
    \end{split}
    \end{align}
    where $U$ is the Tricomi hypergeometric function. 
    The term
    \begin{align}\label{residual term}
        \int_{r=1}^{\infty} dr \frac{e^{2i\omega r}}{(r-1)^{-2i\omega}}\left(\frac{\ell(\ell+1)}{r^2}+\frac{1}{r^3}\right)\left(e^{-i\omega r_*}\horb-\sum_{k=0}^m a_k(r-1)^k\right)
    \end{align}
    extends analytically to all $\omega$ with $0\leq \Im\omega\leq \frac{m}{2}$ and $\omega\neq \frac{ik}{2}$, $k=1,\dots, m$. We may then use the analytic extension of \eqref{recipe for analytic extension of tilde R} to estimate $\tilde{R}$ for $\omega$ with with fixed $\Im\omega\leq \frac{m}{2}$ as $|\Re\omega|\longrightarrow\infty$. 
    
    Note that
    \begin{align}
        \begin{split}
            U(k+1+2i\omega,k+2i\omega,-2i\omega)=(-2i\omega)^{1-2i\omega-k}U(2,2-2i\omega-k,-2i\omega),
        \end{split}
    \end{align}
    \begin{align}
        U(2,2-2i\omega-k,-2i\omega)=U(1,2-2i\omega-k,-2i\omega)-U(1,1-2i\omega-k,-2i\omega),
    \end{align}
    \begin{align}
        U(1,2-2i\omega-k,-2i\omega)=e^{-2i\omega}(-2i\omega)^{k-1+2i\omega}\Gamma(1-k-2i\omega,-2i\omega),
    \end{align}
    \begin{align}
        U(1,1-2i\omega-k,-2i\omega)=e^{-2i\omega}(-2i\omega)^{k+2i\omega}\Gamma(-k-2i\omega,-2i\omega),
    \end{align}
    where $\Gamma(a,b)$ is the incomplete Gamma function.  We then have
    \begin{align}
    \begin{split}
        &e^{2i\omega} \Gamma(k+1+2i\omega)\ell(\ell+1)U(k+1+2i\omega,k+2i\omega,-2i\omega)\\&=\Gamma(k+1+2i\omega)\left[\Gamma(1-k-2i\omega,-2i\omega)+2i\omega \Gamma(-k-2i\omega,-2i\omega)\right]\\&=O(e^{-2\pi|\Re\omega|})
    \end{split}
    \end{align}
    A similar argument is applied to $U(k+1+2i\omega,k-1+2i\omega,-2i\omega)$. The result then follows by rotating the term \eqref{residual term} and proceed as in the argument of Proposition \ref{high frequency behaviour of reflection coefficient proposition}.
\end{proof}

\subsection{Aside: the reflection coefficient ${\trr}$ does not vanish in the fundamental strip}

We now prove the following
\begin{proposition}\label{reflection never vanishes}
    For $\omega$ such that $\Im\omega\in[0,\frac{1}{2})$, we have ${\trr}(\omega)\neq0$.
\end{proposition}
\begin{proof}
     We take $\omega\neq0$. Recall from \eqref{def of U inf}
    \begin{align}\label{i dont know what to call this}
        \begin{split}
            e^{i\omega r_*}\outb=&1+\frac{1}{2i\omega}\int_{r_*}^\infty V_{\ell}(y) e^{i\omega y}\outb(\omega,\ell,y)dy\\
            &-\frac{1}{2i\omega}e^{2i\omega r_*}\int_{r_*}^\infty e^{-2i\omega y}V_{\ell}(y) e^{i\omega y}\outb(\omega,\ell,y)dy.
        \end{split}
    \end{align}
    Assume $\omega$ is such that $\outb=\lambda\hor$ for some $\lambda$. Then $e^{i\omega r_*}\outb$ extends to an analytic function on $r\in[2M,\infty)$. This implies
    \begin{align}
        \lim_{r_*\longrightarrow-\infty}e^{2i\omega r_*}\int_{r_*}^\infty e^{-2i\omega y}V_{\ell}(y) e^{i\omega y}\outb(\omega,\ell,y)dy
    \end{align}
    exists. Taking a $\partial_r$ to $e^{i\omega r_*}\outb$ via \eqref{i dont know what to call this}, we find that we must have
    \begin{align}
        \lim_{r_*\longrightarrow-\infty}\frac{1}{\Omega^2}e^{2i\omega r_*}\int_{r_*}^\infty e^{-2i\omega y}V_{\ell}(y) e^{i\omega y}\outb(\omega,\ell,y)dy
    \end{align}
    must exist. Taking successive $\partial_r$ derivatives, we see that
    \begin{align}
        \lim_{r_*\longrightarrow-\infty}\frac{1}{\Omega^{2n}}e^{2i\omega r_*}\int_{r_*}^\infty e^{-2i\omega y}V_{\ell}(y) e^{i\omega y}\outb(\omega,\ell,y)dy
    \end{align}
    must exist for all $n\in\mathbb{N}$. Now take the Fourier transform of 
    \begin{align}
        F=\int_{-\infty}^{r_*}dy e^{-i\omega y}V_{\ell}(y)\outb(\omega,\ell,y)
    \end{align}
    i.e.~for $\Im(k)\geq0$ we calculate
    \begin{align}
    \begin{split}
        \int_{-\infty}^{\infty}dx\, e^{ikx}F(x)&=\int_{-\infty}^{\infty}dx\, e^{ikx}\int_{-\infty}^{r_*}dy\, e^{-i\omega y}V_{\ell}(y)\outb(\omega,\ell,y)\\&
        =-\frac{1}{ik}\int_{-\infty}^{\infty}dx\, e^{i(k-\omega)x} V_{\ell}(x) \outb(\omega,\ell,x)
        \\&=-\frac{1}{ik}\int_{-\infty}^{\infty}dx\, e^{i(k-2\omega)x} V_{\ell}(x) (e^{i\omega x}\outb(\omega,\ell,x)),
    \end{split}
    \end{align}
and the left hand side is finite for all $k$ with $\Im k\geq0$. Pick $k-2\omega=i\xi$ for $\xi>0$. We have that the $e^{i\omega r_*} \outb=1+O(e^{r_*})$ as $x\longrightarrow-\infty$, while $e^{i\omega r_*} \outb=1+O(\frac{1}{r})$ as $r_*\longrightarrow\infty$, so by taking the limit as $\xi\longrightarrow1$ we find that the right hand side blows up.
\end{proof}

\subsection{The asymptotic behaviour of the transmission coefficient for fixed $\ell$}

\begin{lemma}\label{rudimentary boundedness of T}
    For $\omega$ such that $\Im\omega\geq0$, and all $\ell$, we have
    \begin{align}
        |T(\omega,\ell)|^2\leq 1.
    \end{align}
\end{lemma}
\begin{proof}
    For $\omega\in\mathbb{R}$, the result is already guaranteed by \eqref{unitarity for real frequency}. We then take $\Im\omega=\omega_I>0$. Consider $u=T\hor=\outb-R\out$ and let $G=e^{i\omega r_*}u$. Then we have $G''-2i\omega G'-V_{\ell}G=0$, and
    \begin{align}
        \left[\frac{2|\omega|^2}{\omega_I}|G|^2+(\overline{G}G'+G\overline{G}{}')+\frac{i\omega_R}{\omega_I}(\overline{G}G'-G\overline{G}{}')\right]'-2V_{\ell}|G|^2-2|G'|^2=0.
    \end{align}
    Note that $G|_{r=1}=T$, $G'|_{r=1}=0$, $\lim_{r\longrightarrow\infty}G=1$, $\lim_{r\longrightarrow\infty}G'=0$. Therefore,
    \begin{align}\label{holy grail 2}
        (1-|T|^2)=\frac{\omega_I}{|\omega|^2}\int_{-\infty}^{\infty}dr_*\,V_{\ell}|G|^2+2|G'|^2\geq0.
    \end{align}
\end{proof}

\begin{lemma}\label{asymptotics of T}
    For any fixed $\ell$, let $\omega$ with $\Im\omega\geq0$. Then we have
    \begin{align}
        |{T}-1|\sim O(|\omega|^{-1})
    \end{align}
    as $|\omega|\longrightarrow\infty$.
\end{lemma}
\begin{proof}
    Recall from \eqref{Volterra definition of tilde T},
    \begin{align}
        T(\omega,\ell)=\frac{1}{1+\frac{1}{2i\omega}\int_{-\infty}^{\infty}dr_* V_{\ell} e^{-i\omega r_*}\out}.
    \end{align}
    From Lemma \ref{crude boundedness of horb}, we have
    \begin{align}
        \begin{split}
            \left|\int_{-\infty}^{\infty}dr_* V_{\ell} e^{-i\omega r_*}\out\right|=\frac{1}{2}+O(|\omega|^{-1}).
        \end{split}
    \end{align}
\end{proof}
\begin{proposition}\label{prop: low frequency expansion}
    For any fixed $\ell$, we have
    \begin{align}
        |T(\omega)|= \omega^{\ell+1}(1+o(1))
    \end{align}
    for $\omega\in\mathbb{R}$, $\omega\longrightarrow0$.
\end{proposition}
\begin{proof}
    See \cite{DDS11} (in particular Section 8.2). See also \cite{StarobinskiiChurilov}.
\end{proof}
\begin{corollary}\label{prop: low frequency expansion off the real axis}
    For any fixed $\ell$ and $\delta>0$, we have that $|\omega^{\ell+1}T(\omega)^{-1}|$ is uniformly bounded in $\{z=x+iy: x\in[-\delta,\delta], y\in [0,\delta]\}$.
\end{corollary}
\begin{proof}
    Let $G=e^{i\omega r_*}\hor$. Using \eqref{Volterra definition of T} and Cauchy--Schwartz, we write 
    \begin{align}
        \frac{1}{|T(\omega,\ell)|}\lesssim 1+\frac{1}{|\omega|}\sqrt{\int_{-\infty}^{\infty}dr_* V_{\ell}}\sqrt{\int_{-\infty}^{\infty}dr_* V_{\ell}|G|^2}.
    \end{align}
    Using \eqref{holy grail 2} of Proposition \ref{rudimentary boundedness of T}, we have
    \begin{align}\label{wonder bound on inverse of T}
        \frac{1}{|T(\omega,\ell)|}\lesssim 1+\frac{\left(\ell+\frac{1}{2}\right)^2}{\sqrt{\omega_I}}.
    \end{align}
    Using \eqref{Volterra definition of tilde T} and Lemma \ref{crude boundedness of horb}, we also have
    \begin{align}
    \frac{1}{|T(\omega,\ell)|}\lesssim 1+\frac{\ell\left(\ell+\frac{1}{2}\right)^2(\ell+1)}{|\omega|^2}e^{\frac{\ell(\ell+1)+1}{|\omega|}}.
    \end{align}
    Define
    \begin{align}
        f(\omega)=\frac{\omega^{\ell+1}}{(\omega+i)^{\ell+1}}\left(\frac{1}{T(\omega,\ell)}-1\right).
    \end{align}
    Let $\zeta:=-1/\omega$, and let $\tilde{f}(\zeta)=f(-1/\zeta)$. Since $T$ is holomorphic for all $\omega$ with $\Im\omega>0$ or $\omega\neq0$ with $\Im\omega=0$, we have that $\tilde{f}$ is holomorphic for $\Im\zeta>0$ and continuous on $\Im\zeta\geq0$, $\zeta\neq0$. By Lemma \ref{asymptotics of T}, $\tilde{f}$ is also continuous at $\zeta=0$. In the sector $\mathrm{ph}\zeta\in \left(0,\frac{\pi}{2}\right)$, we have $|\tilde{f}|\lesssim C_{\ell}|\zeta|^2e^{(\ell(\ell+1)+1)|\zeta|}$. For $\zeta$ on the real line, we have by Proposition \ref{prop: low frequency expansion} that $\tilde{f}$ is uniformly bounded. By \eqref{wonder bound on inverse of T}, we also have on the positive imaginary axis (including the origin) that $\tilde{f}$ is uniformly bounded. We may then use the Phragmen--Lindel\"of principle to conclude that $\tilde{f}$ is uniformly bounded for $\zeta$ in the first quadrant. An identical argument gives that $\tilde{f}$ is uniformly bounded for $\zeta$ in the second quadrant. Thus, $\tilde{f}$ is uniformly bounded for $\Im\zeta\geq0$, and so $f$ is uniformly bounded in  $\{z=x+iy: x\in[-\delta,\delta], y\in [0,\delta]\}$.
\end{proof}
\numberwithin{equation}{subsubsection}
\subsection{The asymptotic decay of the transmission coefficient for fixed $\omega$ and \mbox{large $\ell$}}
In the case of fixed $\omega\in\mathbb{R}$ and large $\ell$, the solution $\hor$ must tunnel through the ``forbidden region" where $V_{\ell}>\omega^2$, and as a result the transmitted signal is damped by the exponential decay of the solution inside the forbidden region, and it is a classic result that we then have exponential decay for $T$ at large $\ell$ (see for instance equation 10.17 in \cite{Olver}, or \cite{Ramond}):
\begin{align}\label{barrier tunneling real omega}
    T(\omega,\ell)&\sim e^{-\int_{r_*:V_{\ell}>\omega^2}dr_*\sqrt{V_{\ell}-\omega^2}}.
\end{align}
The case of complex $\omega$ differs in regards that the turning points are now not on the real line. Take $\Im\omega>0$. For $\Im\omega$ small, the turning points are still close to the real line in the complex $r_*$ plane, and in this case the strategy of using Airy's equation to derive a Bessel function-based approximation and obtain \eqref{barrier tunneling real omega} still works. For $\Im\omega$ far away from the real line, we run into the complication that the variable $r$ possesses branch cuts in the complex $r_*$ plane. We can however obtain \eqref{barrier tunneling real omega} using a different strategy. The near-horizon turning point approaches $r=1$ as $\ell\longrightarrow\infty$ for any fixed $\omega$. Considering that the event horizon represents a pole of the radial Schrodinger equation \eqref{ode}, this leads us to believe that the best approximation is a Whittaker model. However, we may still obtain a good approximation of $\hor$ by utilising the Rindler-like geometry of the near horizon region. Indeed, using $r-2M=\rho^2$, we find that
\begin{align}
    g_{Schw}\sim -\rho^2dt^2+4d\rho^2+d\theta^2+\sin\theta^2d\phi^2+\tilde{g},\qquad |\tilde{g}_{\mu\nu}|=o(r-1).
\end{align}
For the Rindler spacetime, the global solution of the separated radial equation is a Modified Bessel function (see \cite{Fulling}). We use Modified Bessel functions below to find a model for the near horizon behaviour of $\hor$. The modified Bessel model reveals that, while the first turning point is not on the real line, we have an effective turning point located at $r=1+(\ell(\ell+1))^{-\frac{1}{2}}|\Re\omega|$. The modified Bessel model provides a good approximation between $r=1$ and a neighbourhood of $r=1+(\ell(\ell+1))^{-\frac{1}{2}}|\Re\omega|$. This neighbourhood overlaps with an intermediate region where we use Legendre polynomials, bridging the horizon and the far region. In the far region, the turning point with large $|r|$ is sufficiently far from the real line that the WKB ansatz provides a good approximation that allows us to close the argument. We present this scheme in detail in the proof of Proposition \ref{holy grail of transmission} below. We concentrate on the case $\Im\omega= \frac{1}{2}$ as it is the only case needed for the example given in Section \ref{sec: counterexample at scrip}. We now state our result, and we proceed to prove it over the Sections \ref{sec: Volterra Olver error estimate}--\ref{sec: transmission behaviour at large ell}.

\begin{proposition}\label{holy grail of transmission}
For any $\omega_I\geq\frac{1}{2}$, there exists a constant $C$ such that for any fixed $\omega$ with $\Im\omega=\omega_I$ and for $\ell>C|\omega|^9$, we have
\begin{align}
    |T(\omega,\ell)|\lesssim \ell^{-\frac{1}{2}\ell}.
\end{align}
\end{proposition}
\subsubsection{WKB error estimates via Volterra equations}\label{sec: Volterra Olver error estimate}

\subsubsection{Large $\ell$ approximation of $\hor$ near the event horizon}\label{sec: near horizon hor}

As mentioned in the preamble to this section, we divide the argument into three regions: the near horizon region, the intermediate region, and the far region. We deal with the near horizon region in this section.
\begin{proposition}\label{prop: approximation prop near horizon} For any $\omega_I\geq\frac{1}{2}$, 
    there exists a constant $C$ such that, for any $\omega$ with $\Im\omega=\omega_I$,  and for $\ell(\ell+1)>C|\omega|^9$, we have
    \begin{align}
        \hor(\omega,\ell,r_*)=\left(\frac{\ell(\ell+1)}{e}\right)^{i\omega}\Gamma(1-2i\omega)(I_{-2i\omega}(\rho)+\epsnh(\rho)),
    \end{align}
    where $\rho:=\sqrt{\ell(\ell+1)(r-1)}$, and $\epsnh$ satisfies
    \begin{align}
      |\epsnh(\rho)|/|I_{-2i\omega}(\rho)|,\; |\epsnh'(\rho)|/|I_{-2i\omega}'(\rho)|=O(C^{-\frac{2}{9}}\deltnh^{\frac{4}{9}}),
    \end{align}    
    for $r\in[1,r_1]$, with $r_1=1+\deltnh(\ell(\ell+1))^{-\frac{1}{3}}$.
\end{proposition}
In order to prove the above result for large $|\Re\omega|$, we will need to approximate Modified Bessel functions by Airy's functions. We start with preliminary statements concerning the modulus and phase of the Airy functions that we use:
\begin{defin}
    Define $\mathrm{Ai}_{\pm}$ by
    \begin{align}
        \mathrm{Ai}_{\pm}(z):=\mathrm{Ai}(e^{\mp\frac{2\pi i}{3}}z).
    \end{align}
\end{defin}
\begin{defin}
Define $E_{\pm}$, $E$ by
\begin{align}\label{Airy weight functions}
    E_0(z):=\left|\exp\left(\frac{2}{3}z^{\frac{3}{2}}\right)\right|\qquad E_+(z):=\left|\exp\left(\frac{2}{3}e^{-i\pi}z^{\frac{3}{2}}\right)\right|, \qquad   E_-(z):=\left|\exp\left(\frac{2}{3}e^{i\pi}z^{\frac{3}{2}}\right)\right|
\end{align}
where we take the square root to be defined via its principal branch. We define
\begin{align}
    M_0(z):=\left(E_+(z)|\mathrm{Ai}_+(z)|^2+E_-(z)|\mathrm{Ai}_-(z)|^2\right)^{\frac{1}{2}},
\end{align}
\begin{align}
     M_{\pm}(z):=\left(E_0(z)|\mathrm{Ai}(z)|^2+E_{\mp}(z)|\mathrm{Ai}_{\mp}(z)|^2\right)^{\frac{1}{2}},
\end{align}
\begin{align}
    N_0(z):=\left(E_+(z)|\mathrm{Ai}'_+(z)|^2+E_-(z)|\mathrm{Ai}'_-(z)|^2\right)^{\frac{1}{2}},
\end{align}
\begin{align}
     N_{\pm}(z):=\left(E_0(z)|\mathrm{Ai}'(z)|^2+E_{\mp}(z)|\mathrm{Ai}'_{\mp}(z)|^2\right)^{\frac{1}{2}},
\end{align}
\begin{align}
    \theta_0(z)=\arctan\frac{E_+(z)|\mathrm{Ai}_+(z)|}{E_-(z)|\mathrm{Ai}_-(z)|},\qquad \theta_{+}(z)=\arctan\frac{E_-(z)|\mathrm{Ai}_-(z)|}{E_0(z)|\mathrm{Ai}(z)|},\qquad \theta_{-}(z)=\arctan\frac{E_0(z)|\mathrm{Ai}(z)|}{E_+(z)|\mathrm{Ai}_+(z)|}.
\end{align}
\begin{align}\label{Airy derivative angles}
    \upomega_0(z)=\arctan\frac{E_+(z)|\mathrm{Ai}'_+(z)|}{E_-(z)|\mathrm{Ai}'_-(z)|},\qquad \upomega_{+}(z)=\arctan\frac{E_-(z)|\mathrm{Ai}'_-(z)|}{E_0(z)|\mathrm{Ai}'(z)|},\qquad \upomega_{-}(z)=\arctan\frac{E_0(z)|\mathrm{Ai}'(z)|}{E_+(z)|\mathrm{Ai}'_+(z)|}.
\end{align}
Define the constants
\begin{align}
    v_0=\sup_{\{\mathrm{ph}z\geq \frac{\pi}{3}\}\cup \{\mathrm{ph}z\leq -\frac{\pi}{3}\}}\{\pi \sqrt{1+|z|}M_0(z)^2\},
\end{align}
\begin{align}
    v_0'=\sup_{\{\mathrm{ph}z\geq \frac{\pi}{3}\}\cup \{\mathrm{ph}z\leq -\frac{\pi}{3}\}}\{\pi N_0(z)M_0(z)\}.
\end{align}
\end{defin}
\begin{remark}
    Note that $v_0$, $v_0'$ are finite numbers. Note also that 
    \begin{align}
        v_0=\sup_{\{\mathrm{ph}z \mp \frac{2\pi}{3}\geq \frac{\pi}{3}\}\cup \{\mathrm{ph}z\mp \frac{2\pi}{3}\leq -\frac{\pi}{3}\}}\{\pi \sqrt{1+|z|}M_{\pm}(z)^2\},
    \end{align}
    \begin{align}
        v_0'=\sup_{\{\mathrm{ph}z\mp \frac{2\pi}{3}\geq \frac{\pi}{3}\}\cup \{\mathrm{ph}z\mp \frac{2\pi}{3}\leq -\frac{\pi}{3}\}}\{\pi N_{\pm}(z)M_{\pm}(z)\}.
    \end{align}
    See Section 8.3, Chapter 11 of \cite{Olver}.
\end{remark}
\begin{proposition}\label{prop: airy approximation of modified bessel for large frequency}
Fix $\beta\geq1$. For large $\alpha>0$, we have the following asymptotic relations, uniformly in $x\in [0,\infty)$:
    \begin{align}
        I_{\beta+i\alpha}(\alpha x)= 2\pi e^{\frac{\alpha\pi}{2}-\frac{i\pi}{2}\beta-\frac{i\pi }{6}}\alpha^{-\frac{1}{3}}\left(\frac{4\zeta}{x^2-1}\right)^{\frac{1}{4}}\mathrm{Ai}\left[e^{-\frac{2\pi i}{3}}\left({\alpha^{\frac{2}{3}}\zeta+\alpha^{-\frac{1}{3}}\Phi}\right)\right]\left(1+O\left(|\alpha|^{-\frac{1}{2}}\right)\right),
    \end{align}
    \begin{align}\label{Funny K estimate}
        K_{\beta+i\alpha}(\alpha x)=\pi e^{-\frac{\alpha \pi}{2}+\frac{i\pi }{2}\beta}\alpha^{-\frac{1}{3}}\left(\frac{4\zeta}{x^2-1}\right)^{\frac{1}{4}}\left[\mathrm{Ai}(\alpha^{\frac{2}{3}}{\zeta}+\alpha^{-\frac{1}{3}}\Phi)+\epsilon\right],
    \end{align}
    where ${\zeta}$ is defined by
    \begin{align}
        \frac{2}{3}\zeta^{\frac{3}{2}}=\sqrt{x^2-1}-\arcsec x
    \end{align}
    for $x\geq1$, and
    \begin{align}
        \frac{2}{3}(-\zeta)^{\frac{3}{2}}=\arcsech x-\sqrt{1-x^2}
    \end{align}
    for $x<1$, and $\Phi$ is defined by
    \begin{align}
        \zeta^{\frac{1}{2}}\Phi=i\beta\arcsec x;\qquad x>1.
    \end{align}
    \begin{align}
        (-\zeta)^{\frac{1}{2}}\Phi=i\beta\arcsech x;\qquad x\leq1.
    \end{align}
    Finally, the error term $\epsilon$ in \eqref{Funny K estimate} is such that
    \begin{align}
        |\epsilon|\leq (|W_0|^2+|W_0^{(+)}|^2)^{\frac{1}{2}}\cdot O(|\alpha|^{-\frac{1}{2}}),\qquad |\epsilon{}'|\leq (|W_0'|^2+|W_0^{(+)}{}'|^2)^{\frac{1}{2}}\cdot O(|\alpha|^{-\frac{1}{2}}),
    \end{align}
    where 
    \begin{align}
         W_0^{(+)}=\frac{\alpha^{\frac{1}{2}}}{(\alpha+\Phi')^{\frac{1}{2}}}\mathrm{Ai}\left[e^{-\frac{2\pi i}{3}}\left({\alpha^{\frac{2}{3}}\zeta+\alpha^{-\frac{1}{3}}\Phi}\right)\right],\qquad  W_0=\frac{\alpha^{\frac{1}{2}}}{(\alpha+\Phi')^{\frac{1}{2}}}\mathrm{Ai}\left({\alpha^{\frac{2}{3}}\zeta+\alpha^{-\frac{1}{3}}\Phi}\right),
    \end{align}
    and the derivative of $\epsilon$ is taken with respect to $\zeta$.
\end{proposition}
\begin{proof}
    See Appendix \ref{Appendix A}.
\end{proof}

\begin{proof}[Proof of Proposition \ref{prop: approximation prop near horizon}]
Without loss of generality, take $\omega_R\leq0$. Consider \eqref{ode} and use the coordinate $\rho=\sqrt{L(r-1)}$, where $L:=\ell(\ell+1)$. Then \eqref{ode} becomes
\begin{align}
    \partial_{\rho}^2u+\left[\frac{2L}{\rho(\rho^2+L)}-\frac{1}{\rho}\right]\partial_\rho u+4\left[\left(\frac{\rho}{L}+\frac{1}{\rho}\right)^2\omega^2-\frac{L}{\rho^2+L}\left(1+\frac{1}{\rho^2+L}\right)\right]u=0.
\end{align}
We rewrite this as
\begin{align}\label{exact ode modified bessel lhs}
    \partial_{\rho}^2u+\frac{1}{\rho}\partial_\rho u+\left(\frac{4\omega^2}{\rho^2}-4\right)u=R,
\end{align}
with $R$ given by
\begin{align}
    R:=\frac{2\rho}{\rho^2+L}\partial_\rho u-4\left[\omega^2\left(\frac{\rho^2}{L^2}+\frac{2}{L}\right)-\frac{\rho^2}{\rho^2+L}+\frac{L}{(\rho^2+L)^2}\right]u.
\end{align}
The equation obtained from \eqref{exact ode modified bessel lhs} by setting $R=0$ is the {modified Bessel equation}, and it is the equation that governs spatial mode solutions of the wave equation in two-dimensional Rindler spacetime \cite{Fulling}. Consider the case $\Re\omega\neq0$, in which case the solutions can be expressed as the linear combination
\begin{align}
    u_{0}=c_1I_{-2i\omega}(2\rho)+c_2 K_{-2i\omega}(2\rho).
\end{align}
We are interested in approximating $\hor$, which behaves like $e^{-i\omega}(r-1)^{-i\omega}$ near $r=1$. Now $I_{-2i\omega}$ has leading order behaviour
\begin{align}
    I_{-2i\omega}(2\rho)\sim\frac{(\ell(\ell+1))^{-i\omega}}{\Gamma(1-2i\omega)}(r-1)^{-i\omega}.
\end{align}
Therefore, we take 
\begin{align}
   c_2=0,\qquad c_1=\left(\frac{\ell(\ell+1)}{e}\right)^{i\omega}\Gamma(1-2i\omega).
\end{align}
Note that the prescription above works even when $\Re\omega=0$. 

Taking $\hor=c_1(I_{-2i\omega}+\epsnh)$, variation of parameters gives
\begin{align}
    \begin{split}
        \epsnh=-2\int_{0}^{\rho} dy\;y\left[{I_{-2i\omega}(2\rho)K_{-2i\omega}(2y)-I_{-2i\omega}(2y)K_{-2i\omega}(2\rho)}\right]R(y),
    \end{split}
\end{align}
noting that $\mathfrak{W}(I_{-2i\omega},K_{-2i\omega})(x)=-\frac{1}{x}$. For large $|\omega_R|$, we apply the change of variable $y=|\omega_R|x$ and get
\begin{align}
    \epsnh=-\frac{\alpha^2}{2}\int_{0}^{\xi}dx\;x\left[I_{\beta+i\alpha}(\alpha\xi)K_{\beta+i\alpha}(\alpha x)-I_{\beta+i\alpha}(\alpha x)K_{\beta+i\alpha}(\alpha\xi)\right]R\left(\frac{\alpha x}{2}\right),
\end{align}
\begin{align}
\begin{split}
    R\left(\frac{\alpha x}{2}\right)=&\frac{4\alpha x}{\alpha^2 x^2+4L}\left[2\left(\frac{d}{dx}I_{-2i\omega}\right)\Big|_{\alpha x}+\epsnh'\right]\\&-4\left[\omega^2\left(\frac{\alpha^2x^2}{4L^2}+\frac{2}{L}\right)-\frac{\alpha^2x^2}{\alpha^2x^2+4L}\right]\left(I_{-2i\omega}(\alpha x)+\epsnh\right),
\end{split}
\end{align}
taking $\xi=|\omega_R|^{-1}\rho$, $\alpha=2|\omega_R|$, $\beta=2\omega_I$. We apply Proposition \ref{prop: airy approximation of modified bessel for large frequency}: let
\begin{align}
    \begin{split}
        \Kinftyi(\zeta_1,\zeta_2):=2\pi^2e^{-\frac{i\pi }{6}}\Big[&(\mathrm{Ai}_+(\hat{\zeta_1})+\epsilon^{(+)}(\hat{\zeta_1}))(\mathrm{Ai}(\hat{\zeta_2})+\epsilon(\hat{\zeta_2}))\\&-(\mathrm{Ai}_+(\hat{\zeta_2})+\epsilon^{(+)}(\hat{\zeta_2}))(\mathrm{Ai}(\hat{\zeta_1})+\epsilon(\hat{\zeta_1})\Big],
    \end{split}
\end{align}
where $\zeta$, $\hat{\zeta}$ are defined in Proposition \ref{prop: airy approximation of modified bessel for large frequency}.
Taking $\alpha$ sufficiently large so that for any $x_1,x_2\geq0$, we find 
\begin{align}
\begin{split}
    &I_{\beta+i\alpha}(\alpha x_1)K_{\beta+i\alpha}(\alpha x_2)-I_{\beta+i\alpha}(\alpha x_2)K_{\beta+i\alpha}(\alpha x_1)\\&={\alpha^{-\frac{2}{3}} }\left(\frac{4\zeta(x_1)}{x_1^2-1}\right)^{\frac{1}{4}}\left(\frac{4\zeta(x_2)}{x_2^2-1}\right)^{\frac{1}{4}}\Kinftyi(\zeta(x_1),\zeta(x_2)).
\end{split}
\end{align}
Thus
\begin{align}
    \begin{split}
        \epsnh=-{\alpha }^{\frac{4}{3}}\int_{0}^{\xi}dx\;&{x}\left(\frac{4\zeta(\xi)}{\xi^2-1}\right)^{\frac{1}{4}}\left(\frac{4\zeta(x)}{x^2-1}\right)^{\frac{1}{4}}\Kinftyi(\zeta(\xi),\zeta(x))R\left(\frac{\alpha x}{2}\right),
    \end{split}
\end{align}
We change coordinates once again with $v=\zeta(x)$, $dx=x\left(\frac{v}{x^2-1}\right)^{\frac{1}{2}}dv$, and get
\begin{align}
    \begin{split}
        \epsnh=-\frac{\alpha^{\frac{4}{3}}}{{2}}\int_{-\infty}^{\zeta}dv\,{x(v)^2}\left(\frac{4\zeta}{x(\zeta)^2-1}\right)^{\frac{1}{4}}\left(\frac{4v}{x(v)^2-1}\right)^{\frac{3}{4}}\Kinftyi(\zeta(\xi),\zeta(x))R\left(\frac{\alpha x(v)}{2}\right).
    \end{split}
\end{align}
We rewrite the above as
\begin{align}
    \begin{split}
        \epsnh=-\frac{\alpha^{\frac{4}{3}}}{{2}}\int_{-\infty}^{\zeta}dv\,{x(v)^2}\left(\frac{4v}{x(v)^2-1}\right)\Kinfty(\zeta(\xi),\zeta(x))R\left(\frac{\alpha x(v)}{2}\right),
    \end{split}
\end{align}
where
\begin{align}
    \Kinfty(\zeta,v)=\left(\frac{4\zeta}{x(\zeta)^2-1}\right)^{\frac{1}{4}}\left(\frac{4v}{x(v)^2-1}\right)^{-\frac{1}{4}}\Kinftyi(\zeta,v).
\end{align}
For large $\alpha$, we use Proposition \ref{prop: airy approximation of modified bessel for large frequency} and Remark \ref{airy approximation of the derivative of modified bessel} to write $R$ as
\begin{align}
    \begin{split}
        &R\left(\frac{\alpha x}{2}\right)\\&=\frac{4\alpha x}{\alpha^2 x^2+4L}\left((1+o(1))4\pi \alpha^{-\frac{2}{3}}e^{\frac{\alpha\pi}{2}-\frac{i\pi }{2}\beta-\frac{5i\pi}{6}}\frac{1}{x}\left(\frac{4\zeta(x)}{x^2-1}\right)^{-\frac{1}{4}}\mathrm{Ai}_+'\left(\hat{\zeta}(x)\right)+\epsnh'\right)\\&-4\left[\omega^2\left(\frac{\alpha^2x^2}{4L^2}+\frac{2}{L}\right)-\frac{\alpha^2x^2}{\alpha^2x^2+4L}\right]\left((1+o(1))2\pi\alpha^{-\frac{1}{3}}e^{\frac{\alpha\pi}{2}-\frac{i\pi }{2}\beta-\frac{i\pi}{6}}\left(\frac{4\zeta(x)}{x^2-1}\right)^{\frac{1}{4}}\mathrm{Ai}_+\left(\hat{\zeta}(x)\right)+\epsnh\right).
    \end{split}
\end{align}
The symbol $'$ on $\epsilon$ refers to the derivative of $\epsilon$ as a function of $\rho$, whereas the symbol $'$ on $\mathrm{Ai}_+$ refers to the derivative of $\mathrm{Ai}_+$ as a function of $\hat{\zeta}$. 

Assume $\epsilon_i$ satisfies the equation
\begin{align}
    \begin{split}
        \epsilon_i=-\frac{\alpha^{\frac{4}{3}}}{\sqrt{2}}\int_{-\infty}^{\zeta}dv\,{x(v)^2}\left(\frac{4v}{x(v)^2-1}\right)\Kinfty(\zeta(\xi),\zeta(x))R_i\left(\frac{\alpha x(v)}{2}\right).
    \end{split}
\end{align}
and we take
\begin{align}
\begin{split}
    R_1\left(\frac{\alpha x}{2}\right)=&\frac{4\alpha x}{\alpha^2 x^2+4L}\left((1+o(1))4\pi \alpha^{-\frac{2}{3}}e^{\frac{\alpha\pi}{2}-\frac{4\pi i}{3}}\frac{1}{x}\left(\frac{4\zeta}{x^2-1}\right)^{-\frac{1}{4}}\mathrm{Ai}_+'\left(\hat{\zeta}(x)\right)\right)\\&+\frac{4\alpha x}{\alpha^2x^2+4L}\epsilon_1'-4\left[\omega^2\left(\frac{\alpha^2x^2}{4L^2}+\frac{2}{L}\right)-\frac{\alpha^2x^2}{\alpha^2x^2+4L}\right]\epsilon_1,
\end{split}
\end{align}
\begin{align}
\begin{split}
    R_2\left(\frac{\alpha x}{2}\right)=&-4\left[\omega^2\left(\frac{\alpha^2x^2}{4L^2}+\frac{2}{L}\right)-\frac{\alpha^2x^2}{\alpha^2x^2+4L}\right]\left((1+o(1))2\pi\alpha^{-\frac{1}{3}}e^{\frac{\alpha\pi}{2}-\frac{2\pi i}{3}}\left(\frac{4\zeta}{x^2-1}\right)^{\frac{1}{4}}\mathrm{Ai}_+\left(\hat{\zeta}(x)\right)\right)\\&+\frac{4\alpha x}{\alpha^2x^2+4L}\epsilon_2'-4\left[\omega^2\left(\frac{\alpha^2x^2}{4L^2}+\frac{2}{L}\right)-\frac{\alpha^2x^2}{\alpha^2x^2+4L}\right]\epsilon_2.
\end{split}
\end{align}
It is clear that
\begin{align}
    \epsnh=\epsilon_1+\epsilon_2.
\end{align}
We apply Theorem \ref{Olver's error approx theorem} to estimate $\epsilon_i, \epsilon_i'$. Starting with $\epsilon_1$, we set
\begin{align}
    \mathrm{K}=\Kinfty,
\end{align}
\begin{align}
    P_0(\zeta)=e^{\frac{\alpha\pi}{2}}\left(\frac{4\zeta}{x(\zeta)^2-1}\right)^{\frac{1}{4}}M_-(\hat{\zeta})E_+^{-1}(\hat{\zeta}), 
\end{align}
\begin{align}
    P_1(\zeta)=\alpha^{\frac{2}{3}}e^{\frac{\alpha\pi}{2}}\left(\frac{4\zeta}{x(\zeta)^2-1}\right)^{\frac{1}{4}}N_-(\hat{\zeta})E_+^{-1}(\hat{\zeta}),
\end{align}
\begin{align}
    Q(v)=e^{-\frac{\alpha\pi}{2}}\left(\frac{4v}{x(v)^2-1}\right)^{-\frac{1}{4}}M_-(\hat{v})E_+(\hat{v}),
\end{align}
\begin{align}
    \upphi(\zeta)=\frac{4\alpha}{\alpha^2 x(\zeta)^2+4L}\cdot\frac{4\zeta x(\zeta)^2}{(x(\zeta)^2-1)},\qquad J(v)=\alpha^{\frac{2}{3}}e^{\frac{\alpha\pi}{2}}\left(\frac{4v}{x(v)^2-1}\right)^{-\frac{1}{4}}\mathrm{Ai}_+',
\end{align}
\begin{align}
    \uppsi_0=-4\alpha^{\frac{4}{3}}\left[\omega^2\left(\frac{\alpha^2x^2}{4L^2}+\frac{2}{L}\right)-\frac{\alpha^2x^2}{\alpha^2x^2+4L}\right]\frac{4vx(v)^2}{x(v)^2-1},
\end{align}
\begin{align}
    \uppsi_1(v)=\frac{4\alpha^{\frac{7}{3}} x(v)}{\alpha^2x(v)^2+4L}\cdot \frac{4vx(v)^2}{x(v)^2-1}.
\end{align}
We note
\begin{align}
    \kappa=\sup Q(v)|J(v)|\leq \alpha^{\frac{2}{3}}\sup |M_-(\hat{v})N_-(\hat{v})|\lesssim \alpha^{\frac{2}{3}},
\end{align}
\begin{align}
     \kappa_0=\sup P_0(v) Q(v)\leq \sup M_-(\hat{v})^2\lesssim 1,
\end{align}
\begin{align}
     \kappa_1=\sup P_1(v) Q(v)\leq \alpha^{\frac{2}{3}}\sup N_-(\hat{v})M_-(\hat{v})\lesssim \alpha^{\frac{2}{3}}.
\end{align}
We now take $1\ll\alpha^2\ll \xi\ll L$, $\alpha^2\xi^2\ll L$,  and estimate
\begin{align}
    \int_{-\infty}^{\zeta}dv\, \frac{4\alpha}{\alpha^2x(v)^2+4L}\frac{4v x(v)^2}{x(v)^2-1}\lesssim \int_{0}^{\xi}dx\,\frac{|v|^{\frac{1}{2}}}{|x^2-1|^{\frac{1}{2}}}\frac{\alpha x}{L}\lesssim \frac{\alpha \xi^{\frac{4}{3}}}{L}
\end{align}
\begin{align}
    \int_{-\infty}^{\zeta}dv \frac{4|\omega|^2|v|x(v)^2}{x(v)^2-1}\frac{\alpha^2 x(v)^2}{L^2}\lesssim \int_{0}^{\xi}dx\,\frac{|v(x)|^{\frac{1}{2}}x}{|x(v)^2-1|^{\frac{1}{2}}}\frac{\alpha^4x^2}{L^2}\lesssim \frac{\alpha^4\xi^{\frac{10}{3}}}{L^2},
\end{align}
\begin{align}
    \int_{-\infty}^{\zeta}dv\, \frac{vx^2}{x^2-1}\frac{|\omega|^2}{L}\lesssim \int_{0}^{\xi}dx\,\frac{\alpha^2|v(x)|^{\frac{1}{2}}x}{L|x^2-1|^{\frac{1}{2}}}\lesssim \frac{\alpha^2\xi^{\frac{4}{3}}}{L},
\end{align}
\begin{align}
    \int_{-\infty}^{\zeta}dv\,\frac{\alpha^2x(v)^2}{\alpha^2x(v)^2+4L}\frac{v x(v)^2}{x(v)^2-1}\lesssim \int_{0}^{\xi}dx\, \frac{|v(x)|^{\frac{1}{2}}}{|x^2-1|^{\frac{1}{2}}}\frac{\alpha^2 x^2}{\alpha^2x^2+4L}\lesssim \frac{\alpha^2 \xi^{\frac{7}{3}}}{L},
\end{align}
\begin{align}
\begin{split}
     \int_{-\infty}^{\zeta}dv\,\frac{4\alpha x(v)}{\alpha^2x(v)^2+4L}\cdot \frac{4vx(v)^2}{x(v)^2-1}&\lesssim \int_{0}^{\xi}dx\, \frac{|x^2-1|^{\frac{1}{2}}}{|v|^{\frac{1}{2}}}\frac{\alpha v(x) x^2}{L|x^2-1|}\\&\lesssim\int_{0}^{\xi}dx\, \frac{\alpha x^{\frac{4}{3}}}{L}\lesssim \frac{\alpha \xi^{\frac{7}{3}}}{L}.
\end{split}
\end{align}
Therefore,
\begin{align}
    \int_{-\infty}^{\zeta}dv \,|\upphi|\lesssim \frac{ \alpha \xi^{\frac{4}{3}}}{L},
\end{align}
\begin{align}
    \alpha^{-\frac{4}{3}}\int_{-\infty}^{\zeta}dv\,|\uppsi_0|\lesssim \frac{\alpha^4 \xi^{\frac{10}{3}}}{L^2}+\frac{\alpha^2 \xi^{\frac{4}{3}}}{L}+\frac{\alpha^2 \xi^{\frac{7}{3}}}{L},\qquad\quad  \int_{-\infty}^{\zeta}dv\,|\uppsi_1|\lesssim \frac{\alpha^{\frac{7}{3}} \xi^{\frac{7}{3}}}{L}.
\end{align}
We now consider $\epsilon_2$. Keeping $P_0$, $P_1$, $Q$, $\uppsi_0$, $\uppsi_1$ as before, we take
\begin{align}
    \upphi=-4\alpha^{\frac{4}{3}}\left[\omega^2\left(\frac{\alpha^2x^2}{4L^2}+\frac{2}{L}\right)-\frac{\alpha^2x^2}{\alpha^2x^2+4L}\right]\frac{4\zeta x(\zeta)^2}{(x(\zeta)^2-1},
\end{align}
\begin{align}
    J(v)=2\pi\alpha e^{\frac{\alpha \pi}{2}}\left(\frac{4v}{x(v)^2-1}\right)^{\frac{1}{4}}\mathrm{Ai}_+(\hat{v}).
\end{align}
As before, we estimate
\begin{align}
    \kappa=\sup Q(v)|J(v)|\lesssim \alpha,\qquad \alpha^{-\frac{4}{3}}\int_{-\infty}^{\zeta}dv\,|\upphi|\lesssim \frac{\alpha^4 \xi^{\frac{10}{3}}}{L^2}+\frac{\alpha^2 \xi^{\frac{4}{3}}}{L}+\frac{\alpha^2 \xi^{\frac{7}{3}}}{L}.
\end{align}
We choose $\rho^3= \deltnh^3L$, which gives $\xi^3=\deltnh^3\alpha^{-3}L$. Recalling that we take $L\geq C|\omega|^9$ with $C>1$, we compute for $\epsilon_1$,
\begin{align}
    \kappa\int_{-\infty}^{\zeta}dv\,|\upphi|\lesssim \alpha^{2}L^{-1}\xi^{\frac{4}{3}}\leq C^{-\frac{1}{9}}\deltnh^\frac{4}{3}L^{-\frac{4}{9}},
\end{align}
\begin{align}
\begin{split}
    \kappa_0\int_{-\infty}^{\zeta}dv\,|\uppsi_0|&\lesssim C^{-\frac{2}{9}}\deltnh^{\frac{10}{3}}L^{-\frac{2}{3}}+C^{-\frac{2}{9}}\deltnh^{\frac{4}{3}}L^{-\frac{13}{27}}+C^{-\frac{1}{9}}\deltnh^{\frac{7}{3}}L^{-\frac{1}{9}}\\&\lesssim C^{-\frac{1}{9}}\deltnh^{\frac{4}{3}}L^{-\frac{1}{9}},
\end{split}
\end{align}
\begin{align}
    \kappa_1\int_{-\infty}^{\zeta}dv\,|\uppsi_1|\lesssim C^{-\frac{2}{27}}\deltnh^{\frac{7}{3}}L^{-\frac{4}{27}}.
\end{align}
For $\epsilon_2$, we have
\begin{align}
    \kappa\int_{-\infty}^{\zeta}dv\,|\upphi|\lesssim C^{-\frac{2}{9}}\deltnh^{\frac{4}{9}}.
\end{align}
We have that $|\mathrm{Ai_+}(\hat{\zeta})|\geq c>0$, since $\hat{\zeta}$ is entirely contained in the first and second quadrant, and therefore the distance between $\hat{\zeta}(\zeta)$ and the nearest zero of $\mathrm{Ai}_+$ is properly bounded below away \mbox{from $0$.} Theorem \ref{Olver's error approx theorem} then implies that for $L>\alpha^9$, $\alpha\gg1$, the following holds:
\begin{align}
    |\epsnh(\alpha \xi)|=O(C^{-\frac{2}{9}}\deltnh^{\frac{4}{9}})\left|e^{\frac{\alpha\pi}{2}}\left(\frac{4\zeta}{x^2-1}\right)^{\frac{1}{4}}\mathrm{Ai}_+(\hat{\zeta})\right|,
\end{align}
\begin{align}
    \left|\frac{d}{dx}\epsnh( \alpha \xi)\right|=O(C^{-\frac{2}{9}}\deltnh^{\frac{4}{9}})\left|\alpha^{\frac{2}{3}}e^{\frac{\alpha\pi}{2}}\left(\frac{4\zeta}{x^2-1}\right)^{\frac{1}{4}}\mathrm{Ai}'_+(\hat{\zeta})\right|.
\end{align}
Therefore,
\begin{align}
    |\epsnh(\rho)|/\left|e^{\frac{\alpha\pi}{2}}\left(\frac{4\zeta(\rho)}{\alpha^{-2}x^2-1}\right)^{\frac{1}{4}}\mathrm{Ai}_+(\hat{\zeta}(x\alpha^{-1})\right|=O(C^{-\frac{2}{9}}\deltnh^{\frac{4}{9}}),
\end{align}
\begin{align}
    |\epsnh'( x)|/\left|\alpha^{-\frac{1}{3}}\left(\frac{4\zeta}{\alpha^{-2}x^2-1}\right)^{\frac{1}{4}}e^{\frac{\alpha\pi}{2}}\mathrm{Ai}_+'(\hat{\zeta}(x\alpha^{-1}))\right|=O(C^{-\frac{2}{9}}\deltnh^{\frac{4}{9}}).
\end{align}
for $L\gtrsim |\omega_R|^9$, $|\omega_R|\gg1$.

For bounded $\omega_R$, we may use the asymptotic forms of $I_{\beta+i\alpha}$, $K_{\beta+i\alpha}$ for large $x$ and for small $x$ (See Section 10.40(i) of \cite{dlmf}),
\begin{align}\label{bounded omega error control functions 1}
    \begin{split}
        |I_{\beta+i\alpha}(x)|\lesssim \sqrt{\frac{x^2}{x^3+1}}e^{x},\qquad\qquad
        x|K_{\beta+i\alpha}(x)|\lesssim \sqrt{{x+1}}e^{-x},
    \end{split}
\end{align}
\begin{align}
    \begin{split}
        |I_{\beta+i\alpha}'(x)|\lesssim \sqrt{\frac{1}{x+1}}e^{x},\qquad\qquad
        |I_{\beta+i\alpha}''(x)|\lesssim \sqrt{\frac{x+1}{x^2}}e^{x},
    \end{split}
\end{align}
We apply Theorem \ref{Olver's error approx theorem} with $P_0$, $Q$ the right hand sides of \eqref{bounded omega error control functions 1}, $P_1=\sqrt{\frac{1}{x+1}}e^{x}$, $P_2=\sqrt{\frac{x+1}{x^2}}e^{x}$, and
\begin{align}
    \mathrm{K}(\rho,y)=y\left[I_{\beta+i\alpha}(2\rho)K_{\beta+i\alpha}(2y)-I_{\beta+i\alpha}(2y)K_{\beta+i\alpha}(2\rho)\right].
\end{align}
It is clear that for $\rho>y$, $\mathrm{K}(\rho,y)\leq P_0(2\rho)Q(2y)$, $\partial_{\rho}\mathrm{K}(\rho,y)\leq P_1(2\rho)Q(2y)$, $\partial^2_{\rho}\mathrm{K}(\rho,y)\leq P_2(2\rho)Q(2y)$. We have
\begin{align}
    \sup_{y\in[0,\rho]} P_0(2y)Q(2y), \sup_{y\in[0,\rho]} P_1(2y)Q(2y)=O(1)
\end{align}
We split $\epsnh$ again to $\epsnh=\epsilon_1+\epsilon_2$, such that
\begin{align}
    \epsilon_1=-2\int_{0}^{\rho}dy\,\mathrm{K}(\rho,y)\left\{\frac{2y}{y^2+L}[2I_{\beta+i\alpha}'(2y)+\epsilon_1']-4\left[\omega^2\left(\frac{y^2}{L^2}+\frac{2}{L}\right)-\frac{y^2}{y^2+L}+\frac{L}{(y^2+L)^2}\right]\epsilon_1\right\},
\end{align}
\begin{align}
    \epsilon_2=-2\int_{0}^{\rho}dy\,\mathrm{K}(\rho,y)\left\{\frac{2y}{y^2+L}\epsilon_2'-4\left[\omega^2\left(\frac{y^2}{L^2}+\frac{2}{L}\right)-\frac{y^2}{y^2+L}+\frac{L}{(y^2+L)^2}\right](I_{\beta+i\alpha}(2y)+\epsilon_2)\right\},
\end{align}
For $\epsilon_1$, take
\begin{align}
    \upphi=\uppsi_1=\frac{2y}{y^2+L}, \qquad J=2I_{\beta+i\alpha}'(2y),\qquad \uppsi_0=-4\left[\omega^2\left(\frac{y^2}{L^2}+\frac{2}{L}\right)-\frac{y^2}{y^2+L}+\frac{L}{(y^2+L)^2}\right],
\end{align}
we have $\sup_{y\in[0,\rho]} |J(2y)|Q(2y)=O(1)$. For $\epsilon_2$, we take
\begin{align}
    \upphi=\uppsi_0=-4\left[\omega^2\left(\frac{y^2}{L^2}+\frac{2}{L}\right)-\frac{y^2}{y^2+L}+\frac{L}{(y^2+L)^2}\right],\qquad J=I_{\beta+i\alpha}(2y),\qquad \uppsi_1=\frac{2y}{y^2+L}
\end{align}

Taking $\rho^3=\deltnh^3  L$ for the same $\deltnh$ as above, and applying Theorem \ref{Olver's error approx theorem} achieves the result. We then take outer boundary of the near horizon region is then 
\begin{align}\label{r_1}
    r_1-1= \deltnh\left(\frac{1}{L}\right)^{\frac{1}{3}}.
\end{align}
\end{proof}
\subsubsection{Large $\ell$ approximation of $\hor$ in the intermediate region}\label{sec: intermediate region hor}
\begin{proposition}\label{prop: approximation prop intermediate region}
    For $\omega_I\geq\frac{1}{2}$, there exists $C>1$ and $0<\delta<1$ chosen suitably small such that for any $\omega$ with $\Im\omega=\omega_I$ and all $\ell$ such that $\ell(\ell+1)>C|\omega|^9$, there exist solutions to \eqref{ode} $u_1$, $u_2$ such that
    \begin{align}
        u_1=P_{\ell}(2r-1)+\epsilon_P,\qquad u_2=Q_{\ell}(2r-1)+\epsilon_Q,
    \end{align}
    where
    \begin{align}
      \sup_{x\in[r_1,r]}|\epsilon_P|/P_{\ell}(2x-1), \;\sup_{x\in[r_1,r]}|\epsilon_P'|/P'_{\ell}(2x-1)=O(\delta+\deltnh^{-\frac{1}{2}}C^{-\frac{2}{9}}L^{-\frac{1}{9}}),
    \end{align}
    \begin{align}
      \sup_{x\in[r_1,r]}|\epsilon_Q|/Q_{\ell}(2x-1), \;\sup_{x\in[r_1,r]}|\epsilon_Q'|/Q'_{\ell}(2x-1)=O(\delta+\deltnh^{-\frac{1}{2}}C^{-\frac{2}{9}}L^{-\frac{1}{9}}),
    \end{align}  
    for $r\in[\tilde{r}_1,r_2]$, with $\tilde{r}_1=1+\frac{1}{2}\deltnh L^{-\frac{1}{3}}$, $r_2=\delta |\omega|^{-1}L^{\frac{1}{4}}$.
\end{proposition}
\begin{proof}
For $r\in\left[\tilde{r}_1, r_2\right]$, we model equation \eqref{ode} after Legendre's equation, taking $u=u_{L}+\epsilon$, with
    \begin{align}
        \begin{split}
            \frac{d^2}{dr_*^2}u_L-V_{\ell}u_L=0.
        \end{split}
    \end{align}
Note that 
    \begin{align}
        u_L(r)=c_pP_{\ell}(2r-1)+c_q Q_{\ell}(2r-1),
    \end{align}
    where $P_{\ell}, Q_{\ell}$ are Legendre functions of the first and second kinds, respectively. We now use the following approximation for Legendre functions (see Section 14.15(iii) in \cite{dlmf})
    \begin{align}
        \begin{split}
            P_{\ell}(\cosh\xi)=\left(\frac{\xi}{\sinh{\xi}}\right)^{\frac{1}{2}}I_0\left[\left(\ell+\frac{1}{2}\right)\xi\right]\left(1+O(\ell^{-1})\right),
        \end{split}
    \end{align}
    \begin{align}
        \begin{split}
            Q_{\ell}(\cosh\xi)=\left(\frac{\xi}{\sinh{\xi}}\right)^{\frac{1}{2}}K_0\left[\left(\ell+\frac{1}{2}\right)\xi\right]\left(1+O(\ell^{-1})\right),
        \end{split}
    \end{align}
    where $\xi$ is such that $2r-1=\cosh\xi$, so we have $\sinh\xi=2\sqrt{r(r-1)}$, $\xi=\log[2r-1+2\sqrt{r(r-1)}]$. For $r\geq \tilde{r}_1$, we have
    \begin{align}
        \left(\ell+\frac{1}{2}\right)\xi\gtrsim L^{\frac{1}{2}}\cdot \left(\frac{1}{2}\deltnh L^{-\frac{1}{3}}+\sqrt{2\deltnh}L^{-\frac{1}{6}}\right)\gg1.
    \end{align}
    Therefore, we can use the large $\xi$ approximation of $I_0$, $K_0$, and find
    \begin{align}\label{asymptotic formula of Legendre in terms of exponential}
    \begin{split}
        P_{\ell}(2r-1)&= \frac{1}{\sqrt{2\pi(\ell+1/2)\sinh\xi}}e^{\left(\ell+\frac{1}{2}\right)\xi}\left(1+O\left(\frac{1}{(\ell+1)\xi}\right)\right)\left(1+O(\ell^{-1})\right)\\&= \frac{1}{(4r(r-1))^{\frac{1}{4}}}\frac{[2r-1+2\sqrt{r(r-1)}]^{\ell+\frac{1}{2}}}{(\pi(2\ell+1))^{\frac{1}{2}}}\left[1+O\left(\deltnh^{-\frac{1}{2}}L^{-\frac{1}{3}}\right)\right],
    \end{split}
    \end{align}
    \begin{align}\label{Q in terms of super exponential}
    \begin{split}
        Q_{\ell}(2r-1)&= \frac{\sqrt{\pi}}{\sqrt{2(\ell+1/2)\sinh\xi}}e^{-\left(\ell+\frac{1}{2}\right)\xi}\left(1+O\left(\frac{1}{(\ell+1)\xi}\right)\right)\left(1+O(\ell^{-1})\right)\\&= \frac{\sqrt{\pi}}{(4r(r-1))^{\frac{1}{4}}}\frac{[2r-1+2\sqrt{r(r-1)}]^{-(\ell+\frac{1}{2})}}{(2\ell+1)^{\frac{1}{2}}}\left[1+O\left(\deltnh^{-\frac{1}{2}}L^{-\frac{1}{3}}\right)\right].
    \end{split}
    \end{align}
    For the derivatives of $P_{\ell}$ at $x>1$, we use (see Section 14.6(i) of \cite{dlmf}),
    \begin{align}
        \frac{d^m}{dx^m}P_{\ell}=\frac{1}{{(x^2-1)}^{\frac{m}{2}}}P_{\ell}^{(m)},
    \end{align}
    together with the uniform asymptotic expansion (see Section 14.15(iii) of \cite{dlmf}),
    \begin{align}
        P_{\ell}^{(m)}=\frac{1}{\ell^m}\left(\frac{\xi}{\sinh\xi}\right)^{\frac{1}{2}}I_{m}\left[\left(\ell+\frac{1}{2}\right)\xi\right](1+O(\ell^{-1})),
    \end{align}
    to find 
    \begin{align}\label{derivative of legendre in terms of bessel}
    \begin{split}
        P_{\ell}'(2r-1)&=\frac{1}{(r(r-1))^{\frac{1}{2}}}\frac{1}{\ell\sqrt{\pi(2\ell+1)\sinh\xi}}e^{(\ell+\frac{1}{2})\xi}\left(1+O\left(\frac{1}{(\ell+1)\xi}\right)\right)\left(1+O(\ell^{-1})\right)\\&=\frac{1}{\sqrt{2}(r(r-1))^{\frac{3}{4}}}\frac{[2r-1+2\sqrt{r(r-1)}]^{\ell+\frac{1}{2}}}{\ell\sqrt{\pi(2\ell+1)}}\left[1+O\left(\deltnh^{-\frac{1}{2}}L^{-\frac{1}{3}}\right)\right],
    \end{split}
    \end{align}
    \begin{align}\label{second derivative of legendre in terms of bessel}
        P_{\ell}''(2r-1)=\frac{1}{\sqrt{2}(r(r-1))^{\frac{5}{4}}}\frac{[2r-1+2\sqrt{r(r-1)}]^{\ell+\frac{1}{2}}}{\ell^2\sqrt{\pi(2\ell+1)}}\left[1+O\left(\deltnh^{-\frac{1}{2}}L^{-\frac{1}{3}}\right)\right].
    \end{align}
    Similarly, for $Q_{\ell}$ we have
    \begin{align}
         \frac{d^m}{dx^m}Q_{\ell}=\frac{1}{{(x^2-1)}^{\frac{m}{2}}}Q_{\ell}^{(m)},
    \end{align}
    \begin{align}
        Q_{\ell}^{(m)}=\ell^m\left(\frac{\xi}{\sinh\xi}\right)^{\frac{1}{2}}K_{m}\left[\left(\ell+\frac{1}{2}\right)\xi\right](1+O(\ell^{-1})).
    \end{align}
    Therefore,
    \begin{align}\label{derivative of Q in terms of Q}
    \begin{split}
        Q_{\ell}'(2r-1)&=\frac{\ell\sqrt{\pi}}{\sqrt{r(r-1)}}\frac{1}{\sqrt{(2\ell+1)\sinh\xi}}e^{-(\ell+\frac{1}{2})\xi}\left(1+O\left(\frac{1}{(\ell+1)\xi}\right)\right)\left(1+O(\ell^{-1})\right)\\&=\frac{\sqrt{\pi\ell}}{\sqrt{2}(r(r-1))^{\frac{3}{4}}}[2r-1+2\sqrt{r(r-1)}]^{-(\ell+\frac{1}{2})}\left[1+O\left(\deltnh^{-\frac{1}{2}}L^{-\frac{1}{3}}\right)\right].
    \end{split}
    \end{align}
    The Wronskian of $P_{\ell}(2r-1)$, $Q_{\ell}(2r-1)$ is
    \begin{align}
        \mathfrak{W}(P_{\ell}(2r-1),Q_{\ell}(2r-1))=-\frac{1}{4r(r-1)}
    \end{align}
    Taking $u=P_{\ell}+\epsilon_P$, variation of parameters gives
    \begin{align}
        \epsilon_P=-4\omega^2\int_{\tilde{r}_1}^{r}dx_{*}\, r(x_{*})(r(x_{*})-1) [P_{\ell}(2r-1)Q_{\ell}(2r(x_{*})-1)-Q_{\ell}(2r-1)P_{\ell}(2r(x_{*})-1)](P_{\ell}(2r(x_{*})-1)+\epsilon).
    \end{align}
     Changing variables to $x=r(x_*)$, we have
    \begin{align}
        \epsilon_P=-4\omega^2\int_{\tilde{r}_1}^{r}dx\, x^2 [P_{\ell}(2r-1)Q_{\ell}(2x-1)-Q_{\ell}(2r-1)P_{\ell}(2x-1)](P_{\ell}(2x-1)+\epsilon).
    \end{align}
    We apply Theorem \ref{Olver's error approx theorem} with $P(x)=P_{\ell}(2x-1)$, $Q(x)=xQ_{\ell}(2r-1)$, $P_1=P_{\ell}'(2x-1)$, $J(x)=P_{\ell}(2x-1)$, $\upphi=\uppsi_0=\omega^2 x$, $\uppsi_1=0$. First, we consider $r\leq 2$, and we find
    \begin{align}
        \kappa=\kappa_0=\sup_{x\in[\tilde{r}_1,2]} xP_{\ell}(2x-1)Q_{\ell}(2x-1)\lesssim \deltnh^{-\frac{1}{2}}L^{-\frac{1}{3}},\qquad         \int_{\tilde{r}_1}^{r}dx |\upphi|\lesssim |\omega|^2.
    \end{align}
    We deduce 
    \begin{align}\label{legendre small area error}
        |\epsilon_P|/P_{\ell}(2r-1), |\epsilon_P'|/P'_{\ell}(2r-1)\lesssim \frac{|\omega|^2}{\deltnh^{\frac{1}{2}}L^{\frac{1}{3}}}\lesssim \deltnh^{-\frac{1}{2}}C^{-\frac{2}{9}}L^{-\frac{1}{9}}.
    \end{align}
    We now consider the estimate on $r\geq 2$. As before, the error is given by
    \begin{align}
        \epsilon_P=-4\omega^2\int_{2}^{r}dx\, x^2 [P_{\ell}(2r-1)Q_{\ell}(2x-1)-Q_{\ell}(2r-1)P_{\ell}(2x-1)](P_{\ell}(2x-1)+\epsilon),
    \end{align}
    and we find 
    \begin{align}\label{legendre large area error}
        \kappa=\kappa_0=\sup_{x\in[2,r]} xP_{\ell}(2x-1)Q_{\ell}(2x-1)\lesssim L^{-\frac{1}{2}},\qquad         \int_{2}^{r}dx |\upphi|\lesssim |\omega|^2r^2.
    \end{align}
    We then choose $r_2=\delta |\omega|^{-1}L^{\frac{1}{2}}$, which gives $\kappa_0\int_{2}^{r}dx |\upphi|\lesssim |\omega|^2r^2\lesssim \delta$. Combining \eqref{legendre small area error} and \eqref{legendre large area error}, we find
    \begin{align}
        \sup_{x\in[\tilde{r}_1,r]}|\epsilon_P|/P_{\ell}(2x-1), \;\sup_{x\in[\tilde{r}_1,r]}|\epsilon_P'|/P'_{\ell}(2x-1)=O(\delta+\deltnh^{-\frac{1}{2}}C^{-\frac{2}{9}}L^{-\frac{1}{9}}).
    \end{align}
    We may repeat the same for $Q$ and find 
    \begin{align}
        \sup_{x\in[\tilde{r}_1,r]}|\epsilon_Q|/Q_{\ell}(2x-1), \sup_{x\in[\tilde{r}_1,r]}|\epsilon_Q'|/Q'_{\ell}(2x-1)=O(\delta+\deltnh^{-\frac{1}{2}}C^{-\frac{2}{9}}L^{-\frac{1}{9}})
    \end{align}
    for $r_2=\delta |\omega|^{-1}L^{\frac{1}{4}}$, $L>C|\omega|^9$.
\end{proof}   
    \subsubsection{Large $\ell$ approximation of $\hor$ near infinity}\label{sec: far horizon hor}
    \begin{proposition}\label{prop: approximation prop far region}
     For $\omega_I\geq\frac{1}{2}$, there exists a constant $C$ such that, for any $\omega$ with $\Im\omega=\omega_I$,  and for $\ell(\ell+1)>C|\omega|^9$, we have
    \begin{align}
        \out=\out^{\mathrm{WKB}}(1+\epsilon_{\infty}^{(1)}),\qquad \outb=\outb^{\mathrm{WKB}}(1+\epsilon_{\infty}^{(2)}),
    \end{align}    
    where
   \begin{align}
        U_{\mathrm{inf}}^{\mathrm{WKB}}=\left(\frac{\omega^2}{\omega^2-f_{\ell}}\right)^{\frac{1}{4}}\exp\left\{\int_{r_*}^{\infty}dy(\sqrt{f_{\ell}-\omega^2}+i\omega)\right\}\exp(i\omega r_*),
    \end{align}
    \begin{align}
        \overline{U}_{\mathrm{inf}}^{\mathrm{WKB}}=\left(\frac{\omega^2}{\omega^2-f_{\ell}}\right)^{\frac{1}{4}}\exp\left\{-\int_{r_*}^{\infty}dy(\sqrt{f_{\ell}-\omega^2}+i\omega)\right\}\exp(-i\omega r_*),
    \end{align}
    the square root above taking the principal value, and we have
    \begin{align}
      |\epsilon_{\infty}^{(1)}|/\out^{\mathrm{WKB}},\;|\epsilon_{\infty}^{(1)}{}'|/(\out^{\mathrm{WKB}})'=O\left(C^{-\frac{5}{18}}\delta^{-\frac{2}{3}}L^{-\frac{1}{9}}\right),
    \end{align}    
    \begin{align}
      |\epsilon_{\infty}^{(2)}|/\outb^{\mathrm{WKB}},\;|\epsilon_{\infty}^{(2)}{}'|/(\outb^{\mathrm{WKB}})'=O\left(C^{-\frac{5}{18}}\delta^{-\frac{2}{3}}L^{-\frac{1}{9}}\right),
    \end{align}  
    for $r\geq \tilde{r}_2$, with $\tilde{r}_2=\frac{\delta}{2} |\omega|^{-1}L^{\frac{1}{4}}$, where $\delta$, $L$ are such that $\tilde{r}_2>20$.
\end{proposition}
\begin{remark}
    We can ensure $\tilde{r}_2>20$ by restricting $\delta$, $C$ to satisfy ${\delta C^{\frac{1}{4}}}>20\times 2^{\frac{9}{4}}$.
\end{remark}
\begin{proof}[Proof of Proposition \ref{prop: approximation prop far region}]
    For the far region, we deploy the WKB approximation.
    We write the equation \eqref{ode} as
    \begin{align}
        u''=(f_{\ell}-\omega^2+g)u,
    \end{align}
    where
    \begin{align}
        f_{\ell}=\ell(\ell+1)\frac{\Omega^2}{r^2},\qquad g=\frac{\Omega^2}{r^3}.
    \end{align}
    The WKB approximation of $\out$ and $\outb$ is
    $U_{\mathrm{inf}}^{\mathrm{WKB}}$ and $\overline{U}_{\mathrm{inf}}^{\mathrm{WKB}}$ are given by
    \begin{align}
        U_{\mathrm{inf}}^{\mathrm{WKB}}=\left(\frac{\omega^2}{\omega^2-f_{\ell}}\right)^{\frac{1}{4}}\exp\left\{\int_{r_*}^{\infty}dy(\sqrt{f_{\ell}-\omega^2}+i\omega)\right\}\exp(i\omega r_*),
    \end{align}
    \begin{align}
        \overline{U}_{\mathrm{inf}}^{\mathrm{WKB}}=\left(\frac{\omega^2}{\omega^2-f_{\ell}}\right)^{\frac{1}{4}}\exp\left\{-\int_{r_*}^{\infty}dy(\sqrt{f_{\ell}-\omega^2}+i\omega)\right\}\exp(-i\omega r_*),
    \end{align}
    where the square root above takes the principal branch. 
    
    Writing $\out=\out^{\mathrm{WKB}}(1+\epsilon_{\infty}^{(1)})$, $\outb=\outb^{\mathrm{WKB}}(1+\epsilon_{\infty}^{(2)})$, with $\epsilon_{\infty}^{(1)}$, $\epsilon_{\infty}^{(2)}$ satisfying
    \begin{align}   
         |\epsilon_{\infty}^{(i)}(r_*)|,\; |\epsilon_{\infty}^{(i)}{}'(r_*)|/|f_{\ell}-\omega^2|^{\frac{1}{2}}\leq \exp\left\{\int_{r_*}^{\infty}dy \left|\frac{1}{(f_{\ell}-\omega^2)^{\frac{1}{4}}}\frac{d^2}{dr_*^2}\frac{1}{(f_{\ell}-\omega^2)^{\frac{1}{4}}}+\frac{g}{(f_{\ell}-\omega^2)^{\frac{1}{2}}}\right|\right\}-1.
    \end{align}
    We compute
    \begin{align}
        \begin{split}
            (f_{\ell}-\omega^2)^{-\frac{1}{4}}\frac{d^2}{dr_*^2}(f_{\ell}-\omega^2)^{-\frac{1}{4}}=\frac{5(f_{\ell}-\omega^2)''}{16(f_{\ell}-\omega^2)^{\frac{3}{2}}}-\frac{[(f_{\ell}-\omega^2)']^2}{4(f_{\ell}-\omega^2)^{\frac{5}{2}}}.
        \end{split}
    \end{align}
    Since $\tilde{r}_2>2 $, there exists $c<1$ such that 
    \begin{align}
        \left|f_{\ell}-\frac{L}{r^2}\right|\leq c\frac{L}{r^2},\quad \left|f_{\ell}'-\frac{2L}{r^3}\right|\leq c\frac{2L}{r^3}, \quad \left|f_{\ell}''-\frac{6L}{r^4}\right|\leq c\frac{6L}{r^4},\quad  \left|g-\frac{1}{r^3}\right|\leq c\frac{1}{r^3}.
    \end{align}
    Note that
    \begin{align}
        \left|\frac{L}{r^2}-\omega^2\right|^2=\left(\frac{L}{r^2}-\omega_R^2+\omega_I^2\right)^2+4\omega_I^2\omega_R^2.
    \end{align}
    We start with
    \begin{align}
        \int_{\tilde{r}_*}^{\infty} dr_* \frac{1}{r^3\left|\frac{L}{r^2}-\omega^2\right|^{\frac{1}{2}}}.
    \end{align}
    with $\tilde{r}_*$ is the value of $r_*$ at $\tilde{r}_2$. For $|\omega_R|<\frac{1}{4}$ we have $\left|\frac{L}{r^2}-\omega^2\right|^2>\frac{1}{256}$, while for $|\omega_R|>\frac{1}{4}$ we have  $\left|\frac{L}{r^2}-\omega^2\right|^2>\omega_R^2$. Therefore,
    \begin{align}
        \left|\frac{L}{r^2}-\omega^2\right|^2\gtrsim |\omega|^2\gtrsim\frac{1}{4}+\omega_R^2.
    \end{align}
    Thus we have
    \begin{align}
        \int_{\tilde{r}_*}^{\infty} dr_* \frac{1}{r^3\left|\frac{L}{r^2}-\omega^2\right|^{\frac{1}{2}}}\lesssim \delta^{-1}|\omega|^{\frac{3}{2}} L^{-\frac{1}{2}}.
    \end{align}
    As for $\int_{\tilde{r}_*}^{\infty} dr_* \frac{L^2}{r^6\left|\frac{L}{r^2}-\omega^2\right|^{\frac{5}{2}}}$, take $u=L r^{-2}$, $du=-2L r^{-3}dr$, then we have
    \begin{align}
        \begin{split}
            \int_{\tilde{r}_*}^{\infty} dr_*\frac{L^2}{r^6\left|\frac{L}{r^2}-\omega^2\right|^{\frac{5}{2}}}\lesssim L^{-\frac{1}{2}}\int_{0}^{u|_{r_*=\tilde{r}_*}} du \frac{u^{\frac{3}{2}}}{\left[(u-\omega_R^2+\omega_I^2)^2+4\omega_I^2\omega_R^2\right]^{\frac{5}{4}}}.
        \end{split}
    \end{align}
    For large $|\omega_R|$, take $u=\omega_R^2v$, then 
    \begin{align}
        \begin{split}
            \int_{\tilde{r}_*}^{\infty} dr_*\frac{L^2}{r^6\left|\frac{L}{r^2}-\omega^2\right|^{\frac{5}{2}}}\lesssim L^{-\frac{1}{2}}\int_{0}^{\omega_R^2v|_{r_*=\tilde{r}_*}} dv \frac{v^{\frac{3}{2}}}{\left[(v-1+\frac{\omega_I^2}{\omega_R^2})^2+\frac{4\omega_I^2}{\omega_R^2}\right]^{\frac{5}{4}}}.
        \end{split}
    \end{align}
    The function $f(v)=v^{\frac{3}{2}}\left[(v-1+\frac{\omega_I^2}{\omega_R^2})^2+\frac{4\omega_I^2}{\omega_R^2}\right]^{-\frac{5}{4}}$ has a unique maximum at 
    \begin{align}
        v_{max}=\frac{\omega_I^2}{4\omega_R^2}-\frac{1}{4}+\frac{5}{4}\sqrt{1+\frac{46\omega_I^2}{25\omega_R^2}+\frac{\omega_I^4}{4\omega_R^4}}=1+\frac{\omega_I^2}{4\omega_R^2}+O\left(\frac{\omega_I^4}{\omega_R^{4}}\right),
    \end{align}
     which implies $f(v_{max})\lesssim \left(\frac{\omega_R}{\omega_I^2}\right)^{\frac{5}{2}}$. There exists $\omega_R^{(0)}>0$ such that for $|\omega_R|\geq \omega_R^{(0)}>0$ we have 
    \begin{align}
        \int_{0}^{2}dv \frac{v^{\frac{3}{2}}}{\left[(v-1+\frac{\omega_I^2}{\omega_R^2})^2+\frac{4\omega_I^2}{\omega_R^2}\right]^{\frac{5}{4}}}\lesssim \omega_R^{\frac{5}{2}},
    \end{align}
    \begin{align}
         \int_{2}^{4\delta^{-2}L^{\frac{1}{2}}}dv \frac{v^{\frac{3}{2}}}{\left[(v-1+\frac{\omega_I^2}{\omega_R^2})^2+\frac{4\omega_I^2}{\omega_R^2}\right]^{\frac{5}{4}}}\lesssim \int_{2}^{4\delta^{-2}L^{\frac{1}{2}}}dv\,\frac{1}{v}\lesssim \log\left(4\delta^{-2}L^{\frac{1}{2}}\right).
    \end{align}
    We then obtain for $|\omega_R|\geq \omega_R^{(0)}>0$
    \begin{align}
        \begin{split}
            \int_{\tilde{r}_*}^{\infty} dr_*\frac{L^2}{r^6\left|\frac{L}{r^2}-\omega^2\right|^{\frac{5}{2}}}&\lesssim  \omega_R^{\frac{5}{2}}L^{-\frac{1}{2}}\log\left(4\delta^{-2}L^{\frac{1}{2}}\right).
        \end{split}
    \end{align}
    For $\omega_R^2\leq\frac{1}{8}$, we have
    \begin{align}
        \int_{0}^{u|_{r_*=\tilde{r}_*}}du \frac{u^{\frac{3}{2}}}{\left[(u-\omega_R^2+\omega_I^2)^2+4\omega_I^2\omega_R^2\right]^{\frac{5}{4}}}\lesssim \int_{0}^{8\delta^{-2}L^{\frac{1}{2}}}du \frac{u^{\frac{3}{2}}}{(u+\frac{1}{4})^{\frac{5}{2}}}\lesssim 1+\log\left(8\delta^{-2}L^{\frac{1}{2}}\right)
    \end{align}
    For $\omega_R$ such that $\frac{1}{8}\leq \omega_R^2\leq (\omega_R^{(0)})^2$, we have
    \begin{align}
        \int_{0}^{4(\omega_R^{(0)})^2+1}du\frac{u^{\frac{3}{2}}}{[(u-\omega_R^2+\omega_I^2)^2+4\omega_I^2\omega_R^2]^{\frac{5}{4}}}\lesssim (4(\omega_R^{(0)})^2+1)^{\frac{5}{2}},
    \end{align}
    \begin{align}
    \begin{split}    \int_{4(\omega_R^{(0)})^2+1}^{u|_{r_*=\tilde{r}_*}}du\frac{u^{\frac{3}{2}}}{[(u-\omega_R^2+\omega_I^2)^2+4\omega_I^2\omega_R^2]^{\frac{5}{4}}}&\lesssim \int_{4(\omega_R^{(0)})^2+1}^{8\delta^{-2}L^{\frac{1}{2}}}du\frac{u^{\frac{1}{4}}}{(u-2\omega_R^2)^{\frac{5}{4}}}\\&\lesssim \int_{2(\omega_R^{(0)})^2+1}^{8\delta^{-2}L^{\frac{1}{2}}}dv\,\frac{1}{v}\\&\lesssim \log\left(8\delta^{-2}L^{\frac{1}{2}}\right).
    \end{split}
    \end{align}
    We conclude 
    \begin{align}
        \begin{split}
            \int_{\tilde{r}_*}^{\infty} dr_*\frac{L^2}{r^6\left|\frac{L}{r^2}-\omega^2\right|^{\frac{5}{2}}}&\lesssim |\omega|^{\frac{5}{2}}L^{-\frac{1}{2}}\log\left(4\delta^{-2}L^{\frac{1}{2}}\right)\\&\lesssim C^{-\frac{5}{18}}\delta^{-\frac{4}{9}}L^{-\frac{1}{9}}
        \end{split}
    \end{align}
    for $L\geq C|\omega|^9$. We now turn to the term $\int_{\tilde{r}_*}^{\infty} dr_* \frac{L}{r^4\left|\frac{L}{r^2}-\omega^2\right|^{\frac{3}{2}}}$, and we apply the substitution $u=Lr^{-2}$ to find
    \begin{align}
        \int_{\tilde{r}_*}^{\infty} dr_* \frac{L}{r^4\left|\frac{L}{r^2}-\omega^2\right|^{\frac{3}{2}}}\lesssim L^{-\frac{1}{2}}\int_{0}^{u|_{r_*=\tilde{r}_*}} du \frac{u^{\frac{1}{2}}}{\left[(u-\omega_R^2+\omega_I^2)^2+4\omega_I^2\omega_R^2\right]^{\frac{3}{4}}}.
    \end{align}
    We repeat the argument above and find for large $\omega_R$
    \begin{align}
        \int_{\tilde{r}_*}^{\infty} dr_* \frac{L}{r^4\left|\frac{L}{r^2}-\omega^2\right|^{\frac{3}{2}}}\lesssim \omega_R^{\frac{3}{2}}L^{-\frac{1}{2}}\log \left (4\delta^{-2}L^{\frac{1}{2}}\right), 
    \end{align}
    whereas for bounded $\omega_R$
    \begin{align}
        \int_{\tilde{r}_*}^{\infty} dr_* \frac{L}{r^4\left|\frac{L}{r^2}-\omega^2\right|^{\frac{3}{2}}}\lesssim L^{-\frac{1}{2}}\log\left(4 \delta^{-2}L^{\frac{1}{2}}\right).
    \end{align}
    Therefore,
    \begin{align}
        \begin{split}
             \int_{\tilde{r}_*}^{\infty} dr_* \frac{L}{r^4\left|\frac{L}{r^2}-\omega^2\right|^{\frac{3}{2}}}&\lesssim |\omega|^{\frac{3}{2}}L^{-\frac{1}{2}}\log\left(4\delta^{-2}L^{\frac{1}{2}}\right)\\&\lesssim C^{-\frac{1}{6}}\delta^{-\frac{2}{3}}L^{-\frac{1}{6}}
        \end{split}
    \end{align}
    The argument above shows that as $L\longrightarrow\infty$,
    \begin{align}
        |\epsilon_{\infty}^{(1)}|/\out^{\mathrm{WKB}},\;|\epsilon_{\infty}^{(1)}{}'|/(\out^{\mathrm{WKB}})'=O\left(C^{-\frac{5}{18}}\delta^{-\frac{2}{3}}L^{-\frac{1}{9}}\right),
    \end{align}
    \begin{align}
        |\epsilon_{\infty}^{(2)}|/\outb^{\mathrm{WKB}},\;|\epsilon_{\infty}^{(2)}{}'|/(\outb^{\mathrm{WKB}})'=O\left(C^{-\frac{5}{18}}\delta^{-\frac{2}{3}}L^{-\frac{1}{9}}\right).
    \end{align}
    \end{proof}
    \begin{remark}\label{happy fortuitous wronskian remark}
        Note that we have the following exact relation,
         \begin{align}\label{happy fortuitous wronskian}
            \mathfrak{W}(\outb^{\mathrm{WKB}},\out^{\mathrm{WKB}})=-2i\omega.
        \end{align}
    \end{remark}
    \subsubsection{The upper bound estimate on the transmission coefficient}\label{sec: transmission behaviour at large ell}
We are in position to estimate the transmission coefficient.
\begin{lemma}\label{Wronskian matching lemma}
    Assume that for $x\in[x_1,x_2]$ with $x_2>x_1$ the functions $f,g_1,g_2, \epsilon, \epsilon_1, \epsilon_2\in C^1$ satisfy
    \begin{align}
        f+\epsilon=c_1(g_1+\epsilon_1)+c_2(g_2+\epsilon_2)
    \end{align}
    for some constants $c_1,c_2$. Then $c_1$, $c_2$ are given by
    \begin{align}\label{wronskian matching c1}
        c_1=\frac{\mathfrak{W}(f,g_2)\left[1+\frac{f g_2}{\mathfrak{W}(f,g_2)}\frac{f'}{f}\left(\frac{\epsilon_2}{g_2}+\frac{\epsilon'}{f'}+\frac{\epsilon' \epsilon_2}{f' g_2}\right)-\frac{f g_2}{\mathfrak{W}(f,g_2)}\frac{g_2'}{g_2}\left(\frac{\epsilon_2'}{g_2'}+\frac{\epsilon}{f}+\frac{\epsilon \epsilon_2'}{f g_2'}\right)\right]}{\mathfrak{W}(g_1,g_2)\left[1+\frac{g_1 g_2}{\mathfrak{W}(g_1,g_2)}\frac{g_1'}{g_1}\left(\frac{\epsilon_2}{g_2}+\frac{\epsilon_1'}{g_1'}+\frac{\epsilon_1' \epsilon_2}{g_1' g_2}\right)-\frac{g_1 g_2}{\mathfrak{W}(g_1,g_2)}\frac{g_2'}{g_2}\left(\frac{\epsilon_2'}{g_2'}+\frac{\epsilon_1}{g_1}+\frac{\epsilon_1 \epsilon_2'}{g_1 g_2'}\right)\right]},
    \end{align}
    \begin{align}\label{wronskian matching c2}
        c_2=-\frac{\mathfrak{W}(f,g_1)\left[1+\frac{f g_1}{\mathfrak{W}(f,g_1)}\frac{f'}{f}\left(\frac{\epsilon_1}{g_1}+\frac{\epsilon'}{f'}+\frac{\epsilon' \epsilon_1}{f' g_1}\right)-\frac{f g_1}{\mathfrak{W}(f,g_1)}\frac{g_1'}{g_1}\left(\frac{\epsilon_1'}{g_1'}+\frac{\epsilon}{f}+\frac{\epsilon \epsilon_1'}{f g_1'}\right)\right]}{\mathfrak{W}(g_1,g_2)\left[1+\frac{g_1 g_2}{\mathfrak{W}(g_1,g_2)}\frac{g_1'}{g_1}\left(\frac{\epsilon_2}{g_2}+\frac{\epsilon_1'}{g_1'}+\frac{\epsilon_1' \epsilon_2}{g_1' g_2}\right)-\frac{g_1 g_2}{\mathfrak{W}(g_1,g_2)}\frac{g_2'}{g_2}\left(\frac{\epsilon_2'}{g_2'}+\frac{\epsilon_1}{g_1}+\frac{\epsilon_1 \epsilon_2'}{g_1 g_2'}\right)\right]}.
    \end{align}
\end{lemma}

\begin{proof}[Proof of Proposition \ref{holy grail of transmission}]

We start by using Proposition \ref{prop: approximation prop near horizon} and Proposition \ref{prop: approximation prop intermediate region} to express $\hor$ in terms of $P_{\ell}(2r-1)$ and $Q_{\ell}(2r-1)$ in the region $r\in[r_1,\tilde{r}_1]$, leading to
    \begin{align}
        (e^{-1}\ell(\ell+1))^{i\omega}\Gamma(1-2i\omega)(I_{-2i\omega}+\epsnh)|_{\sqrt{\ell(\ell+1)(r-1)}}=c_1(P_{\ell}(2r-1)+\epsilon_P)+c_2(Q_{\ell}(2r-1)+\epsilon_Q).
    \end{align}
We compute the derivative of $I_{-2i\omega}(2\sqrt{L(r-1)})$ for $r\in [\tilde{r}_1,r_1]$ (See Section 10.40(i) of \cite{dlmf}),
\begin{align}
\begin{split}
     \frac{d}{dr}I_{-2i\omega}(2\sqrt{\ell(\ell+1)(r-1)})&=2\sqrt{\frac{L}{r-1}}I_{-2i\omega}'(2\sqrt{\ell(\ell+1)(r-1)})\\&= 2\sqrt{\frac{L}{r-1}} \frac{e^{2\sqrt{L(r-1)}}}{\sqrt{4\pi\sqrt{L(r-1)}}}(1+O(\deltnh^{-\frac{1}{2}}C^{-\frac{2}{9}}L^{-\frac{1}{9}}))\\&= 2\sqrt{\frac{L}{r-1}}I_{-2i\omega}(1+O(\deltnh^{-\frac{1}{2}}C^{-\frac{2}{9}}L^{-\frac{1}{9}})).
\end{split}
\end{align}
The derivative of $Q_{\ell}(2r-1)$ for $r\in [\tilde{r}_1,r_1]$ is computed using \eqref{Q in terms of super exponential} and \eqref{derivative of Q in terms of Q},
\begin{align}
    \begin{split}
        \frac{d}{dr}Q_{\ell}(2r-1)= \frac{\sqrt{2}\ell(2\ell+1)^{\frac{1}{2}}}{(r(r-1))^{\frac{1}{2}}}Q_{\ell}(2r-1)\left(1+O\left(\deltnh^{-\frac{1}{2}}L^{-\frac{1}{3}}\right)\right).
    \end{split}
\end{align}
Therefore,
\begin{align}
\begin{split}
    &\mathfrak{W}(I_{-2i\omega}(2\sqrt{\ell(\ell+1)(r-1)}),Q_{\ell}(2r-1))\\&= I_{-2i\omega}(2\sqrt{\ell(\ell+1)(r-1)})Q_{\ell}(2r-1))\left(\frac{\ell(2\ell+1)^{\frac{1}{2}}}{(r(r-1))^{\frac{1}{2}}}-\frac{1}{2}\sqrt{\frac{L}{r-1}}\right)(1+O(\deltnh^{-\frac{1}{2}}C^{-\frac{2}{9}}L^{-\frac{1}{9}}))
\end{split}
\end{align}
For $r=r_1$, 
\begin{align}
    Q_{\ell}(2r_1-1)= \sqrt{\frac{\pi}{2}}L^{-\frac{1}{6}}e^{-4L^{\frac{1}{3}}-8L^{\frac{1}{6}}-\frac{11}{3}}\left(1+O\left(\deltnh^{-\frac{1}{2}}L^{-\frac{1}{3}}\right)\right),
\end{align}
\begin{align}
    I_{-2i\omega}(2\sqrt{\ell(\ell+1)(r_1-1)})= \frac{1}{\sqrt{4\pi L^{\frac{1}{3}}}}e^{2 L^{\frac{1}{3}}}(1+O(\deltnh^{-\frac{1}{2}}C^{-\frac{2}{9}}L^{-\frac{1}{9}})).
\end{align}
Therefore,
\begin{align}
    \mathfrak{W}(I_{-2i\omega}(2\sqrt{\ell(\ell+1)(r-1)}),Q_{\ell}(2r-1))= \frac{1}{4\sqrt{2}}L^{\frac{1}{3}}e^{-2L^{\frac{1}{3}}-8L^{\frac{1}{6}}-\frac{11}{3}}(1+O(\deltnh^{-\frac{1}{2}}C^{-\frac{2}{9}}L^{-\frac{1}{9}})),
\end{align}
\begin{align}
    \mathfrak{W}(P_{\ell}(2r-1),Q_{\ell}(2r-1))= \frac{1}{4}L^{\frac{1}{3}}\left(1+O\left(\deltnh^{-\frac{1}{2}}L^{-\frac{1}{3}}\right)\right).
\end{align}
Note that $I'/I|_{r=r_1}=O(\deltnh^{-\frac{1}{2}}L^{\frac{1}{3}})$, $Q'/Q|_{r=r_1}=O(\deltnh^{-\frac{1}{2}}L^{\frac{7}{12}})$. Therefore, we can ensure 
\begin{align}
    \frac{I_{-2i\omega}|_{2L^{\frac{1}{2}}\sqrt{(r-1)}}Q_{\ell}(2r-1)}{\mathfrak{W}(I_{-2i\omega}|_{2L^{\frac{1}{2}}\sqrt{(r-1)}},Q_{\ell}(2r-1))}\frac{\frac{d}{dr}\left(I_{-2i\omega}|_{2L^{\frac{1}{2}}\sqrt{(r-1)}}\right)}{I_{-2i\omega}|_{2L^{\frac{1}{2}}\sqrt{(r-1)}}}\Bigg|_{r=r_1},
\end{align}
\begin{align}
    \frac{I_{-2i\omega}|_{2L^{\frac{1}{2}}\sqrt{(r-1)}}Q_{\ell}(2r-1)}{\mathfrak{W}(I_{-2i\omega}|_{2L^{\frac{1}{2}}\sqrt{(r-1)}},Q_{\ell}(2r-1))}\frac{Q_{\ell}'(2r-1)}{Q_{\ell}(2r-1)}\Bigg|_{r=r_1},
\end{align}
are of order $O(1)$ as $\ell\longrightarrow\infty$ by choosing $\deltnh=1$ and $C$ sufficiently large. By Propositions \ref{sec: near horizon hor} and \ref{prop: approximation prop intermediate region}, we can also choose $C$ large enough so that 
\begin{align}
    \frac{\epsnh|_{2L^{\frac{1}{2}}\sqrt{(r-1)}}}{I_{-2i\omega}|_{2L^{\frac{1}{2}}\sqrt{(r-1)}}}\Bigg|_{r=r_1},\quad \frac{\frac{d}{dr}\left(\epsnh|_{2L^{\frac{1}{2}}\sqrt{(r-1)}}\right)}{\frac{d}{dr}\left(I_{-2i\omega}|_{2L^{\frac{1}{2}}\sqrt{(r-1)}}\right)}\Bigg|_{r=r_1},
\end{align}
\begin{align}
    \frac{\epsilon_P(2r-1)}{P(2r-1)},\quad \frac{\epsilon_P'(2r-1)}{P'(2r-1)},\quad \frac{\epsilon_Q(2r-1)}{Q(2r-1)}, \quad \frac{\epsilon_Q'(2r-1)}{Q'(2r-1)}
\end{align}
are small enough so that Lemma \ref{Wronskian matching lemma},
\begin{align}
    c_1=\frac{1}{\sqrt{2}}(e^{-1}\ell(\ell+1))^{i\omega}\Gamma(1-2i\omega)e^{-2L^{\frac{1}{3}}-8L^{\frac{1}{6}}-\frac{11}{3}}(1+\eta_{r_1}^{(1)}),\qquad \eta_{r_1}^{(1)}=O(\delta).
\end{align}
Similarly, we have for $c_2$
\begin{align}
    P_{\ell}(2r_1-1)= \frac{1}{\sqrt{2\pi}}L^{-\frac{1}{6}}e^{4L^{\frac{1}{3}}+8L^{\frac{1}{6}}+\frac{11}{3}}\left(1+O\left(L^{-\frac{1}{3}}\right)\right),
\end{align}
\begin{align}
\begin{split}
    &\mathfrak{W}(I_{-2i\omega}(2\sqrt{\ell(\ell+1)(r-1)}),P_{\ell}(2r-1))\\&= I_{-2i\omega}(2\sqrt{\ell(\ell+1)(r-1)})P_{\ell}(2r-1)\left(\frac{2}{\ell(r(r-1))^{\frac{1}{2}}}-\frac{1}{2}\sqrt{\frac{L}{r-1}}\right)(1+O(C^{-\frac{2}{9}}L^{-\frac{1}{9}}))\\&= -\frac{1}{4\sqrt{2}\pi}L^{\frac{1}{3}}e^{6L^{\frac{1}{3}}+8L^{\frac{1}{6}}+\frac{11}{3}}(1+O(C^{-\frac{2}{9}}L^{-\frac{1}{9}})),
\end{split}
\end{align}
thus we find by Lemma \ref{Wronskian matching lemma},
\begin{align}
    c_2= \frac{1}{\sqrt{2}}(e^{-1}\ell(\ell+1))^{i\omega}\Gamma(1-2i\omega)\frac{1}{\sqrt{2}\pi}e^{6L^{\frac{1}{3}}+8L^{\frac{1}{6}}+\frac{11}{3}}(1+\eta^{(2)}_{r_1}),\qquad \eta_{r_1}^{(2)}=O(\delta).
\end{align}
We note that at $r=r_2$, $c_1P_{\ell}$ dwarfs $c_2Q_{\ell}$. Thus, it suffices to consider a change of basis at $r_2=\delta L^{\frac{1}{4}}|\omega|^{-1}$. Using Propositions \ref{prop: approximation prop intermediate region}, \ref{prop: approximation prop far region}, we have
\begin{align}
        \begin{split}
            P_{\ell}(2r_2-1)+\epsilon_P(2r-1)=\left\{\tilde{c}_1\left(U_{\mathrm{inf}}^{\mathrm{WKB}}+\epsilon_{\infty}^{(1)}\right)+\tilde{c}_2\left(\overline{U}_{\mathrm{inf}}^{\mathrm{WKB}}+\epsilon_{\infty}^{(2)}\right)\right\}\Big|_{2r_2-1}.
        \end{split}
    \end{align}
We compute
    \begin{align}
        \partial_{r_*}\out^{\mathrm{WKB}}=\out^{\mathrm{WKB}}\left[-\frac{L\Omega^2(3\Omega^2-1)}{4r^3(\omega^2-f_{\ell})}-\sqrt{f_{\ell}-\omega^2}\right],
    \end{align}
    \begin{align}
        \partial_{r_*}\outb^{\mathrm{WKB}}=\outb^{\mathrm{WKB}}\left[-\frac{L\Omega^2(3\Omega^2-1)}{4r^3(\omega^2-f_{\ell})}+\sqrt{f_{\ell}-\omega^2}\right].
    \end{align}
    We already note
    \begin{align}
        \begin{split}
            \sqrt{f_{\ell}-\omega^2}\Big|_{r=r_2}\sim \delta^{-1}L^{\frac{1}{4}}|\omega|^2,\qquad \qquad \frac{L\Omega^2(3\Omega^2-1)}{4r^3(\omega^2-f_{\ell})}\Bigg|_{r=r_2}\sim \frac{L^{\frac{1}{4}}|\omega|^3}{\delta^3\omega^2-\delta L^{\frac{1}{2}}|\omega|^2+L^{-\frac{3}{4}}|\omega|^3}\sim \delta^{-1}L^{-\frac{1}{2}}|\omega|.
        \end{split}
    \end{align}
     We use \eqref{derivative of legendre in terms of bessel} to compute the first derivative of $P_{\ell}(2r-1)$ to find 
\begin{align}
\begin{split}
    &\mathfrak{W}(\out^{\mathrm{WKB}},P_{\ell}(2r-1))|_{r=r_2}\\&=P_{\ell}(2r_2-1)\out^{\mathrm{WKB}}(r_2)\left[\frac{L}{4r_2^3(\omega^2-f_{\ell}(r_2))}+\sqrt{f_{\ell}(r_2)-\omega^2}+\frac{1}{\ell(2r_2(r_2-1))^{\frac{1}{2}}}\left(1+O\left(L^{-\frac{1}{3}}\right)\right)\right].
\end{split}
\end{align}
Once again, we can choose $C$ so that 
\begin{align}
    \frac{P_{\ell}(2r-1),\out^{\mathrm{WKB}}}{\mathfrak{W}(P_{\ell}(2r-1),\out^{\mathrm{WKB}})}\frac{2P_{\ell}'(2r-1)}{P_{\ell}(2r-1)},\quad \frac{P_{\ell}(2r-1),\out^{\mathrm{WKB}}}{\mathfrak{W}(P_{\ell}(2r-1),\out^{\mathrm{WKB}})}\frac{\frac{d}{dr}\out^{\mathrm{WKB}}}{\out^{\mathrm{WKB}}},\quad \frac{\out^{\mathrm{WKB}}\outb^{\mathrm{WKB}}}{\mathfrak{W}(\out^{\mathrm{WKB}},\outb^{\mathrm{WKB}})}\frac{\frac{d}{dr}\out^{\mathrm{WKB}}}{\out^{\mathrm{WKB}}}
\end{align}
are $O(1)$, and
\begin{align}
    \begin{split}
        \frac{\epsilon_{\infty}^{(1)}}{\out^{\mathrm{WKB}}},\quad \frac{\epsilon_{\infty}^{(1)}{}'}{\out^{\mathrm{WKB}{}'}},\quad \frac{\epsilon_{\infty}^{(2)}}{\outb^{\mathrm{WKB}}},\quad \frac{\epsilon_{\infty}^{(2)}{}'}{\outb^{\mathrm{WKB}{}'}},\quad \frac{2P_{\ell}'(2r-1)}{P_{\ell}(2r-1)}
    \end{split}
\end{align}
are small enough so that 
\begin{align}
\begin{split}
    &\tilde{c}_2=\frac{1}{2i\omega}P_{\ell}(2r_2-1)\out^{\mathrm{WKB}}(r_2)\left[\frac{L}{4r_2^3(\omega^2-f_{\ell}(r_2))}+\sqrt{f_{\ell}(r_2)-\omega^2}+\frac{1}{\ell(2r_2(r_2-1))^{\frac{1}{2}}}\left(1+O\left(L^{-\frac{1}{3}}\right)\right)\right](1+\eta_{r_2}),
\end{split}
\end{align}
\begin{align}
    \eta_{r_2}=O(\delta).
\end{align}
We choose $\delta$ small enough, $C$ large enough so that
$|\eta_{r_1}^{(1)}|<\frac{1}{2}, |\eta_{r_1}^{(2)}|<\frac{1}{2}$. We then conclude
\begin{align}
    |T|\lesssim c_1^{-1}\tilde{c}_2^{-1}.
\end{align}

We now estimate $\tilde{c}_2^{-1}$. Using \eqref{asymptotic formula of Legendre in terms of exponential}, and choosing $\delta$ small enough so that $(4\delta)^{-1}|\omega|>2$, we find
\begin{align}
\begin{split}
    P_{\ell}(2r_2-1)&\sim \frac{|\omega|^{\frac{1}{2}}}{2\sqrt{\pi}\delta^{\frac{1}{2}}L^{\frac{3}{8}}}(4\delta L^{\frac{1}{4}}|\omega|^{-1}-1)^{L^{\frac{1}{2}}}(1+o(1))\\&\sim\frac{|\omega|^{\frac{1}{2}}}{2\sqrt{\pi}\delta^{\frac{1}{2}}L^{\frac{3}{8}}}(4\delta L^{\frac{1}{4}}|\omega|^{-1})^{L^{\frac{1}{2}}}e^{-\frac{|\omega|L^{\frac{1}{4}}}{4\delta}}e^{-\frac{|\omega|^2}{32\delta^2}}(1+o(1)).
\end{split}
\end{align}
For the contribution of the far region, we have 
\begin{align}
    \begin{split}
        \exp\left\{-\bigintsss_{r_*}^{\infty}dy(\sqrt{f_{\ell}-\omega^2}+i\omega)\right\}&=\exp\left\{-\bigintsss_{r_*(r_2)}^{\infty}dy \frac{\sqrt{2}f_{\ell}}{\sqrt{\sqrt{(f_{\ell}-\omega_R^2+\frac{1}{4})^2+\omega_R^2}+f_{\ell}-\omega_R^2+\frac{1}{4}}+\frac{1}{2}}\right\}\\&\lesssim \exp\left\{-\bigintsss_{r_2}^{\infty}dr\,\frac{\delta^{2}|\omega|^{-2}L^{\frac{1}{2}}}{r^2}\right\}\\&\lesssim \exp\left\{-\frac{\delta L^{\frac{1}{4}}}{2|\omega|}\right\}.
    \end{split}
\end{align}
In the above we used that the derivative of ${\sqrt{\sqrt{(f_{\ell}-\omega_R^2+\frac{1}{4})^2+\omega_R^2}+f_{\ell}-\omega_R^2+\frac{1}{4}}+\frac{1}{2}}$ never vanishes for $r>\frac{3}{2}$, hence the maximum over the interval $[r_2,\infty]$ occurs at $r=r_2$. Note also that
\begin{align}
    |\omega^{-\frac{1}{2}}(\omega^2-f_{\ell})^{\frac{1}{4}}(r_2)|\lesssim \delta^{-\frac{1}{2}}L^{\frac{1}{8}},\qquad |e^{-i\omega r_*(r_2)}|\lesssim \left[\frac{\delta L^{\frac{1}{4}}}{|\omega|}\right]^{\frac{1}{2}}e^{-\frac{\delta L^{\frac{1}{4}}}{2|\omega|}},
\end{align}
\begin{align}
    \left|\frac{L}{4r_2^3(\omega^2-f_{\ell}(r_2))}+\sqrt{f_{\ell}(r_2)-\omega^2}+\frac{1}{\ell(2r_2(r_2-1))^{\frac{1}{2}}}\right|^{-1}\lesssim L^{\frac{1}{4}}|\omega|^{-1}.
\end{align}
By choosing $\delta$ sufficiently small and $C$ sufficiently large, we find
\begin{align}
    |T(\omega,\ell)|\lesssim \left(\frac{\ell}{|\omega|^2}\right)^{-\frac{1}{2}\ell}.
\end{align}
\end{proof}

\numberwithin{equation}{subsection}

\section{Rigidity of asymptotics: proof of Theorems \ref{exponential decay theorem'}--\ref{construction of counter example Scrip}}

\subsection{Energy estimates near the event horizons}

We record the following standard local energy estimates:

\begin{proposition}\label{prop: standard local energy estimates future backwards}
    A solution $\psi$ \eqref{wave equation in Kruskal coordinates} in the region $\{(U,V,\theta^A);U\leq0, V\leq V_0, t\geq0\}$ is such that $\psi|_{\overline{\hp}\cap\{(V,\theta^A);V\in[0,V_0]\}}\in H^k(\overline{\hp}\cap\{(V,\theta^A);V\in[0,V_0]\})$, $\psi|_{\{(U,V_0,\theta^A); U\leq0, t\geq0\}}\in H^k(\{(U,V_0,\theta^A); U\leq0, t\geq0\})$ if and only if $\psi|_{\overline{\Sigma}\cap\{R\leq 2V_0\}}\in H^{k}(\overline{\Sigma}\cap\{R\leq 2V_0\})$, $n_{\overline{\Sigma}}\psi|_{t=0,R\leq 2V_0}\in H^{k-1}(\overline{\Sigma}\cap\{R\leq 2V_0\})$, and we have for $r_0$ such that $t=0$, $V=V_0$ constants $c,C>0$ with 
    \begin{align}
    \begin{split}
        &c\left\{\|\psi|_{\{(U,V_0,\theta^A); U\leq0, t\geq0\}}\|^2_{H^k(\{(U,V_0,\theta^A);  U\leq0, t\geq0\})}+\|\psi|_{\overline{\hp}\cap\{(V,\theta^A);V\in[0,V_0]\}}\|^2_{H^k(\overline{\hp}\cap\{(V,\theta^A);V\in[0,V_0]\})}\right\}
        \\&\leq\|\psi|_{t=0}\|_{H^k(\overline{\Sigma}\cap\{(r,\theta^A);r\in[1,r_0]\})}^2+\|n_{\Sigma}\psi|_{t=0}\|_{H^{k-1}(\overline{\Sigma}\cap\{(r,\theta^A);r\in[1,r_0]\})}^2\\&\leq C\;\;\left\{\|\psi|_{\{(U,V_0,\theta^A); U\leq0, t\geq0\}}\|^2_{H^k(\{(U,V_0,\theta^A);  U\leq0, t\geq0\})}+\|\psi|_{\overline{\hp}\cap\{(V,\theta^A);V\in[0,V_0]\}}\|^2_{H^k(\overline{\hp}\cap\{(V,\theta^A);V\in[0,V_0]\})}\right\}.
    \end{split}
    \end{align}
    In particular, $\psi$ is such that $\psi|_{{\hp}\cap\{(V,\theta^A);V\in(0,V_0]\}}\in H^k({\hp}\cap\{(V,\theta^A);V\in(0,V_0]\})$, $\psi|_{\{(U,V_0,\theta^A); U\leq0, t\geq0\}}\in H^k(\{(U,V_0,\theta^A); U\leq0, t\geq0\})$ if and only if 
    $\psi|_{\overline{\Sigma}\cap\{R\leq 2V_0\}}\in H^{k}({\Sigma}\cap\{R\leq 2V_0\})$, $n_{\overline{\Sigma}}\psi|_{t=0,R\leq 2V_0}\in H^{k-1}({\Sigma}\cap\{R\leq 2V_0\})$, and we have $c,C>0$ for which
    \begin{align}
    \begin{split}
        &c\left\{\|\psi|_{\{(U,V_0,\theta^A); U\leq0, t\geq0\}}\|^2_{H^k(\{(U,V_0,\theta^A);  U\leq0, t\geq0\})}+\|\psi|_{{\hp}\cap\{(V,\theta^A);V\in(0,V_0]\}}\|^2_{H^k({\hp}\cap\{(V,\theta^A);V\in(0,V_0]\})}\right\}\\&\leq \|\psi|_{t=0}\|_{H^k({\Sigma}\cap\{(r,\theta^A);r\in(1,r_0]\})}^2+\|n_{\Sigma}\psi|_{t=0}\|_{H^{k-1}({\Sigma}\cap\{(r,\theta^A);r\in(1,r_0]\})}^2\\&\leq C\left\{\|\psi|_{\{(U,V_0,\theta^A); U\leq0, t\geq0\}}\|^2_{H^k(\{(U,V_0,\theta^A);  U\leq0, t\geq0\})}+\|\psi|_{{\hp}\cap\{(V,\theta^A);V\in(0,V_0]\}}\|^2_{H^k({\hp}\cap\{(V,\theta^A);V\in(0,V_0]\})}\right\}.
    \end{split}
    \end{align}
\end{proposition}

\begin{proposition}\label{prop: standard local energy estimates past forwards}
    A solution $\psi$ to \eqref{wave equation in Kruskal coordinates} in the region $\{(U,V,\theta^A);U\leq U_0, V\geq 0, t\leq0\}$ is such that $\psi|_{\overline{\hm}\cap\{(U,\theta^A);U\in[U_0,0]\}}\in H^k(\overline{\hm}\cap\{(U,\theta^A);U\in[U_0,0]\})$, $\psi|_{\{(U_0,V,\theta^A); V\geq0, t\leq0\}}\in H^k(\{(U_0,V,\theta^A); V\geq0, t\leq0\})$ if and only if $\psi|_{\overline{\Sigma}\cap\{R\leq -2U_0\}}\in H^{k}(\overline{\Sigma}\cap\{R\leq -2U_0\})$, $n_{\overline{\Sigma}}\psi|_{t=0,R\leq -2U_0}\in H^{k-1}(\overline{\Sigma}\cap\{R\leq -2U_0\})$, and we have for $r_0$ such that $t=0$, $U=U_0$ constants $c,C>0$ with 
    \begin{align}
    \begin{split}
        &c\left\{\|\psi|_{\{(U_0,V,\theta^A); V\geq0, t\leq0\}}\|^2_{H^k(\{(U_0,V,\theta^A);  V\geq0, t\leq0\})}+\|\psi|_{\overline{\hm}\cap\{(U,\theta^A);U\in[U_0,0]\}}\|^2_{H^k(\overline{\hm}\cap\{(U,\theta^A);U\in[U_0,0]\})}\right\}\\
        &\leq\|\psi|_{t=0}\|_{H^k(\overline{\Sigma}\cap\{(r,\theta^A);r\in[1,r_0]\})}^2+\|n_{\Sigma}\psi|_{t=0}\|_{H^{k-1}(\overline{\Sigma}\cap\{(r,\theta^A);r\in[1,r_0]\})}^2\\&\leq C\left\{\|\psi|_{\{(U_0,V,\theta^A); V\geq0, t\leq0\}}\|^2_{H^k(\{(U_0,V,\theta^A);  V\geq0, t\leq0\})}+\|\psi|_{\overline{\hm}\cap\{(U,\theta^A);U\in[U_0,0]\}}\|^2_{H^k(\overline{\hm}\cap\{(U,\theta^A);U\in[U_0,0]\})}\right\}.
    \end{split}
    \end{align}
    In particular, $\psi$ is such that $\psi|_{{\hm}\cap\{(U,\theta^A);U\in[U_0,0)\}}\in H^k({\hp}\cap\{(V,\theta^A);V\in(0,V_0]\})$, $\psi|_{\{(U,V_0,\theta^A); U\leq0, t\geq0\}}\in H^k(\{(U,V_0,\theta^A); U\leq0, t\geq0\})$ if and only if 
    $\psi|_{\overline{\Sigma}\cap\{R\leq -2U_0\}}\in H^{k}({\Sigma}\cap\{R\leq -2U_0\})$, $n_{\overline{\Sigma}}\psi|_{t=0,R\leq -2U_0}\in H^{k-1}({\Sigma}\cap\{R\leq -2U_0\})$, and we have $c,C>0$ for which
    \begin{align}
    \begin{split}
        &c\left\{\|\psi|_{\{(U_0,V,\theta^A); V\geq0, t\leq0\}}\|^2_{H^k(\{(U_0,V,\theta^A);  V\geq0, t\leq0\})}+\|\psi|_{{\hm}\cap\{(U,\theta^A);V\in[U_0,0)\}}\|^2_{H^k({\hm}\cap\{(U,\theta^A);U\in[U_0,0)\})}\right\}
        \\&\leq\|\psi|_{t=0}\|_{H^k({\Sigma}\cap\{(r,\theta^A);r\in(1,r_0]\})}^2+\|n_{\Sigma}\psi|_{t=0}\|_{H^{k-1}({\Sigma}\cap\{(r,\theta^A);r\in(1,r_0]\})}^2\\&\leq C\left\{\|\psi|_{\{(U_0,V,\theta^A); V\geq0, t\leq0\}}\|^2_{H^k(\{(U_0,V,\theta^A);  V\geq0, t\leq0\})}+\|\psi|_{{\hm}\cap\{(U,\theta^A);V\in[U_0,0)\}}\|^2_{H^k({\hm}\cap\{(U,\theta^A);U\in[U_0,0)\})}\right\}.
    \end{split}
    \end{align}
\end{proposition}

We now show that for a solution to \eqref{full wave equation} arising from Cauchy or characteristic data of finite $\partial_t$-energy and finite local energy at $\hp$, the radiation field $\psi_{\hp}$ is itself square integrable.
\begin{proposition}\label{the radiation field at hp is in L2}
    Let $\psi$ be a solution to \eqref{wave equation} on $\{u\geq u_0, v\geq v_0\}$ for finite $u_0,v_0$. If $\psi|_{v=v_0, u\geq u_0}\in H^1({v=v_0, u\geq u_0}, dUd\mathrm{Vol}_{S^2})$, $\psi|_{u=u_0, v\geq v_0}\in H^1({u=u_0, v\geq v_0}, dvd\mathrm{Vol}_{S^2})$, then we have
    \begin{align}
        \int_{v_0}^{\infty}d\bar{v} d\mathrm{Vol}_{S^2}\, |\psi_{\hp}|^2\lesssim \|\psi|_{v=v_0, u\geq u_0}\|^2_{H^1({v=v_0, u\geq u_0}, dUd\mathrm{Vol}_{S^2})}+\|\psi|_{u=u_0, v\geq v_0}\|^2_{H^1({u=u_0, v\geq v_0}, dvd\mathrm{Vol}_{S^2})}.
    \end{align}
\end{proposition}
\begin{proof}
    This follows by a redshift estimate (see for instance Theorem 3.7.1 of \cite{DRSR14}) in the region $\{r\leq r_0, v\geq v_0\}$ for $r_0<\frac{3}{2}$ (i.e.~$r_0<3M$, if $2M$ is restored). We may then estimate the boundary term at $r=r_0$ by averaging and using the Morawetz integrated local energy decay estimate.
\end{proof}
\begin{corollary}\label{standard redshift estimate}
    Assume $\psi$ is a solution to \eqref{wave equation} arising from data in $\mathcal{E}^N_{\Sigma}$. Then the radiation field $\psi_{\hp}$ at $\hp$ satisfies
    \begin{align}
        \int_{\hp}dvd\mathrm{Vol}_{S^2}\; |\psi_{\hp}|^2\leq \|(\psi,\partial_t\psi)|_{t=0}\|^2_{\mathcal{E}^N_{\Sigma}}.
    \end{align}
\end{corollary}
\begin{proof}
    An energy estimate applied to \eqref{wave equation in Kruskal coordinates} in the region $\{t\geq0, V\leq1, U\leq0\}$ shows that $\partial_V\psi_{\hp}\in L^2(\{V\in[0,1], dVd\mathrm{Vol}_{S^2}\})$ and that $\psi, \mathring{\slashed{\nabla}}\psi, \partial_U\psi|_{V=1}\in L^2(\{V=1, U\leq0,t\geq0\},dUd\mathrm{Vol}_{S^2})$. We then have
    \begin{align}
        |\psi_{\hp}|^2\leq \left(\int_{-\infty}^{v}dvd\mathrm{Vol}_{S^2}|\partial_v\psi_{\hp}|\right)^2\leq  e^{v}  \int_{-\infty}^{v}dvd\mathrm{Vol}_{S^2} e^{-v}|\partial_v\psi_{\hp}|^2,
    \end{align}
    which gives $\psi_{\hp}\in L^2(\{v\in (-\infty,0]\},dvd\mathrm{Vols}_{S^2})$. The result then follows by Proposition \ref{the radiation field at hp is in L2}.
\end{proof}

The following proposition states that regularity at $\hp$ is equivalent to the regularity of $\psi_{\hm}$ in Kruskal coordinates near $\mathcal{B}$, provided the solution arises either from Cauchy data or scattering data such that the $\partial_t$-energy is finite.

\begin{proposition}\label{prop: the thing that nobody noticed apparently}
    Assume that $\psi$ solves \eqref{wave equation} in the region
    \begin{align}
        \mathcal{N}=\{(u,v)\times S^2; (u,v)\in (u_0,\infty)\times (-\infty,v_0)\},
    \end{align}
    such that 
    \begin{align}\label{odd regularity assumption}
        \int_{\hm\cap \mathcal{N}}dud\mathrm{Vol}_{S^2}\;e^u(\partial_u\psi|_{\hm})^2<\infty,
    \end{align}
    \begin{align}\label{odd degenerate regularity assumption}
        \int_{-\infty}^{\infty}\int_{S^2}dv\,d\mathrm{Vol}_{S^2}\, |\partial_v \psi|_{u=u_0}|^2+\frac{\Omega^2}{r^2}|\mathring{\slashed{\nabla}}\psi|_{u=u_0}|^2+\frac{\Omega^2}{r^3}|\psi|_{u=u_0}|^2<\infty.
    \end{align}
    Then for any $v\leq v_0$, we have
    \begin{align}
        \int_{\underline{\mathscr{C}}_{v}\cap\mathcal{N}}dud\mathrm{Vol}_{S^2} \frac{1}{\Omega^2}(\partial_u\psi)^2\lesssim\text{left hand side of }\eqref{odd regularity assumption}+\text{left hand side of }\eqref{odd degenerate regularity assumption}.
    \end{align}
\end{proposition}
\begin{proof}
    We rewrite \eqref{wave equation} using the coordinate system $(U,v)$ with $U=-e^{-u}$ and get
    \begin{align}\label{wave equation in U v}
        \partial_v\partial_U\psi-\frac{V}{r^3}e^{-r}\lapo\psi+\frac{V}{r^4}e^{-r}\psi=0,
    \end{align}
    where $V=e^{v}$. We then derive
    \begin{align}\label{hellish divergence}
    \begin{split}
        &\partial_v(\partial_U\psi)^2-2\mathring{\slashed{\nabla}}\left[\frac{V e^{-r}}{r^3}\partial_U\mathring{\slashed{\nabla}}\psi\right]+\partial_U\left[\frac{Ve^{-r}}{r^3}\left(|\mathring{\slashed{\nabla}}\psi|^2+\frac{1}{r}|\psi|^2\right)\right]\\&+\frac{V^2 e^{-2r}}{r^4}\left(1+\frac{3}{r}\right)|\mathring{\slashed{\nabla}}\psi|^2-\frac{V^2e^{-2r}}{r^5}\left(1+\frac{4}{r}\right)|\psi|^2=0.
    \end{split}
    \end{align}
    Note that $\Omega^2=\frac{-UV}{r}e^{-r}$. 
    
    Now take $\{(\psi_{u_0+\epsilon})_{n}\}_{n=1}^{\infty}$ is a sequence of smooth, compactly supported functions with $(\psi_{u_0+\epsilon})_n|_{U=0}=\psi|_{\hm, u=u_0+\epsilon}$, such that $(\psi_{u_0+\epsilon})_{n}$ tends to $\psi|_{\mathscr{C}_{u_0+\epsilon}\cap\{v<v_0\}}$ in the norm
    \begin{align}
        \int_{-\infty}^{v_0-\epsilon'}\int_{S^2}dv\,d\mathrm{Vol}_{S^2}\, |\partial_v f|^2+\frac{\Omega^2}{r^2}|\mathring{\slashed{\nabla}}f|^2+\frac{\Omega^2}{r^3}|\psi|^2. 
    \end{align}
    We now apply \eqref{hellish divergence} to the solution $\psi_n$ to \eqref{wave equation} arising out of data $(\psi_{u_0+\epsilon})_{n}$ on $\mathscr{C}_{u_0+\epsilon}\cap\{v\leq v_0-\epsilon'\}$ and $\psi|_{\hm}$ on $\{V=0, u\geq u_0+\epsilon\}$. Suppressing the integral sign over $S^2$, we write
    \begin{align}
    \begin{split}
        &\int_{-e^{-u_0-\epsilon}}^{0}dU\,d\mathrm{Vol}_{S^2}\,|\partial_U\psi_n|^2+\int_{-\infty}^{v_0-\epsilon'}dv\,d\mathrm{Vol}_{S^2}\, e^{v-1}\left(|\mathring{\slashed{\nabla}}\psi_n|_{\hp}|^2+|\psi_n|_{\hp}|^2\right)\\&\int_{u\geq u_0+\epsilon, v\leq v_0-\epsilon'}dUdv\,d\mathrm{Vol}_{S^2}\,\frac{ e^{2v-2r}}{r^4}\left[\left(1+\frac{3}{r}\right)|\mathring{\slashed{\nabla}}\psi_n|^2+\left(1+\frac{4}{r}\right)|\psi_n|^2\right]\\&=\int_{-e^{-u_0-\epsilon}}^{0}dU\,d\mathrm{Vol}_{S^2}\,|\partial_U\psi|_{\hm}|^2+\int_{-\infty}^{v_0-\epsilon'}dv\,d\mathrm{Vol}_{S^2}\, \frac{e^{v-1}}{r^3}\left(|\mathring{\slashed{\nabla}}(\psi_{u_0+\epsilon})_n|^2+|(\psi_{u_0+\epsilon})_n|^2\right).
    \end{split}
    \end{align}
    For fixed $u$, we have $e^v\sim \Omega^2$ in $\mathcal{N}$, and thus we have
    \begin{align}
        \begin{split}
            &\int_{-\infty}^{v_0-\epsilon'}dv\,d\mathrm{Vol}_{S^2}\, \frac{e^{v-1}}{r^3}\left(|\mathring{\slashed{\nabla}}(\psi_{u_0+\epsilon})_n|^2+|(\psi_{u_0+\epsilon})_n|^2\right)\\&\lesssim \int_{-\infty}^{v_0+\epsilon'}dv\,d\mathrm{Vol}_{S^2}\, {\Omega^2}\left(|\mathring{\slashed{\nabla}}(\psi_{u_0+\epsilon})_n|^2+|(\psi_{u_0+\epsilon})_n|^2\right).
        \end{split}
    \end{align}
    By $\partial_t$-energy conservation, we have
    \begin{align}
        \begin{split}
            &\int_{-\infty}^{v_0-\epsilon'} du\,d\mathrm{Vol}_{S^2}\,\Omega^2\left[|\mathring{\slashed{\nabla}}\psi_n|^2+|\psi_n|^2\right]\\\lesssim &
        \int_{-\infty}^{v_0-\epsilon'}dv\,d\mathrm{Vol}_{S^2}\, |\partial_v (\psi_{u_0+\epsilon})_n|^2+\frac{\Omega^2}{r^2}|\mathring{\slashed{\nabla}}(\psi_{u_0+\epsilon})_n|^2+\frac{\Omega^2}{r^3}|(\psi_{u_0+\epsilon})_n|^2\\&+\int_{u_0+\epsilon}^{\infty}du\,d\mathrm{Vol}_{S^2}\,|\partial_u\psi|_{\hm}|^2.
        \end{split}
    \end{align}
    We conclude
    \begin{align}
        \begin{split}
            &\int_{u_0+\epsilon}^{\infty}du\,d\mathrm{Vol}_{S^2}\,\frac{1}{\Omega^2}|\partial_u\psi_n|^2\\&\lesssim \int_{u_0+\epsilon}^{\infty}du\,d\mathrm{Vol}_{S^2}\,(1+e^u)|\partial_u\psi|_{\hm}|^2\\&+\int_{-\infty}^{v_0-\epsilon'}dv\,d\mathrm{Vol}_{S^2}\, |\partial_v (\psi_{u_0+\epsilon})_n|^2+\frac{\Omega^2}{r^2}|\mathring{\slashed{\nabla}}(\psi_{u_0+\epsilon})_n|^2+\frac{\Omega^2}{r^3}|(\psi_{u_0+\epsilon})_n|^2.
        \end{split}
    \end{align}
    The result now follows by the linearity of \eqref{wave equation} and by \eqref{odd regularity assumption}.
\end{proof}

\begin{proposition}\label{tranversally commuted data}
    Assume that $\psi$ is a solution to \eqref{wave equation} such that for $n\geq0$,
    \begin{align}
        \lapo^{n}\psi_{\hp}\in\mathcal{E}^T_{\hp}, \qquad\quad \lapo^{n}\psi_{\Scrip}\in\mathcal{E}^T_{\Scrip}.
    \end{align}
    Then we have for finite $u,v$,
    \begin{align}
        \int_{-\infty}^{v}\int_{S^2}d\bar{v}d\mathrm{Vol}_{S^2}\,\Omega^2|\partial_u^{n}\psi(u,\bar{v})|^2\lesssim \|\lapo^{n}\psi_{\hp}\|^2_{\mathcal{E}^T_{\hp}}+\|\lapo^{n}\psi_{\Scrip}\|^2_{\mathcal{E}^T_{\Scrip}}.
    \end{align}
\end{proposition}
\begin{proof}
    Commuting \eqref{wave equation} with $\partial_u^{k}$ gives
    \begin{align}
        \partial_v\partial_u^{k+1}\psi=\sum_{j=0}^{k} \binom{n}{k}\partial_{u}^{n-j}\frac{\Omega^2}{r^2}\partial_u^j\lapo\psi-\sum_{j=0}^{k} \binom{n}{j}\partial_{u}^{n-j}\frac{\Omega^2}{r^3}\partial_u^j\psi.
    \end{align}
    We may rewrite the above as
    \begin{align}\label{yet another infernal thing}
        \partial_v\partial_u^{k+1}\psi-\frac{\Omega^2}{r^2}\partial_u^{k+1}\psi=\sum_{j=0}^{k} f_{n-j}(r)\Omega^2\partial_u^j\lapo\psi-\sum_{j=0}^{k} \tilde{f}_{n-j}(r)\Omega^2\partial_u^j\psi,
    \end{align}
    where $f_{j}$, $\tilde{f}_j$ are smooth functions of $r$. Since \eqref{wave equation} commutes with any power of $\lapo$, $\partial_t$-energy conservation implies
    \begin{align}
        \int_{-\infty}^{v}\int_{S^2}d\bar{v}d\mathrm{Vol}_{S^2}\,\Omega^2|\lapo^n\psi(u,\bar{v})|^2\lesssim \|\lapo^{n}\psi_{\hp}\|^2_{\mathcal{E}^T_{\hp}}+\|\lapo^{n}\psi_{\Scrip}\|^2_{\mathcal{E}^T_{\Scrip}}.
    \end{align}
    Assume for the sake of an induction argument that
    \begin{align}
        \int_{-\infty}^{v}\int_{S^2}d\bar{v}d\mathrm{Vol}_{S^2}\,\Omega^2\sum_{j=0}^{k}|\lapo^{n-j}\partial_u^j\psi|^2\lesssim \|\lapo^{n}\psi_{\hp}\|^2_{\mathcal{E}^T_{\hp}}+\|\lapo^{n}\psi_{\Scrip}\|^2_{\mathcal{E}^T_{\Scrip}}.
    \end{align}
    It is immediate that we can integrate \eqref{yet another infernal thing}, use the above and apply Gr\"onwall's inequality to show
    \begin{align}
        \int_{-\infty}^{v}\int_{S^2}d\bar{v}d\mathrm{Vol}_{S^2}\,\Omega^2 |\lapo^{n-k-1}\partial_u^{k+1}\psi|^2\lesssim \|\lapo^{n}\psi_{\hp}\|^2_{\mathcal{E}^T_{\hp}}+\|\lapo^{n}\psi_{\Scrip}\|^2_{\mathcal{E}^T_{\Scrip}}.
    \end{align}
\end{proof}
\begin{proposition}\label{the thing that nobody noticed commuted}
    Assume that $\psi$ solves \eqref{wave equation} in the region
    \begin{align}
        \mathcal{N}=\{(u,v)\times S^2; (u,v)\in (u_0,\infty)\times (-\infty,v_0)\},
    \end{align}
    such that 
    \begin{align}\label{odd regularity assumption n}
        \sum_{k=0}^{n}\int_{\hm\cap \mathcal{N}}dud\mathrm{Vol}_{S^2}\;e^u(\partial_u(e^u\partial_u)^k\psi|_{\hm})^2<\infty,
    \end{align}
    \begin{align}\label{odd degenerate regularity assumption n}
        &\sum_{k=0}^{n}\int_{-\infty}^{\infty}\int_{S^2}dv\,d\mathrm{Vol}_{S^2}\, |\partial_v \lapo^k\psi|_{u=u_0}|^2+\frac{\Omega^2}{r^2}|\mathring{\slashed{\nabla}}\lapo^k\psi|_{u=u_0}|^2+\frac{\Omega^2}{r^3}|\lapo^k\psi|_{u=u_0}|^2<\infty.
    \end{align}
    Then for any $v\leq v_0$, we have
    \begin{align}
        \sum_{k=0}^{n}\int_{\underline{\mathscr{C}}_{v}\cap\mathcal{N}}du d\mathrm{Vol}_{S^2}\frac{1}{\Omega^2}\left(\partial_u\left(\frac{1}{\Omega^2}\partial_u\right)^k\psi\right)^2&\lesssim\text{left hand side of }\eqref{odd regularity assumption n}+\text{left hand side of }\eqref{odd degenerate regularity assumption n}.
    \end{align}
\end{proposition}
\begin{proof}
Commuting \eqref{wave equation in U v} with $\partial_u^k$ gives
\begin{align}
    \partial_v\partial_U^{k+1}\psi-\frac{V}{r^3}e^{-r}\lapo\partial_U^k\psi+\frac{V}{r^4}e^{-r}\partial_U^k\psi=VF_{k-1},
\end{align}
where
\begin{align}
    F_{k-1}=\sum_{j=0}^{k-1}\binom{n}{k}\left[\partial_
    U^{k-j}\left(\frac{e^{-r}}{r^3}\right)\lapo\partial_U^j\psi-\partial_
    U^{k-j}\left(\frac{e^{-r}}{r^4}\right)\partial_U^j\psi\right].
\end{align}
Multiplying by $\partial_U^{k+1}\psi$ and integrating by parts over $\mathcal{N}$ gives 
\begin{align}\label{f1}
    \begin{split}
        &\int_{-e^{-u_0}}^{0}dU\,d\mathrm{Vol}_{S^2}\,|\partial_U^{k+1}\psi|^2+\int_{-\infty}^{v_0}dv\,d\mathrm{Vol}_{S^2}\, e^{v-1}\left(|\mathring{\slashed{\nabla}}\partial_U^k\psi_n|_{\hp}|^2+|\partial_U^k\psi_n|_{\hp}|^2\right)\\&\int_{u\geq u_0, v\leq v_0}dUdv\,d\mathrm{Vol}_{S^2}\,\frac{ e^{2v-2r}}{r^4}\left[\left(1+\frac{3}{r}\right)|\mathring{\slashed{\nabla}}\partial_U^k\psi_n|^2+\left(1+\frac{4}{r}\right)|\partial_U^k\psi_n|^2\right]\\&=\int_{-e^{-u_0}}^{0}dU\,d\mathrm{Vol}_{S^2}\,|\partial_U^{k+1}\psi|_{\hm}|^2+\int_{-\infty}^{v_0}dv\,d\mathrm{Vol}_{S^2}\, \frac{e^{v-1}}{r^3}\left(|\mathring{\slashed{\nabla}}\partial_U^k\psi|_{u_0}|^2+|\partial_U^k\psi|_{u_0}|^2\right)\\&+\int_{u\geq u_0, v\leq v_0}dUdv\,d\mathrm{Vol}_{S^2}\,V 2\partial_{U}^{k+1}\psi F_{k-1}.
    \end{split}
    \end{align}
    Proposition \ref{tranversally commuted data} implies that
    \begin{align}\label{f2}
        \int_{-\infty}^{v_0}dv\,d\mathrm{Vol}_{S^2}\, \frac{e^{v-1}}{r^3}\left(|\mathring{\slashed{\nabla}}\partial_U^{k-1}\psi|_{u_0}|^2+|\partial_U^{k-1}\psi|_{u_0}|^2\right)\lesssim \int_{-\infty}^{v_0}dv\,d\mathrm{Vol}_{S^2}\,\Omega^2|\lapo^n \psi|^2.
    \end{align}
    We then apply an inductive argument on $k$, where for each $k$ we use \eqref{f1} and \eqref{f2} and argue as in Proposition \ref{prop: the thing that nobody noticed apparently}.
\end{proof}

\begin{proposition}\label{the thing that nobody noticed commuted reversed in time}
    Assume that $\psi$ solves \eqref{wave equation} in the region
    \begin{align}
        \mathcal{N}=\{(u,v)\times S^2; (u,v)\in (u_0,\infty)\times (-\infty,v_0)\},
    \end{align}
    such that 
    \begin{align}\label{odd regularity assumption n reversed in time}
        \int_{\hp\cap \mathcal{N}}dud\mathrm{Vol}_{S^2}\;e^u(\partial_u(e^u\partial_u)^n\psi|_{\hp})^2<\infty,
    \end{align}
    \begin{align}\label{odd degenerate regularity assumption n reversed in time}
        \sum_{k=0}^{n}\int_{-\infty}^{\infty}\int_{S^2}du\,d\mathrm{Vol}_{S^2}\, |\partial_u \lapo^k\psi|_{v=v_0}|^2+\frac{\Omega^2}{r^2}|\mathring{\slashed{\nabla}}\lapo^k\psi|_{v=v_0}|^2+\frac{\Omega^2}{r^3}|\lapo^k\psi|_{v=v_0}|^2<\infty.
    \end{align}
    Then for any $u\leq u_0$, we have
    \begin{align}
        \sum_{k=0}^{n}\int_{\mathscr{C}_{u}\cap\mathcal{N}}dvd\mathrm{Vol}_{S^2} \frac{1}{\Omega^2}\left(\partial_v\left(\frac{1}{\Omega^2}\partial_v\right)^k\psi\right)^2&\lesssim\text{left hand side of }\eqref{odd regularity assumption n reversed in time}+\text{left hand side of }\eqref{odd degenerate regularity assumption n reversed in time}.
    \end{align}
\end{proposition}

\subsection{Rigidity of asymptotics at $\Scrip$ for fixed $\ell$}

\begin{proposition}\label{rigidity of asymptotics for n=0}
    Theorem \ref{exponential decay theorem'} holds for $n=0$.
\end{proposition}
\begin{proof}
    Assume $\psi$ is a solution \eqref{fixed ell mode wave equation 2} arising from via $\mathscr{B}^T_+(0,\psi_{\Scrip})$ for $\partial_u\psi_{\Scrip}\in \mathcal{E}_{\Scrip}^T$ and satisfying \eqref{regularity assumption}. An energy estimate applied to \eqref{wave equation in Kruskal coordinates} in the region $\{V\in[0,1], U\in[-1,0]\}$ shows that $(1+e^{\frac{u}{2}})\partial_u\psi_{\hm}\in L^2((0,\infty), du)$. By Paley--Wiener, we have that $a_{\hm}$ is the boundary value of a function that is holomorphic in the strip $\Im\omega\in(-\frac{1}{2},0)$ with 
    \begin{align}
        \sup_{0< \eta< \frac{1}{2}}\|(\omega-i\eta) a_{\hm}(\omega-i\eta)\|_{L^2_{\omega}}<\infty.
    \end{align}
    Since $\psi_{\Scrip}\in \mathcal{E}_{\Scrip}^T$, we also have by By Proposition \ref{scattering matrix relations among fixed frequency radiation fields} that $\omega T^{-1}a_{\hm}\in L^2(\mathbb{R})$. Therefore, we also have by Proposition \ref{prop: low frequency expansion} that $\omega^{-\ell} a_{\hm}\in L^2(\mathbb{R})$. We now show that 
    \begin{align}
        \sup_{0< \eta< \frac{1}{2}}\|(\omega-i\eta)^{-\ell} a_{\hm}(\omega-i\eta)\|_{L^2_{\omega}}<\infty.
    \end{align}
    We know that $\|(\omega-i\eta)^{-\ell} a_{\hm}(\omega-i\eta)\|_{L^2_{\omega}}<\infty$ for $\eta\in([0,\frac{1}{2})$, since $T\neq0$ for $\Im\omega>0$. Now denote by $F_k$ the inverse Fourier transform of $\omega^{-k}a_{\hm}$, and assume that $(1+e^{\frac{u}{2}})F_k\in L^2(\mathbb{R})$ and that $F_{k+1}\in L^2(\mathbb{R})$. We show that $(1+e^{\frac{u}{2}})F_{k+1}\in L^2(\mathbb{R})$. Indeed, for $u_2>u_1$ Cauchy--Schwartz gives
    \begin{align}
        |F_{k+1}(u_1)|^2\leq|F_{k+1}(u_2)|^2+ (e^{-u_1}-e^{-u_2})\int_{u_1}^{u_2}du\, e^{u}|F_k|^2 
    \end{align}
    Since $F_{k+1}\in L^2(\mathbb{R})$, there exists a sequence $\{u^*_n\}_{n=0}^{\infty}$ with $u_n^{*}\longrightarrow\infty$ monotonically as $n\longrightarrow\infty$, such that $|F_{k+1}(u^*_n)|\longrightarrow0$ as $n\longrightarrow\infty$. Therefore, we can take $u_2=u_n^*$ and take $n\longrightarrow\infty$ to get
    \begin{align}\label{i dont know}
         e^{u_1}|F_{k+1}(u_1)|^2\leq \int_{u_1}^{\infty}du\, e^{u}|F_k|^2.
    \end{align}
    We then find for $u'> u$
    \begin{align}
    \begin{split}
        \int_{u}^{u'}d\bar{u}\,e^{\bar{u}}|F_{k+1}|^2\leq  e^{u'}|F_{k+1}(u')|^2-\int_{u}^{u'}d\bar{u}\, e^{\bar{u}}F_{k+1}\partial_u F_{k+1},
    \end{split}
    \end{align}
    We then use Young's inequality and \eqref{i dont know} to get
    \begin{align}
        \begin{split}
            \int_{u}^{u'}d\bar{u}\,e^{\bar{u}}|F_{k+1}|^2\lesssim \int_{u}^{\infty}d\bar{u}\, e^{\bar{u}}|F_k|^2.
        \end{split}
    \end{align}
    We then have 
    \begin{align}
        \sup_{0< \eta< \frac{1}{2}}\|(\omega-i\eta)^{-k-1} a_{\hm}(\omega-i\eta)\|_{L^2_{\omega}}<\infty.
    \end{align}
    Repeating this argument over $k=0\dots \ell-1$, we then have $\sup_{0< \eta< \frac{1}{2}}\|\omega^{-\ell} a_{\hm}(\cdot-i\eta)\|_{L^2_{\Re\omega}}<\infty$. We then argue using Corollary \ref{prop: low frequency expansion off the real axis} and Lemma \ref{asymptotics of T} that 
    \begin{align}
    \begin{split}
        \int_{-\infty}^{\infty}d\omega\, |(\omega-i\eta) T(-\omega+i\eta)^{-1}|^2|a_{\hm}(\omega-i\eta)|^2\lesssim_{\ell}& \int_{\omega\in\mathbb{R},|\omega|\geq 1 }d\omega\,|\omega-i\eta|^2 a_{\hm}(\omega-i\eta)|^2\\&+\int_{\omega\in\mathbb{R},|\omega|\leq 1}d\omega\, |\omega-i\eta|^{-2\ell}|a_{\hm}(\omega-i\eta)|^2.
    \end{split}
    \end{align}
    By Proposition \ref{scattering matrix relations among fixed frequency radiation fields}, we have that $\omega a_{\Scrip}$ is the boundary value of a function that is holomorphic in the strip $\Im\omega\in(-\frac{1}{2},0)$ and is such that 
    \begin{align}
        \sup_{0< \eta< \frac{1}{2}}\|(\omega-i\eta)^{-\ell} a_{\Scrip}(\omega-i\eta)\|_{L^2_{\omega}}\lesssim \|\psi_{\Scrip}\|^2_{\mathcal{E}^{T}_{\Scrip}}+\int_{-\infty}^{\infty}du\, e^{-u}|\psi_{\hm}|^2,
    \end{align}
    and the result follows.
\end{proof}
\begin{lemma}\label{dumb lemma}
    A function $f$ satisfies for $k=0,\dots,n$
    \begin{align}
        \int_{-\infty}^{\infty}dx |f^{(k+1)}|^2<\infty,
    \end{align}
    \begin{align}
        \int_{-\infty}^{\infty}dx\, e^x(\partial_x(e^x\partial_x)^kf)^2<\infty,
    \end{align}
    if and only if $\mathcal{F}[f']$ is holomorphic in $\Im\omega \left(-n-\frac{1}{2},0\right)$ and satisfies
    \begin{align}
        \sup_{\eta\in(0,k+\frac{1}{2})}\left\|\left[\prod_{q=0}^{k}(\omega-iq-i\eta)\right]\hat{f}(\omega-i\eta)\right\|_{L^2_{\omega}(\mathbb{R})}^2<\infty,
    \end{align}
    where $\hat{f}=(-i\omega)^{-1}\mathcal{F}[f']$.
    and there exist constants $c,C>0$ such that
    \begin{align}
    \begin{split}
        &c\;\Bigg\|\left[\prod_{q=0}^{k}(\omega-iq)\right]\hat{f}(\omega)\Bigg\|_{PW(-k-\frac{1}{2},0)}^2\\&\leq \sum_{j=0}^{k}\int_{-\infty}^{\infty}dx |f^{(j+1)}|^2+\int_{-\infty}^{\infty}dx\, e^x(\partial_x(e^x\partial_x)^jf)^2\\&\leq C\;\Bigg\|\left[\prod_{q=0}^{k}(\omega-iq)\right]\hat{f}(\omega)\Bigg\|_{PW(-k-\frac{1}{2},0)}^2.
    \end{split}
    \end{align}
\end{lemma}
\begin{proof}
    This follows from Theorem \ref{Paley Wiener L2} and the identity
    \begin{align}
        (e^x\partial_x)^k=e^{nx}\prod_{q=0}^{k-1}(\partial_x+q).
    \end{align}
\end{proof}
\begin{proposition}\label{Theorem 1.1 for n geq 0}
     Theorem \ref{exponential decay theorem'} holds for all $n=\mathbb{N}_{\geq0}$.
\end{proposition}
\begin{proof}
    Since $e^{\frac{u}{2}}\partial_u(e^{u}\partial_u)^k\psi_{\hp}\in L^2(\mathbb{R})$ for all $k=1,\dots n$, and since $\partial_u^{k}\psi_{\hp}\in L^2(\mathbb{R})$ for all $k=1,\dots n+1$, Lemma \ref{dumb lemma} implies 
    \begin{align}
        \prod_{q=0}^{k}(\omega+iq)a_{\hm}(\omega,\ell)
    \end{align}
    is a holomorphic function of $\omega$ for $\Im\omega\in(-k-\frac{1}{2},0)$ for all $k=0,\dots,n$, square integrable on $\Im\omega=\omega_I$ for $\omega_I\in(-k-\frac{1}{2},0)$, and satisfies 
    \begin{align}
        \sup_{\lambda\in(0,k+\frac{1}{2})}\Big\|\prod_{q=0}^{k}(\,\cdot+iq-i\lambda)a_{\hm}(\,\cdot-i\lambda,\ell)\Big\|^2_{L^2(\mathbb{R})}<\infty.
    \end{align}
    Since $\omega a_{\Scrip}=\omega T(-\omega)^{-1}a_{\hp}$ and $\partial_u^{k}\psi_{\Scrip}\in L^2(\mathbb{R})$ for all $k=1,\dots n+1$, we can apply the argument leading to Proposition \ref{rigidity of asymptotics for n=0} to find that 
    \begin{align}
        \prod_{q=0}^{k}(\omega+iq)a_{\Scrip}(\omega,\ell)
    \end{align}
    is a holomorphic function of $\omega$ for $\Im\omega\in(-k-\frac{1}{2},0)$ for all $k=0,\dots,n$, square integrable on $\Im\omega=\omega_I$ for $\omega_I\in(-k-\frac{1}{2},0)$, and satisfies 
    \begin{align}
        \sup_{\lambda\in(0,k+\frac{1}{2})}\Big\|\prod_{q=0}^{k}(\,\cdot+iq-i\lambda)a_{\Scrip}(\,\cdot-i\lambda,\ell)\Big\|^2_{L^2(\mathbb{R})}<\infty.
    \end{align}
    Indeed, the $L^2$ norm over horizontal lines in the complex $\omega$ plane is log-convex when seen as a function of $\Im\omega$, thus
    \begin{align}\label{tatata}
    \begin{split}
        \sup_{\eta\in(0,\frac{1}{4})}\Big\|\prod_{q=0}^{k}(\,\cdot+iq-i\eta)a_{\Scrip}(\,\cdot-i\eta,\ell)\Big\|^2_{L^2(\mathbb{R})}\lesssim& \;\Big\|\prod_{q=0}^{k}(\,\cdot+iq)a_{\Scrip}(\,\cdot,\ell)\Big\|^2_{L^2(\mathbb{R})}\\&+\Big\|\prod_{q=0}^{k}\left(\,\cdot+iq-\frac{i}{4}\right)a_{\Scrip}\left(\,\cdot-\frac{i}{4},\ell\right)\Big\|^2_{L^2(\mathbb{R})},
    \end{split}
    \end{align}
    and the estimate \eqref{wonder bound on inverse of T} implies
    \begin{align}\label{dadada}
    \begin{split}
        \sup_{\eta\in(0,\frac{1}{4})}\Big\|\prod_{q=0}^{k}(\,\cdot+iq-i\eta)a_{\Scrip}(\,\cdot-i\eta,\ell)\Big\|^2_{L^2(\mathbb{R})}\lesssim&\;\Big\|\prod_{q=0}^{k}(\,\cdot+iq)a_{\Scrip}(\,\cdot,\ell)\Big\|^2_{L^2(\mathbb{R})}\\&+\Big\|\prod_{q=0}^{k}\left(\,\cdot+iq-\frac{i}{4}\right)a_{\hm}\left(\,\cdot-\frac{i}{4},\ell\right)\Big\|^2_{L^2(\mathbb{R})}.
    \end{split}
    \end{align}
    We then have
    \begin{align}
    \begin{split}
        \sup_{\eta\in(0,k+\frac{1}{2})}\Big\|\prod_{q=0}^{k}(\,\cdot+iq-i\eta)a_{\Scrip}(\,\cdot-i\eta,\ell)\Big\|^2_{L^2(\mathbb{R})}\lesssim&\;\Big\|\prod_{q=0}^{k}(\,\cdot+iq)a_{\Scrip}(\,\cdot,\ell)\Big\|^2_{L^2(\mathbb{R})}\\&+\sup_{\eta\in(0,k+\frac{1}{2})}\Big\|\prod_{q=0}^{k}\left(\,\cdot+iq-i\eta\right)a_{\hm}\left(\,\cdot-i\eta,\ell\right)\Big\|^2_{L^2(\mathbb{R})}.
    \end{split}
    \end{align}
    Therefore, by Theorem \ref{Paley Wiener L2}
    \begin{align}
        \sum_{j=0}^{k}\int_{-\infty}^{\infty}du\, e^u(\partial_u(e^u\partial_u)^j\psi_{\Scrip})^2\lesssim \sum_{j=0}^{k}\int_{-\infty}^{\infty}du\,|\partial_u^{j+1}\psi_{\Scrip}|^2+\int_{-\infty}^{\infty}du\,e^u(\partial_u(e^u\partial_u)^j\psi_{\hm})^2.
    \end{align}
\end{proof}
\begin{lemma}\label{clunky lemma}
    Assume for $n\geq\mathbb{N}_{\geq0}$ that $\psi$ is a solution to \eqref{fixed ell mode wave equation 2} with 
    \begin{align}
        \sum_{k=0}^{n}\int_{-\infty}^{\infty}du\, (\partial_u^{k+1}\psi_{\hm})^2+\int_{-\infty}^{\infty}du\,e^u(\partial_u(e^u\partial_u)^k\psi_{\hm})^2<\infty.
    \end{align}
    Assume moreover that $\partial_U^k\psi_{\hm}|_{U=0}=0$ for $k=1,\dots,n$. Then there exist constants $C,c>0$ such that
    \begin{align}
    \begin{split}
        &c\sum_{k=1}^{n+1}\int_{-\infty}^{\infty}du\, (1+e^{(2n+1)u})|\partial_u^j\psi_{\hm}|^2\\&\leq \sum_{k=0}^{n}\int_{-\infty}^{\infty}du\, (\partial_u^{k+1}\psi_{\hm})^2+\int_{-\infty}^{\infty}du\,e^u(\partial_u(e^u\partial_u)^k\psi_{\hm})^2\\&\leq C\sum_{k=1}^{n+1}\int_{-\infty}^{\infty}du\, (1+e^{(2n+1)u})|\partial_u^j\psi_{\hm}|^2.
    \end{split}
    \end{align}
    An analogous result holds with $\hm$ replaced by $\hp$, $u$ replaced by $-v$.
\end{lemma}
\begin{proof}
    We estimate
    \begin{align}\label{magic trick 1}
        \begin{split}
          \partial_U^j\psi_{\hm}&=(-1)^{j+1}\int_{U}^{0}dU_1\int_{U_1}^{0}dU_2\cdots\int_{U_{j}}^{0} dU_{j+1}\,\partial_U^{n+1}\psi_{\hm}\\\implies|\partial_U^j\psi_{\hm}|&\lesssim |U|^{n-j+\frac{1}{2}}\sqrt{\int_{U}^{0}d\bar{U}|\partial_U^{n+1}\psi_{\hm}|^2}.
        \end{split}
      \end{align}
      Since $\partial_u^k=\sum_{j=0}^{k} c^k_j U^j\partial_U^j$ for some constants $\{c^k_j\}_{j=1}^k$, we then have for $U\in[-1,0]$
        \begin{align}\label{magic trick 2}
        \begin{split}
            |\partial_u^k\psi_{\hm}|&\lesssim e^{-nu-\frac{1}{2}u}\sqrt{\int_{U}^{0}d\bar{U}|\partial_U^{n+1}\psi_{\hm}|^2}.
        \end{split}
        \end{align}\
        Therefore, $\omega^{k} a_{\hm}$ is holomorphic for $\omega$ such that $\Im\omega\in\left(-n-\frac{1}{2},0\right)$ and we have for any $\lambda<1$ and for $k=0,\dots,n+1$,
        \begin{align}\label{magic trick 3}
            \sup_{\eta\in \left(0,n+\frac{\lambda}{2}\right)}\|(\omega-i\eta)^k a_{\hm}(\omega-i\eta)\|_{L^2_{\omega}(\mathbb{R})}^2\lesssim \int_{-\infty}^{0}d\bar{U}|\partial_U^{n+1}\psi_{\hm}|^2+\sum_{k=1}^{n+1}\int_{-\infty}^{\infty}du\, |\partial_u^k\psi_{\hm}|^2.
        \end{align}
        We also have
        \begin{align}
                \sup_{\lambda\in(0,n+\frac{1}{2})}\left\|\left[\prod_{q=0}^{n}(\,\cdot+iq-i\lambda)\right]a_{\hm}(\,\cdot-i\lambda,\ell)\right\|^2_{L^2(\mathbb{R})}\lesssim \int_{-\infty}^{0}d\bar{U}|\partial_U^{n+1}\psi_{\hm}|^2+\sum_{k=1}^{n+1}\int_{-\infty}^{\infty}du\, |\partial_u^k\psi_{\hm}|^2 .
        \end{align}
        Since $\left[\prod_{q=0}^n(\omega+iq)\right]^{-1}$ is holomorphic for $\omega$ such that $\Im\omega\in \left(-n-\frac{1}{2},0\right)$, we conclude that
        \begin{align}
            \sup_{\eta\in \left(0,n+\frac{1}{2}\right)}\|(\omega-i\eta)^k a_{\hm}(\omega-i\eta)\|_{L^2_{\omega}(\mathbb{R})}^2\lesssim \int_{-\infty}^{0}d\bar{U}|\partial_U^{n+1}\psi_{\hm}|^2+\sum_{k=1}^{n+1}\int_{-\infty}^{\infty}du\, |\partial_u^k\psi_{\hm}|^2.
        \end{align}
        It is now clear that 
        \begin{align}
            \sum_{k=0}^{n}\int_{-\infty}^{\infty}du\, (\partial_u^{k+1}\psi_{\hm})^2+\int_{-\infty}^{\infty}du\,e^u(\partial_u(e^u\partial_u)^k\psi_{\hm})^2\lesssim \sum_{k=1}^{n+1}\int_{-\infty}^{\infty}du\, (1+e^{(2n+1)u})|\partial_u^j\psi_{\hm}|^2.
        \end{align}
\end{proof}
\begin{corollary}\label{magic corollary}
    For $n\in\mathbb{N}_{\geq0}$, let $\psi$ be as in Theorem \ref{exponential decay theorem'}. Assume that for 
    \begin{align}
        \partial_U^k\psi_{\hm}|_{U=0}=0\;\;\forall k=0,\dots, n.
    \end{align}
    Then $\psi_{\Scrip}\in \spacescriexpon{\ell}$, and we have for $k=1,\dots, n+1$
    \begin{align}
        \int_{-\infty}^{\infty}du\, (1+e^{(2n+1)u})|\partial_u^k\psi_{\Scrip}|^2\lesssim \sum_{j=1}^{n+1}\int_{-\infty}^{\infty}du\,|\partial_u^{j}\psi_{\Scrip}|^2+\int_{-\infty}^{\infty}du\, (1+e^{(2n+1)u})|\partial_u^j\psi_{\hm}|^2.
    \end{align}
\end{corollary}
\begin{proof}
    
        We use Lemma \ref{clunky lemma} and argue as in Proposition \ref{Theorem 1.1 for n geq 0} to obtain \eqref{tatata} and \eqref{dadada}.
\end{proof}
\subsection{The case of vanishing radiation at $\Scrip$}
\begin{proof}[Proof of Theorem \ref{construction of counterexample I}]
    For $\epsilon>0$, $\alpha>0$ and for any $n\in\mathbb{N}_{\geq0}$, consider
    \begin{align}\label{counterexample roa at hp}
        \omega a_{\hp}=-i^{n+\frac{\alpha-1}{2}}\frac{\prod_{k=0}^n\omega+(k+\frac{1}{2})i}{\left[\omega+(n+\frac{3}{2})i\right]^{2n+\frac{3+\alpha}{2}}}\cos(\epsilon\omega^2).
    \end{align}
    We note that $\omega^{k+1} a_{\hp}\in L^2(\mathbb{R})$ for $\omega\in\mathbb{R}$ and $k=0,\dots,n$. When $\omega\notin \mathbb{R}$, we have
    \begin{align}\label{cos identity for omega2}
        \cos(\epsilon\omega^2)=\cos(\epsilon(\omega_R^2-\omega_I^2))\cosh(2\epsilon\omega_{R}\omega_{I})-i\sin(\epsilon(\omega_R^2-\omega_I^2))\sinh(2\epsilon\omega_{R}\omega_{I}).
    \end{align}
    In particular, though $a_{\hp}$ is holomorphic for $\Im\omega\in(-n-\frac{3}{2},\infty)$, we have for any $\lambda \in(-n-\frac{3}{2},\infty), \lambda\neq0$,
    \begin{align}\label{lack of exponential decay for the ansatz}
        \omega a_{\hp}|_{\Im\omega=\lambda}\notin L^2_{\Re\omega}(\mathbb{R}).
    \end{align}
    In fact, we have for integer $m$ that $v^m\partial_v^{k}\psi_{\hp}\in L^2(\mathbb{R})$ if and only if $m+k<n+\frac{\alpha}{2}$. In particular, we have $\partial_v^{m+1}\psi_{\hp}\in L^2(\mathbb{R})$ if and only if $0\leq m< n+\frac{\alpha}{2}$.
    
    We find $a_{\hm}$ from \eqref{relation microlocal hm to hp scrip}, setting $a_{\Scrip}=0$,
    \begin{align}\label{sneaky reflection relation}
        a_{\hm}(\omega,\ell)=-\tilde{R}(-\omega,\ell)a_{\hp}(\omega,\ell),
    \end{align}
    and we have by Corollary \ref{analyticity of the wronskians}, Lemma \ref{hint at a pole of R tilde}, and the definition of $a_{\hp}$, that
    \begin{align}
        (\omega+ik)(\omega+i(k-1))\dots (\omega+i)\omega \tilde{R}(-\omega,\ell)a_{\hp}
    \end{align}
     is holomorphic for $\Im\omega\in(-k-\frac{1}{2},0]$. Given Proposition \ref{high frequency behaviour of reflection coefficient proposition n} and equation \eqref{cos identity for omega2}, we can choose $\epsilon$ such that 
    \begin{align}
        (\omega+ik)(\omega+i(k-1))\dots (\omega+i)\omega \tilde{R}(-\omega,\ell)\omega a_{\hm}\Big|_{\Im\omega=\lambda}\in L^2_{\Re\omega}(\mathbb{R})\;\;\forall k=0,\dots,n\;\;\forall \lambda\in \left(-k-\frac{1}{2},0\right],
    \end{align}
    \begin{align}
        \sup_{\lambda \in \left(-k-\frac{1}{2},0\right]}\Big\|\prod_{q=0}^k(\,\cdot+(\lambda+q)i)a_{\hm}(\,\cdot+(\lambda+q)i)\Big\|^2_{L^2(\mathbb{R})}<\infty.
    \end{align}
    By Theorem \ref{Paley Wiener L2}, we have
    \begin{align}
        \int_{-\infty}^{\infty}du\, e^{\bar{u}}|\partial_u(e^{\bar{u}}\partial_u)^k\partial_u\psi_{\hm}|^2<\infty\quad \forall k=0,\dots,n.
    \end{align}
    The result now follows by Proposition \ref{the thing that nobody noticed commuted}. 
\end{proof}
\begin{remark}
    The condition that $\partial_v^{k}\psi_{\hp}\in L^2(\mathbb{R})$ is only essential for $k=1$. This is because Proposition \ref{high frequency behaviour of reflection coefficient proposition 2} and equation \eqref{sneaky reflection relation} immediately imply that $\omega^ka_{\hm}\in L^2(\mathbb{R})$ for all $k\geq 1$.
\end{remark}
\begin{remark}
    We can use Proposition \ref{the thing that nobody noticed commuted} to show that for any $n\in\mathbb{N}_{\geq0}$ we can use \eqref{counterexample roa at hp} to construct a counterexample satisfying \eqref{regularity assumption n intro} and that $\partial_v^{k}\psi_{\hp}\in L^2(\mathbb{R})$ for $k=1,\dots,m$ and any $m$ with $1\leq m\leq n+1$. The counterexample will in particular satisfy that $v^j\partial_u^k\psi_{\hp}\notin L^2(\mathbb{R})$ for $j+k\geq m+\frac{\alpha}{2}$.  
\end{remark}

\subsection{The case of solutions of unbounded support in $\ell$ with vanishing radiation at $\hp$}\label{sec: counterexample at scrip}

\begin{proof}[Proof of Theorem \ref{construction of counter example Scrip}] 
Take
\begin{align}
    \omega a_{\Scrip}(\omega,\ell)=\frac{i}{(\ell+1)^{2n+\frac{3}{2}}}e^{-\cosh \frac{\pi}{4(n+1)}\omega}e^{i\ell\omega},    
\end{align}
\begin{align}
    \psi_{\Scrip}^{(\ell)}=\int^u \mathcal{F}^{-1}[-i\omega a_{\Scrip}(\omega,\ell)],\qquad \partial_u\psi_{\Scrip}(u,\theta,\phi)=\sum_{\ell=0}^{\infty}\sum_{\ell=-m}^m \partial_u\psi_{\Scrip}^{(\ell)}(u)Y^{\ell}_{m}(\theta,\phi).
\end{align}
    We already note that $\partial_u^k\psi_{\Scrip}\in L^2(\Scrip)$ for all $k\geq1$. We also have
    \begin{align}
         \int_{u_0}^{\infty} dud\mathrm{Vol}_{S^2} |\partial_u^k\lapo^n\psi_{\Scrip}|^2<\infty.
    \end{align}
    for all $k\geq1$. However, $\partial_u^k\lapo^{j}\psi_{\Scrip}\notin L^2(\Scrip)$ for all $j> n$ and all $k\geq 1$.
    Moreover, we have
    \begin{align}
         \int_{u_0}^{\infty} dud\mathrm{Vol}_{S^2} |u|^{2m}|\partial_u^k\lapo^j\psi_{\Scrip}|^2<\infty
    \end{align}
    if and only if $m+2j<2n$. By Theorem 9.5.3 of \cite{DRSR14}, The solution $\psi$ to \eqref{wave equation} arising via $\mathscr{B}^T_+$
    from scattering data $(0,\psi_{\Scrip})$ is such that 
    \begin{align}\label{borrowing from drsr 1}
        \psi_{\hm}^{(\ell)}=\int^u \mathcal{F}^{-1}[-iT(-\omega,\ell)\omega a_{\Scrip}(\omega,\ell)],\qquad \partial_u\psi_{\hm}(u,\theta,\phi)=\sum_{\ell=0}^{\infty}\sum_{\ell=-m}^m \partial_u\psi_{\hm}^{(\ell)}(u)Y^{\ell}_{m}(\theta,\phi).
    \end{align}
    Note also that for $\omega\in\mathbb{R}$ we have for any $k\geq1$ that $\sum_{\ell=0}^{L}\sum_{m=-\ell}^{\ell}\left(\int_{-\infty}^{\infty}d\omega |\omega^k a_{\Scrip}(\omega,\ell)|\right)^2$ converges as $L\longrightarrow\infty$. Therefore, $\partial_u^k\psi_{\Scrip}(u,\theta,\phi)|_{u}\in L^2(S^2)$. By Lemma \ref{rudimentary boundedness of T} and \eqref{borrowing from drsr 1}, the same applies to $\partial_u^k\psi_{\hm}$. 
    
    For each $\ell$, we have that $e^{\frac{u}{2}}\partial_u(e^u\partial_u)^k\psi_{\hm}^{(\ell)}\in L^2(\mathbb{R})$ for $k=0,\dots,n$, and that 
    \begin{align}
        \mathcal{F}[e^{\frac{u}{2}}\partial_u(e^u\partial_u)^k\psi_{\hm}^{(\ell)}]=\prod_{q=0}^{k}\left[\omega-\left(q+\frac{1}{2}\right)i\right]T\left(-\omega+\left(k+\frac{1}{2}\right)i\right)a_{\hm}\left(\omega-\left(k+\frac{1}{2}\right)i,\ell\right).
    \end{align}
    If we prove that 
    \begin{align}\label{the new plancherel}
        \sum_{\ell=0}^{L}\sum_{\ell=-m}^m \int_{-\infty}^{\infty} d\omega \left[\prod_{q=0}^{k}\left|\omega-\left(q+\frac{1}{2}\right)i\right|^2\right]\left|T\left(-\omega+\left(k+\frac{1}{2}\right)i\right)a_{\hm}\left(\omega-\left(k+\frac{1}{2}\right)i,\ell\right)\right|^2,
    \end{align}
    converges as $L\longrightarrow\infty$, we will have proven that $e^{\frac{u}{2}}\partial_u(e^u\partial_u)^k\psi_{\hm}\in L^2(\mathscr{H}^-)$.
    
    For $\Im\omega=-k-\frac{1}{2}$, denote $\Re\omega=\omega_R$. We find
    \begin{align}
    \begin{split}
    \left|\prod_{q=0}^{k}\left[\omega_R-\left(q+\frac{1}{2}\right)i\right]a_{\Scrip}\left(\omega_R-\left(k+\frac{1}{2}\right)i,\ell\right)\right|\lesssim  \frac{|\omega|^k}{(\ell+1)^{2n+\frac{3}{2}}} e^{-\cosh \omega}e^{\frac{1}{2}\epsilon(k+1)\ell}.
    \end{split}
\end{align}
We now apply Proposition \ref{holy grail of transmission}; on any interval $\omega_R\in[N,N+1]$ and $\ell\geq C|\omega_R-\frac{i(k+1)}{2}|^{\frac{9}{2}}$, 
\begin{align}
\begin{split}
    &\int_{N}^{N+1} \dd\omega_R \sum_{\ell\geq C(N+1)^{\frac{9}{2}}} \left[\prod_{q=0}^{k}\left|\omega_R-\left(q+\frac{1}{2}\right)i\right|^2\right]\left|a_{\hm}\left(\omega_R-\left(k+\frac{1}{2}\right)i,\ell\right)\right|^2\\&\leq  (N+1)^{2(k+1)}e^{-\cosh\frac{\pi}{2(n+1)} N} \sum_{\ell\geq C(N+1)^{\frac{9}{2}}}\frac{1}{(\ell+1)^{4n+3}}\ell^{-\ell} e^{\epsilon(k+1)\ell}\\&\lesssim (N+1)^{2(k+1)}e^{-\frac{\pi}{2(n+1)}\cosh N}.
\end{split}
\end{align}
For $\ell\leq C|\omega_R-\left(k+\frac{1}{2}\right)i|^\frac{9}{2}$ on any interval $\omega_R\in[N,N+1]$, Lemma \ref{rudimentary boundedness of T} gives 
\begin{align}
\begin{split}
   &\int_{N}^{N+1} \dd\omega_R \sum_{\ell\leq C(N+1)^{\frac{9}{2}}} \left[\prod_{q=0}^{k}\left|\omega_R-\left(q+\frac{1}{2}\right)i\right|^2\right]\left|a_{\hm}\left(\omega_R-\left(q+\frac{1}{2}\right)i,\ell\right)\right|^2\\&\lesssim (N+1)^{2(k+1)} e^{\epsilon (k+1) (N+1)^\frac{9}{2}}e^{-\frac{\pi}{2(n+1)}\cosh N}.
\end{split}
\end{align}
This shows that \eqref{the new plancherel} holds. The result follows by Proposition \ref{the thing that nobody noticed commuted}.
\end{proof}

\section{Non-degenerate scattering theory for fixed $\ell$}

\subsection{The case of spherically symmetric solutions}

We now prove Theorem \ref{non degenerate backwards scattering hp spherical symmetry} and Corollary \ref{corollary non degenerate backwards scattering hp spherical symmetry}.

\begin{proof}[Proof of Theorem \ref{non degenerate backwards scattering hp spherical symmetry}]
    We start with the case of ${}^{(\ell)}\mathcal{E}^{N,n}_{\Sigma}$. For a solution to \eqref{fixed ell mode wave equation 2} with $\ell=0$ arising from data $(\uppsi,\uppsi')\in{}^{(\ell=0)}\mathcal{E}^{N,n}_{\Sigma}$, we have that $\partial_v^{k}\psi_{\hp}\in L^2(\mathbb{R})$ for $k=1,\dots,n$. 
    Corollary \ref{standard redshift estimate}, applied for evolution towards the past and future to estimate $\|\psi_{\mathscr{H}^{\pm}}\|_{H^1(\mathbb{R})}$, implies
      \begin{align}\label{continuity of the forwards map spherical symmetry}
          \|\psi_{\hp}\|_{{}^{(\ell=0)}\mathcal{E}_{\hp}^{N,n}}^2+\|\tilde{\psi}_{\Scrip}\|^2_{\spacescriexpn{\ell=0}}\lesssim \|(\uppsi,\uppsi')\|_{{}^{(\ell=0)}\mathcal{E}_{\Sigma}^{N,n}}^2.
      \end{align}
    By Corollary \ref{prop: low frequency expansion off the real axis} and Lemma \ref{lem: matrix element identities at R}, we have that $\omega T(-\omega)^{-1}\tilde{R}(-\omega)$ is uniformly bounded. Let 
    \begin{align}\label{the scrip pair of psi hp}
         F_{\Scrip}[\psi_{\hp}]:= \int^{u}du\,\mathcal{F}^{-1}\left[\omega \frac{\tilde{R}(-\omega)}{T(-\omega)}a_{\hp}\right].
    \end{align}
    The solution $\psi_{\nwarrow}$ arising from 
      \begin{align}
          \mathscr{B}^{T,n}_+\left(\psi_{\hp},F_{\Scrip}[\psi_{\hp}]\right)
      \end{align}
      has no radiation at $\hm$ by Proposition \ref{scattering matrix relations among fixed frequency radiation fields}. By Proposition \ref{prop: the thing that nobody noticed apparently}, the solution $\psi_{\nwarrow}$ satisfies \eqref{regularity assumption}. Let
      \begin{align}
          \psi_{\nearrow}:=\psi-\psi_{\nwarrow}.
      \end{align}
      Then $\psi_{\nearrow}$ has no outgoing radiation at $\hp$. Moreover, the radiation field of $\psi_{\nearrow}$ at $\hm$ is identical to that of $\psi$, i.e.~$\psi_{\hm}$. Since $e^{\frac{u}{2}}\psi_{\hm}\in L^2(\mathbb{R})$, $\partial_u^{k+1}\psi_{\hm}\in L^2(\mathbb{R})$, and $e^{\frac{u}{2}}\partial_u(e^u\partial_u)^k\psi_{\hp}\in L^2(\mathbb{R})$ for all $k=0,\dots,n$, we may apply Theorem \ref{exponential decay theorem'} to find that the radiation field of $\psi_{\nearrow}$, which we will denote by $\tilde{\psi}_{\Scrip}$, satisfies $\tilde{\psi}_{\Scrip}\in\spacescriexpn{\ell=0}$, and we have
      \begin{align}\label{one more label}
          \|\tilde{\psi}_{\Scrip}\|^2_{\spacescriexpn{\ell=0}}\lesssim \sum_{k=0}^{n}\|e^{\frac{u}{2}}\partial_u(e^u\partial_u)^k\psi_{\hm}\|_{L^2(\mathbb{R})}^2+\sum_{k=0}^{n+1}\|\partial_u^{k}\psi_{\hm}\|^2_{L^2(\mathbb{R})}.
      \end{align}
      We now show that $\tilde{\psi}_{\Scrip}\in\spacescriexpon{\ell=0}$. Since $\psi_{\hm}$ arises from data in ${}^{(\ell)}\mathcal{E}^{N,n}_{\Sigma}$, we have that $\partial_U^k\psi_{\hm}|_{U=0}=0$ for $k=0,\dots, n$. We apply Lemma \ref{clunky lemma} and Corollary \ref{magic corollary} to find
        \begin{align}
            \begin{split}
                \|\tilde{\psi}_{\Scrip}\|^2_{\spacescriexpon{\ell=0}}\lesssim \sum_{j=1}^{n+1}\int_{-\infty}^{\infty}du\,|\partial_u^{j}\psi_{\Scrip}|^2+\int_{-\infty}^{\infty}du\, (1+e^{(2n+1)u})|\partial_u^j\psi_{\hm}|^2
            \end{split}
        \end{align}
        We combine the above with Lemma \ref{clunky lemma} and Proposition \ref{prop: standard local energy estimates past forwards} to conclude that
        \begin{align}\label{one more label 2}
          \|\tilde{\psi}_{\Scrip}\|^2_{\spacescriexpon{\ell=0}}\lesssim \|(\uppsi,\uppsi')\|_{{}^{(\ell=0)}\mathcal{E}_{\Sigma}^{N,n}}^2.
      \end{align}
      We define
      \begin{align}
          {}^{(\ell=0)}\mathscr{F}_+^{N,n}(\uppsi,\uppsi'):=(\psi_{\hp},\tilde{\psi}_{\Scrip}),
      \end{align}
      and we have that ${}^{(\ell=0)}\mathscr{F}_+^{N,n}$ is continuous by \eqref{continuity of the forwards map spherical symmetry}.
      
      We now establish the backwards map ${}^{(\ell=0)}\mathscr{B}_+^{N,n}$. For data $(\psi_{\hp},\tilde{\psi}_{\Scrip})\in {}^{(\ell=0)}\mathcal{E}_{\hp}^{N,n}\oplus \spacescriexpon{\ell=0}$, we define $\psi_{\nwarrow}$ to be the solution to \eqref{fixed ell mode wave equation 2} with $\ell=0$ arising via $\mathscr{B}^{T,n}_{+}(\psi_{\hp},F_{\Scrip}[\psi_{\hp}])$, where $F_{\Scrip}[\psi_{\hp}]$ is given by \eqref{the scrip pair of psi hp}. We define $\psi_{\nearrow}$ to be the solution to \eqref{fixed ell mode wave equation 2} with $\ell=0$ arising via $\mathscr{B}^{T,n}_{+}(0,\tilde{\psi}_{\Scrip})$. By Propositions \ref{prop: standard local energy estimates future backwards} and \ref{the thing that nobody noticed commuted} and the continuity of $\mathscr{B}^{T,n}_{+}$, we have
      \begin{align}
          \|(\psi_{\nwarrow}|_{\overline{\Sigma}}, n_{\overline{\Sigma}}\psi_{\nwarrow}|_{\overline{\Sigma}})\|_{{}^{(\ell=0)}\mathcal{E}_{\overline{\Sigma}}^{N,n}}^2\lesssim \|\psi_{\hp}\|_{{}^{(\ell=0)}\mathcal{E}_{\overline{\hp}}^{N,n}}^2.
      \end{align}
      Note the continuous inclusion $ {}^{(\ell=0)}\mathcal{E}_{\hp}^{N,n}\hookrightarrow {}^{(\ell=0)}\mathcal{E}_{\overline{\hp}}^{N,n}$, and if $\psi_{\hp}\in {}^{(\ell=0)}\mathcal{E}_{\overline{\hp}}^{N,n}$ is such that $\partial_V^k\psi_{\hp}|_{V=0}=0$ for $k=0,\dots,n$, then $\psi_{\hp}\in {}^{(\ell=0)}\mathcal{E}_{{\hp}}^{N,n}$. Therefore,
      \begin{align}
          \|(\psi_{\nwarrow}|_{\overline{\Sigma}}, n_{\overline{\Sigma}}\psi_{\nwarrow}|_{\overline{\Sigma}})\|_{{}^{(\ell=0)}\mathcal{E}_{\overline{\Sigma}}^{N,n}}^2\lesssim \|\psi_{\hp}\|_{{}^{(\ell=0)}\mathcal{E}_{{\hp}}^{N,n}}^2.
      \end{align}
      By Propositions \ref{the thing that nobody noticed commuted} and \ref{the thing that nobody noticed commuted reversed in time}, the solution $\psi_{\nwarrow}$ is $H^k_{loc}(\{U\leq0, V\geq0\})$ in the smooth structure defined by the coordinate system $(U,V,\theta^A)$. In particular, since $\partial_U^{k}\psi_{\nwarrow}|_{\hm}=0$, $\partial_V^{k}\psi_{\nwarrow}|_{\hm}=0$, we have that $\partial_R^k(\psi_{\nwarrow}|_{\overline{\Sigma}})=0$ for $k=0,\dots,n$, and $\partial_R^k(n_{\overline{\Sigma}}\psi_{\nwarrow}|_{\overline{\Sigma}})=0$ for $k=0,\dots,n-1$. Thus 
      \begin{align}
          \|(\psi_{\nwarrow}|_{{\Sigma}}, n_{{\Sigma}}\psi_{\nwarrow}|_{{\Sigma}})\|_{{}^{(\ell=0)}\mathcal{E}_{{\Sigma}}^{N,n}}^2\lesssim \|\psi_{\hp}\|_{{}^{(\ell=0)}\mathcal{E}_{{\hp}}^{N,n}}^2.
      \end{align}
      Since $\tilde{T}(-\omega,\ell)$ is uniformly bounded for $\Im\omega\geq0$ by Lemma \ref{rudimentary boundedness of T}, Theorem \ref{Paley Wiener L2} implies
      \begin{align}
          \sum_{k=0}^{n+1}\int_{-\infty}^{\infty}du\, (1+e^{(2n+1)u})|\partial_u^k\psi_{\nearrow}|_{\hm}|^2\lesssim \|\tilde{\psi}_{\Scrip}\|^2_{\spacescriexpon{\ell=0}}.
      \end{align}
      In particular, we have that $\partial_U^k\psi_{\nearrow}|_{\hm}=0$ for $k=0,\dots,n$. Propositions \ref{the thing that nobody noticed commuted} and \ref{the thing that nobody noticed commuted reversed in time} give us that the solution $\psi_{\nwarrow}$ is $H^k_{loc}(\{U\leq0, V\geq0\})$ in the smooth structure defined by the coordinate system $(U,V,\theta^A)$. In particular, since $\partial_U^{k}\psi_{\nearrow}|_{\hm}=0$, $\partial_V^{k}\psi_{\nearrow}|_{\hm}=0$, we have that $\partial_R^k(\psi_{\nearrow}|_{\overline{\Sigma}})=0$ for $k=0,\dots,n$, and $\partial_R^k(n_{\overline{\Sigma}}\psi_{\nearrow}|_{\overline{\Sigma}})=0$ for $k=0,\dots,n-1$. Thus, Propositions \ref{prop: standard local energy estimates future backwards}, \ref{prop: standard local energy estimates past forwards}, \ref{the thing that nobody noticed commuted}, \ref{the thing that nobody noticed commuted reversed in time}, and the continuity of $\mathscr{B}^{T,n}_{+}$ imply
      \begin{align}
          \|(\psi_{\nearrow}|_{\Sigma}, n_{\Sigma}\psi_{\nearrow}|_{\Sigma})\|_{{}^{(\ell=0)}\mathcal{E}_{\Sigma}^{N,n}}^2\lesssim\|\tilde{\psi}_{\Scrip}\|_{\spacescriexpon{\ell=0}}^2.
      \end{align}
      We define $\psi=\psi_{\nearrow}+\psi_{\nwarrow}$, and define
      \begin{align}\label{whatever}
          {}^{(\ell=0)}\mathscr{B}_+^{N,n}(\uppsi,\uppsi'):=(\psi|_{\Sigma},n_{\Sigma}{\psi}|_{\Sigma}).
      \end{align}
      The fact that ${}^{(\ell=0)}\mathscr{F}_+^{N,n}\circ{}^{(\ell=0)}\mathscr{B}_+^{N,n}=\mathrm{Id}_{{}^{(\ell=0)}\mathcal{E}_{\hp}^{N,n}\oplus \spacescriexpn{\ell=0}}$, ${}^{(\ell=0)}\mathscr{B}_+^{N,n}\circ{}^{(\ell=0)}\mathscr{F}_+^{N,n}=\mathrm{Id}_{{}^{(\ell=0)}\mathcal{E}_{\Sigma}^{N,n}}$ is a straightforward consequence of the fact that $\mathscr{F}^{T,n}_+$ and $\mathscr{B}^{T,n}_+$ are inverse to one another, and we leave the details to the reader.

      It remains to establish the maps ${}^{(\ell=0)}\mathscr{F}_+^{N,n}$, ${}^{(\ell=0)}\mathscr{B}_+^{N,n}$ for solutions associated with Cauchy data in ${}^{(\ell=0)}\mathcal{E}_{\overline{\Sigma}}^{N,n}$. Let $\chi$ be a smooth function that is identically equal to $1$ for $v\geq0$ and identically equal to $0$ for $v\leq -1$. For Cauchy data $(\uppsi,\uppsi')\in {}^{(\ell=0)}\mathcal{E}_{\overline{\Sigma}}^{N,n}$, we decompose the resulting solution $\psi$ into
      \begin{align}
          \psi=\psi_{\nearrow}+\psi_{\nwarrow}+\psi_{\mathcal{B}},
      \end{align}
      where $\psi_{\nwarrow}$ arises via $\mathscr{B}^{T,n}_+$ from scattering data $(\chi\psi_{\hp},F_{\Scrip}[\chi\psi_{\hp}])$, $\psi_{\mathcal{B}}$ arises via $\mathscr{B}^{T,n}_+$ from scattering data $((1-\chi)\psi_{\hp},0)$, and $\psi_{\nearrow}:=\psi-\psi_{\nwarrow}-\psi_{\mathcal{B}}$ has no radiation on $\hp$.

      The solution $\psi_{\nearrow}$ satisfies
      \begin{align}
          \psi_{\nearrow}|_{\hm}=\psi_{\hm}-\psi_{\mathcal{B}}|_{\hm}.
      \end{align}
      Note that the solution $\psi_{\mathcal{B}}$ is identically $0$ for $v\geq0$ by $\partial_t$-energy conservation, and we have by Proposition \ref{prop: standard local energy estimates future backwards} applied to $\psi_{\mathcal{B}}$ and $\partial_t$-energy conservation
      \begin{align}
          \begin{split}
              \sum_{k=0}^{n}\|e^{\frac{u}{2}}\partial_u(e^u\partial_u)^k\psi_{\mathcal{B}}|_{\hm}\|_{L^2(\mathbb{R})}^2+\sum_{k=1}^{n+1}\|\partial_u^{k}\psi_{\mathcal{B}}|_{\hm}\|^2_{L^2(\mathbb{R})} \lesssim\|(\uppsi,\uppsi')\|_{{}^{(\ell=0)}\mathcal{E}_{\overline{\Sigma}}^{N,n}}^2.
          \end{split}
      \end{align}
      Therefore, Theorem \ref{exponential decay theorem'} implies
      \begin{align}
           \|\tilde{\psi}_{\Scrip}\|^2_{\spacescriexpn{\ell=0}}\lesssim \|(\uppsi,\uppsi')\|_{{}^{(\ell=0)}\mathcal{E}_{\overline{\Sigma}}^{N,n}}^2.
      \end{align}
      By Proposition \ref{prop: standard local energy estimates future backwards} and the continuity of $\mathscr{F}^{T,n}_+$, we have
      \begin{align}
          \|\psi_{\hp}\|_{\spacehpbar}^2\lesssim \|(\uppsi,\uppsi')\|_{{}^{(\ell=0)}\mathcal{E}_{\overline{\Sigma}}^{N,n}}^2.
      \end{align}
      As before, we define
      \begin{align}
          {}^{(\ell=0)}\mathscr{F}_+^{N,n}(\uppsi,\uppsi'):=(\psi_{\hp},\tilde{\psi}_{\Scrip}).
      \end{align}
      The backwards map is constructed as follows: given data $(\psi_{\hp},\tilde{\psi}_{\Scrip})\in {}^{(\ell=0)}\mathcal{E}_{\overline{\hp}}^{N,n}\oplus \spacescriexpn{\ell=0}$, we define $\psi_{\nwarrow}$, $\psi_{\mathcal{B}}$ exactly as we did in the case of the forwards map, and define $\psi_{\nearrow}$ to be the solution to \eqref{fixed ell mode wave equation 2} with $\ell=0$ arising via $\mathscr{B}^{T,n}_{+}(0,\tilde{\psi}_{\Scrip})$. We take $\psi=\psi_{\nwarrow}+\psi_{\nearrow}+\psi_{\mathcal{B}}$. By Propositions \ref{prop: standard local energy estimates future backwards} and \ref{the thing that nobody noticed commuted} and the continuity of $\mathscr{B}^{T,n}_{+}$, we have
      \begin{align}
          \|(\psi_{\nwarrow}|_{\overline{\Sigma}}, n_{\overline{\Sigma}}\psi_{\nwarrow}|_{\overline{\Sigma}})\|_{{}^{(\ell=0)}\mathcal{E}_{\overline{\Sigma}}^{N,n}}^2\lesssim \|\psi_{\hp}\|_{{}^{(\ell=0)}\mathcal{E}_{\overline{\hp}}^{N,n}}^2.
      \end{align}
      By Proposition \ref{prop: standard local energy estimates future backwards} applied to $\psi_{\mathcal{B}}$ and $\partial_t$-energy conservation
      \begin{align}
          \begin{split}
            \left\|(\psi_{\mathcal{B}}|_{\overline{\Sigma}},n_{\overline{\Sigma}}\psi_{\mathcal{B}}|_{\overline{\Sigma}})\right\|_{\mathcal{E}^{N,n}_{\overline{\Sigma}}}^2\lesssim \|\psi_{\hp}\|^2_{\spacehpbarn}.
          \end{split}
      \end{align} 
      Since $e^{\frac{u}{2}}\partial_u(e^u\partial_u)^k=e^{(n+\frac{1}{2})u}\prod_{q=0}^{n}(\partial_u+q)$, the fact that $\tilde{T}(-\omega,\ell)$ is uniformly bounded for $\Im\omega\geq0$ by Lemma \ref{rudimentary boundedness of T} and Theorem \ref{Paley Wiener L2} immediately imply
      \begin{align}
         \sum_{k=0}^{n}\|e^{\frac{u}{2}}\partial_u(e^u\partial_u)^k\psi_{\nearrow}|_{\hm}\|_{L^2(\mathbb{R})}^2+\sum_{k=0}^{n+1}\|\partial_u^{k}\psi_{\nearrow}|_{\hm}\|^2_{L^2(\mathbb{R})}\lesssim  \|\tilde{\psi}_{\Scrip}\|^2_{\spacescriexpn{\ell=0}}.
      \end{align}
      We then define
      \begin{align}\label{whatever again}
          {}^{(\ell=0)}\mathscr{B}_+^{N,n}(\uppsi,\uppsi'):=(\psi|_{\overline{\Sigma}},n_{\overline{\Sigma}}{\psi}|_{\overline{\Sigma}}).
      \end{align}  
      We leave the argument that ${}^{(\ell=0)}\mathscr{B}^{N,n}_+$ and ${}^{(\ell=0)}\mathscr{F}^{N,n}_+$ are inverses to one another to the reader.
\end{proof}

\begin{proof}[Proof of Corollary \ref{corollary non degenerate backwards scattering hp spherical symmetry}]
    Assume $\psi$ solves \eqref{fixed ell mode wave equation 2} with $\ell=0$ and satisfies \eqref{regularity assumption n intro}. Define $\psi_{\Scrip}^{(1)}:=F_{\Scrip}[\psi_{\hp}]$ as in \eqref{the scrip pair of psi hp}, $\psi^{(1)}$ to be the solution arising via $\mathscr{B}^{T,n}_+$ from scattering data $(\psi_{\hp},\psi_{\Scrip}^{(1)})$. Then $\psi^{(2)}=\psi-\psi^{(1)}$ has no radiation at $\hp$. By Proposition \ref{prop: the thing that nobody noticed apparently}, we have that $\psi^{(2)}$ satisfies \eqref{regularity assumption}, and therefore $\psi_{\Scrip}^{(2)}$, the radiation field of $\psi^{(2)}$ at $\Scrip$, is in $\spacescriexpn{\ell=0}$ by Theorem \ref{exponential decay theorem'}. Taking $a_{\Scrip}^{(1)}$ to be the Fourier transform of $F_{\Scrip}[\psi_{\hp}]$, we have that for $k=1,\dots,n+1$,
    \begin{align}\label{the other thing}
        \omega^k a^{(1)}_{\hp} = \frac{T(-\omega,\ell=0)}{\tilde{R}(-\omega,\ell=0)} \omega^k a_{\Scrip}^{(1)}.
    \end{align}
    It then follows by Propositions \ref{high frequency behaviour of reflection coefficient proposition 2} and \ref{prop: low frequency expansion} and Lemma \ref{asymptotics of T} that $\frac{\omega^{k+1} }{\sqrt{1+\omega^2}}e^{4M\pi|\omega|}a^{(1)}_{\Scrip}(\omega)\in L^2(\mathbb{R})$.

    Conversely, assume $\psi_{\Scrip}=\psi_{\Scrip}^{(1)}+\psi_{\Scrip}^{(2)}$ where $\psi_{\Scrip}^{(2)}\in\spacescriexp{\ell=0}$, and $\frac{\omega^{k+   1} }{\sqrt{1+\omega^2}}e^{4M\pi|\omega|}\mathcal{F}[\partial_u\psi_{\Scrip}^{(1)}](\omega)\in L^2(\mathbb{R})$ for $k=1,\dots,n$. Let $\omega a_{\Scrip}^{(1)}=\mathcal{F}[\partial_u\psi_{\Scrip}^{(1)}]$, and let $\psi^{(1)}$ be the solution to \eqref{fixed ell mode wave equation 2} with $\ell=0$ arising $\mathscr{B}^{T,n}_+$ from scattering data $(\psi_{\hp},\psi_{\Scrip}^{(1)})$, where $\psi_{\hp}$ is defined in frequency space via \eqref{the other thing}. Then $\psi^{(1)}|_{\hm}=0$, and Proposition \ref{the thing that nobody noticed commuted} implies $\psi^{(1)}$ satisfies \eqref{regularity assumption n intro}. Let $\psi^{(2)}$ be the solution to \eqref{fixed ell mode wave equation 2} with $\ell=0$ arising $\mathscr{B}^T_+$ from scattering data $(0,\psi_{\Scrip}^{(2)})$, then $e^{\frac{u}{2}}\partial(e^u\partial_u)^k\psi_{\hm}^{(2)}\in L^2(\mathbb{R})$ for $k=1,\dots,n$, and Proposition \ref{the thing that nobody noticed commuted} implies $\psi^{(2)}$ satisfies \eqref{regularity assumption n intro}. Thus $\psi:=\psi^{(1)}+\psi^{(2)}$ satisfies \eqref{regularity assumption n intro} and attains $\psi_{\Scrip}=\psi_{\Scrip}^{(1)}+\psi_{\Scrip}^{(2)}$ as its radiation field at $\Scrip$.
\end{proof}
\subsection{Scattering for $\ell>0$ via time reversal isometry}
\subsubsection{The scattering map based at $\hp$}
We now prove Theorem \ref{non degenerate backwards scattering hp}.  As mentioned in the statement of Theorem \ref{non degenerate backwards scattering hp}, we make the assumption that 
    \begin{align}\label{formal assumption on low frequency expansion of tilde R}
        \tilde{R}(\omega,\ell)=\widetilde{W}_{\ell}^{(\ell)}(\omega)+\widetilde{H}_{\ell}^{(\ell)}(\omega),
    \end{align}
    where $\widetilde{W}_{\ell}^{(\ell)}$ and $\omega^{-\ell}\widetilde{H}_{\ell}^{(\ell)}$ are square integrable on $[-1,1]$.
\begin{defin}
    Define the spaces
    \begin{align}
    \cosspacen:=H^{n+1}((2M,\infty),\Omega^{-1}r^2dr),\qquad \sinespacen:=H^n((2M,\infty),\Omega^{-1}r^2dr).
\end{align}
\end{defin}

The following lemmata shows that time-symmetric and time-antisymmetric solutions can be identified by a simple criterion applied to their radiation fields 
\begin{lemma}\label{how to recognize time symmetric solutions}
    Let $\psi$ be a solution to \eqref{wave equation} via any one of the maps $\mathscr{B}^{T,n}_{\pm}$. Assume that $\psi_{\hp}(x)=\psi_{\hm}(-x)$ and that $\psi_{\Scrip}(x)=\psi_{\Scrim}(-x)$. Then $\psi$ is time-symmetric, i.e.
    \begin{align}
        \psi(-t,r,\theta^A)=\psi(t,r,\theta^A).
    \end{align}
\end{lemma}
\begin{proof}
    Let $\psi$ arise via $\mathscr{B}^{T,n}_{+}$. Assume without loss of generality that $\psi$ is smooth and that $\psi_{\Scrip}$ is attained as a pointwise limit towards $\Scrip$. Then since \eqref{Schwarzschild metric in Boyer Lindquist} is time-reversal invariant, the function $\isp$ defined by $\isp(t,r,\theta^A):=\psi(-t,r,\theta^A)$ is also a solution to \eqref{wave equation}. Moreover, using the coordinates \eqref{EF DNG} we have
    \begin{align}
        \isp|_{(u,v,\theta^A)}=\psi|_{(-v,-u,\theta^A)}.
    \end{align}
    In particular, $\isp$ attains the limit $\psi_{\hm}(-v)$ as $u\longrightarrow\infty$, and the limit $\psi_{\Scrim}(-u)$ as $v\longrightarrow\infty$. Thus $\isp-\psi$ has vanishing radiation at both $\overline{\hp}$ and $\Scrip$ by hypothesis, which means that $\isp-\psi=0$.
\end{proof}

An identical argument yields
\begin{lemma}\label{how to recognize time antisymmetric solutions}
    Let $\psi$ be a solution to \eqref{wave equation} via any one of the maps $\mathscr{B}^{T,n}_{\pm}$. Assume that $\psi_{\hp}(x)=-\psi_{\hm}(-x)$ and that $\psi_{\Scrip}(x)=-\psi_{\Scrim}(-x)$. Then $\psi$ is time-antisymmetric, i.e.
    \begin{align}
        \psi(-t,r,\theta^A)=-\psi(t,r,\theta^A).
    \end{align}
\end{lemma}

We are ready to prove  Theorem \ref{non degenerate backwards scattering hp}, starting with the following:
\begin{proposition}\label{prop: cosine isomorphism}
    Assume that $\tilde{R}$ satisfies \eqref{formal assumption on low frequency expansion of tilde R}. Then there exist Hilbert space isomorphisms
    \begin{align}
        {}^{(\ell)}\mathscr{F}_{+,\hp}^{N,\cos,n}:\cosspacen\longrightarrow \cosspacehpn,\qquad {}^{(\ell)}\mathscr{B}_{+,\hp}^{N,\cos,n}: \cosspacehpn\longrightarrow\cosspacen,
    \end{align}
    such that 
    \begin{align}
         {}^{(\ell)}\mathscr{F}_{+,\hp}^{N,\cos,n}\circ {}^{(\ell)}\mathscr{B}_{+,\hp}^{N,\cos,n}=\mathrm{Id}_{\cosspacehpn},\qquad {}^{(\ell)}\mathscr{B}_{+,\hp}^{N,\cos,n}\circ {}^{(\ell)}\mathscr{F}_{+,\hp}^{N,\cos,n}=\mathrm{Id}_{\cosspacen}.
    \end{align}
\end{proposition}
\begin{proof}
    We start with the forwards map $\cosmapforfutn$. Starting from data $\uppsi$ in $\cosspacen$, the forward scattering maps $\mathscr{F}_{\pm}^{N,n}$ acting on $(\uppsi,0)$ produce a solution $\psi$ such that the radiation fields $\psi_{\hm}$, $\psi_{\hp}$, $\psi_{\Scrim}$, $\psi_{\Scrip}$ are in $\mathcal{E}_{\hm}^{T,n}$, $\mathcal{E}_{\hp}^{T,n}$, $\mathcal{E}_{\Scrim}^{T,n}$, $\mathcal{E}_{\Scrip}^{T,n}$, respectively. Moreover, we have by Corollary \ref{standard redshift estimate} that 
    \begin{align}
        \int_{-\infty}^{\infty} dv\,|\psi_{\hp}|^2\lesssim \|\uppsi\|^2_{H^1(\Sigma)}.
    \end{align}
    Proposition \ref{prop: standard local energy estimates future backwards} gives for $k=0,\dots,n$,
    \begin{align}\label{the tenth label}
        \int_{-\infty}^{v_0}dv\, e^{-v}(\partial_v(e^{-v}\partial_v)^k\psi_{\hp})^2\lesssim \|\uppsi\|^2_{H^{k+1}((1,2v_0),\Omega^{-1}dr)}.
    \end{align}
    Since $\partial_V^k\psi_{\hp}|_{V=0}=0$ for $k=0,\dots,n$, we can apply Lemma \ref{clunky lemma} and Theorem \ref{Paley Wiener L2} to find
    \begin{align}
        \sum_{k=0}^{n}\int_{-\infty}^{\infty}dv\, \left|\left(1+e^{-\frac{(2n+1)v}{4M}}\right)\partial_v^{k}\psi_{\hp}\right|^2\lesssim \|\uppsi\|^2_{H^{k+1}((1,\infty),\Omega^{-1}dr)}.
    \end{align}
    The continuity of the map $\mathscr{F}^{T,n}_+$ implies
    \begin{align}
        \sum_{k=1}^{n+1}\int_{-\infty}^{\infty}du\, (\partial_u^k\psi_{\Scrip})^2\lesssim \|\uppsi\|^2_{H^{k+1}((1,\infty),\Omega^{-1}dr)}.
    \end{align}
    Since $\psi$ arises from data such that $\partial_t\psi|_{t=0}=0$, we must have
    \begin{align}
        \psi(-t,r)=\psi(t,r).
    \end{align}
    This immediately implies for any $x\in\mathbb{R}$
    \begin{align}
        \psi_{\hm}(x)=\psi_{\hp}(-x),
    \end{align}
    which gives
    \begin{align}\label{cosine microlocal identity hm}
        a_{\hm}(\omega,\ell)=a_{\hp}(-\omega,\ell).
    \end{align}
    We define
    \begin{align}
        \cosmapforfutn(\uppsi):=\psi_{\hp}.
    \end{align}
    Using \eqref{relation microlocal hm to hp scrip} and \eqref{cosine microlocal identity hm}, we have
    \begin{align}\label{cosine microlocal identity Scrip}
        \omega a_{\Scrip}=\frac{\omega}{T(-\omega)}\left[a_{\hp}(-\omega)+\tilde{R}(-\omega)a_{\hp}(\omega)\right].
    \end{align}
    Since $\omega^{k} a_{\hp}$ is square integrable for $k=0,\dots,n+1$, and $\omega a_{\Scrip}$ is square integrable for $k=1,\dots,n+1$, assumption \eqref{formal assumption on low frequency expansion of tilde R} implies for $k=0,\dots,n$,
    \begin{align}
        \int_{-\infty}^{\infty}d\omega \;\omega^{-2\ell+2k}[a_{\hp}(\omega)+\widetilde{W}_{\ell}^{(\ell)}(\omega)a_{\hp}(-\omega)]^2\lesssim \|\uppsi\|^2_{H^{k+1}((1,\infty),\Omega^{-1}dr)}.
    \end{align}
    We then conclude
    \begin{align}
        \|\psi_{\hp}\|_{\cosspacehpn}^2\lesssim \|\uppsi\|^2_{\cosspacen}.
    \end{align}
    
    We now construct the backwards map $\cosmapforbacn$.
    Let $\psi_{\hp}\in \cosspacehpn$, and define $\psi_{\Scrip}$ in frequency space via  \eqref{cosine microlocal identity Scrip}. By the definition of $\cosspacehpn$, we have $\psi_{\hp}\in\mathcal{E}_{\hp}^{T,n}$, $\psi_{\Scrip}\in\mathcal{E}_{\Scrip}^{T,n}$, and therefore we may construct via $\mathscr{B}^{T,n}_+$ a unique solution $\psi$ realising $\psi_{\hp}$, $\psi_{\Scrip}$ as radiation fields at $\hp$, $\Scrip$. Define
    \begin{align}
        \cosmapforbacn(\psi_{\hp}):=\psi|_{t=0}.
    \end{align}
    The solution $\psi$ extends to define radiation fields $\psi_{\hm}$, $\psi_{\Scrim}$ on $\hm$, $\Scrim$ via the map $\mathscr{B}^{T,n}_-$. Using \eqref{relation microlocal hm to hp scrip}, we find that $\psi_{\hm}$ satisfies \eqref{cosine microlocal identity hm}. Similarly, using \eqref{relation microlocal hp to hm scrim}, we find that $a_{\Scrim}(\omega,\ell)=a_{\Scrip}(-\omega,\ell)$. Thus $\psi(-t,r)=\psi(t,r)$ by Lemma \ref{how to recognize time symmetric solutions}, and therefore $\partial_t\psi|_{t=0}=0$. In particular, note
    \begin{align}\label{mirror magic}
        \left\|\left(1+e^{\frac{(2n+1)u}{2}}\right)\partial_u^k\psi_{\hm}\right\|_{L^2(\hm)}= \left\|\left(1+e^{-\frac{(2n+1)v}{2}}\right)\partial_v^k\psi_{\hp}\right\|_{L^2(\hp)}.
    \end{align}
    for $k=0,\dots,n$. Equation \ref{mirror magic}, together with and the fact that $\psi_{\Scrip}\in \mathcal{E}^{T,n}_{\Scrip}$ imply that  \eqref{regularity assumption n intro} is satisfied by Proposition \ref{the thing that nobody noticed commuted}, and we have
    \begin{align}
        \begin{split}
\sum_{j=0}^k\int_{u_0}^{\infty}du\,\frac{1}{\Omega^2}\left|\partial_u\left(\frac{1}{\Omega^2}\partial_u\right)^j\phi(u,v)\right|^2\lesssim&\; \sum_{j=0}^{k}\|e^{-\frac{v}{2}}(\partial_v(e^{-v}\partial_v)^j)\partial_v\psi_{\hp}\|_{L^2(\hp)}\\&+\sum_{j=0}^{k}\int_{-\infty}^{\infty}du |\partial_u^{j+1}\psi_{\Scrip}|^2.
\end{split}
    \end{align}
    Proposition \ref{prop: standard local energy estimates future backwards} gives
    \begin{align}
    \begin{split}
        \|\psi|_{t=0}\|_{H^{k+1}(\{r\in(1,2)\},\Omega^{-1}dr)}&\lesssim \sum_{j=0}^k\int_{-\infty}^{0}dv e^v(\partial_v(e^v\partial_v)^j\psi_{\hp})^2+\sum_{j=0}^{k}\int_{-\infty}^{\infty}du |\partial_u^j\psi_{\Scrip}|^2\\&\lesssim \|\psi_{\hp}\|^2_{\cosspacehpn}.
    \end{split}
    \end{align}
    Note that $\cosmapforfutn$ is the restriction of $\mathscr{F}^{T,n}_+$ to $\cosspacen\oplus\{0\}$. Moreover, for any $\psi_{\hp}\in\cosspacehpn$, the relation \eqref{cosine microlocal identity Scrip} defines a unique element $\psi_{\Scrip}\in \mathcal{E}^{T,n}_{\Scrip}$, and therefore $\cosmapforbacn$ is the restriction of $\mathscr{B}_+^{T,n}$ to the space
    \begin{align}
        \left\{(\psi_{\hp}, \psi_{\Scrip})\in \mathcal{E}_{\hp}^{T,n}\oplus \mathcal{E}_{\hp}^{T,n}:\; \mathcal{F}[\partial_u\psi_{\Scrip}]=\frac{1}{T(-\omega)}\left[\mathcal{F}[\partial_v\psi_{\hp}](-\omega)+\tilde{R}(-\omega)\mathcal{F}[\partial_v\psi_{\hp}](\omega)\right]\right\}.
    \end{align}
    Since $\mathscr{F}_+^{T,n}$ and $\mathscr{B}_+^{T,n}$ are inverses to one another, we immediately conclude that $\cosmapforfutn$, $\cosmapforbacn$ are inverses to one another.
\end{proof}
An identical argument leads to 
\begin{proposition}\label{prop: sine isomorphism}
        Assume that $\tilde{R}$ satisfies \eqref{formal assumption on low frequency expansion of tilde R}. Then there exist Hilbert space isomorphisms
    \begin{align}
        {}^{(\ell)}\mathscr{F}_+^{N,\sin,n}:\sinespace\longrightarrow \sinespacehpn,\qquad {}^{(\ell)}\mathscr{B}_+^{N,\sin,n}: \sinespacehpn\longrightarrow\sinespacen,
    \end{align}
    such that 
    \begin{align}
         {}^{(\ell)}\mathscr{F}_+^{N,\sin,n}\circ {}^{(\ell)}\mathscr{B}_+^{N,(\sin)}=\mathrm{Id}_{\sinespacehpn},\qquad {}^{(\ell)}\mathscr{B}_+^{N,\sin,n}\circ {}^{(\ell)}\mathscr{F}_+^{N,\sin,n}=\mathrm{Id}_{\sinespacen}.
    \end{align}
\end{proposition}
\begin{proof}
    We may repeat the argument leading to Proposition \ref{prop: cosine isomorphism}, noting that in this case we use the relation $\psi(-t,r)=-\psi(t,r)$, thus we use $a_{\hm}(\omega,\ell)=-a_{\hp}(-\omega,\ell)$ in place of \eqref{cosine microlocal identity hm}, and 
    \begin{align}
        a_{\Scrip}=\frac{1}{T(-\omega)}\left[-a_{\hp}(-\omega)+\tilde{R}(-\omega)a_{\hp}(\omega)\right],
    \end{align}
    in place of \eqref{cosine microlocal identity Scrip} throughout the argument, and the result follows.
\end{proof}
\begin{corollary}
    There exists Hilbert space maps
    \begin{align}
        {}^{(\ell)}\mathscr{F}_{+,\hp}^{N,n}:\mathcal{E}_{\Sigma}^N\longrightarrow \mathcal{E}^N_+,\qquad {}^{(\ell)}\mathscr{B}_{+,\hp}^{N,n}:\mathcal{E}^N_+\longrightarrow\mathcal{E}_{\Sigma}^N,
    \end{align}
    where the space $\mathcal{E}^{N,n}_+$ is defined by
    \begin{align}
        \mathcal{E}^{N,n}_+:=\cosspacehpn\oplus\sinespacehpn.
    \end{align}
    The maps $ {}^{(\ell)}\mathscr{F}^{N,n}_{+,\hp},  {}^{(\ell)}\mathscr{B}^{N,n}_{+,\hp}$ are such that
    \begin{align}
        {}^{(\ell)}\mathscr{F}_{+,\hp}^{N,n}\circ{}^{(\ell)}\mathscr{B}_{+,\hp}^{N,n}=\mathrm{Id}_{\mathcal{E}^{N,n}_+},\qquad {}^{(\ell)}\mathscr{B}_{+,\hp}^{N,n}\circ{}^{(\ell)}\mathscr{F}_{+,\hp}^{N,n}=\mathrm{Id}_{\mathcal{E}_{\Sigma}^{N,n}}.
    \end{align}
\end{corollary}
\begin{proof}
    Since $\mathcal{E}_{\Sigma}^{N,n}=\cosspacen\oplus \sinespacen$, and by Propositions \ref{prop: cosine isomorphism} and \ref{prop: sine isomorphism}, the map ${}^{(\ell)}\mathscr{F}_{+,\hp}^{N,n}$ is simply the canonical isomorphism between $\mathcal{E}_{\Sigma}^{N,n}$ and $\cosspacehpn\oplus\sinespacehpn$, with the map ${}^{(\ell)}\mathscr{B}_{+,\hp}^{N,n}$ being the canonical inverse.
\end{proof}
\begin{corollary}
    A radiation field $\psi_{\hp}\in\mathcal{E}^{T,n}_{\hp}$ arises from initial data in $\mathcal{E}_{\Sigma}^{N,n}$ if and only if $\psi_{\hp}=\psi_{\hp}^{\cos}+\psi_{\hp}^{\cos}$, where
    \begin{align}
        \psi_{\hp}^{\cos}\in\cosspacehpn,\qquad \psi_{\hp}^{\sin}\in\sinespacehpn.
    \end{align}
\end{corollary}
\begin{defin}
    Define
    \begin{align}
    \begin{split}
        \cspacehp=\Big\{&f:\mathbb{R}\longrightarrow\mathbb{R}; f^{(k)}\in L^2(\mathbb{R})\;\;\forall k=0,\dots,n+1\\&e^{-\frac{v}{2}}\partial_v(e^{-v}\partial_v)^kf\in L^2(\mathbb{R})\;\;\forall k=0,\dots,n\\&\omega^{-\ell}\mathcal{F}[f]\in L^2_{loc}(\mathbb{R})\Big\},
    \end{split}
    \end{align}
    equipped with the norm
    \begin{align}
        \begin{split}
            \|f\|_{\cspacehp}^2=&\sum_{k=0}^{n+1}\left\|\left(1+e^{-\frac{(2n+1)v}{2}}\right)f^{(k)}\right\|_{L^2(\mathbb{R})}^2+\int_{-1}^{1}d\omega\; \omega^{-2\ell}\left|\mathcal{F}[f](\omega)\right|^2.
        \end{split}
    \end{align}
\end{defin}
 Recall that we defined $\kspacehp$ in \eqref{def of kspace intro} to be $\kspacehp:=\cosspacehpn\cap\sinespacehpn$. We now show 
\begin{proposition}
The space $\kspacehp$ satisfies
    \begin{align}
        \kspacehp=\big\{(v,-v): v\in \cspacehp\big\},
    \end{align}
\end{proposition}
\begin{proof}
    If $(v,-v)\in \kspacehp$, then 
    \begin{align}
    \begin{split}
        \int_{-1}^{1}d\omega\,\omega^{-2\ell}|\mathcal{F}[v](\omega)|^2\leq& \int_{-1}^{1}d\omega\,\omega^{-2\ell}|\mathcal{F}[v](\omega)+\widetilde{W}_{\ell}^{(\ell)}(\omega)\mathcal{F}[v](-\omega)|^2\\&+\int_{-1}^{1}d\omega\,\omega^{-2\ell}|\mathcal{F}[v](\omega)-\widetilde{W}_{\ell}^{(\ell)}(\omega)\mathcal{F}[v](-\omega)|^2\\\leq&\|v\|_{\cosspacehpn}^2+\|v\|_{\sinespacehpn}^2.
    \end{split}
    \end{align}
    Thus $v\in\cspacehp$. Conversely, if $v\in \cspacehp$, then it is immediate that $v\in\cosspacehpn\cap\sinespacehpn$. 
\end{proof}
\begin{proposition}
    The space $\kspacehp$ is a closed subspace of $\cosspacehpn\oplus\sinespacehpn$.
\end{proposition}
\begin{proof}
    It is immediate that if a sequence $\{(v_n,-v_n)\}$ in $\kspacehp$ is Cauchy in the norm of $\cosspacehpn\oplus \sinespacehpn$, then $\{v_n\}$ will be a Cauchy sequence in the norm of $\cspacehp$. Thus $\{v_n\}$ converges to a limit $v\in\cspacehp$, which means $(v,-v)\in\kspacehp$.
\end{proof}
We recall the definition of $\spacehpn$:
\begin{defin}\label{def of spacehpn}
    Define the space $\spacehp$ by
    \begin{align}
    \begin{split}
        f\in \spacehpn\quad \Longleftrightarrow  &\quad f=f^{\cos}+f^{\sin}, \qquad f^{\cos}\in \cosspacehpn, \;\;f^{\sin}\in \sinespacehpn,\\&\quad \left\|(f^{\cos},f^{\sin})\right\|_{\cosspacehpn\oplus \sinespacehpn}=\inf_{g\in \cspacehp}\left\|(f^{\cos}+g,f^{\sin}-g)\right\|_{\cosspacehpn\oplus \sinespacehpn}.
    \end{split}
    \end{align}
    We equip $\spacehp$ with the norm
    \begin{align}
        \left\|f\right\|_{\spacehpn}:=\left\|(f^{\cos},f^{\sin})\right\|_{\cosspacehpn\oplus \sinespacehpn}.
    \end{align}
\end{defin}

To characterise the image of $\cosspacebar$ under ${}^{(\ell)}\mathscr{F}^N_{+}$, we make the assumption that 
    \begin{align}\label{formal assumption on low frequency expansion of tilde R ell plus 1}
        \tilde{R}(\omega,\ell)=\widetilde{W}_{\ell}^{(\ell+1)}(\omega)+\widetilde{H}_{\ell}^{(\ell+1)}(\omega),
    \end{align}
    where $\widetilde{W}_{\ell}^{(\ell+1)}$ and $\omega^{-\ell-1}\widetilde{H}_{\ell}^{(\ell+1)}$ are square integrable on $[-1,1]$.
\begin{proposition}\label{prop: cosine isomorphism bar}
        Assume that $\tilde{R}$ satisfies \eqref{formal assumption on low frequency expansion of tilde R ell plus 1}. Then there exist Hilbert space isomorphisms
    \begin{align}
        {}^{(\ell)}\mathscr{F}_{+,\overline{\hp}}^{N,\cos}:\cosspacebar\longrightarrow \cosspacehpbar,\qquad {}^{(\ell)}\mathscr{B}_{+,\overline{\hp}}^{N,\cos}: \cosspacehpbar\longrightarrow\cosspacebar,
    \end{align}
    such that 
    \begin{align}
         {}^{(\ell)}\mathscr{F}_{+,\overline{\hp}}^{N,\cos}\circ {}^{(\ell)}\mathscr{B}_{+,\overline{\hp}}^{N,\cos}=\mathrm{Id}_{\cosspacehpbar},\qquad {}^{(\ell)}\mathscr{B}_{+,\overline{\hp}}^{N,\cos}\circ {}^{(\ell)}\mathscr{F}_{+,\overline{\hp}}^{N,\cos}=\mathrm{Id}_{\cosspacebar}.
    \end{align}
\end{proposition}
\begin{proof}
    We follow a very similar argument to that leading to Proposition \ref{prop: cosine isomorphism}, making use of assumption \eqref{formal assumption on low frequency expansion of tilde R ell plus 1} instead of \eqref{formal assumption on low frequency expansion of tilde R}. For $\uppsi$ in $\cosspacebarn$, the solution $\psi$ arising from Cauchy data $(\uppsi,0)$ on $\overline{\Sigma}$ induces a radiation field $\psi_{\hp}$ on $\overline{\hp}$ which satisfies by Proposition \ref{prop: standard local energy estimates future backwards}  for $k=0,\dots,n$,
    \begin{align}
        \int_{-\infty}^{v_0}dv\, e^{-v}(\partial_v(e^{-v}\partial_v)^k\psi_{\hp})^2\lesssim \|\uppsi\|^2_{H^{k+1}((1,2v_0),\Omega^{-1}dr)}.
    \end{align}
    We define $\psi_{\Scrip}$ via \eqref{cosine microlocal identity Scrip}. We use the fact that $\omega^k a_{\hp}$ is square integrable for $k=1,\dots,n+1$ to find by \eqref{formal assumption on low frequency expansion of tilde R ell plus 1}
    \begin{align}
        \int_{-\infty}^{\infty}d\omega \;\omega^{-2\ell-2+2k}[a_{\hp}(\omega)+\widetilde{W}_{\ell}^{(\ell+1)}(\omega)a_{\hp}(-\omega)]^2\lesssim \|\uppsi\|^2_{H^{k+1}((1,\infty),\Omega^{-1}dr)}.
    \end{align}
    and we conclude $\|\psi_{\hp}\|_{\cosspacehpbarn}^2\lesssim \|\uppsi\|^2_{\cosspacebarn}$. We then define $\cosmapforfutbarn(\uppsi):=\psi_{\hp}$.
    
    As for the backwards map, for $\psi_{\hp}\in \cosspacehpn$ we define $\psi_{\Scrip}$ in frequency space via  \eqref{cosine microlocal identity Scrip}, then as in the case of $\cosmapforbacn$, we have that $\psi$ is time-symmetric by Lemma \ref{how to recognize time symmetric solutions}, and thus
    \begin{align}\label{mirror magic 2}
        \|e^{\frac{u}{2}u}(\partial_u(e^u\partial_u)^k)\psi_{\hm}\|_{L^2\left(\overline{\hm}\right)}= \|e^{-\frac{v}{2}}(\partial_v(e^{-v}\partial_v)^k)\partial_v\psi_{\hp}\|_{L^2\left(\overline{\hp}\right)}.
    \end{align}
    We then argue as in Proposition \ref{prop: cosine isomorphism}.
\end{proof}
\begin{proposition}\label{prop: sine isomorphism bar}
        Assume that $\tilde{R}$ satisfies \eqref{formal assumption on low frequency expansion of tilde R}. Then there exist Hilbert space isomorphisms
    \begin{align}
        {}^{(\ell)}\mathscr{F}_+^{N,\sin,n}:\sinespacebarn\longrightarrow \sinespacehpbarn,\qquad {}^{(\ell)}\mathscr{B}_+^{N,\sin,n}: \sinespacehpbarn\longrightarrow\sinespacebarn,
    \end{align}
    such that 
    \begin{align}
         {}^{(\ell)}\mathscr{F}_+^{N,\sin,n}\circ {}^{(\ell)}\mathscr{B}_+^{N,\sin,n}=\mathrm{Id}_{\sinespacehpbarn},\qquad {}^{(\ell)}\mathscr{B}_+^{N,\sin,n}\circ {}^{(\ell)}\mathscr{F}_+^{N,\sin,n}=\mathrm{Id}_{\sinespacebarn}.
    \end{align}
\end{proposition}

\begin{corollary}
    A radiation field $\psi_{\hp}\in\mathcal{E}^{T,n}_{\overline{\hp}}$ arises from initial data in $\mathcal{E}_{\overline{\Sigma}}^{N,n}$ if and only if $\psi_{\hp}=\psi_{\hp}^{\cos}+\psi_{\hp}^{\cos}$, where
    \begin{align}
        \psi_{\hp}^{\cos}\in\cosspacehpbarn,\qquad \psi_{\hp}^{\sin}\in\sinespacehpbarn.
    \end{align}
\end{corollary}
\begin{defin}
    Define
    \begin{align}
    \begin{split}
        \cspacehpbar=\Big\{&f:\mathbb{R}\longrightarrow\mathbb{R}; f^{(k)}\in L^2(\mathbb{R})\;\;\forall k=1,\dots,n+1,\\&e^{-\frac{v}{2}}\partial_v(e^{-v}\partial_v)^kf\in L^2(\mathbb{R})\;\;\forall k=0,\dots,n\\&\omega^{-\ell-1}\mathcal{F}[f]\in L^2_{loc}(\mathbb{R})\Big\},
    \end{split}
    \end{align}
    equipped with the norm
    \begin{align}
        \begin{split}
            \|f\|_{\cspacehpbar}^2=&\sum_{k=1}^{n+1}\|f^{(k)}\|_{L^2(\mathbb{R})}^2+\sum_{k=0}^{n}\|e^{-\frac{v}{4M}}\partial_v(e^{-\frac{v}{2M}}\partial_v)^kf\|_{L^2(\mathbb{R})}^2+\int_{-1}^{1}d\omega\; \omega^{-2\ell-2}\left|\mathcal{F}[f](\omega)\right|^2.
        \end{split}
    \end{align}
\end{defin}
\begin{proposition}
The space $\kspacehpbar$ satisfies
    \begin{align}
        \kspacehpbar=\big\{(v,-v): v\in \cspacehpbar\big\},
    \end{align}
\end{proposition}
\begin{proof}
    If $(v,-v)\in \kspacehpbar$, then $v\in L^2(\mathbb{R})$, and we have
    \begin{align}
        \omega^{-\ell-1}\left[\mathcal{F}[v'](\omega)-\widetilde{W}_{\ell}^{(\ell+1)}(\omega)\mathcal{F}[v'](\omega)\right]=-i\omega^{-\ell}\left[\mathcal{F}[v](\omega)+\widetilde{W}_{\ell}^{(\ell+1)}(\omega)\mathcal{F}[v](\omega)\right].
    \end{align}
    We estimate
    \begin{align}
    \begin{split}
        \int_{-1}^{1}d\omega\,\omega^{-2\ell}|\mathcal{F}[v](\omega)|^2\leq& \int_{-1}^{1}d\omega\,\omega^{-2\ell}|\mathcal{F}[v](\omega)+\widetilde{W}_{\ell}^{(\ell+1)}(\omega)\mathcal{F}[v](-\omega)|^2\\&+\int_{-1}^{1}d\omega\,\omega^{-2\ell}|\mathcal{F}[v](\omega)-\widetilde{W}_{\ell}^{(\ell)}(\omega)\mathcal{F}[v](-\omega)|^2\\&+\int_{-1}^{1}d\omega\,\omega^{-2\ell}|\widetilde{W}_{\ell}^{(\ell+1)}(\omega)\mathcal{F}[v](-\omega)-\widetilde{W}_{\ell}^{(\ell)}(\omega)\mathcal{F}[v](-\omega)|^2.
    \end{split}
    \end{align}
    Note that by \eqref{formal assumption on low frequency expansion of tilde R} and \eqref{formal assumption on low frequency expansion of tilde R ell plus 1}, we get
    \begin{align}
        \widetilde{W}_{\ell}^{(\ell+1)}(\omega)-\widetilde{W}_{\ell}^{(\ell)}(\omega)=\widetilde{H}_{\ell}^{(\ell+1)}(\omega)-\widetilde{H}_{\ell}^{(\ell)}(\omega),
    \end{align}
    and we have that $\omega^{-\ell}\widetilde{H}_{\ell}^{(\ell)}$, $\omega^{-\ell}\widetilde{H}_{\ell}^{(\ell+1)}$, $\mathcal{F}[v]$ are all square integrable. Thus $v\in\cspacehpbar$. Conversely, if $v\in \cspacehp$, then it is immediate that $v\in\cosspacehpbarn\cap\sinespacehpbarn$. 
\end{proof}
\begin{proposition}
    The space $\kspacehpbar$ is a closed subspace of $\cosspacehpbarn\oplus\sinespacehpbarn$.
\end{proposition}

\begin{defin}\label{def of spacehpbarn}
    The space $\spacehpbarn$ is defined via Definition \ref{def of spacehpn} with $\hp$ replaced by $\overline{\hp}$ and $\cosspacehpn$ replaced by $\cosspacehpbarn$.
\end{defin}

We are ready to prove Theorem \ref{non degenerate backwards scattering hp}:
\begin{proof}[Proof of Theorem \ref{non degenerate backwards scattering hp}]
    Let $(\uppsi,\uppsi')\in {}^{(\ell)}\mathcal{E}_{\Sigma}^{N,n}$, $\psi$ the resulting solution to \eqref{fixed ell mode wave equation 2}. We use Propositions \ref{prop: cosine isomorphism} and \ref{prop: sine isomorphism} to find  $\psi_{\hp}^{\cos}\in\cosspacehpn$, $\psi_{\hp}^{\sin}\in\sinespacehpn$ such that $\psi_{\hp}=\psi_{\hp}^{\cos}+\psi_{\hp}^{\sin}$. Let $({}^{*}{\psi}^{\cos}_{\hp},{}^{*}{\psi}^{\sin}_{\hp})$ be the unique element of $\cosspacehpn\oplus\sinespacehpn$ such that $({}^{*}{\psi}^{\cos}_{\hp},{}^{*}{\psi}^{\sin}_{\hp})\perp \kspacehp$ with ${}^{*}{\psi}^{\cos}_{\hp}+{}^{*}{\psi}^{\sin}_{\hp}=\psi_{\hp}$. We now apply Proposition \ref{prop: cosine isomorphism} to find time-symmetric $\tilde{\psi}^{\cos}$ solving \eqref{fixed ell mode wave equation 2} and attaining ${}^{*}{\psi}^{\cos}_{\hp}$ as its radiation field at $\hp$. Similarly, we apply Proposition \ref{prop: sine isomorphism} to find time-antisymmetric ${}^{*}{\psi}^{\sin}$ solving \eqref{fixed ell mode wave equation 2} and attaining ${}^{*}{\psi}^{\sin}_{\hp}$ as its radiation field at $\hp$. Then the solution ${}^{*}{\psi}={}^{*}{\psi}^{\cos}+{}^{*}{\psi}^{\sin}$ to \eqref{fixed ell mode wave equation 2} attains $\psi_{\hp}$ as its radiation field at $\hp$. Propositions \ref{prop: cosine isomorphism} and \ref{prop: sine isomorphism} and the Pythagorean theorem imply
    \begin{align}\label{i used pythagorean}
    \begin{split}
        \left\|\psi_{\hp}\right\|^2_{\spacehpn}&\lesssim \left\|{}^{*}{\psi}^{\cos}_{\hp}\right\|^2_{\cosspacehpn}+\left\|{}^{*}{\psi}^{\sin}_{\hp}\right\|^2_{\sinespacehpn}\\&\lesssim \left\|{\psi}^{\cos}_{\hp}\right\|^2_{\cosspacehpn}+\left\|{\psi}^{\sin}_{\hp}\right\|^2_{\sinespacehpn}\\&\lesssim \left\|(\uppsi,\uppsi')\right\|_{\mathcal{E}_{\Sigma}^{N,n}}^2.
    \end{split}
    \end{align}
    Therefore, $\psi_{\hp}\in \spacehpn$.
    
    Let $\tilde{\psi}=\psi-{}^{*}\psi$. Then $\tilde{\psi}$ has no radiation at $\hp$, and Theorem \ref{exponential decay theorem'} and Corollary \ref{magic corollary} imply that 
    $\tilde{\psi}_{\Scrip}\in \spacescriexpon{\ell}$, and we have by $\ref{magic corollary}$ and Lemma \ref{clunky lemma}
    \begin{align}
        \begin{split}
            \|\tilde{\psi}_{\Scrip}\|^2_{\spacescriexpon{\ell}}
           \lesssim&\; \sum_{j=0}^{k}\int_{-\infty}^{\infty}du\,|\partial_u^{j+1}\psi_{\Scrip}|^2+|\partial_u^{j+1}{}^{*}\psi_{\Scrip}|^2\\&+\sum_{j=0}^{k}\int_{-\infty}^{\infty}du\,(1+e^{(2n+1)u})\left[|\partial_u^{j+1}\psi_{\hm}|^2+|\partial_u^{j+1}{}^{*}\psi_{\hm}|^2\right].
        \end{split}
    \end{align}
    Lemma \ref{clunky lemma} and Proposition \ref{prop: standard local energy estimates future backwards} gives
    \begin{align}
        \sum_{j=0}^{k}\int_{-\infty}^{\infty}du\,e^u(\partial_u(e^u\partial_u)^j{\psi}_{\hm})^2\lesssim \|\uppsi\|_{H^{k+1}((1,r_0),\Omega^{-1}dr)}^2+\|\uppsi'\|_{H^{k}((1,r_0),\Omega^{-1}dr)}^2.
    \end{align}
    The identity \eqref{mirror magic} applies to ${}^*\psi$, and Proposition \ref{prop: standard local energy estimates future backwards} and the estimate \eqref{i used pythagorean} lead to
    \begin{align}
    \begin{split}
         \sum_{j=0}^{k}\int_{-\infty}^{\infty}du\,e^u(\partial_u(e^u\partial_u)^j{{}^{*}\psi}_{\hm})^2\lesssim &\; \sum_{j=0}^{k}\int_{-\infty}^{\infty}du\,e^u(\partial_u(e^u\partial_u)^j{}^{*}\psi^{\cos}_{\hm})^2+e^u(\partial_u(e^u\partial_u)^j{}^{*}\psi^{\sin}_{\hm})^2\\&\lesssim \|(\uppsi,\uppsi')\|_{\mathcal{E}_{\Sigma}^{N,n}}^2.
    \end{split}
    \end{align}
    Note now that by \eqref{i used pythagorean} we have
    \begin{align}
        \begin{split}
            \sum_{j=0}^{k}\int_{-\infty}^{\infty}du\,|\partial_u^{j+1}{}^{*}\psi_{\Scrip}|^2\lesssim&\;\sum_{j=0}^{k}\int_{-\infty}^{\infty}du\,|\partial_u^{j+1}{}^{*}\psi_{\Scrip}^{\cos}|^2+|\partial_u^{j+1}{}^{*}\psi_{\Scrip}^{\sin}|^2\\\lesssim &\; \left\|{}^{*}\psi^{\cos}_{\hp}\right\|^2_{\cosspacehpn}+\left\|{}^{*}\psi^{\sin}_{\hp}\right\|^2_{\sinespacehpn}\\= &\;\|{}^{*}\psi_{\hp}\|^2_{\spacehpn}\\\lesssim &\;\|(\uppsi,\uppsi')\|_{\mathcal{E}_{\Sigma}^{N,n}}^2.
        \end{split}
    \end{align}
    We conclude
    \begin{align}\label{continuity in spacescriexpn for ell}
        \|\tilde{\psi}_{\Scrip}\|^2_{\spacescriexpn{\ell}}\lesssim \|(\uppsi,\uppsi')\|_{\mathcal{E}_{\Sigma}^{N,n}}^2.
    \end{align}
    We then define
    \begin{align}
        \mathscr{F}^{N,n}_{+}(\uppsi,\uppsi')=(\psi_{\hp},\tilde{\psi}_{\Scrip}),
    \end{align}
    noting that \eqref{i used pythagorean} and \eqref{continuity in spacescriexpn for ell} imply that $\mathscr{F}^{N,n}_{+}$ is continuous.

    We now establish $\mathscr{B}^{N,n}_{+}$. Given data $(\psi_{\hp},\tilde{\psi}_{\Scrip})\in \spacehpn\oplus \spacescriexpon{\ell}$, We define ${}^{*}\psi^{\cos}$, ${}^{*}\psi^{\sin}$ to be the solutions to \eqref{fixed ell mode wave equation 2} arising from $\psi_{\hp}^{\cos}$, $\psi_{\hp}^{\sin}$ via Propositions \ref{prop: cosine isomorphism} and \ref{prop: sine isomorphism} respectively, and we define ${}^{*}\psi:={}^{*}\psi^{\cos}+{}^{*}\psi^{\sin}$. Taking 
    \begin{align}
        {}^{*}\uppsi:={}^{*}\psi^{\cos}|_{t=0}, \qquad {}^{*}\uppsi':={}^{*}\psi^{\sin}|_{t=0}
    \end{align}
    Propositions \ref{prop: cosine isomorphism} and \ref{prop: sine isomorphism} imply
    \begin{align}\label{backwards continuity ell 1}
        \|({}^{*}\uppsi,{}^{*}\uppsi')\|^2_{\mathcal{E}^{N,n}_{\Sigma}}\lesssim \|\psi_{\hp}\|^2_{\spacehpn}.
    \end{align}
    Define $\tilde{\psi}$ to be the solution to \eqref{fixed ell mode wave equation 2} arising via $\mathscr{B}^{T,n}_+(0,\tilde{\psi}_{\Scrip})$. 
    Define
    \begin{align}
       \tilde{\uppsi}:=\tilde{\psi}|_{t=0}, \qquad \tilde{\uppsi}':=\tilde{\psi}|_{t=0}.
    \end{align}
    Then since $\tilde{\psi}_{\Scrip}\in \spacescriexpon{\ell}$, Lemma \ref{rudimentary boundedness of T} and Theorem \ref{Paley Wiener L2} imply
    \begin{align}
        \sum_{j=1}^{n+1}\int_{-\infty}^{\infty}du\,(1+e^{(2n+1)v})(\partial_u^j{\tilde{\psi}}_{\hm})^2\lesssim \|\tilde{\psi}_{\Scrip}\|_{\spacescriexpon{\ell}}^2. 
    \end{align}
     Propositions \ref{prop: standard local energy estimates future backwards}, \ref{prop: standard local energy estimates past forwards}, \ref{the thing that nobody noticed commuted}, and \ref{the thing that nobody noticed commuted reversed in time}, and the continuity of $\mathscr{B}^{T,n}_+$ give
    \begin{align}\label{backwards continuity ell 2}
        \|\tilde{\uppsi}\|^2_{H^{k+1}((1,\infty),\Omega^{-1}dr)}+\|\tilde{\uppsi}'\|^2_{H^{k}((1,\infty),\Omega^{-1}dr)}\lesssim \big\|\tilde{\psi}_{\Scrip}\big\|^2_{\spacescriexpn{\ell}}.
    \end{align}
    We then define
    \begin{align}
        \mathscr{B}^{N,n}_{+}(\psi_{\hp},\tilde{\psi}_{\Scrip})=({}^{*}\uppsi+\tilde{\uppsi},{}^{*}\uppsi'+\tilde{\uppsi}'),
    \end{align}
    and we have that $\mathscr{B}^{N,n}_{+}$ is continuous by \eqref{backwards continuity ell 1} and \eqref{backwards continuity ell 2}.

    Finally, since any element of $\cosspacehpn\oplus\sinespacehpn$ has a unique projection onto $\kspacehp$, and since $\cosmapforfutn$, $\cosmapforbacn$ are inverses to one another and $\sinmapforfutn$, $\sinmapforbacn$ are also inverses to one another, it follows that $\mathscr{F}^{N,n}_{+}$ and $\mathscr{B}^{N,n}_{+}$ are inverse to one another.

    An analogous argument applies to show the isomorphism $\mathcal{E}^{N,n}_{\overline{\Sigma}}\simeq \spacehpbarn\oplus\spacescriexpn{\ell}$.
\end{proof}
\subsubsection{The scattering map based at $\Scrip$}\label{sec: scattering maps based at infinity}
\begin{defin}
    We define
    \begin{align}
        \fancyapcos[\hat{f}](\omega):=\frac{1}{T(-\omega)}[\hat{f}(-\omega)+R(-\omega)\hat{f}(\omega)],
    \end{align}
    \begin{align}
        \fancyapsin[\hat{f}](\omega):=\frac{1}{T(-\omega)}[\hat{f}(-\omega)-R(-\omega)\hat{f}(\omega)].
    \end{align}
\end{defin}
We now prove Theorem \ref{non degenerate backwards scattering Scrip}:
\begin{proposition}\label{prop: cosine isomorphism scrip}
     There exist Hilbert space isomorphisms
    \begin{align}
        {}^{(\ell)}\mathscr{F}_{+,\Scrip}^{N,\cos,n}:\cosspacebarn\longrightarrow \cosspacescripn,\qquad {}^{(\ell)}\mathscr{B}_{+,\Scrip}^{N,\cos,n}: \cosspacescripn\longrightarrow\cosspacebarn,
    \end{align}
    such that 
    \begin{align}
         {}^{(\ell)}\mathscr{F}_{+,\Scrip}^{N,\cos,n}\circ {}^{(\ell)}\mathscr{B}_{+,\Scrip}^{N,\cos,n}=\mathrm{Id}_{\cosspacescripn},\qquad {}^{(\ell)}\mathscr{B}_{+,\Scrip}^{N,\cos,n}\circ {}^{(\ell)}\mathscr{F}_{+,\Scrip}^{N,\cos,n}=\mathrm{Id}_{\cosspacebarn}.
    \end{align}
\end{proposition}
\begin{proof}
    Assume that $\uppsi\in \cosspacebarn$. Then the maps $\mathscr{F}^{T,n}_{\pm}$ applied to Cauchy data $(\uppsi,0)$ on $\overline{\Sigma}$ give a solution $\psi$ to \eqref{fixed ell mode wave equation 2} which induces radiation fields $\psi_{\hp}$, $\psi_{\Scrip}$, $\psi_{\hm}$, $\psi_{\Scrim}$ in $\mathcal{E}^T_{\hp}$, $\mathcal{E}^T_{\Scrip}$, $\mathcal{E}^T_{\hm}$, $\mathcal{E}^T_{\Scrim}$ respectively. Moreover, since $\psi$ is time-symmetric, we have
    \begin{align}
        a_{\Scrim}(\omega)=a_{\Scrip}(-\omega),\qquad a_{\hm}(\omega)=a_{\hp}(-\omega).
    \end{align}
    Using \eqref{relation microlocal scrim to hp scrip}, we find
    \begin{align}
        a_{\hp}(\omega)=\fancyapcos[a_{\Scrip}].
    \end{align}
    Since $(\uppsi,0)\in \mathcal{E}_{\overline{\Sigma}}^{N,n}$, we have by Proposition \ref{prop: standard local energy estimates future backwards} that $e^{-\frac{v}{2}}(\partial_v(e^{-v}\partial_v)^k)\partial_v\psi_{\hp}\in L^2(\mathbb{R})$ for $k=1,\dots,n$. Theorem \ref{Paley Wiener L2} implies $\left[\prod_{q=0}^{k}\omega-iq\right] \fancyapcos[a_{\Scrip}]\in PW\left(0,\frac{2k+1}{2}\right)$, and therefore $\psi_{\hp}\in\cosspacescripn$. Proposition \ref{prop: standard local energy estimates future backwards} and the continuity of $\mathscr{F}^{T,n}$ imply
    \begin{align}
        \|\psi_{\Scrip}\|_{\cosspacescripn}\lesssim \|\uppsi\|_{\cosspacebarn}.
    \end{align}

    Conversely, let $\psi_{\Scrip}\in\cosspacescripn$. Let $\psi_{\hp}=\int^v\mathcal{F}^{-1}\left[\fancyapcos[a_{\Scrip}]\right]$. Then from \ref{def of E cos scrip}, $\psi_{\hp}\in\mathcal{E}^{T,n}_{{\hp}}$, and $\psi_{\Scrip}\in \mathcal{E}^{T,n}_{\Scrip}$, and thus there exists a unique solution $\psi$ realising $\psi_{\hp}$, $\psi_{\Scrip}$ as radiation fields on $\overline{\hp}$, $\Scrip$ respectively. The relations \eqref{relation microlocal scrip to hm scrim}, \eqref{relation microlocal scrim to hp scrip} imply $a_{\Scrip}(\omega)=a_{\Scrim}(-\omega)$, $a_{\hp}(\omega)=a_{\hm}(-\omega)$, and therefore $\psi$ is time symmetric by Lemma \ref{how to recognize time symmetric solutions}.Since $\psi_{\Scrip}\in \cosspacescripn$, \eqref{def of E cos scrip} and Lemma \ref{dumb lemma} imply that $e^{-\frac{v}{2}}(\partial_v(e^{-v}\partial_v)^k)\partial_v\psi_{\hp}\in L^2(\mathbb{R})$ for $k=0,\dots,n$, and time symmetry implies $e^{\frac{u}{2}}(\partial_u(e^{u}\partial_u)^k)\partial_v\psi_{\hm}\in L^2(\mathbb{R})$ for $k=1,\dots,n$. Propositions \ref{prop: standard local energy estimates future backwards} and \ref{the thing that nobody noticed commuted} imply 
    \begin{align}
        \|\psi|_{t=0}\|^2_{\cosspacebarn}\lesssim \|\psi_{\Scrip}\|^2_{\cosspacescripn}.
    \end{align}
    The result follows by taking
    \begin{align}
        \cosmapforbacscripn[\psi_{\Scrip}]:=\psi|_{t=0}.
    \end{align}
\end{proof}
In an entirely analogous fashion, we have
\begin{proposition}\label{prop: sine isomorphism scrip}
     There exist Hilbert space isomorphisms
    \begin{align}
        {}^{(\ell)}\mathscr{F}_{+,\Scrip}^{N,\sin,n}:\sinespacebarn\longrightarrow {}^{(\ell)}\mathcal{E}_{\Scrip}^{N,\sin,n},\qquad {}^{(\ell)}\mathscr{B}_{+,\Scrip}^{N,\sin,n}: {}^{(\ell)}\mathcal{E}_{\Scrip}^{N,\sin,n}\longrightarrow\sinespacebarn,
    \end{align}
    such that 
    \begin{align}
         {}^{(\ell)}\mathscr{F}_{+,\Scrip}^{N,\sin,n}\circ {}^{(\ell)}\mathscr{B}_{+,\Scrip}^{N,\sin,n}=\mathrm{Id}_{\sinespacescripn},\qquad {}^{(\ell)}\mathscr{B}_{+,\Scrip}^{N,\sin,n}\circ {}^{(\ell)}\mathscr{F}_{+,\Scrip}^{N,\sin,n}=\mathrm{Id}_{\sinespacebarn}.
    \end{align}
\end{proposition}
\begin{proposition}\label{prop: cosine isomorphism scrip without B}
     There exist Hilbert space isomorphisms
    \begin{align}
        {}^{(\ell)}\mathscr{F}_{+,\Scrip}^{N,\cos,n}:\cosspacen\longrightarrow \cosspacescripon,\qquad {}^{(\ell)}\mathscr{B}_{+,\Scrip}^{N,\cos,n}: \cosspacescripo\longrightarrow\cosspacen,
    \end{align}
    \begin{align}
        {}^{(\ell)}\mathscr{F}_{+,\Scrip}^{N,\sin,n}:\sinespacen\longrightarrow \sinespacescripon,\qquad {}^{(\ell)}\mathscr{B}_{+,\Scrip}^{N,\sin,n}: \sinespacescripon\longrightarrow\sinespacen,
    \end{align}
    such that 
    \begin{align}
         {}^{(\ell)}\mathscr{F}_{+,\Scrip}^{N,\cos,n}\circ {}^{(\ell)}\mathscr{B}_{+,\Scrip}^{N,\cos,n}=\mathrm{Id}_{\cosspacescripon},\qquad {}^{(\ell)}\mathscr{B}_{+,\Scrip}^{N,\cos,n}\circ {}^{(\ell)}\mathscr{F}_{+,\Scrip}^{N,\cos,n}=\mathrm{Id}_{\cosspacen},
    \end{align}
    \begin{align}
         {}^{(\ell)}\mathscr{F}_{+,\Scrip}^{N,\sin,n}\circ {}^{(\ell)}\mathscr{B}_{+,\Scrip}^{N,\sin,n}=\mathrm{Id}_{\sinespacescripon},\qquad {}^{(\ell)}\mathscr{B}_{+,\Scrip}^{N,\sin,n}\circ {}^{(\ell)}\mathscr{F}_{+,\Scrip}^{N,\sin,n}=\mathrm{Id}_{\sinespacen}.
    \end{align}
\end{proposition}
\begin{proof}
    The proof is entirely analogous to the proof of Proposition \ref{prop: cosine isomorphism scrip}, and we omit the details.
\end{proof}
\begin{corollary}
    A radiation field $\psi_{\Scrip}\in\mathcal{E}^{T,n}_{\Scrip}$ arises from initial data in $\mathcal{E}_{\overline{\Sigma}}^{N,n}$ if and only if $\psi_{\Scrip}$ is in the algebraic sum
    \begin{align}
        \cosspacescripn+\sinespacescripn.
    \end{align}
    A radiation field $\psi_{\Scrip}\in\mathcal{E}^{T,n}_{\Scrip}$ arises from initial data in $\mathcal{E}_{{\Sigma}}^{N,n}$ if and only if $\psi_{\Scrip}$ is in the algebraic sum
    \begin{align}
        \cosspacescripon+\sinespacescrip.
    \end{align}
\end{corollary}

\appendix
\numberwithin{equation}{section}
\section{Asymptotics of modified Bessel functions}\label{Appendix A}

In this section we study the high order approximation of the modified Bessel functions $I_{\beta+i\alpha}$, $K_{\beta+i\alpha}$ via Airy's equation and prove Proposition \ref{prop: airy approximation of modified bessel for large frequency}.
\begin{proof}[Proof of Proposition \ref{prop: airy approximation of modified bessel for large frequency}]
Let $\alpha\in\mathbb{R}_{>0}$. Recall that  $I_{\beta+i\alpha}(x)$, $K_{\beta+i\alpha}(x)$ satisfy the equation
\begin{align}\label{modified bessel appendix jj}
    u''+\frac{1}{x}u'-\left(\frac{(\beta+i\alpha)^2}{x^2}+1\right)u=0.
\end{align}
Take $w(x)=\sqrt{x}u(\alpha x)$. Then $w$ satisfies
\begin{align}
    w''=\left[\alpha^2\left(\frac{x^2-1}{x^2}\right)+\alpha\frac{2\beta i}{x^2}+\frac{4\beta^2-1}{4x^2}\right]w.
\end{align}
The equation above is of Liouville--Green form, with turning points at $x=\pm1$. It is possible to obtain a large $\alpha$ approximation that is uniform in $x$ outside of a neighbourhood of $x=0$ via the use of the Liouville--Green approximation by Airy's functions to take into account the turning point at $x=1$, following Section 11.1 of \cite{Olver}. Take 
\begin{align}
    f_0(x):=\frac{x^2-1}{x^2},\qquad f_1:=\frac{2\beta i}{x^2}, \qquad g:=\frac{4\beta^2-1}{4x^2}.
\end{align}
For $x\geq1$, define $\zeta$ by
\begin{align}\label{zeta for x greater than 1}
    \frac{2}{3}\zeta^{\frac{3}{2}}=\int_{1}^{x}dx\,{f_0}^{\frac{1}{2}}=\sqrt{x^2-1}-\arcsec x.
\end{align}
For $x<1$, we take
\begin{align}\label{zeta for x less than 1}
\begin{split}
    \frac{2}{3}(-\zeta)^{\frac{3}{2}}=\int_{x}^{1}dx\,{(-f_0)}^{\frac{1}{2}}&=\arcsech x-\sqrt{1-x^2}\\&=\log\left(\frac{1+\sqrt{1-x^2}}{x}\right)-\sqrt{1-x^2}.
\end{split}
\end{align}
Taking
\begin{align}
    \hat{f}:=\frac{f_0}{\zeta},\qquad\qquad \phi:=\frac{\zeta f_1}{f_0}=2i\beta \frac{\zeta}{x^2-1}. 
\end{align}
We define
\begin{align}
    \Phi:=\frac{1}{2\zeta^{\frac{1}{2}}}\int_{0}^{\zeta}dv\,\frac{\phi(v)}{v^{\frac{1}{2}}}
\end{align}
for $\zeta\geq0$. For $\zeta<0$ we take
\begin{align}
    \Phi=\frac{1}{2(-\zeta)^{\frac{1}{2}}}\int_{\zeta}^{0}dv\frac{\phi(v)}{(-v)^{\frac{1}{2}}}.
\end{align}
For $\zeta\geq0$, we compute
\begin{align}
\begin{split}
    \int_{0}^{\zeta} dv\,\frac{\phi(v)}{v^{\frac{1}{2}}}=2i\beta \int_{0}^{\zeta}dv\frac{\sqrt{v}}{x(v)^2-1}=2i\beta \int_{1}^{x}dy\frac{1}{y\sqrt{y^2-1}}=2i\beta\arcsec x.
\end{split}
\end{align}
For $\zeta<0$,
\begin{align}\label{x equation for Phi times sqrt zeta}
\begin{split}
    \int_{\zeta}^0 dv\,\frac{\phi(v)}{(-v)^{\frac{1}{2}}}=2i\beta \int_{\zeta}^{0}dv\frac{\sqrt{-v}}{1-x(v)^2}=2i\beta\int_{x}^{1}dy\frac{1}{y\sqrt{1-y^2}}&=2i\beta\arcsech x\\&=2i\beta\log\frac{1+\sqrt{1-x^2}}{x}.
\end{split}
\end{align}
Note that $\phi$, $\Phi$ are analytic at $\zeta=0$.

Taking $W(\zeta)$ such that $W|_{x}=\hat{f}^{\frac{1}{4}}(x)w(x)$, we have the equation
\begin{align}\label{LG equation}
    \frac{d^2}{d\zeta^2}W=\{\alpha^2\zeta+\alpha\phi(\zeta)+h(\zeta)\}W,
\end{align}
where 
\begin{align}\label{error functional 1 original Airy}
    h(\zeta)=\frac{g}{\hat{f}}-\frac{1}{\hat{f}^{\frac{3}{4}}}\frac{d^2}{dx^2}\frac{1}{\hat{f}^{\frac{1}{4}}}=\frac{g}{\hat{f}}-\frac{5}{16\hat{f}^3}(\hat{f}')^2+\frac{1}{4\hat{f}^2}\hat{f}''.
\end{align}
Note that $\Phi$ is purely imaginary, while $\zeta$ is real by construction. Since $i\Phi>0$ for $\zeta<0$, we propose as approximate solutions to \eqref{LG equation} a combination 
\begin{align}
    c_+W_0^{(+)}+c_-W_0^{(-)},
\end{align}
where for
\begin{align}
    \hat{\zeta}:=\alpha^{\frac{2}{3}}\zeta+\alpha^{-\frac{1}{3}}\Phi    
\end{align}
we have
\begin{align}
    W_0^{(+)}=\frac{\alpha^{\frac{1}{2}}}{(\alpha+\Phi')^{\frac{1}{2}}}\mathrm{Ai}_+(\hat{\zeta}),\qquad W_0^{(-)}=\frac{\alpha^{\frac{1}{2}}}{(\alpha+\Phi')^{\frac{1}{2}}}\mathrm{Ai}_-(\hat{\zeta}),
\end{align}
where we denote
\begin{align}
    \mathrm{Ai}_{\pm}(\hat{\zeta}):=\mathrm{Ai}(\hat{\zeta}e^{\mp\frac{2\pi i}{3}}).
\end{align}
Each of $W_0^{(\pm)}$ satisfies the equation
\begin{align}\label{modified Airy's a la Olver}
    \frac{d^2}{d\zeta^2}W_0^{(\pm)}=\{\alpha^2+\alpha\phi+h_0(\zeta)\}W_0^{(\pm)},
\end{align}
where (see \cite{Olver} Chapter 11, equation 11.08)
\begin{align}
    h_0=\zeta\Phi'^2+2\Phi\Phi'+\frac{\Phi\Phi'^2}{\alpha}+\frac{3\Phi''^2-2\Phi'\Phi'''-2\alpha\Phi'''}{4(\alpha+\Phi')^2}.
\end{align}
So far we have
\begin{align}
    I_{\beta+i\alpha}(\alpha x)=\left(\frac{4\zeta}{x^2-1}\right)^{\frac{1}{4}}\left[c_+\left(\frac{\alpha^{\frac{1}{2}}}{(\alpha+\Phi')^{\frac{1}{2}}}\mathrm{Ai}(\hat{\zeta}e^{\frac{-2\pi i}{3}})+\epsilon^{(+)}\right)+c_-\left(\frac{\alpha^{\frac{1}{2}}}{(\alpha+\Phi')^{\frac{1}{2}}}\mathrm{Ai}(\hat{\zeta}e^{\frac{2\pi i}{3}})+\epsilon^{(-)}\right)\right].
\end{align}
We choose $c_+$ and $c_-$ by matching the asymptotics of $W_0^{(\pm)}$ for $\zeta\longrightarrow-\infty$ with those of $I_{\beta+i\alpha}$ at small $x$. 
Using equation 10.25.2 of \cite{dlmf}, we have for small $x$
\begin{align}\label{series expansion of modified Bessel for small x}
    I_{\pm(\beta+i\alpha)}(\alpha x)=\left(\frac{\alpha x}{2}\right)^{\pm(\beta+i\alpha)}\left(1+O\left(|\alpha|x^2\right)\right).
\end{align}
We note that as $\zeta\longrightarrow-\infty$ (see Appendix A in \cite{Melrose}),
\begin{align}\label{negative large z asymptotic of Ai pm}
    \mathrm{Ai}(\hat{\zeta}e^{\pm\frac{2\pi i}{3}})\sim \frac{1}{4\pi^{\frac{3}{2}}{(e^{\pm\frac{2\pi i}{3}}\hat{\zeta})^{\frac{1}{4}}}}\exp\left\{\pm i\frac{2}{3}(-\hat{\zeta})^{\frac{3}{2}}\right\},
\end{align}
and we have
\begin{align}
    \frac{2}{3}(-\hat{\zeta})^{\frac{3}{2}}=\frac{2}{3}\left(\alpha^{\frac{2}{3}}|\zeta|-\alpha^{-\frac{1}{3}}\Phi\right)^{\frac{3}{2}}\sim \frac{2}{3}\alpha |\zeta|^{\frac{3}{2}}-\sqrt{|\zeta|}\Phi.
\end{align}
Therefore,
\begin{align}
    \exp\left\{\pm i\frac{2}{3}(-\hat{\zeta})^{\frac{3}{2}}\right\}\sim \left(\frac{x}{2}\right)^{\mp (\beta+i\alpha)}e^{\mp i\alpha},
\end{align}
and we must choose 
\begin{align}
    c_-=0. 
\end{align}
To find $c_+$, we match the asymptotic expansion of $\sqrt{x}I_{\beta+i\alpha}(\alpha x)$ as $x\rightarrow0$ with the asymptotics of $(\alpha+\Phi')^{-\frac{1}{2}}\mathrm{Ai}_{+}(\hat{\zeta})$ as $\zeta\rightarrow-\infty$. Using \eqref{zeta for x less than 1} and \eqref{x equation for Phi times sqrt zeta} we find
\begin{align}
\begin{split}
    -i\frac{2}{3}(-\hat{\zeta})^{\frac{3}{2}}&\sim-\frac{2}{3}i\alpha |\zeta|^{\frac{3}{2}}+\beta\log\left(\frac{x}{1+\sqrt{1-x^2}}\right)\\&=i\alpha\sqrt{1-x^2}+(\beta+i\alpha)\log\frac{x}{1+\sqrt{1-x^2}}.
\end{split}
\end{align}
Therefore, for small $x$ we have
\begin{align}
    \exp\left\{-i\frac{2}{3}(-\hat{\zeta})^{\frac{3}{2}}\right\}\sim \left(\frac{x}{2}\right)^{\beta+i\alpha}.
\end{align}
We then find that 
\begin{align}
    c_+=\frac{2\sqrt{2}\pi^{\frac{3}{2}}e^{\frac{\pi}{12}i}e^{-i\alpha}\alpha^{\beta+i\alpha}\left(1-\frac{2i\beta}{3\alpha}\right)^{\frac{3}{4}}\alpha^{\frac{1}{6}}}{\Gamma(\beta+i\alpha)}.
\end{align}
To leading order in $\alpha$ for $\alpha\gg1$, the Stirling approximation of $\Gamma(\beta+i\alpha)$ leads to
\begin{align}
    c_+=2\pi\alpha^{-\frac{1}{3}}e^{\frac{\alpha\pi}{2}}e^{-\frac{i\pi}{2}\beta-\frac{i\pi}{6}}(1+o(1)).
\end{align}

To find an approximation of $K_{\beta+i\alpha}$, we note that $K_{\beta+i\alpha}$ is recessive for $x>1$, and therefore we propose as an approximation of $\sqrt{x}K_{\beta+i\alpha}(\alpha x)$ the function
\begin{align}
    W_0=\frac{\alpha^{\frac{1}{2}}}{(\alpha+\Phi')^{\frac{1}{2}}}\mathrm{Ai}(\hat{\zeta}),
\end{align}
noting the connection formula, valid for all $z\in\mathbb{C}$,
\begin{align}
    \mathrm{Ai}(z)+e^{-\frac{2\pi i}{3}}\mathrm{Ai}_{+}(z)+e^{\frac{2\pi i}{3}}\mathrm{Ai}_{-}(z)=0.
\end{align}
At large $x$, $K_{\beta+i\alpha}(\alpha x)$ is approximated by
\begin{align}
    K_{\beta+i\alpha}(\alpha x)=\sqrt{\frac{\pi}{2\alpha x}}e^{-\alpha x}.
\end{align}
The large $\hat{\zeta}$ behaviour when $\zeta$ has positive real and imaginary parts is
\begin{align}
    \mathrm{Ai}(\hat{\zeta})\sim \frac{1}{2\sqrt{\pi}}\hat{\zeta}^{-\frac{1}{4}}e^{-\frac{2}{3}\hat{\zeta}^{\frac{3}{2}}}.
\end{align}
By matching the large $x$ asymptotics, we obtain
\begin{align}
     K_{\beta+i\alpha}(\alpha x)=c \left(\frac{4\zeta}{x^2-1}\right)^{\frac{1}{4}}\left[\frac{\alpha^{\frac{1}{2}}}{(\alpha+\Phi')^{\frac{1}{2}}}\mathrm{Ai}(\hat{\zeta})+\epsilon\right],
\end{align}
where
\begin{align}
    c={\pi}{\alpha^{-\frac{1}{3}}} e^{-\frac{\alpha \pi}{2}+\frac{i\pi }{2}\beta}.
\end{align}
We now estimate the error terms. Variation of parameters on \eqref{modified Airy's a la Olver} gives us
\begin{align}
    \epsilon^{(+)}=\int_{\zeta_{-\infty}}^{\zeta}dv\,\frac{W_0^{(+)}(\zeta)W_0(v)-W_0^{(+)}(v)W_0(\zeta)}{\mathfrak{W}(W_0^{(+)},W_0)}(h-h_0)(W_0^{(+)}+\epsilon^{(+)}),
\end{align}
\begin{align}\label{approximation of K in terms of Airys}
    \epsilon=\int_{\zeta}^{\infty}dv\,\frac{W_0^{(+)}(\zeta)W_0(v)-W_0^{(+)}(v)W_0(\zeta)}{\mathfrak{W}(W_0^{(+)},W_0)}(h-h_0)(W_0+\epsilon).
\end{align}
We will choose $\zeta_{-\infty}$ so that $\epsilon^{(+)}$ is uniformly small for $\zeta\geq \zeta_{-\infty}$, and for $\zeta\leq 2\zeta_{-\infty}$ we can estimate $I_{\beta+i\alpha}$ using its Maclaurin series expansion with an error that remains small as $\alpha$ grows.
We apply Theorem \ref{Olver's error approx theorem} to the above equations, and follow an identical argument to that of Theorem 9.1 of Chapter 11 in    \cite{Olver}, which we outline here, and refer the reader to Sections 8.1--8.3, 9.1, and 11.1--11.4, Chapter 11 of \cite{Olver} for the full details. First, note that
\begin{align}
    \mathfrak{W}(W_0^{(+)},W_0)=e^{\frac{i\pi}{6}}\frac{\alpha^{\frac{2}{3}}}{2\pi} 
\end{align}
Consider
\begin{align}
    W_0^{(+)}(\zeta)W_0(v)-W_0^{(+)}(v)W_0(\zeta).
\end{align}
Taking
\begin{align}
    \mathrm{K}(\zeta,v)=[W_0^{(+)}(\zeta)W_0(v)-W_0^{(+)}(v)W_0(\zeta)].
\end{align}
Then we use
\begin{align}
    |\mathrm{Ai}_{+}|= E_+^{-1}M_{-}\cos\theta_{-},\qquad |\mathrm{Ai}|= E_0^{-1}M_{-}\sin\theta_{-},
\end{align}
and the fact that $E_0^{-1}=E_+$ to find (using the monotonicity of $E_+$ and that $\zeta\geq v$)
\begin{align}
    |\mathrm{K}(\zeta,v)|\leq E_+^{-1}(\hat{\zeta})E_+(\hat{v})M_-(\hat{\zeta})M_-(\hat{v}).
\end{align}
Similarly,
\begin{align}
     |\partial_{\zeta}\mathrm{K}(\zeta,v)|\leq E_+(\hat{\zeta})E_+^{-1}(\hat{v})M_-(\hat{\zeta})M_-(\hat{v}).
\end{align}
An analogous bound holds for $|\partial^2_{\zeta}\mathrm{K}|$ since $\mathrm{K}$ satisfies \eqref{modified Airy's a la Olver}. Therefore, taking in Theorem \ref{Olver's error approx theorem}
\begin{align}
    \upphi=\uppsi_0=\frac{2\pi e^{-\frac{i\pi}{6}}}{\alpha^{\frac{2}{3}}}(h-h_0),\qquad J=W_{0}^{(+)},\qquad \varepsilon=\epsilon^{(+)},
\end{align}
\begin{align}
    P_0(\zeta)=E_+^{-1}(\hat{\zeta})M_-(\hat{\zeta}),\qquad Q(v)=E_+(\hat{v})M_-(\hat{v}),\qquad P_1(\zeta)=E_+^{-1}(\hat{\zeta})N_-(\hat{\zeta}),
\end{align}
we find that $\kappa=\kappa_0<\infty$, since $v_0$, $v_+$ are both finite. Moreover, 
\begin{align}
    \left|\frac{\epsilon^{(+)}}{M_-}\right|,\; \left|\frac{\epsilon^{(+)}{}'}{\alpha^{\frac{2}{3}}N_-}\right|\lesssim \frac{1}{E_+}\left[\exp\left(\int_{\zeta_{-\infty}}^{\zeta}dv\frac{2\pi}{\alpha^{\frac{2}{3}}}|h-h_0|\right)-1\right].
\end{align}
Therefore, we will have an error estimate on $\epsilon^{(+)}$ if we estimate 
\begin{align}
    \int_{\zeta_{-\infty}}^{\zeta}dv|h-h_0|.
\end{align}
We now estimate $h_0$. Note that
\begin{align}
    \Phi'=\frac{1}{2\zeta}(\phi-\Phi),\qquad \Phi''=-\frac{1}{4\zeta^2}(\phi-\Phi)+\frac{1}{2\zeta}\phi',
\end{align}
\begin{align}
    \Phi'''=\frac{3}{8\zeta^3}(\phi-\Phi)-\frac{3}{4\zeta^2}\phi'+\frac{1}{2\zeta}\phi''.
\end{align}
For $x<1$,
\begin{align}
    \phi'=\frac{2i}{x^2-1}-6i\frac{x^2}{(1-x^2)^{\frac{5}{2}}}[-\log x+\log(1+\sqrt{1-x^2})-\sqrt{1-x^2}],
\end{align}
\begin{align}
\begin{split}
    \phi''=\left(\frac{-\zeta}{1-x^2}\right)^{\frac{1}{2}}{x}^{\frac{3}{2}}\Bigg[&4i\frac{1}{(x^2-1)^2}+12i\left(-\log x-\sqrt{1-x^2}+\log\left(1+\sqrt{1-x^2}\right)\right)\\&+6ix\left(-1-\frac{x^2}{\sqrt{1-x^2}}+\frac{x}{\sqrt{1-x^2}(1+\sqrt{1-x^2})}\right)\Bigg]
\end{split}
\end{align}
Noting that as $\zeta\longrightarrow-\infty$, we have $x\sim e^{\zeta}$, and
\begin{align}
    |\phi|,\; |\Phi|\lesssim \zeta, \quad |\phi'|,\;|\Phi'|\lesssim 1,\quad |\Phi''|\lesssim |\zeta|^{-1},\quad |\Phi'''|,\;|\phi''|\lesssim |\zeta|^{\frac{3}{2}}e^{\frac{3}{2}\zeta}.
\end{align}
Therefore, we have as $\zeta\longrightarrow-\infty$,
\begin{align}
    |h_0|\lesssim |\zeta|.
\end{align}
We now estimate $h$, where we have
\begin{align}
    \hat{f}'=-\frac{1}{\zeta}(\hat{f})^{\frac{1}{2}}+\frac{1}{\zeta}f_0'=\frac{1}{x\zeta}\left[\frac{2}{x^2}-\left(\frac{x^2-1}{\zeta}\right)^{\frac{1}{2}}\right],
\end{align}
\begin{align}
\begin{split}
    \hat{f}''=&-\frac{6}{x^4\zeta}+\frac{2}{x^{\frac{7}{2}}\zeta^2}\left(\frac{x^2-1}{\zeta}\right)^{\frac{1}{2}}+\frac{1}{x^2\zeta}\left(\frac{x^2-1}{\zeta}\right)^{\frac{1}{2}}+\frac{1}{x^{\frac{3}{2}}\zeta^2}\left(\frac{x^2-1}{\zeta}\right)\\&+\frac{1}{2x\zeta}\left(\frac{x^2-1}{\zeta}\right)^{-\frac{1}{2}}\left[\frac{2x}{\zeta}-\frac{(x^2-1)^{\frac{3}{2}}}{x^{\frac{1}{2}}\zeta^{\frac{5}{2}}}\right].
\end{split}
\end{align}
Inspecting \eqref{error functional 1 original Airy}, we find that for $\zeta\longrightarrow-\infty$,
\begin{align}
    |h|\lesssim |\zeta|.
\end{align}
As for $x\longrightarrow\infty$, we have $x\sim \zeta^{\frac{3}{2}}$, and
\begin{align}
    |\phi|\sim \zeta^{-2},\quad |\phi'|\lesssim \zeta^{-3},\quad |\phi''|\lesssim \zeta^{-4},
\end{align}
\begin{align}
    |\Phi|\sim \zeta^{-\frac{1}{2}},\quad |\Phi'|\lesssim \zeta^{-\frac{3}{2}},\quad |\Phi''|\lesssim \zeta^{-\frac{5}{2}},\quad |\Phi'''|\lesssim \zeta^{-\frac{7}{2}}.
\end{align}
Therefore, we obtain for $\zeta\longrightarrow\infty$
\begin{align}
    |h_0|\lesssim \zeta^{-2}.
\end{align}
Similarly, we find
\begin{align}
    |h|\lesssim \zeta^{-2}.
\end{align}
Putting everything together, we finally obtain
\begin{align}
    \int_{\zeta_{-\infty}}^{\infty}dv |h-h_0|\lesssim |\zeta_{-\infty}|^{2}\sim (\log x|_{\zeta_{-\infty}})^{\frac{4}{3}}.
\end{align}
We must choose $\zeta_{-\infty}$ so that
\begin{align}
    x^2|_{\zeta_{-\infty}}|\alpha|,\quad \frac{1}{|\alpha|}|\zeta_{-\infty}|^{2}
\end{align}
are simultaneously small as $\alpha\longrightarrow\infty$. This is achieved by choosing $\zeta_{-\infty}\sim -|\alpha|^{\frac{1}{12}}$. Therefore,
\begin{align}
    |\epsilon^{(+)}|/|M_-|,\,|\epsilon^{(+)}{}'|/(\alpha^{\frac{2}{3}}|N_-|)=E_+^{-1}\cdot O\left({|\alpha|}^{-\frac{1}{2}}\right).
\end{align}
We follow an identical argument to estimate $\epsilon$ using \eqref{approximation of K in terms of Airys} up to $\zeta_{-\infty}$ and use the Maclaurin series of $K_{\beta+i\alpha}$ at $x=0$ on $(-\infty,2\zeta_{-\infty})$, leading to the estimate
\begin{align}
    |\epsilon|/|M_-|,\,|\epsilon{}'|/(\alpha^{\frac{2}{3}}|N_-|)=E_0^{-1}\cdot O\left({|\alpha|}^{-\frac{1}{2}}\right).
\end{align}
We conclude
\begin{align}
    |\epsilon^{(+)}|\leq (|W_0|^2+|W_0^{(+)}|^2)^{\frac{1}{2}}\cdot O(|\alpha|^{-\frac{1}{2}}),\qquad \alpha^{-\frac{2}{3}}|\epsilon^{(+)}{}'|\leq (|W_0'|^2+|W_0^{(+)}{}'|^2)^{\frac{1}{2}}\cdot O(|\alpha|^{-\frac{1}{2}}),
\end{align}
\begin{align}
    |\epsilon|\leq (|W_0|^2+|W_0^{(+)}|^2)^{\frac{1}{2}}\cdot O(|\alpha|^{-\frac{1}{2}}),\qquad \alpha^{-\frac{2}{3}}|\epsilon{}'|\leq (|W_0'|^2+|W_0^{(+)}{}'|^2)^{\frac{1}{2}}\cdot O(|\alpha|^{-\frac{1}{2}}).
\end{align}
We may refine the error estimate for $W_0^{(+)}$ by noting that $\hat{\zeta}$ remains in the first and second quadrants for all $\alpha\geq0$, whereas the zeros of $\mathrm{Ai}_+$ are all contained in the line $\mathrm{ph}\,\hat{\zeta}=-\frac{i\pi}{3}$ and are bounded away from $\hat{\zeta}=0$. Therefore, $\theta_-$ is such that $|\cos\theta_-|> c>0$, and the same argument applies to $\upomega_-$, and so we have
\begin{align}
    |\epsilon^{(+)}|\leq |W_0^{(+)}|\cdot O(|\alpha|^{-\frac{1}{2}}),\qquad \alpha^{-\frac{2}{3}}|\epsilon^{(+)}{}'|\leq |W_0^{(+)}{}'|\cdot O(|\alpha|^{-\frac{1}{2}}).
\end{align}
\end{proof}
\begin{remark}\label{airy approximation of the derivative of modified bessel}
    To leading order in $\alpha^{-1}$, we have
    \begin{align}
    \begin{split}
        \left(\frac{d}{dx}I_{\beta+i\alpha}\right)\Big|_{\alpha x}=&2\pi \alpha^{-\frac{2}{3}}e^{\frac{\alpha\pi}{2}-\frac{i\pi}{2}\beta-\frac{5i\pi }{6}}\frac{1}{x}\left(\frac{4\zeta}{x^2-1}\right)^{-\frac{1}{4}}\mathrm{Ai}_+'\left[e^{-\frac{2\pi i}{3}}\left(\alpha^{\frac{2}{3}}\zeta+\alpha^{-\frac{1}{3}}\Phi\right)\right]\\&\times\left(1+O\left(|\alpha|^{-\frac{1}{2}}\right)\right).
    \end{split}
    \end{align}
\end{remark}

\printbibliography

\end{document}